\documentclass[11pt,reqno]{amsart}
\usepackage{amsthm,amsfonts,fancyhdr,amssymb,euscript,datetime,tikz,graphicx,float}
\usepackage[colorlinks,linkcolor=blue,citecolor=blue]{hyperref}

\newcommand{\bl}{\begin{lemma}}
\newcommand{\el}{\end{lemma}}
\def\beaa{\begin{eqnarray*}} 
\def\eeaa{\end{eqnarray*}}
\def\ba{\begin{array}}
\def\ea{\end{array}}
\def\be#1{\begin{equation} \label{#1}}
\def \eeq{\end{equation}}

\newcommand{\gd}{{g \mkern-8mu /\ \mkern-5mu }}

\newcommand{\di}{\mbox{$d \mkern-9.2mu /$\,}}

\newcommand\smallO{
  \mathchoice
    {{\scriptstyle\mathcal{O}}}
    {{\scriptstyle\mathcal{O}}}
    {{\scriptscriptstyle\mathcal{O}}}
    {\scalebox{.6}{$\scriptscriptstyle\mathcal{O}$}}
  }

\def\a{{\alpha}}

\def\be{{\beta}}

\def\ga{\gamma}
\def\Ga{\Gamma}
\def\de{\delta}

\def\si{\sigma}
\def\Si{\Sigma}
\def\om{\omega}
\def\Om{\Omega}
\def\th{\theta}

\def\nab{\nabla}
\def\varep{\varepsilon}

\def\pr{{\partial}}

\def\les{\lesssim}

\def\rh{{\rho}}

\def\ind{{\in \mkern-16mu /\ \mkern-4mu}}
\def\prd{{\partial \mkern-9mu /\ \mkern-4mu}}

\def\XX{{\mathcal{X}}}

\providecommand{\lrpar}[1]{\left( #1\right)}

\def\Ldo{{\overset{\circ}{\Ld}}}

\def\CC{{\mathcal C}}

\def\MM{{\mathcal M}}

\def\II{{\mathcal I}}
\def\FF{{\mathcal F}}
\def\EE{{\mathcal E}}
\def\HH{{\mathcal H}}
\def\LL{{\mathcal L}}
\def\GG{{\mathcal G}}
\def\TT{{\mathcal T}}

\def\OO{{\mathcal O}}

\def\DD{{\mathcal D}}
\def\PP{{\mathcal P}}

\def\HHb{\underline{\mathcal H}}

\def\D{{\bf D}}

\def\g{{\bf g}}

\def\SSS{{\mathbb{S}}}
\def\RRR{{\mathbb R}}

\def\f12{{\frac 1 2}}

\DeclareMathOperator{\Div}{\mathrm{div}}
\DeclareMathOperator*{\Curl}{\mathrm{curl}}
\def\half{\frac{1}{2}}

\newcommand{\pd}{\pd \mkern-9mu/\ \mkern-7mu}
\newcommand{\Lied}{\mathcal{L} \mkern-9mu/\ \mkern-7mu}

\newcommand{\DDd}{\DD \mkern-10mu /\ \mkern-5mu}

\newcommand{\Du}{\underline{D}}
\newcommand{\RRRic}{\mathrm{Ric}}

\newcommand{\Divd}{\Div \mkern-17mu /\ }
\newcommand{\Divdo}{{\overset{\circ}{\Div \mkern-17mu /\ }}}
\newcommand{\Curld}{\Curl \mkern-17mu /\ }

\newcommand{\Nd}{\nabla \mkern-13mu /\ }

\newcommand{\Ld}{\triangle \mkern-12mu /\ }
\newcommand{\iin}{\in \mkern-16mu /\ \mkern-5mu}

\newcommand{\trchi}{{\tr \chi}}
\newcommand{\trchib}{{\tr \chib}}

\def\ni{\noindent}

\def\Lb{{\,\underline{L}}}

\def\tr{\mathrm{tr}}

\def\chih{{\widehat \chi}}
\def\chib{{\underline \chi}}
\def\chibh{{\underline{\chih}}}

\def\etab{{\underline \eta}}
\def\omb{{\underline{\om}}}

\def\aa{{\underline{\a}}}

\def\th{\theta}

\def\f{\widetilde{f}}

\def\Rbf{{\mathbf{R}}}

\def\Gammad{{\Gamma \mkern-11mu /\,}}

\newcommand{\gac}{{\overset{\circ}{\ga}}}

\newcommand{\ab}{{\underline{\alpha}}}
\newcommand{\beb}{{\underline{\beta}}}

\newcommand{\mfm}{{\mathfrak{m}}}
\newcommand{\mfb}{{\mathfrak{B}}}
\newcommand{\mfW}{{\mathbb{W}}}
\newcommand{\tgd}{{\tilde{\gd}}}

\newcommand{\sql}{{\sqrt{\half l(l+1)-1}}}
\newcommand{\sqlp}{{\sqrt{\half l'(l'+1)-1}}}

\newcommand{\sqll}{{\sqrt{l(l+1)}}}

\newcommand{\tm}{{\tilde m}}
\newcommand{\Ef}{{\mathbf{E}}}
\newcommand{\Pf}{{\mathbf{P}}}
\newcommand{\Lf}{{\mathbf{L}}}
\newcommand{\Gf}{{\mathbf{G}}}

\newtheorem{theorem}{Theorem}[section]
\newtheorem{lemma}[theorem]{Lemma}

\newtheorem{corollary}[theorem]{Corollary}
\newtheorem{definition}[theorem]{Definition}
\newtheorem{remark}[theorem]{Remark}

\numberwithin{equation}{section}

\makeatletter
\def\@setthanks{\vspace{-\baselineskip}\def\thanks##1{\@par##1\@addpunct.}\thankses}
\makeatother

\begin{document}

\title[Obstruction-free gluing for the Einstein equations]{Obstruction-free gluing \\ for the Einstein equations}
\author[S. Czimek and I. Rodnianski]{Stefan Czimek$^{(1)}$ and Igor Rodnianski$^{(2)}$} 
\thanks{\noindent$^{(1)}$ Mathematisches Institut, Universit\"at Leipzig, Augustusplatz 10, 04109 Leipzig, Deutschland,  \texttt{stefan.czimek@uni-leipzig.de}. \\
$^{(2)}$ Department of Mathematics, Princeton University, Fine Hall, Washington Road, Princeton, NJ 08544, USA, \texttt{irod@math.princeton.edu}. } 
%
\begin{abstract} In this paper we develop a new approach to the gluing problem in General Relativity,  that is, the problem of matching two solutions of the 
Einstein equations along a spacelike or characteristic (null) hypersurface. In contrast to the previous constructions, the new perspective actively utilizes the nonlinearity of the constraint equations. As a result, we are able to remove the 10-dimensional spaces of obstructions to the null and spacelike (asymptotically flat) gluing problems, previously known in the literature.
In particular, we show that any asymptotically flat spacelike initial data can be glued to the Schwarzschild initial data of mass $M$ for any $M>0$ sufficiently large. More generally, compared to the celebrated result of Corvino-Schoen, our methods allow us to choose ourselves the Kerr spacelike initial data that is being glued onto.  As in our earlier work, our primary focus is the analysis of the null problem, where we develop a new technique of combining low-frequency linear analysis with high-frequency {\it nonlinear} control. The corresponding spacelike results are derived a posteriori by solving a 
characteristic initial value problem.
\end{abstract}
\maketitle
\setcounter{tocdepth}{2}
\tableofcontents

\section{Introduction} \label{SECintro} 

\ni This paper is concerned with gluing constructions for the Einstein equations of general relativity. Generally, the  \emph{gluing problem of general relativity} is to ``connect" two given spacetimes $(\MM_1,\mathbf{g}_2)$ and $(\MM_2,\mathbf{g}_2)$ as solution to the Einstein equations. In other words, the goal is to construct a solution $(\MM,\mathbf{g})$ to the Einstein equations such that $(\MM_1,\mathbf{g}_2)$ and $(\MM_2,\mathbf{g}_2)$ isometrically embed into $(\MM,\mathbf{g})$.

A special case of this gluing problem are \emph{initial data gluing problems} where two (subsets of) initial data sets for the Einstein equations of the same type -- either spacelike or null -- shall be ``connected" one to another as solution of the corresponding initial data constraint equations. Once such initial data gluing is achieved, it suffices to evolve the Einstein equations forward to get a gluing of the corresponding pieces of spacetimes in the above sense (where the spacetime pieces are determined by the domain of dependence property of the Einstein equations).

In the present article we study the \emph{null} initial data gluing problem and then deduce \emph{a posteriori} the corresponding results for \emph{spacelike} initial data gluing by evolution of the Einstein equations. However, as spacelike gluing has a long history in mathematical relativity and Riemannian geometry, we first discuss in Section \ref{SECintroSPACELIKE} the spacelike gluing problem, recalling in particular the classical result by Corvino-Schoen on spacelike gluing to Kerr and explaining how our present work improves on it. In Section \ref{SECintroNULL} we then set up the null gluing problem and revisit the previous null gluing results established by the authors in collaboration with Aretakis in \cite{ACR1,ACR2,ACR3}. In Section \ref{SECintroSecondVersionMainTHM} we give the first version of the main  result of this paper and explain our approach. In Section \ref{SECtoymodelINTRO} we illustrate the ideas by discussing a model problem.

\subsection{The spacelike gluing problem}\label{SECintroSPACELIKE}

\ni Consider two given spacelike initial data sets $(\Si_1,g_1,k_1)$ and $(\Si_2,g_2,k_2)$. The \emph{spacelike gluing problem} is to ``connect" them as solution to the spacelike constraint equations, that is, to construct spacelike initial data $(\Si,g,k)$ such that $(\Si_1,g_1,k_1)$ and $(\Si_2,g_2,k_2)$ isometrically embed into $(\Si,g,k)$.

As the spacelike constraint equations are of elliptic nature, it was expected in the past that solutions display some form of elliptic rigidity, i.e. that solutions are completely or almost completely determined from their restriction onto an open set. It thus came as a surprise to the community when Corvino \cite{Corvino} and Corvino-Schoen \cite{CorvinoSchoen} (see also \cite{ChruscielDelay}) proved that \emph{asymptotically flat spacelike initial data}, that is, $\Si$ being diffeomorphic to $\RRR^3 \setminus \overline{B(0,1)}$ outside a compact set and $(g,k)$ admitting the following expansions as $\vert x \vert \to \infty$,
\begin{align*}
\begin{aligned}
g_{ij}(x)-e_{ij} = \OO(\vert x \vert^{-1}), \,\, k_{ij}(x) = \OO(\vert x \vert^{-2}),
\end{aligned}
\end{align*}
\emph{can be glued up to a $10$-dimensional obstruction space}. The latter should be understood in the sense that the data $(\Si_1,g_1,k_1)$ and $(\Si_2,g_2,k_2)$ can be chosen to be arbitrary \emph{up to the values of certain $10$ integral quantities.} In particular, as it turns out, $(\Si_2,g_2,k_2)$
can be chosen to be (a slice of) {\it{one}} of the elements of the ($10$-dimensional) Kerr family.

The gluing in these results takes place far out in the asymptotically flat region, namely, across the annulus $A_{[R,2R]}$ bounded by the coordinate spheres $S_R$ and $S_{2R}$, for large $R>0$. The scale-invariance of the constraint equations allows to rescale the problem to an equivalent \emph{small data} (i.e. close to Minkowski) gluing problem across the annulus $A_{[1,2]}$. We emphasize already that, in a similar vein, the main results of this paper are stated in terms of the appropriate small data gluing problem between two spheres $S_1$ and $S_2$ (except for our application to spacelike gluing, see Theorem \ref{THMintro1} below).  
 
 The $10$-dimensional obstruction space can be interpreted geometrically in terms of the $10$ ADM parameters of \emph{energy} $\mathbf{E}_{\mathrm{ADM}}$, \emph{linear momentum} $\Pf_{\mathrm{ADM}}$, \emph{angular momentum} $\Lf_{\mathrm{ADM}}$, and \emph{center-of-mass} $\mathbf{C}_{\mathrm{ADM}}$.

The Corvino-Schoen spacelike gluing \cite{CorvinoSchoen} can be interpreted as follows. Any solution with well-defined ADM parameters $(\mathbf{E}_{\mathrm{ADM}},\mathbf{P}_{\mathrm{ADM}},\mathbf{L}_{\mathrm{ADM}},\mathbf{C}_{\mathrm{ADM}})$ can be glued, far out in the asymptotically flat region, to the initial data $(\Si_{\mathrm{Kerr}},g_{\mathrm{Kerr}},k_{\mathrm{Kerr}})$ induced on a spacelike hypersurface $\Si_{\mathrm{Kerr}}$ in a Kerr black hole spacetime $(\MM_{\mathrm{Kerr}},\g_{\mathrm{Kerr}})$. However, the corresponding ADM parameters $(\mathbf{E}_{\mathrm{ADM}},\mathbf{P}_{\mathrm{ADM}},\mathbf{L}_{\mathrm{ADM}},\mathbf{C}_{\mathrm{ADM}})_{\mathrm{Kerr}}$ of $(\Si_{\mathrm{Kerr}},g_{\mathrm{Kerr}},k_{\mathrm{Kerr}})$ \emph{cannot be chosen freely but are determined from the original solution}. In fact, they can be calculated to be
\begin{align*}
\begin{aligned}
(\mathbf{E}_{\mathrm{ADM}},\mathbf{P}_{\mathrm{ADM}},\mathbf{L}_{\mathrm{ADM}},\mathbf{C}_{\mathrm{ADM}})_{\mathrm{Kerr}} = (\mathbf{E}_{\mathrm{ADM}},\mathbf{P}_{\mathrm{ADM}},\mathbf{L}_{\mathrm{ADM}},\mathbf{C}_{\mathrm{ADM}}) + \text{small corrections.}
\end{aligned}
\end{align*}

We note that the analogous result for \emph{null} gluing to Kerr, including an appropriate geometric interpretation of the $10$-parameter space, is proved in \cite{ACR3}.
\vskip 1pc

In this paper we propose a novel \emph{nonlinear gluing method} which allows in particular to \emph{eliminate this $10$-dimensional obstruction space}; see also Section \ref{SECintroSecondVersionMainTHM} for a literature comparison. In particular, we obtain the following result.

\begin{theorem}[Obstruction-free spacelike gluing, version 1] \label{THMintro1} Let $(\Si,g,k)$ be asymptotically flat spacelike initial data with well-defined ADM parameters $(\mathbf{E}_{\mathrm{ADM}},\mathbf{P}_{\mathrm{ADM}},\mathbf{L}_{\mathrm{ADM}},\mathbf{C}_{\mathrm{ADM}})$. It is possible to glue, far out in the asymptotic region, $(\Si,g,k)$ to any Kerr initial data $(g^{\mathrm{Kerr}},k^{\mathrm{Kerr}})$ provided that the following inequalities are satisfied,
\begin{align}
\begin{aligned}
\Ef_{\mathrm{ADM}}^{\mathrm{Kerr}} - \mathbf{E}_{\mathrm{ADM}} >0, \,\, \Ef_{\mathrm{ADM}}^{\mathrm{Kerr}} - \mathbf{E}_{\mathrm{ADM}} > C\cdot \left\vert \Pf_{\mathrm{ADM}}^{\mathrm{Kerr}} - \mathbf{P}_{\mathrm{ADM}}\right\vert,
\end{aligned}\label{EQenergyMomentumIneq}
\end{align}
where $C>0$ is a (potentially) large constant.

In particular, it is possible to glue any asymptotically flat spacelike initial data set with well-defined ADM parameters far out in the asymptotic region to the Schwarzschild spacelike initial data of mass $M>0$ for any sufficiently large $M>0$.

\end{theorem}

\ni \emph{Remarks on Theorem \ref{THMintro1}.}
\begin{enumerate}
\item It is interesting to compare the inequalities \eqref{EQenergyMomentumIneq}, which can be rewritten as 
$$
\Delta E >C |\Delta P|,
$$
with the general form of the positive mass theorem \cite{SchoenYau1,Witten} asserting 
that \emph{the ADM energy-momentum vector is timelike}:
$$
E\ge |P|.
$$
\item The precise quantitative version of the Theorem, see Corollary \ref{CORspacelike}, requires a lower bound $\Delta E\ge c>0$,  which should be much larger (independent) than the inverse of the radius $R$ where the gluing takes place. The constant $C$ then can be chosen to be large and 
universal. Its value however is not sharp. The question of sharpness of $C$ is not pursued in this paper.
\item Theorem \ref{THMintro1} shows that it is possible to \emph{{extract all angular momentum}} from given initial data and \emph{{transfer it into the mass}}. The reader should compare this to the considerations of Penrose \cite{PenroseNuovo} on the extraction of angular momentum of black holes, and of Christodoulou \cite{ChristodoulouMASS} on the reversible and irreversible transformations and the irreducible mass of a black hole.
\end{enumerate}

\ni As mentioned above, in this paper we continue to take the point of view established in \cite{ACR1,ACR2,ACR3} that results for spacelike gluing can be established as corollaries of the corresponding \emph{null} gluing results. Indeed, this strategy is successfully employed in \cite{ACR3} and \cite{ACR1} to provide alternative proofs of the Corvino-Schoen spacelike gluing to Kerr \cite{CorvinoSchoen} and the Carlotto-Schoen localization of spacelike initial data \cite{CarlottoSchoen}, respectively. In the latter case, it also resulted in a solution of an open problem on the sharp decay rates.

In the next section we introduce the null gluing problem and recall the main results of \cite{ACR1,ACR2}. In Section \ref{SECintroSecondVersionMainTHM} we state the first version of the main theorem of this paper and give an overview of its proof. In Section \ref{SECtoymodelINTRO} we illustrate the main ideas of our approach applied to a model problem.

\subsection{The null gluing problem} \label{SECintroNULL} \ni In the following we first introduce in Section \ref{SECSUBnullframework} the geometric \emph{double null framework} and the notion of null initial data, and discuss the characteristic initial value problem for the Einstein equations. In Section \ref{SUBSECcharseed} we discuss the hierarchical structure of the null structure equations and define characteristic seeds. In Section \ref{SECSUBnullgluingsetup} we formulate the null gluing problem and recapitulate the main results of \cite{ACR1,ACR2}.

\subsubsection{The double null framework, null structure equations and the characteristic Cauchy problem} \label{SECSUBnullframework}

Consider a spacelike $2$-sphere $S$ in a spacetime $(\MM,\g)$, and let ${u}_0$ and ${v}_0$ with $v_0>u_0$ be two real numbers. Let $u$ and $v$ be two \emph{optical functions} of $(\MM,\g)$, that is, satisfying each the eikonal equation (here $\mathbf{D}$ is the covariant derivative on $(\MM,\g)$)
\begin{align*} 
\begin{aligned} 
\vert \mathbf{D}u \vert_{\g}^2 = 0, \,\, \vert \mathbf{D}v \vert_{\g}^2 = 0,	
\end{aligned} 
\end{align*}
such that $S=\{ {u}={u}_0, {v}={v}_0\}$, and for real numbers ${u}_1$ and ${v}_1$, the level sets
$\HH_{{u}_1} := \{ {u}= {u}_1\}$ and $\HHb_{{v}_1} := \{ {v}= {v}_1\}$ are \emph{outgoing} and \emph{ingoing} null hypersurfaces, respectively. The union of these null hypersurfaces form locally a so-called \emph{double null foliation} of $(\MM,\g)$. Denote the spacelike intersection $2$-spheres by $S_{u_1,v_1} := \{ {u}={u}_1, {v}={v}_1\}$, and let $\gd$, $\Nd$ and $r(u,v)$ denote their induced Riemannian metric, the associated covariant derivative and area radius, respectively. The reference optical functions on Minkowski spacetime are given by $u=\half (t-r)$ and $v=\half(t+r)$.

On the sphere $S=S_{u_0,v_0}$ we define ($2$ patches of) local angular coordinates $(\th^1,\th^2)$, and extend them everywhere by propagating them first along the null generators of $\HH_{u_0}$ and then of $\HHb_v$ for all $v$, as indicated in the Figure \ref{FIGdoublenullS2pic1} below. The resulting coordinate system $( u,  v, \th^1,\th^2)$ is called a \emph{double null coordinate system}.

\begin{figure}[H]
\begin{center}
\includegraphics[width=6.5cm]{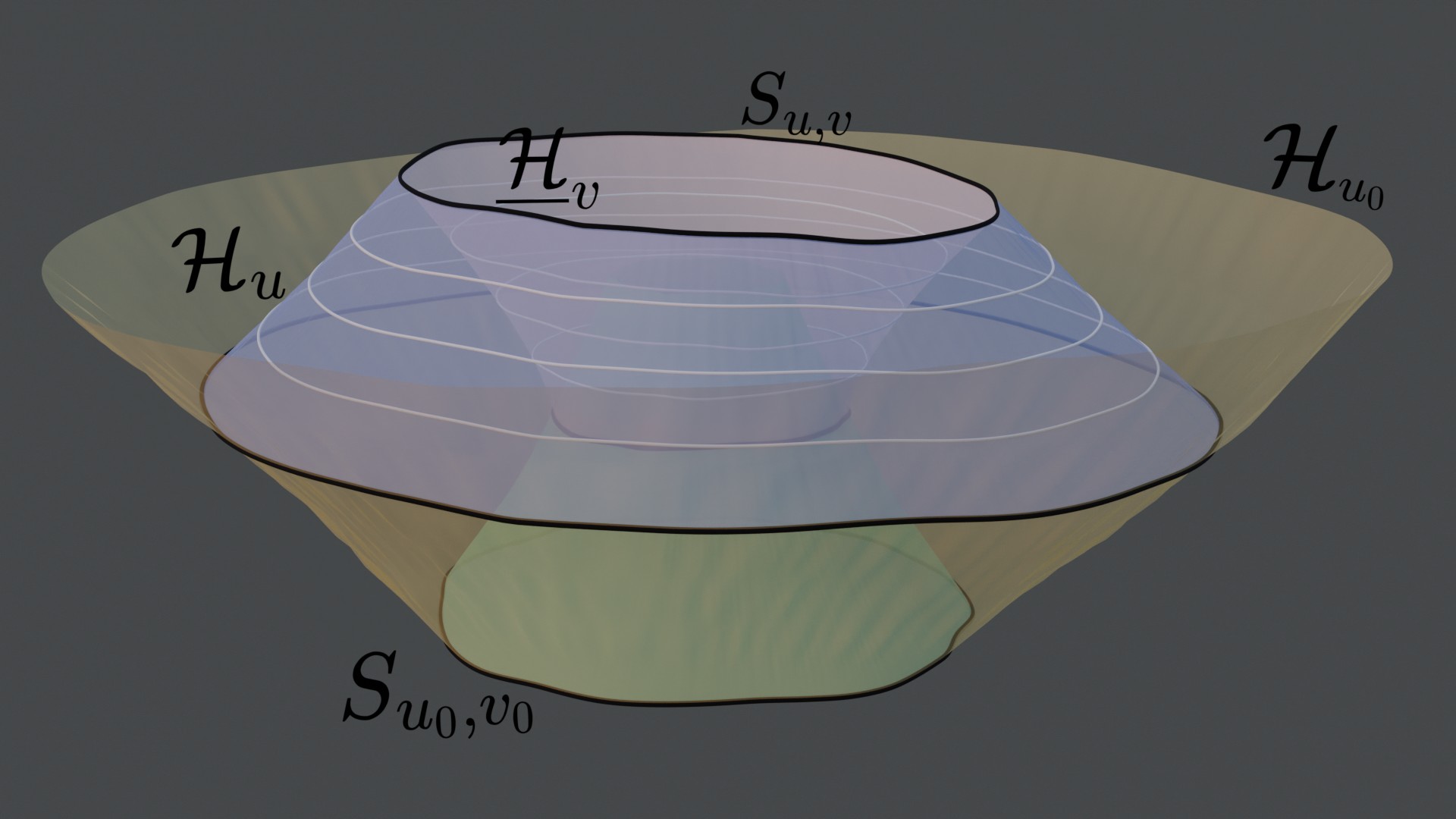} \,\, \includegraphics[width=6.5cm]{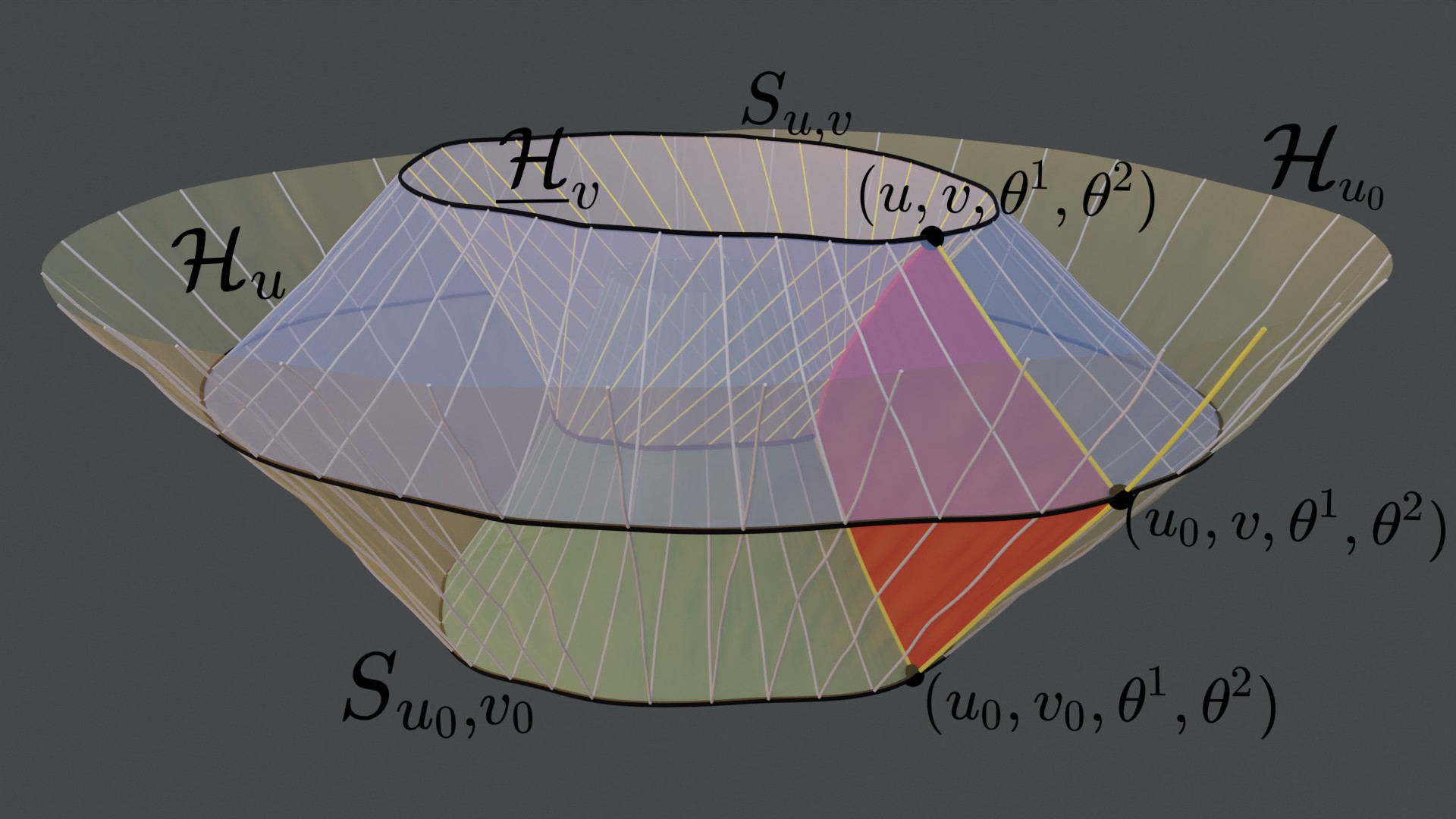} 
\vspace{0.4cm}
\caption{The double null foliation formed by the level sets of the optical functions $u$ and $v$ (left), and the construction of local double null coordinates (right).}
\end{center}
\label{FIGdoublenullS2pic1}
\end{figure}
\vspace{-0.2cm}

\ni Using the angular coordinates $(\th^1,\th^2)$ we can define on each $S_{u,v}$ the round (unit) metric $\gac := (d\th^1)^2 + (\sin\th^1)^2 (d\th^2)^2$. On each $S_{u,v}$ consider the following conformal decomposition of $\gd$,
\begin{align} 
	\begin{aligned} 
		\gd = \phi^2 \gd_c, 
	\end{aligned} \label{EQconfDecompintro}
\end{align}
where the conformal factor $\phi$ and the conformal metric $\gd_c$ are determined by \eqref{EQconfDecompintro} with the condition that $\det \gd_c = \det \gac$ on $S_{u,v}$. The conformal metric $\gd_c$ is the first of two essential quantities along null hypersurfaces used in the definition null initial data and the study of the null gluing problem.

Relative to a double null coordinate system the metric $\g$ takes the form
\begin{align} 
\begin{aligned} 
\g = -4 {\Om}^2 d u d v + \gd_{AB} \lrpar{d\th^A + {b}^A d v}\lrpar{d\th^B + {b}^B d v}.
\end{aligned} \label{EQintroDoubleNull123}
\end{align}
Here the scalar function $\Om$ is the so-called \emph{null lapse} and is related to the inverse foliation density of the spheres $S_{u,v}$ in the null hypersurfaces $\HH_u$, see below. The null lapse $\Om$ is the second of the two essential quantities used for the definition of null initial data and the study of the null gluing problem. The $S_{ u,  v}$-tangent vectorfield ${b}$ in \eqref{EQintroDoubleNull123} is the so-called \emph{shift vector}; by construction $b$ vanishes on $\HH_{u_0}$ so that it does not play an important role for the null gluing problem along $\HH_{u_0}$ later.

We follow the standard notation of \cite{ChrForm} to decompose the connection form and the Riemann curvature tensor $\Rbf$ of $(\MM,\g)$ into the so-called \emph{Ricci coefficients} and \emph{null curvature components}, respectively. Namely, define the two null vector pairs
\begin{align*} 
\begin{aligned} 
(L,\Lb):= (\pr_{ v} + {b}^A\pr_{\th^A}, \pr_u), \,\,\, (\widehat{ L},\widehat{ \Lb}) := ( \Om^{-1}L, \Om^{-1}\Lb),	
\end{aligned} 
\end{align*}
where the so-called \emph{normalized} null vectorfields $\widehat{L}$ and $\widehat{\Lb}$ satisfy $\g(\widehat{L},\widehat{\Lb}) =-2$ and we note that $\widehat{L}(v)=\Om^{-1}$. Denote for $S_{u,v}$-tangent vectorfields $X$ and $Y$ the \emph{Ricci coefficients} 
\begin{align} 
\begin{aligned} 
{\chi}(X,Y) :=& \g(\D_X \widehat{ L},Y), & {\chib}(X,Y) :=& \g(\D_X \widehat{ \Lb},Y), & {\zeta}(X) :=& \half \g(\D_X \widehat{L}, \widehat{\Lb}), \\
\eta :=& \zeta + \di \log \Om, & \om :=& D \log \Om, & \omb :=& \Du \log \Om,
\end{aligned} \label{EQdefRicciINTRO}
\end{align}
where $\di$ is the exterior derivative on $S_{u,v}$, and $D:= \Lied_L$, $\Du :=\Lied_\Lb$ denote the projections of the Lie derivatives along $L$ and $\Lb$, respectively, onto $S_{u,v}$. Let moreover $\etab:= -\zeta + \di\log \Om$, and note that $\eta+ \etab = 2\di\log\Om$.

The \emph{null curvature components} are given for $S_{u,v}$-tangent vectorfields $X$ and $Y$ by
\begin{align} 
\begin{aligned} 
\alpha(X,Y) :=& \Rbf(X,\widehat{ L}, Y, \widehat{ L}), & \beta(X) :=& \half \Rbf(X, \widehat{ L},\widehat{\Lb},\widehat{ L}), & \rh :=& \frac{1}{4} \Rbf(\widehat{\Lb}, \widehat{ L}, \widehat{\Lb}, \widehat{ L}), \\
 \sigma \iin(X,Y) :=& \half \Rbf(X,Y,\widehat{\Lb}, \widehat{ L}), &
\beb(X) :=& \half \Rbf(X, \widehat{\Lb},\widehat{\Lb},\widehat{ L}), & \ab(X,Y) :=& \Rbf(X,\widehat{\Lb}, Y, \widehat{\Lb});
\end{aligned} \label{EQNCCintrodef}
\end{align}
where $\iin$ denotes the area form on $(S_{u,v},\gd)$. We split $\chi$ and $\chib$ into tracefree and trace part,
\begin{align*} 
\begin{aligned} 
\chi = \chih +\half \trchi \gd, \,\, \chib= \chibh + \half \trchib \gd,
\end{aligned} 
\end{align*}
where $\trchi:= \gd^{AB}\chi_{AB}$ and $\chih:= \chi - \half \trchi \gd$.

The Einstein equations together with the embedding equations for a double null foliation stipulate that the metric components, Ricci coefficients and null curvature components satisfy the so-called \emph{null structure equations}; see Section \ref{SECdoublenull} for the full system of equations.

The null structure equations are typically viewed as \emph{transport-elliptic} system of equations. The elliptic part arises if one views, as is often done in evolution problems for the Einstein equations, the Ricci coefficients and the metric as determined from curvature; for example, through the Gauss and Gauss-Codazzi equations,
\begin{align*} 
\begin{aligned} 
K + \frac{1}{4}\trchi\trchib - \half (\chih,\chibh) =& -\rho, \\
\Divd \chih -\half \di \tr \chi + \chih \cdot \zeta - \half \trchi \zeta =& - \beta, 
\end{aligned} 
\end{align*}
where $(\Divd \chih)_A := \Nd^C \chih_{AC}$, $(\chih,\chibh):= \chih^{AB}\chibh_{AB}$, and $(\chih \cdot \zeta)_A:= \chih_{AB}\zeta^B$.
Here we take the opposite point of view that these equations \emph{determine} null curvature components from the metric components and Ricci coefficients on each sphere $S_{u,v}$.

The metric and the Ricci coefficients themselves satisfy a coupled system of null transport equations, such as, for example, the Raychaudhuri equation and the first variation equation along $\HH$,
\begin{align*} 
\begin{aligned} 
&D \trchi + \frac{\Om}{2} (\trchi)^2 - \om \trchi = - \Om \vert \chih \vert^2_{\gd}, \\
&D\phi=\frac{\Om\trchi\phi}{2},\,\,
\widehat{D\gd}=2\Om\chih.
\end{aligned}
\end{align*}

The null structure equations can be brought into a hierarchical order. This order and the solvability of the system are discussed in Section \ref{SUBSECcharseed}. 

Higher transversal derivatives of metric components and Ricci coefficients (for example, $\Du\Du\chibh$ along $\HH_u$) satisfy corresponding \emph{higher-order null transport equations}; these equations can easily be computed. For the purposes of this paper it suffices to consider the \emph{$2$-jet of null structure equations}, that is, the null transport equations for up to $2$nd order derivatives of the metric components.

As mentioned above, given metric components and Ricci coefficients on $S_{u,v}$, the further specification of $\rho$ and $\beta$ on $S_{u,v}$ is \emph{redundant}. More generally, considering all null structure equations to remove the redundancy of induced $2$-jet data on a sphere, we arrive at the following definition of \emph{$C^2$-sphere data} (defined for the first time in \cite{ACR1,ACR2}) where the $C^2$ indicates that it fully determines the $2$-jet of $\mathbf{g}$ on the sphere, i.e. all derivatives of $\mathbf{g}$ up to order $2$.

\begin{definition}[$C^2$-Sphere data] \label{DEFspheredata2} Let $S$ be a $2$-sphere. Sphere data $x$ on $S$ consists of choice of a unit round metric $\gac$ on $S$, see Section \ref{SECSUBnullframework}, and the following tuple of $S$-tangent tensors,
	\begin{align} 
		\begin{aligned} 
			x = (\Om,\gd, \Om\trchi, \chih, \Om\trchib, \chibh, \eta, \om, D\om, \omb, \Du\omb, \a, \ab),
		\end{aligned} \label{EQspheredataDEF}
	\end{align}
	where
	\begin{itemize}
		\item $\Om>0$ is a positive scalar function and $\gd$ is a Riemannian metric, 
		\item $\Om\trchi, \Om\trchib,\om, D\om, \omb, \Du\omb$ are scalar functions,
		\item $\eta$ is a vectorfield,
		\item $\chih$, $\chibh$, $\a$ and $\ab$ are symmetric $\gd$-tracefree $2$-tensors.
	\end{itemize}
\end{definition}

\ni The above definition of $C^2$-sphere data is used in our formulation of the null gluing problem below. Each quantity in \eqref{EQspheredataDEF} is subject to a null transport equation along $\HH$ and $\HHb$, respectively. By definition, sphere data is not subject to any constraints on $S$. Sphere data is \emph{gauge-dependent} in the sense that it is sensitive to sphere diffeomorphisms of $S$ and, in case $S$ lies in a vacuum spacetime $\MM$, perturbations of $S$ to nearby in $\MM$. The reference Minkowski sphere data on $S_{u,v}$ is denoted as follows, with $r:=v-u$,
\begin{align} 
\begin{aligned} 
\mathfrak{m}_{u,v} = \lrpar{1,r^2 \gac, \frac{2}{r}, 0, -\frac{2}{r}, 0, 0, 0, 0, 0, 0, 0, 0}.
\end{aligned} \label{EQrefMinkowskiSphereData}
\end{align} 

It is well-known (see, for example, \cite{Rendall,LukChar,LukRod1}) that the prescription of metric coefficients, Ricci coefficients, and null curvature components along two transversely-intersecting null hypersurfaces $\HH_{u_0}$ and $\HHb_{v_0}$, satisfying the respective null structure equations along $\HH_{u_0}$ and $\HHb_{v_0}$ (as well as simple compatibility conditions on the intersection $S_{u_0,v_0}$) leads to a well-posed \emph{characteristic initial  value} problem for the Einstein equations. 

\subsubsection{The null structure equations and the characteristic seed} \label{SUBSECcharseed}

\ni In the following we discuss the null structure equations along an outgoing null hypersurface $\HH$ as those are relevant for the null gluing problem along $\HH$ introduced below. Analogous statements hold for the null structure equations along ingoing null hypersurfaces $\HHb$.

\emph{A priori}, the construction of solutions to the coupled system of null transport equations in the null structure equations (discussed above in Section \ref{SECSUBnullframework}) may seem highy non-trivial. However, Sachs \cite{SachsIVP} (see also Christodoulou \cite{ChrForm}) made the following important observation. The null structure equations have the advantageous property that they can be rewritten in a \emph{hierarchical form} which allows to \emph{freely prescribe} a so-called \emph{characteristic seed} from which a solution can be constructed by straight-forward integration of a sequence of null transport equations along $v$ (i.e. along the null generators of $\HH$). Through this construction, the space of solutions to the null structure equations is parametrized in a $1$-to-$1$ fashion by the free characteristic seeds. This is in stark contrast to the general situation for spacelike initial data.

\begin{definition}[Characteristic seed for $\HH$] \label{DEFcharSEEDH} Let $\HH := [1,2] \times \SSS^2$. A characteristic seed along $\HH$ consists of the following two free prescriptions:
\begin{enumerate}
	\item The \emph{characteristic seed on the sphere $S_1:= \{1\}\times \SSS^2$}, consisting of 
	\begin{itemize}
		\item an induced Riemannian metric $\gd$, 
		\item four scalar functions $\trchi,\trchib,\omb,\Du\omb$,
		\item an $S_1$-tangential vectorfield $\eta$,
		\item two $\gd$-tracefree $S_1$-tangential symmetric $2$-tensors $\chibh$ and $\ab$.
	\end{itemize}
	\item The \emph{characteristic seed along the hypersurface $\HH=\cup_{1\leq v \leq 2} S_v$}, consisting of 
	\begin{itemize}
		\item a conformal class $\mathrm{conf}(\gd)$ of induced metrics on $S_v$, compatible with the prescribed $\gd$ on $S_1$, that is, $\gd\vert_{S_1} \in \mathrm{conf}(\gd) \vert_{S_1}$,
		\item a scalar function $\Om$, which will be the \emph{null lapse}.
	\end{itemize}
\end{enumerate}
\end{definition}
\ni We remark that prescribing the conformal class $\mathrm{conf}(\gd)$ along $\HH$ is equivalent to prescribing the conformal metric $\gd_c$ in \eqref{EQconfDecompintro}. The characteristic seed along $\HH$ is the key to null gluing constructions.

Given a characteristic seed, the null transport equations of the null structure equations are integrated in hierarchical sequence to yield metric components, Ricci coefficients and null curvature components in the following order, 
\begin{align}
	\begin{aligned}
		\phi, \gd, \Om\trchi, \Om\chih, \eta, \Om\trchib, \Om\chibh, \omb, \ab, \Du\omb.
	\end{aligned}\label{EQsequence}
\end{align}
The remaining Ricci coefficients and null curvature components can then be directly calculated from the characteristic seed and \eqref{EQsequence} by the remaining null structure equations. 
In particular, in this sense the quantities \eqref{EQspheredataDEF} fully encode the solution of the null structure equations -- this is the notation we are using in this paper.

The null structure equations are presented in their hierarchical order in Sections \ref{SECphiEstimatesPrecise} to \ref{SECestimateDuomb} where solutions are constructed and estimated from a prescribed characteristic seed.

Returning to the characteristic initial value problem for the Einstein equations, by the above observation it follows in particular that the next free specification along a pair of hypersurfaces $\HH$ and $\HHb$, intersecting transversely at a $2$-dimensional surface $S$, constitutes well-posed \emph{characteristic initial data},
\begin{align*} 
	\begin{aligned} 
		\lrpar{\Omega, \gd_c} \text{ on } \HH \cup \HHb \,\, \text{ and } \lrpar{\gd, \trchi, \trchib, \eta} \text{ on } S,
	\end{aligned} 
\end{align*}
where $\gd_c$ is required to be compatible with $\gd$ on $S$ as in Definition \ref{DEFcharSEEDH}.

One important detail for the construction of solutions to the null structure equations from a characteristic seed is that the quantity
\begin{align*}
\begin{aligned}
\vert \Om\chih \vert^2_{\gd} := \gd^{AC}\gd^{BD} \lrpar{\Om\chih}_{AB}\lrpar{\Om\chih}_{CD},
\end{aligned}
\end{align*}
is \emph{conformally invariant}, that is, 
\begin{align}
\begin{aligned}
\vert \Om\chih \vert^2_{\gd} = \vert \Om\widehat{\tilde{\chi}} \vert^2_{\tilde{\gd}},
\end{aligned}\label{EQshearINVARIANCE}
\end{align}
where the right-hand side can be directly calculated from the given characteristic seed along $\HH$.

\subsubsection{The null gluing problem and previous results} \label{SECSUBnullgluingsetup}

\ni The null gluing problem for the Einstein equations was first introduced in \cite{ACR1,ACR2,ACR3}. In the general terms of the beginning of this introduction, the \emph{null} initial data gluing problem asks that two spacetimes $(\MM_1,\g_1)$ and $(\MM_2,\g_2)$ be ``joined" by gluing the induced data on an outgoing null hypersurface $\HH_1 \subset \MM_1$ to the induced data on an outgoing null hypersurface $\HH_2 \subset \MM_2$ as solution to the null structure equations.

Using the notions of sphere data and null structure equations introduced above, we can phrase the null gluing problem as follows. 

\emph{Given two spacetimes $(\MM_1,\g_1)$ and $(\MM_2,\g_2)$, is it possible to pick $2$-spheres $S_1\subset \MM_1$ and $S_2\subset\MM_2$ together with respective local double null coordinate systems such that the induced sphere data $x_1$ on $S_1$ can be glued to the induced sphere data $x_2$ on $S_2$ as solution to the null structure equations? In other words, does there exist a solution $x$ to the null structure equations along a hypersurface $\HH_{[1,2]}=[1,2]\times \mathbb{S}^2$ such that its restriction to $S_1$ and $S_2$ agrees with the sphere data $x_1$ and $x_2$, respectively?}

In \cite{ACR1,ACR2} the above null gluing problem was solved in the asymptotically flat regime. In the following rough statement we present the corresponding rescaled result; see also Figure \ref{FIGperturbative} below.
\vskip 5pc
\begin{theorem}[Perturbative null gluing of \cite{ACR1,ACR2}, version 1] \label{THMintroPASTnull}

Consider spheres $S_1\subset\MM_1$ and $S_2\subset\MM_2$ with induced sphere data $x_1$ and $x_2$ close to the reference Minkowski sphere data on the spheres of radius $1$ and $2$, respectively. There exists a perturbation of the sphere $S_2$ to a nearby sphere $S_2'$ in $\MM_2$ and a local gauge transformation of the double null coordinates around $S_2'$ such that there is a solution $x$ to the null structure equations along $\HH_{[1,2]}=[1,2]\times \mathbb{S}^2$ with $x\vert_{S_1}=x_1$ and such that \emph{up to a $10$-dimensional space} it holds that $x\vert_{S_2} = x_2'$ where $x_2'$ denotes the sphere data on $S_2'$ with respect to the new double null coordinates in $\MM_2$.

Specifically, the $10$-dimensional space is spanned by $10$ sphere integrals denoted by $\mathbf{E},\Pf^m,\Lf^m,\Gf^m$ for $m=-1,0,1$ (stated in Definition \ref{DEFchargesEPLG} below) and the following statement holds. If \emph{a posteriori} the constructed solution $x$ satisfies (for $m=-1,0,1$)
\newline 
$$
\lrpar{\mathbf{E},\Pf^m,\Lf^m,\Gf^m}(x\vert_{S_2}) = \lrpar{\mathbf{E},\Pf^m,\Lf^m,\Gf^m}(x'_2),
$$ \,\,
then it holds that $x\vert_{S_2} = x'_2$ on $S_2$.
\end{theorem}

\ni The proof of Theorem \ref{THMintroPASTnull} is based on an analysis of the \emph{linearized} null structure equations around Minkowski spacetime. It is shown in \cite{ACR1,ACR2} that the linearized equations admit an $\infty$-dimensional space of obstructions to null gluing in the form of \emph{exact conservation laws} along $\HH$. We make two short remarks on the null gluing of Theorem \ref{THMintroPASTnull}.

\begin{enumerate}
\item It is possible to glue transversally to the space of conservation laws by using the freedom of choosing the characteristic seed along $\HH$. In fact, the gluing conditions turn into linearly independent \emph{weighted integral conditions} on the characteristic seed; we refer to \cite{ACR1,ACR2} for details and discussion of the relevant hierarchy of weights in the integrals. 

\item Within the space of conservation laws, all but $10$ of the conserved charges can be adjusted by applying \emph{linearized gauge transformations} and \emph{linearized sphere perturbations} at $S_2$ (see \cite{DHR,ACR1,ACR2}), which leads exactly to the codimension-$10$ result stated in Theorem \ref{THMintroPASTnull}. The existence of $\infty$-many conservation laws lies in stark contrast to the situation in spacelike gluing where the linearized equations directly admit only a finite-dimensional space of obstructions due to the operators being Fredholm.
\end{enumerate}

In this paper we use a different phrasing of Theorem \ref{THMintroPASTnull} (see Theorem \ref{THMcharGluing10d} below) where the perturbation and subsequent gauge transformation are applied to the sphere $S_1$ instead of $S_2$.
\begin{figure}[H]
\begin{center}
\includegraphics[width=9.5cm]{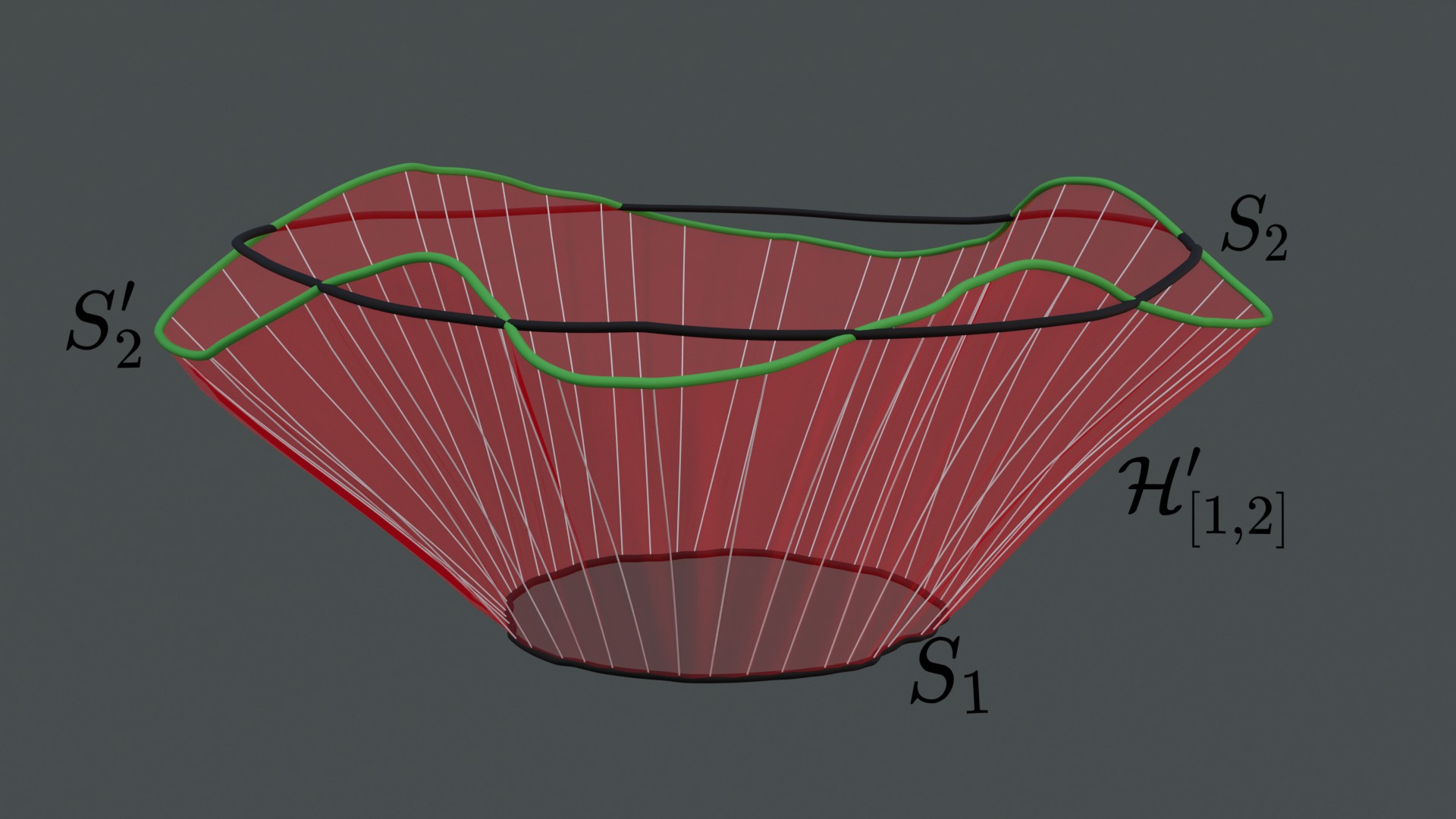} 
\vspace{0.4cm} 
\caption{Null gluing constructed in \cite{ACR1,ACR2} from the sphere $S_1\subset \MM_1$ along the null hypersurface $\HH'_{[1,2]}$ (red) to the perturbed sphere $S_2'\subset \MM_2$.}\label{FIGperturbative}
\end{center}
\end{figure}

\ni The $10$ integral charges $\Ef,\Pf^m,\Lf^m,\Gf^m$ in Theorem \ref{THMintroPASTnull} are given by the following expressions.
\begin{definition}[Charges $\mathbf{E},\mathbf{P},\mathbf{L},\mathbf{G}$] \label{DEFchargesEPLG} Given sphere data $x$ on a sphere $S$, first define the $S$-tangent vectorfield $\mfb$ and the scalar function $\mathfrak{m}$ on $S$ by
\begin{align}
\begin{aligned}
\mfb:= \frac{\phi^3}{2\Om^2} \lrpar{\di(\Om\trchi)+\Om\trchi\lrpar{\eta-2\di\log\Om}}, \,\,
\mathfrak{m} := \phi^3 \lrpar{K+\frac{1}{4}\trchi\trchib} -\phi \Divd \mfb.
\end{aligned}\label{EQdefMUBETA}
\end{align}
Then define from $\mathfrak{m}$ and $\mfb$ the charges $\Ef,\Pf,\Lf,\Gf$ as follows, for $m=-1,0,1$,
\begin{align}
\begin{aligned}
\mathbf{E}:= \mathfrak{m}^{(0)}, \,\, \mathbf{P}^m:= \mathfrak{m}^{(1m)}, \,\, \mathbf{L}^m := \mfb^{(1m)}_H, \,\, \mathbf{G}^m := \mfb^{(1m)}_E,
\end{aligned}\label{EQdefNonlinCharges}
\end{align}
where for scalar functions $f$ and vectorfields $X$ on $S$ the projections onto the tensor spherical harmonics (with respect to the round unit metric $\gac$, as in the rest of this paper) are defined by
\begin{align*}
\begin{aligned}
f^{(lm)} := \int_S f \, Y^{lm} d\mu_{\SSS^2}, \,\,
X_E^{(lm)} := \int_S \gac(X, E^{(lm)}) d\mu_{\SSS^2}, \,\, X_H^{(lm)} := \int_S \gac(X, H^{(lm)}) d\mu_{\SSS^2},
\end{aligned}
\end{align*}
where $Y^{lm}$ are the standard real-valued spherical harmonics with respect to $\gac$, and the electric and magnetic tensor spherical harmonics $E^{(lm)}$ and $H^{(lm)}$ are respectively defined by
\begin{align*}
\begin{aligned}
E^{(lm)} := -\frac{1}{\sqrt{l(l+1)}} \di Y^{lm}, \,\,\, H^{(lm)} := \frac{1}{\sqrt{l(l+1)}}{}^\ast \di Y^{lm} \,\,\,\text{ for } l\geq1, \, -l\leq m \leq l, 
\end{aligned}
\end{align*}
where ${}^{\ast}$ denotes the Hodge dual; see also Appendix \ref{APPconstructionW}.
\end{definition}
\ni \emph{Remarks on Definition \ref{DEFchargesEPLG}.}
\begin{enumerate}
\item The above definition of $\Ef,\Pf,\Lf,\Gf$ slightly differs from the $10$ charges defined in \cite{ACR1,ACR2}. They are designed in such a way as to 
take advantage of the structure of the nonlinear terms in their corresponding  null transport equations. Nevertheless, the linearizations of $\Ef,\Pf,\Lf,\Gf$ at Minkowski are importantly in precise agreement with the $10$ corresponding linearized charges introduced in \cite{ACR1,ACR2} (modulo numerical factors).
\item In Appendix \ref{AppDerivation} we derive the null transport equations for $\Ef,\Pf,\Lf,\Gf$ by first deriving null transport equations for $\mathfrak{m}$ and $\mfb$ and then projecting them according to \eqref{EQdefNonlinCharges}.
\item The quantities $\mathfrak{m}$ and $\mfb$ can also be expressed as follows,
\begin{align*}
\begin{aligned}
\mfb= \frac{\phi^3}{\Om} \lrpar{\be + \Divd \chih + \chih \cdot (\eta-\di\log\Om)}, \,\, \mathfrak{m} = \phi^3 \lrpar{-\rho + \half (\chih,\chibh)} -\phi \Divd \mfb,
\end{aligned}
\end{align*}
where we remark that $\mathfrak{m}$ is reminiscent of the classical \emph{mass function} \cite{CHRmassaspect}. In Section \ref{SECintroSecondVersionMainTHM} we consider a simplified version of the null transport equation of $\mfb$ to illustrate several central aspects of our proof.
\item Following \cite{ACR3}, the charges $\Ef,\Pf,\Lf$ can be interpreted far out in the asymptotically flat region as (local) energy $\mathbf{E}$, linear momentum $\Pf$, and angular momentum $\Lf$. The charge $\Gf$ is interpreted in \cite{ACR3} on the sphere of radius $R$ (for $R>0$ large) as
\begin{align*}
\begin{aligned}
\Gf = \text{center-of-mass} - R \cdot \text{linear momentum}.
\end{aligned}
\end{align*}
For \emph{strongly} asymptotically flat initial data with $\Pf_{\mathrm{ADM}}=0$, $\Gf$ is accordingly interpreted as the center-of-mass; see \cite{ACR3} for details. In asymptotically flat spacelike initial data with well-defined, finite center-of-mass and non-vanishing asymptotic linear momentum $\Pf_{\mathrm{ADM}}$, $\Gf$ is expected to grow proportionally to $R$. 
The above interpretation of $\Ef,\Pf,\Lf,\Gf$ is based on the proximity of these quantities to their ADM counterparts and was 
essential for the \emph{null gluing to Kerr} in \cite{ACR3}.

\item The Minkowski reference sphere data $\mathfrak{m}$ (see \eqref{EQrefMinkowskiSphereData}) has vanishing charges $(\Ef,\Pf,\Lf,\Gf)$. Let $\varep>0$ be a real number. For sphere data $x$ on a sphere $S$ which is $\varep$-close to Minkowski reference sphere data (in the norm $\XX(S)$ to be introduced in Section \ref{SECspacesNORMS}) it holds that the $10$ charges $\Ef,\Pf,\Lf,\Gf$ are well-defined and bounded of size $\varep$.
\end{enumerate}

\ni As mentioned at the end of Section \ref{SUBSECcharseed} above, one can more generally prescribe a characteristic seed along the union $\HH\cup\HHb$ of two transversely intersecting null hypersurfaces $\HH$ and $\HHb$ having as intersection a spacelike $2$-sphere. The null gluing problem extends in straight-forward fashion to this setting, and the following result is proved in \cite{ACR1,ACR2}; see also Figure \ref{FIGbifurcate} below.

\begin{theorem}[Bifurcate null gluing of \cite{ACR1,ACR2}, version 1] \label{THMbifurcateintro} Consider spheres $S_1\subset \MM_1$ and $S_2\subset\MM_2$ with induced sphere data $x_1$ and $x_2$ close to the reference Minkowski sphere dat aon the spheres of radius $1$ and $2$, respectively. There exists a solution $x$ to the null structure equations along the bifurcate null hypersurface $\HHb\cup\HH$ emanating from the auxiliary $2$-sphere $S_{\mathrm{aux}}$ such that $x\vert_{S_1} = x_1$, and it holds that $x\vert_{S_2} = x_2$ \emph{up to the $10$-dimensional space} spanned by $\Ef,\Pf,\Lf,\Gf$. If \emph{a posteriori} it holds that
\begin{align}
	\lrpar{\mathbf{E},\Pf,\Lf,\Gf}(x\vert_{S_2}) = \lrpar{\mathbf{E},\Pf,\Lf,\Gf}(x_2),
\end{align}
then it holds that $x\vert_{S_2} = x_2$ on $S_2$.
\end{theorem}

\begin{figure}[H]
	\begin{center}
		\includegraphics[width=9.5cm]{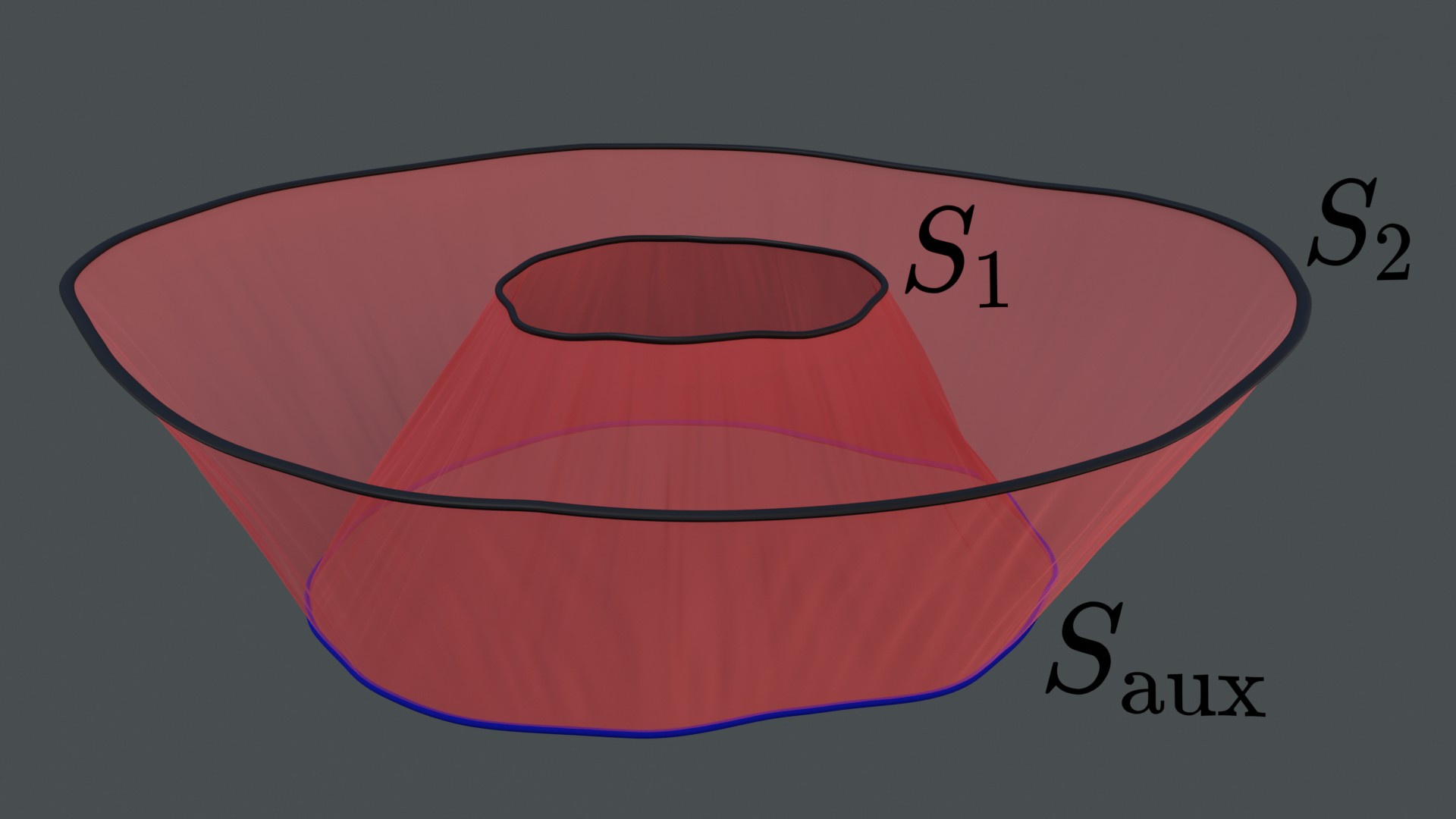} 
		\vspace{0.4cm} 
		\caption{Illustration of the codimension-$10$ null gluing constructed in \cite{ACR1,ACR2} along a bifurcate null hypersurface (red) emanating from an auxiliary sphere $S_{\mathrm{aux}}$ (blue).} \label{FIGbifurcate}
	\end{center}
\end{figure}

\ni \emph{Remarks on Theorem \ref{THMbifurcateintro}.}
\begin{enumerate}
\item In constrast to Theorem \ref{THMintroPASTnull}, Theorem \ref{THMbifurcateintro} does not require us to perturb the sphere $S_2\subset \MM_2$ or apply a gauge transformation. This stems from the additional degree of freedom of choosing the sphere data on $S_{\mathrm{aux}}$ together with the fact that -- except for $\Ef,\Pf,\Lf,\Gf$ -- the respective linearly conserved charges along $\HH$ and $\HHb$  can be adjusted as part of the null gluing along the respectively transversal null hypersurface. We refer to \cite{ACR1,ACR2} for details. 
\item The constructed solution to the null structure equations is sufficiently regular to apply the local existence results in \cite{LukChar} and \cite{LukRod1}. With thus constructed characteristic data on $\HH\cup \HHb$ we can solve the Einstein equations locally forward to construct an Einstein vacuum spacetime $(\MM,\g)$ containing $\HH\cup \HHb$. In this spacetime, we can now choose a spacelike hypersurface $\Si\subset\MM$ which then {\it{glues}} $S_1$ to $S_2$ as solution to the spacelike constraint equations. It is in this sense that our null gluing results imply the corresponding spacelike gluing results; see \cite{ACR1,ACR2,ACR3}.
\end{enumerate}

\subsection{Main theorem and overview of the proof}  \label{SECintroSecondVersionMainTHM}
We are now in position to state a first version of the main theorem of this paper.
\begin{theorem}[Main theorem: Obstruction-free null gluing, version 1] \label{THMmainIntrov1}
Let $\varep>0$ be a real number, and let $(\MM_1,\g_1)$ and $(\MM_2,\g_2)$ be two given spacetimes. Consider spheres $S_1\subset \MM_1$ and $S_2\subset\MM_2$ with respective sphere data $x_1$ and $x_2$ being $\varep$-close to the Minkowski reference sphere data on the spheres of radius $1$ and $2$, respectively. Let $(\mathbf{E},\mathbf{P},\mathbf{L},\mathbf{G})(x_A)$ for $A=1,2$ denote the charges of Definition \ref{DEFchargesEPLG} calculated from the sphere data $x_A$, and define
\begin{align*}
	\begin{aligned}
		\lrpar{\triangle \mathbf{E},\triangle \mathbf{P},\triangle \mathbf{L},\triangle \mathbf{G}} := \lrpar{\mathbf{E},\mathbf{P},\mathbf{L},\mathbf{G}}(x_2) - \lrpar{\mathbf{E},\mathbf{P},\mathbf{L},\mathbf{G}}(x_1).
	\end{aligned}
\end{align*}
Assume that for three real numbers $C_1,C_2,C_3>0$,
\begin{subequations}
\begin{align}
\triangle \mathbf{E} =& C_1\, \varep, \label{EQsmallnessMain21intro} \\ 
\vert \triangle \mathbf{L} \vert \leq& C_2 \varep^2,  \label{EQsmallnessMain1002intro} \\
\varep^{-1} (\triangle \mathbf{E}) >& C_3 \lrpar{\varep^{-1} \vert \triangle \mathbf{P}\vert + \varep^{-1} \vert \triangle \mathbf{G}\vert}. \label{EQsmallnessMain2intro}
\end{align}
\end{subequations}
There are real numbers $\tilde{\varep}>0$ and $\tilde{C}_3>0$ such that if $0<\varep<\tilde{\varep}$ and $C_3>\tilde{C}_3$, then there exist a perturbation from $S_1$ to a nearby sphere $S_1'$ in $\MM_2$ and a local gauge transformation at $S_1'$ yielding the sphere data $x_1'$ on $S_1'$ such that there exists a solution $x$ to the null structure equations along the hypersurface $\HH_{[1,2]}'$ leading from $S_1'$ to $S_2$ such that 
\begin{align}
\begin{aligned}
x\vert_{S_1'} = x'_1, \,\, x\vert_{S_2} = x_2.
\end{aligned}\label{EQfullgluing1intro}
\end{align}
In particular, as consequence of \eqref{EQfullgluing1intro}, \emph{the charges $\Ef,\Pf,\Lf,\Gf$ are glued on $S_2$}.
\end{theorem}

\ni In the following we give first a general discussion of the proof of Theorem \ref{THMmainIntrov1} before turning in Section \ref{SECtoymodelINTRO} to a model problem to illustrate the ideas.

Our starting point for the discussion of the proof of Theorem \ref{THMmainIntrov1} is the next observation. In the literature there have been three approaches to initial data gluing for the Einstein equations,
\begin{enumerate}
\item Gluing constructions centered on connected-sum gluing were developed by Chru\'sciel--Isenberg--Pollack \cite{CIP1,CIP2}, Chru\'sciel--Mazzeo \cite{ChruscielMazzeo}, Isenberg--Maxwell--Pollack \cite{IMP3}, Isenberg--Mazzeo--Pollack \cite{IMP1,IMP2}. In this context we note also the works \cite{SchoenYauPSCM,GromovL} on codimension-$3$ surgery for manifolds of positive scalar curvature.
\item Spacelike gluing along $3$-dimensional spacelike hypersurfaces pioneered by Corvino \cite{Corvino} and Corvino--Schoen \cite{CorvinoSchoen}, and further developed and refined by Chru\'sciel--Delay \cite{ChruscielDelay1,ChruscielDelay}, Chru\'sciel--Pollack \cite{ChruscielPollack}, Cortier \cite{Cortier}, Hintz \cite{Hintz}. Another milestone in this direction was the result \cite{CarlottoSchoen} by Carlotto--Schoen. 
\item Null gluing along $3$-dimensional null hypersurfaces recently established by the present authors in \cite{ACR1,ACR2,ACR3} in collaboration with Aretakis. 
\end{enumerate} 
In all of the above approaches, the gluing problem has always been treated \emph{linearly}. That is, the constraint equations are first linearized and the properties (such as injectivity, surjectivity, and flexibility) of the linearized operator are studied. Then an implicit function theorem is applied to show that, based on the properties of the linearized gluing problem, the nonlinear gluing can be achieved. In particular, all the obstructions to gluing that have been found are \emph{linear} obstructions, that is, obstructions to the linearized gluing problem.

In this paper we depart from the purely linear analysis and propose a new method for nonlinear gluing. In this approach we combine the linear analysis 
with the exploration and manipulation of the quadratic terms in the constraint equations and a novel iteration procedure, to show that the \emph{linear obstructions can be eliminated and no higher order obstructions appear} and thus establish that the gluing problem is {\it{obstruction-free}}.

In order to do that we show that we can choose specific characteristic seeds in a way such that they \emph{control both the linear and the nonlinear parts of the theory}.  The part of the characteristic seed that controls the quadratic terms consists of \emph{high-frequency terms in the $v$-variable} which also have \emph{larger amplitude} than the corresponding part used to control the problem at the linear level. This comes perhaps as a surprise because there is a legitimate danger that such large-amplitude high-frequency terms interact with the original linear theory ansatz and eventually destroy the linear part of the construction. However, the point here is that it does not, and part of it has to do exactly with the fact that we use a high-frequency ansatz which, in conjunction with the transport-nature of the hierarchy of the null structure equations, \emph{does not influence the problem at the linear level but only at the nonlinear level}. 

Formally, to combine the control of the linear and of the nonlinear terms, we set up an iteration scheme and prove that it converges. The limit of this iteration scheme is then the desired gluing solution. Recall that the classical implicit function theorem states that if the linear operator is surjective, one can use linear theory to produce an iteration scheme where error terms produced by the linear ansatz can be controlled by the next-order linear ansatz. In contrast, in our nonlinear constraint problem we need to precisely manipulate and control \emph{quadratic terms} to achieve certain properties. In doing so, we generally have error terms \emph{from linear terms and from quadratic terms}. We control some of those error terms not with the linear ansatz from the next iteration step \emph{but from the quadratic ansatz at the same step of the iteration} -- of course, this subsequently produces new linear and nonlinear error terms at the next order, and at the next order the same procedure has to be repeated.

\begin{remark} High-frequency characteristic seeds appeared previously in general relativity in the context of the black hole formation problem, see, for example, the work of Christodoulou \cite{ChrForm}, and also the analysis of high-frequency limits of solutions to the Einstein equations in the context of Burnett's conjecture, see, for example,  \cite{LukRodnianskiHIGH}. We should also note that both of these works exploited the designed largeness 
(but not the precise control) of the quadratic term $|\Omega \chih|^2$ (see below). These, to our knowledge, are the only instances of previous works which, either 
at the level of the constraint or evolutions equations, exploited actively (rather than in the form of cancellations) the nonlinear structure of the Einstein 
equations.
\end{remark}

\begin{remark} It is interesting to compare our approach to the methods of \emph{convex integration} (see, for example, \cite{Nash,Nash1,Gromov,DeLellisSz}) where one solves various nonlinear constraints by utilizing high-frequency iteration schemes. One major difference is that in the context of convex integration, while it is possible to prove the existence of solutions to nonlinear constraints with the help of high-frequency approximations, one necessary property of such iteration scheme procedures is that the frequency scale has to grow at each step, and often
results in a \emph{non-smooth solution}. The smooth isometric embedding result of Nash in \cite{Nash1} is however one of the exceptions to the above.
The construction in this paper features an infinite iteration procedure involving only 2 different high frequency scales. As a result, the limiting solution is automatically \emph{regular}. \end{remark}
\subsection{Model problem}\label{SECtoymodelINTRO} In the following we analyze a \emph{model problem} to illustrate the above ideas by a practical example. This model null gluing problem is based on the null transport equation for the vectorfield $\mfb$ introduced in Definition \ref{DEFchargesEPLG} (see also Appendix \ref{AppDerivation} for the full null transport equation for $\mfb$). In our discussion we focus on the lowest-order terms and omit higher-order terms.

Consider the following setup. Consider a vectorfield $B$ along $\HH=[1,2]\times \SSS^2$ subject to the null transport equation
\begin{align}
\begin{aligned}
D B = \di\lrpar{\vert \chih\vert^2_\gd} + \Divd_\gd \chih.
\end{aligned}\label{EQtoyODE}
\end{align}
The \emph{null gluing problem for $B$} is to pick the shear $\chih$ along $\HH$ such that for given vectorfields $B_1$ on $S_1$ and $B_2$ on $S_2$ (both $\varep$-small for a real number $\varep>0$) the solution of \eqref{EQtoyODE} satisfies 
\begin{align}
\begin{aligned}
B\vert_{S_1} = B_1, \,\, B\vert_{S_2} = B_2.
\end{aligned}\label{EQtoyboundaryconditions}
\end{align}
Moreover, in this model we stipulate that $\gd$ is determined from $\chih$ through the first variation equation
\begin{align}
\begin{aligned}
D\gd = \chih,
\end{aligned}\label{EQtoyfirstvariation}
\end{align} 
with some $\gd\vert_{S_1}$ (being $\varep$-close to the unit round metric $\gac$) given as data on $S_1$. In particular, we can calculate $\gd$ along $\HH$ by integrating $\chih$.

We remark that in the full null gluing problem for the Einstein equations, the role of $\chih$ is taken by $\Om\chih$; in other words, we set $\Om=1$ in this model problem.

In the following we discuss the null gluing problem for $B$ in three steps.
\begin{enumerate}
\item \emph{Linear gluing.} We show how to solve the null gluing problem for $B^{[l\geq2]}$, that is, prescribing all $(l\geq2)$-modes of $B$ (in particular, both electric and magnetic) on $S_1$ and $S_2$ of order $\varep>0$. This corresponds to the technique used in the codimension-$10$ gluing \cite{ACR1,ACR2}.
\item \emph{Nonlinear gluing I.} We illustrate how to go beyond \cite{ACR1,ACR2} and glue -- in addition to $B^{[l\geq2]}$ -- the electric $(l=1)$-modes of $B$ denoted by $B^{[1]}_E$ of order $\varep>0$.
\item \emph{Nonlinear gluing II.} We show how to glue, in addition to the above, the remaining magnetic $(l=1)$-modes of $B$, denoted by $B^{[1]}_H$, of order $\varep^2>0$. This completes the null gluing of $B$ along $\HH$ at orders which correspond to asymptotic flatness.
\end{enumerate}

\ni Before turning to the three steps, we note that integrating \eqref{EQtoyODE} yields the following central representation formula for our discussion of the null gluing problem for $B$ which should be kept in mind,
\begin{align*}
\begin{aligned}
B\vert_{S_2} - B\vert_{S_1} = \int\limits_1^2 \lrpar{\di\lrpar{\vert \chih\vert^2_\gd} + \Divd_\gd \chih}
\end{aligned}
\end{align*}

\ni\textbf{(1) Linear gluing.} For the construction of gluing of $B^{[\geq2]}$ we choose the following linear ansatz for $\chih$. For a smooth (low-frequency) $2$-covariant symmetric $\gac$-tracefree tensorfield $V$ along $\HH$, let
\begin{align}
\begin{aligned}
\chih := \varep V.
\end{aligned}\label{EQlinearomchihansatz}
\end{align}
By integrating \eqref{EQtoyfirstvariation}, it follows that $\gd = r^2 \gac + \OO(\varep)$ along $\HH$ (with $r=v$). First we clearly have by \eqref{EQlinearomchihansatz} that 
\begin{align}
\begin{aligned}
\vert \chih \vert^2_\gd = \OO(\varep^2).
\end{aligned}\label{EQsmallnesslinearansatz1}
\end{align}
Second, we have that
\begin{align}
\begin{aligned}
\lrpar{\Divd_\gd\chih}^{[1]} =& \lrpar{\Divd_{r^2\gac + \OO(\varep)} \lrpar{\varep V} }^{[1]} = \underbrace{\lrpar{\Divd_{r^2\gac} \lrpar{\varep V} }^{[1]}}_{=0} +\OO(\varep^2),\\
\lrpar{\Divd_\gd\chih}^{[\geq2]} =&\lrpar{\Divd_{r^2\gac}\lrpar{\varep V}}^{[\geq2]} +\OO(\varep^2),
\end{aligned}\label{EQsmallnesslinearansatz2}
\end{align}
where we used that the image of the operator $\Divd_{r^2\gac}$ has vanishing $(l=1)$-modes (see Appendix \ref{APPconstructionW}). 
Importantly, the operator $\Divd_{r^2\gac}$ is a bijection between $2$-covariant $\gac$-tracefree symmetric $2$-tensors and symmetric $\gac$-tracefree $2$-covariant tensors of modes $l\geq2$. Plugging \eqref{EQsmallnesslinearansatz1} and \eqref{EQsmallnesslinearansatz2} into \eqref{EQtoyODE} and integrating, this allows to solve the null gluing problem for $B^{[\geq2]}$ at order $\varep>0$ by stipulating the following integral condition on the tensor $V$,
\begin{align*}
\begin{aligned}
B_2^{[\geq2]} - B_1^{[\geq2]} = \int\limits_1^2 \lrpar{\Divd_{r^2\gac}\lrpar{\varep V}}^{[\geq2]} +\OO(\varep^2)
=& \int\limits_1^2 \lrpar{r^{-2}\Divd_{\gac}\lrpar{\varep V}}^{[\geq2]} +\OO(\varep^2)\\
=& \varep \cdot \Divd_{\gac}\lrpar{\int\limits_1^2r^{-2} V^{[\geq2]}} + \OO(\varep^2),
\end{aligned}
\end{align*}
where we applied the scaling of the divergence operator. By the ellipticity of $\Divd_{\gac}$, the above integral condition on $V$ is straight-forward to solve. On the other hand, \eqref{EQsmallnesslinearansatz1} and \eqref{EQsmallnesslinearansatz2} show that the null gluing problem for $B^{[1]}$ cannot be solved by the ansatz \eqref{EQlinearomchihansatz} at order $\varep>0$ because in this setting, $B^{[1]}$ is linearly \emph{conserved} along $\HH$,
\begin{align*}
\begin{aligned}
B_2^{[1]} - B_1^{[1]} = \int\limits_1^2 \lrpar{\Divd_{r^2\gac}\lrpar{\varep V}}^{[1]} + \OO(\varep^2) = \OO(\varep^2).
\end{aligned}
\end{align*}
This finishes our discussion of the null gluing of $B^{[\geq2]}$ at order $\varep>0$.\\

\ni\textbf{(2) Nonlinear gluing I.} In the following we consider the additional null gluing of $B^{[1]}_E$. Introduce the following generalization of the ansatz \eqref{EQlinearomchihansatz},
\begin{align}
\begin{aligned}
\chih = \varep V + \varep^{1/2} \sin\lrpar{\frac{v}{\varep}} \mfW_0,
\end{aligned}\label{EQnonlinearomchihansatz1}
\end{align}
where $V$ and $\mfW_0$ are smooth (low-frequency) symmetric $2$-covariant $\gac$-tracefree tensorfields along $\HH$. By integrating \eqref{EQtoyfirstvariation} we get that $\gd = r^2 \gac + \OO(\varep) - \varep^{3/2}\cos\lrpar{\frac{v}{\varep}} \mfW_0$. First, with \eqref{EQnonlinearomchihansatz1} we calculate that
\begin{align}
\begin{aligned}
\lrpar{\Divd_\gd\chih}^{[1]}_{E} =& \lrpar{\Divd_{r^2 \gac + \OO(\varep) - \varep^{3/2}\cos\lrpar{\frac{v}{\varep}} \mfW_0}\lrpar{\varep V + \varep^{1/2} \sin\lrpar{\frac{v}{\varep}} \mfW_0}}^{[1]}_{E} \\
=& \OO(\varep^2) + \OO(\varep) \cdot  \Nd \lrpar{\varep^{1/2} \sin\lrpar{\frac{v}{\varep}} \mfW_0}.
\end{aligned}\label{EQexpansionEdivd}
\end{align}
It plays an important role in this paper that \emph{the integration of high-frequency functions such as $\sin\lrpar{\frac{v}{\varep}}$ yields extra $\varep$-smallness}. For the term \eqref{EQexpansionEdivd} this results in the property that after integration, its contribution to $B\vert_{S_2}-B\vert_{S_1}$ is of order $\OO\lrpar{\varep^2}$,
\begin{align*}
\begin{aligned}
\int\limits_1^2 \lrpar{\Divd_\gd\lrpar{\chih}}^{[1]}_{E} = \OO(\varep^2) + \int\limits_1^2 \OO(\varep) \cdot  \Nd \lrpar{\varep^{1/2} \sin\lrpar{\frac{v}{\varep}} \mfW_0} = \OO(\varep^2) + \OO(\varep^{5/2}) = \OO(\varep^2).
\end{aligned}
\end{align*}
Second, with \eqref{EQnonlinearomchihansatz1} we importantly calculate that, at lowest-order,
\begin{align*}
\begin{aligned}
\vert \chih\vert^2_\gd = \varep \sin^2\lrpar{\frac{v}{\varep}} \vert \mfW_0\vert^2_{r^2\gac} = \half \varep \vert \mfW_0\vert^2_{r^2\gac} - \half \varep \cos\lrpar{\frac{2v}{\varep}} \vert \mfW_0\vert^2_{r^2\gac} .
\end{aligned}
\end{align*}
While the second term on the right-hand side is high-frequency and thus contributes only at order $\OO(\varep^2)$ to $B^{[1]}_E\vert_{S_2}-B^{[1]}_E\vert_{S_1}$, the first term contributes at order $\varep$. Indeed, using that $(\di f)^{(1m)}_E = -\sqrt{2}f^{(1m)}$ by Fourier analysis,
\begin{align*}
\begin{aligned}
\int\limits_1^2 \lrpar{\di \lrpar{\vert\chih\vert^2_\gd}}^{(1m)}_E = - \sqrt{2} \int\limits_1^2 \half \varep \lrpar{\vert \mfW_0\vert^2_{r^2\gac}}^{(1m)} + \OO(\varep^2).
\end{aligned}
\end{align*}
The above shows that the null gluing problem for $B^{[1]}_E$ can be solved at order $\varep$ if we can prescribe the values of the integral, for $m=-1,0,1$,
\begin{align}
\begin{aligned}
\int\limits_1^2 \lrpar{\vert \mfW_0\vert^2_{r^2\gac}}^{(1m)}.
\end{aligned}\label{EQintegralsquarelone}
\end{align}
We prove that this is indeed possible. In fact, it can be done with the explicit  ansatz for $\mfW_0$,
\begin{align*}
\begin{aligned}
\mfW_0 := f_{30} \psi^{30} + \sum\limits_{m=-1}^1 f_{2m}\psi^{2m},
\end{aligned}
\end{align*}
where $\psi^{30}$ and $\psi^{2m}$ are fixed $v$-independent tensor spherical harmonics, see Appendix \ref{APPconstructionW}, and $f_{30}, f_{2m}$ are appropriately chosen low-frequency scalar functions dependent only on $v$. For further details see Section \ref{SECproofW}.

We add a comment to illustrate that the explicit construction of $\mfW_0$ given in the paper can also be replaced by a more abstract argument. We may first reduce the statement to simply finding a symmetric $\gac$-tracefree 2-tensor $\mfW_0$ on $\Bbb S^2$ with the property (by rescaling) that
$$
(|\mfW_0|^2)^{(0)}=1,
$$
and 
$$
|(|\mfW_0|^2)^{(1m)}|\le \delta, \text{ for } m=-1,0,1,
$$
are prescribed, where $\delta>0$ is a sufficiently small constant.

Let $\xi_k$ be a sequence of smooth bump functions with disjoint supports and the property that these supports form a mesh  of a sufficiently small size 
(say, $\kappa\ll \delta$) of $\Bbb S^2$ and let $z_k$ be a sequence of constant symmetric $\gac$-tracefree $2$-tensors (using trivialization of the corresponding 
bundle over ${\text{sppt}} \,\xi_k$), so that 
 $$
 \mfW_0=\sum_k \xi_k z_k,\,\,\, |\mfW_0|^2= \sum_k \xi_k^2 |z_k|^2.
  $$
Let 
$$
a_k^0= \int_{\Bbb S^2} \frac{1}{\sqrt{4\pi}} \xi_k^2,\,\,\, a_k^j= \int_{\Bbb S^2}Y^{1 (j-2)} \xi_k^2,\, j=1,2,3.
$$
We can normalize the choice of $\xi_k$ so that for some choice of a partition of $\Bbb S^2$ into disjoint $\kappa$-size open sets $U_k$: 
$\Bbb S^2=\cup \overline U_k$ and ${\text{sppt}} \,\xi_k\subset U_k$, we have 
$$
a_k^0=\frac 1{4\pi}\int_{U_k} 1,
$$
so that in particular 
$$
\sum_k a_k^0=1.
$$
It follows that we need to find a sequence of non-negative numbers $x_k=|z_k|^2$ such that 
$$
(Sx)^j:=\sum_k a_k^j x_k=w^j,\,\,\, w^0=(|\mfW_0|^2)^{(0)}=1,\,\,\, w^{j}=(|\mfW_0|^2)^{(1(j-2))}.
$$
We observe that the sequence $x_k=1$ produces the solution 
$$
w^0=1,\,\,\, |w^j|\les \kappa, \,\,\, j=1,2,3,
$$
where the latter is just a Riemann sum estimate of the vanishing integrals $\int_{\Bbb S^2}Y^{1 (j-2)}$. This can be shown as follows.
For a sequence of points $p_k\in U_k$,
\begin{align*}
\int_{\Bbb S^2}Y^{1 (j-2)}&=\sum_k \int_{U_k}Y^{1 (j-2)}=\sum_k Y^{1 (j-2)}(p_k)  \int_{U_k} 1 +\OO(\kappa)
\\&= \sqrt{4\pi}\sum_k Y^{1 (j-2)}(p_k)  \int_{\Bbb S^2} \xi^2_k  +\OO(\kappa)=\sqrt{4\pi}\sum_k \int_{\Bbb S^2}\xi_k^2Y^{1 (j-2)}+\OO(\kappa).
\end{align*}
To find the desired solution for a given 4-tuple $w^0=1$ and $|w^j|\le \delta$, we need to show that $S$ is invertible and, if it is, we can choose a positive 
solution. We compute the conjugate operator $S^*$
$$
(S^*\omega)_k =\sum_{j=0}^3 a_k^j \omega^j.
$$
$S^*$ has a trivial kernel, which follows from a discretized version of the statement that the functions $Y^0, Y^{1m}$ are linearly independent.
Arguing as above, we see that 
$$
\kappa^{-1} (S^*\omega)_k=\frac{1}{\sqrt{4\pi}}Y^0(p_k) \omega^0+\sum_m \frac{1}{\sqrt{4\pi}}Y^{1m}(p_k) \omega^{m+2}+\OO(\kappa |\omega|).
$$
Note that the linear independence may already be established for $k$ ranging over a finite (much smaller, in fact independent of, than $\approx 1/\kappa^2$ terms in the sequence) set. As a result, the operator $SS^*:\Bbb R^4\to \Bbb R^4$ is invertible with the norm independent of $\kappa$.
In fact, explicitly computing $SS^*$
$$
(SS^*)_{ij}=\sum_k a_k^i a_k^j,
$$
we see that it is a Riemann sum approximation of the inner product of the spherical harmonics $Y^0, Y^{1m}$. As a consequence,
$$
(SS^*)_{ij}=\delta_{ij}+\OO(\kappa).
$$
Now, solving the problem 
$$
SS^*\omega=w,
$$
and defining 
$$
x:=S^*\omega,
$$
we find the desired solution. To verify that $x$ is positive, we first observe that since $w^0=1$ and $|w^j|\le \delta$ and $\kappa\ll\delta$, 
it follows that 
$$
\omega^0=1+\OO(\kappa), \,\,\, |\omega^j|\le \delta.
$$
As a result,
$$
x_k=(S^*\omega)_k=1+\OO(\delta),
$$
as required.\\

\ni \textbf{(3) Nonlinear gluing II.} In the following we consider the additional null gluing of $B^{[1]}_H$. We first observe that because $(\di f)_H=0$ for any scalar function $f$, the high-frequency ansatz \eqref{EQnonlinearomchihansatz1} leads only to a change of order $\varep^2$ of $B^{[1]}_H$,
\begin{align}
\begin{aligned}
B^{(1m)}_H\vert_{S_2}-B^{(1m)}_H\vert_{S_1} =\int\limits_1^2 \lrpar{\Divd_\gd\chih}_H^{(1m)} = \OO(\varep^2).
\end{aligned}\label{EQintegralnonlinear2intro}
\end{align}
However, at this point we note that the null gluing problem for $B^{(1m)}_H$ \emph{at order $\varep^2$} is in accordance with the \emph{decay rates of asymptotically flat spacelike initial data with bounded angular momentum}. Nevertheless, an inspection of the above integral \eqref{EQintegralnonlinear2intro} shows that the contribution $\OO(\varep^2)$ stems in fact solely from the \emph{linear ansatz} \eqref{EQlinearomchihansatz} (which was determined by the gluing of $B^{[\geq2]}$), while the high-frequency terms in the ansatz \eqref{EQnonlinearomchihansatz1} contributes only $\OO(\varep^3)$ to the integral. Indeed, this is already visible from \eqref{EQexpansionEdivd} by recalling the principle of high-frequency improvement and noting that $\sin\lrpar{\frac{v}{\varep}}\cos\lrpar{\frac{v}{\varep}}$ is a high-frequency function of frequency $1/\varep$ (so that its integration improves smallness by $\varep$). 

To solve the null gluing problem for $B^{[1]}_H$ at order $\varep^2$, we introduce the following modified high-frequency ansatz,
\begin{align}
\begin{aligned}
\chih = \varep V + \varep^{1/2} \sin\lrpar{\frac{v}{\varep}} \mfW_0 + \varep^{3/4}\lrpar{\sin\lrpar{\frac{v}{\sqrt{\varep}}} \mfW_1 + \cos\lrpar{\frac{v}{\sqrt{\varep}}} \mfW_2 },
\end{aligned}\label{EQnonlinearomchihansatz2}
\end{align}
where $\mfW_1$ and $\mfW_2$ are smooth (low-frequency) $2$-covariant symmetric $\gac$-tracefree tensorfields on $\HH$. By integrating \eqref{EQtoyfirstvariation} we get that
\begin{align}
\begin{aligned}
\gd= r^2 \gac + \OO(\varep) -\varep^{3/2} \cos\lrpar{\frac{v}{\varep}} \mfW_0 +\varep^{5/4} \lrpar{-\cos\lrpar{\frac{v}{\sqrt{\varep}}} \mfW_1+ \sin\lrpar{\frac{v}{\sqrt{\varep}}} \mfW_2}
\end{aligned}\label{EQmetricexpressionnonlin2}
\end{align}
Plugging \eqref{EQnonlinearomchihansatz2} and \eqref{EQmetricexpressionnonlin2} into $\Divd_\gd\chih$, we get, at lowest-order,
\begin{align*}
\begin{aligned}
\int\limits_1^2 \lrpar{\Divd_\gd \chih}^{(1m)}_H =& \varep^2 \int\limits_1^2 \lrpar{-\cos^2\lrpar{\frac{v}{\sqrt{\varep}}}\mfW_1 \cdot \Nd \mfW_2 + \sin^2\lrpar{\frac{v}{\sqrt{\varep}}} \mfW_2 \cdot \Nd \mfW_1}_H^{(1m)} + \OO(\varep^2) \\
=& - \varep^2 \int\limits_1^2 \lrpar{\mfW_1 \cdot \Nd \mfW_2}_H^{(1m)} + \OO(\varep^2),
\end{aligned}
\end{align*}
where we integrated by parts the second integrand, using that $(\di f)_H=0$ for any scalar function. Here the term $\OO(\varep^2)$ does not depend on $\mfW_0, \mfW_1$ and $\mfW_2$ but only on $\varep V$ of \eqref{EQlinearomchihansatz}, and we used that products of high-frequency terms with different frequencies are again high-frequency and thus subject to high-frequency improvement in integration.

Thus we reduced the gluing problem of $B^{[1]}_H$ at order $\varep^2$ to the prescription of the integral (the actual expression is a sum of two terms of this kind)
\begin{align}
\begin{aligned}
\int\limits_1^2 \lrpar{\mfW_1 \cdot \Nd \mfW_2}_H^{(1m)}.
\end{aligned}\label{EQprescriptionnonlin2}
\end{align}
We show that this is once again possible. And, again, it can be done with the explicit ansatz
\begin{align*}
\begin{aligned}
\mfW_1 = \tilde{f}_{30} \psi^{30}, \,\, \mfW_2 = \sum\limits_{m=-1}^1 \tilde{f}^{2m} \psi^{2m},
\end{aligned}
\end{align*}
where $\tilde{f}_{30}, \tilde{f}_{2m}$ are appropriately chosen low-frequency scalar functions in $v$. For further details see Section \ref{SECproofW}.

\vskip .5pc

We note moreover that, by similar arguments as above,
\begin{align*}
\begin{aligned}
\int\limits_1^2 \lrpar{\vert\chih\vert^2_\gd}^{(1m)} =&- \sqrt{2} \int\limits_1^2 \half \varep \lrpar{\vert \mfW_0\vert^2_{r^2\gac}}^{(1m)} + \OO(\varep^{5/4}),
\end{aligned}
\end{align*}
that is, at lowest-order $\varep$, the added high-frequency terms with $\mfW_1$ and $\mfW_2$ are not counteracting the adjustment of $B^{[1]}_E$ through $\mfW_0$.

\vskip .5pc
To carry out the gluing procedure {\it{at all orders}} of $\varep$ requires an iteration, combined from linear and 2 nonlinear steps as above. 
The details of this iteration are discussed in the body of the paper. Perhaps one of its surprising elements is the fact each successive step 
can be done with the linear ansatz of low $v$-frequency and the nonlinear ansatzes of the {\it{same}} 2 $v$-frequencies $\frac 1{\sqrt\varep}, 
\frac 1\varep$, attached to a finite number of spherical harmonics. For further details on the iteration process see Section \ref{SECgluingEP}.

\subsection{Organisation of the paper} The paper is structured as follows.
\begin{itemize}
\item In Section \ref{SECprelimANDmainResults}, we present the null structure equations, define norms, 
give the statement of the perturbative codimension-$10$ gluing of \cite{ACR1,ACR2}, and state the 
precise versions of the main theorem of this paper and its application to spacelike gluing.
\item In Section \ref{SECgluingEP}, we prove the main theorem by employing an iteration scheme which combines the codimension-$10$ gluing of Theorem \ref{THMintroPASTnull} (proved in \cite{ACR1,ACR2}) and the null gluing of $(\Ef,\Pf,\Lf,\Gf)$ in Theorem \ref{THMmain0}.
\item In Section \ref{SECproofSPACELIKE} we prove the obstruction-free spacelike gluing result, Corollary \ref{CORspacelike}.
\item In Section \ref{SECconstructionSolution} we construct and bound specific high-frequency solutions to the null structure equations by integrating the corresponding null transport equations in hierarchical order.
\item In Section \ref{SECproofW}, we complete the proof  of Theorem \ref{THMmain0} by showing how  to use
the construction in the previous section to glue the  $(\Ef,\Pf,\Lf,\Gf)$ charges and proving the necessary  estimates.
\item In Appendix \ref{APPconstructionW} we recall the basics from Fourier theory of spherical harmnics. 
\item In Appendix \ref{AppDerivation} we derive the null transport equations for $(\Ef,\Pf,\Lf,\Gf)$ along $\HH$.

\end{itemize}
\vskip .5pc

\subsection{Acknowledgements} The authors would like to thank Andy Strominger for inspiring comments. I.R. is supported by a Simons Investigator Award.

\section{Preliminaries and statement of main theorem} \label{SECprelimANDmainResults}

\ni The notation of this paper follows \cite{ACR1,ACR2,ChrForm}. For real numbers $A,B$, the relation $A\les B$ indicates that there is a universal constant $C>0$ such that $A\leq CB$. The expression $A\les_q B$ indicates that $A\leq CB$ holds for a constant $C>0$ depending on the quantity $q$. 

For real numbers $\varep>0$ and $\a \geq0$, let $\OO(\varep^{\a})$ denote terms such that $\OO(\varep^\a)/\varep^{\a}$ stays bounded as $\varep\to0$. Let $\varphi(v): [1,2]\to[0,1]$ be a smooth cut-off function vanishing near $v=1$ and $v=2$,
\begin{align}
\begin{aligned}
\varphi\vert_{[1+\frac{1}{8}, 2-\frac{1}{8}]} \equiv 1, \,\, \varphi\vert_{[1,1+\frac{1}{16}] \cup [2-\frac{1}{16},2]} \equiv 0.
\end{aligned}\label{EQcutoffFCTdef}
\end{align} 

\subsection{Null structure equations} \label{SECdoublenull} As discussed in Section \ref{SECSUBnullframework}, the double null geometry and the Einstein equations stipulate the following \emph{null structure equations} between metric components, Ricci coefficients and null curvature components. Before stating them, we introduce the following notation following Chapter 1 of \cite{ChrForm}.
\begin{itemize}
\item For two $S_{u,v}$-tangential $1$-forms $X$ and $Y$,
\begin{align*} \begin{aligned}
(X,Y):=& \gd(X,Y), & ({}^\ast X)_A :=& \ind_{AB}X^B, \\
(X \widehat{\otimes} Y)_{AB} :=& X_A Y_B + X_B Y_A - (X \cdot Y)\gd_{AB}, & \Divd X :=& \Nd^A X_A, \\
 (\Nd \widehat{\otimes} Y)_{AB} :=& \Nd_A Y_B + \Nd_B Y_A - (\Divd Y)\gd_{AB}, & \Curld X :=& \ind^{AB}\Nd_A X_B,
\end{aligned}\end{align*}
where $\ind$ denotes the area $2$-form of $(S_{u,v},\gd)$.
\item For two symmetric $S_{u,v}$-tangential $2$-tensors $V$ and $W$, 
\begin{align*}
\tr V := \gd^{AB} V_{AB}, \,\, \widehat{V} := V - \half \tr V \gd, \,\, V \wedge W := \ind^{AB} V_{AC}W^C_{\,\,\,B}.
\end{align*}
\item For a symmetric $S_{u,v}$-tangential $2$-tensor $V$ and a $1$-form $X$,
\begin{align*} 
\begin{aligned} 
(V \cdot X)_A := V_{AB}X^B.
\end{aligned} 
\end{align*}
\item For a symmetric $S_{u,v}$-tangential  $2$-tensor $V$,
\begin{align*} 
\begin{aligned} 
\Divd V_A := \Nd^B V_{BA}.
\end{aligned} 
\end{align*}
\item For a symmetric $S_{u,v}$-tangential tensor $W$, let $\widehat{D}W$ and $\widehat{\Du}W$ denote the tracefree parts of $DW$ and $\Du W$, respectively, with respect to $\gd$.
\end{itemize}

\ni We are now in position to state the \emph{null structure equations}, see also Chapter 1 of \cite{ChrForm}.\\

\begin{subequations}
\ni \textbf{First variation equations}
\begin{align}
\begin{aligned}
D\gd = 2\Om\chi, \,\, \Du \gd =& 2 \Om \chib.
\end{aligned}\label{EQfirstvariation001}
\end{align}
In particular, using \eqref{EQconfDecompintro} and \eqref{EQdefRicciINTRO}, the first variation equations imply that
\begin{align}
\begin{aligned}
D\phi=\frac{\Om\trchi\phi}{2}, \,\, \widehat{D\gd}=2\Om\chih, \,\, \Du\phi=\frac{\Om\trchib\phi}{2}, \,\, \widehat{\Du\gd}=2\Om\chibh.
\end{aligned}\label{EQcompactRAYpart1}
\end{align}
 
\ni \textbf{Raychaudhuri equations}
\begin{align}
\begin{aligned}
D\trchi + \frac{\Om}{2}(\trchi)^2 - \om \trchi = -\Om\vert\chih\vert^2, \,\, \Du \trchib + \frac{\Om}{2} (\trchib)^2 - \omb \trchib = - \Om \vert \chibh \vert^2_{\gd}.
\end{aligned}\label{EQfirstvariation001part2}
\end{align}

\ni In this paper we combine \eqref{EQcompactRAYpart1} and \eqref{EQfirstvariation001part2} in the forms 
\begin{align}
\begin{aligned}
D\lrpar{\Om^{-2}D\phi} = -\frac{\phi}{2\Om^2} \vert \Om\chih \vert^2_{\gd}, \,\, D\lrpar{\frac{\phi\Om\trchi}{\Om^2}}=&-\frac{\phi}{\Om^2} \vert \Om\chih \vert^2.
\end{aligned}\label{EQcompactRAY}
\end{align}
\textbf{Null transport equations for Ricci coefficients}
\begin{align} \begin{aligned}
D\chih =& \Om \vert \chih \vert^2 \gd + \om \chih - \Om \a, &\Du \chibh =& \Om \vert \chibh \vert^2 \gd + \omb \chibh - \Om \ab,\\
D \eta =& \Om (\chi \cdot \etab - \beta), & \Du \etab =& \Om (\chib \cdot \eta + \beb), \\
D \omb =& \Om^2(2 (\eta, \etab) - \vert \eta \vert^2 -\rh), & \Du \om=& \Om^2(2 (\eta,\etab) - \vert \etab \vert^2 - \rh), \\
D\etab =& - \Om (\chi \cdot \etab -\be) + 2 \di \om, & \Du\eta =& - \Om (\chib \cdot \eta + \beb) + 2 \di \omb.
\end{aligned} \label{EQtransportEQLnullstructurenonlinear}\end{align}
where $\etab:= -\eta + 2\di\log\Om$, and 
\begin{align}\begin{aligned}
D (\Om \trchib) =& 2 \Om^2 \Divd \etab + 2 \Om^2 \vert \etab \vert^2 - \Om^2 (\chih, \chibh) - \half \Om^2 \trchi \trchib + 2 \Om^2 \rh,\\ 
\Du (\Om \trchi) =& 2 \Om^2 \Divd \eta + 2 \Om^2 \vert \eta \vert^2 - \Om^2 (\chih, \chibh) - \half \Om^2 \trchi \trchib + 2 \Om^2 \rh,
\end{aligned} \label{EQtransporttrchitrchib1}\end{align}
as well as
\begin{align} \begin{aligned}
D(\Om \chibh) =& \Om^2\lrpar{(\chih, \chibh) \gd + \half \trchi \chibh + \Nd \widehat{\otimes}\etab + \etab \widehat{\otimes}\etab - \half \trchib \chih}, \\
\Du(\Om \chih) =& \Om^2 \lrpar{(\chih, \chibh)\gd + \half \trchib \chih + \Nd \widehat{\otimes} \eta + \eta \widehat{\otimes} \eta - \half \trchi \chibh}.
\end{aligned} \label{EQchihequations1} \end{align} 

\ni \textbf{Gauss equation}
\begin{align}
	K + \frac{1}{4}\trchi\trchib - \half (\chih,\chibh) = -\rho, \label{EQGaussEQ}
\end{align}
\ni where $K$ denotes the Gauss curvature of $(S,\gd)$ and can be calculated as follows,
\begin{align*}
\begin{aligned}
K = \half \mathrm{R}_{\mathrm{scal}}(\gd) = \half \gd^{BC}\lrpar{\pr_A\Gammad^A_{BC}-\pr_C\Gammad^A_{BA}+ \Gammad^E_{BC}\Gammad^A_{AE}-\Gammad^E_{BA}\Gammad^A_{CE}},
\end{aligned}
\end{align*}
where $\mathrm{R}_{\mathrm{scal}}(\gd)$ denotes the scalar curvature of $\gd$, and $\Gammad^A_{BC}$ the Christoffel symbols of $\gd$,
\begin{align*}
\begin{aligned}
\Gammad^A_{BC} := \half \gd^{AD} \lrpar{\pr_B \gd_{CD}+\pr_C\gd_{BD}-\pr_D\gd_{BC}}.
\end{aligned}
\end{align*}

\ni \textbf{Gauss-Codazzi equations}
\begin{align}\begin{aligned}
\Divd \chih -\half \di \trchi + \chih \cdot \zeta -\half \trchi \zeta =& -\beta \label{EQGaussCodazzi}, \\
\Divd \chibh - \half \di \trchib -\chibh \cdot \zeta +\half \trchib \zeta =& \beb.
\end{aligned}\end{align}
\ni \textbf{Curl equations}
\begin{align}
\begin{aligned}
\Curld \eta=& - \half \chih \wedge \chibh - \si, & \Curld \etab =& - \Curld \eta = - \Curld \zeta.
\end{aligned}\label{EQcurlEquations}
\end{align}

\ni\textbf{Null transport equation for $\Du\omb$} (see Appendix B of \cite{ACR1} for a derivation)
\begin{align} 
\begin{aligned} 
D\Du\omb =&-12 \Om^2 (\eta-\di\log\Om,\di\omb)+2\Om^2\omb\lrpar{(\eta,-3\eta+4\di\log\Om)-\rho}\\
&+4\Om^3\chib(\eta,\di\log\Om) +\Om^3 \lrpar{\beb,7\eta-3\di\log\Om} + \frac{3}{2} \Om^3\trchib\rh + \Om^3 \Divd\beb + \frac{\Om^3}{2} (\chih, \ab).
\end{aligned} \label{EQDUOMU1}
\end{align}
We omit the statement of the analogous null transport equation for $D\om$ along $\HHb$.\\

The null curvature components also satisfy null transport equations along $\HH$ and $\HHb$, the so-called \emph{null Bianchi equations}, see Chapter 1 of \cite{ChrForm}. However, given \eqref{EQtransportEQLnullstructurenonlinear},\eqref{EQGaussEQ},\eqref{EQGaussCodazzi},\eqref{EQcurlEquations}, the only relevant null Bianchi equation for the purpose of null gluing are the following.\\

\ni \textbf{Null Bianchi equations for $\ab$ and $\a$.} 
\ni By Proposition 1.2 in \cite{ChrForm}, the following \emph{null Bianchi equations} hold,
\begin{align} \begin{aligned}
\widehat{\Du} \a - \half \Om \tr \chib \a + 2 \omb \a + \Om \left( -\Nd \widehat{\otimes} \be - (4 \eta + \zeta) \widehat{\otimes} \be + 3 \chih \rh + 3 {}^\ast \chih \si \right) =&0, \\
\widehat{D} \aa - \half \Om \tr \chi \aa + 2 \om \aa + \Om \left( \Nd \widehat{\otimes} \beb + (4 \etab - \zeta) \widehat{\otimes} \beb + 3 \chibh \rh -3 {}^\ast \chibh \si \right) =&0.
\end{aligned} \label{EQnullBianchiEquations}\end{align}

\end{subequations}

\subsection{Function spaces and norms} \label{SECspacesNORMS} We start by defining the standard tensor spaces on spheres and null hypersurfaces.
\begin{definition}[Tensor spaces on $2$-spheres] \label{DEFsphereSPACES} For two real numbers $v\geq u$, let $S_{u,v}$ be a $2$-sphere equipped with a round unit metric $\gac$. For integers $m\geq0$ and tensors $T$ on $S_{u,v}$, define
\begin{align*} 
\begin{aligned} 
\Vert T \Vert^2_{H^m(S_{u,v})} :=   \sum\limits_{i=0}^m  \left\Vert \Nd^i T \right\Vert^2_{L^2(S_{u,v})},
\end{aligned} 
\end{align*}
where the covariant derivative $\Nd$ and the measure in $L^2(S_{u,v})$ are with respect to the round metric $\ga=(v-u)^2\gac$. Let $H^m(S_{u,v}) := \{ T: \Vert T \Vert_{H^m(S_{u,v})} < \infty \}$.
\end{definition}

\begin{definition}[Tensor spaces along null hypersurfaces] \label{DEFnullHHspaces} Let $m\geq0$ and $l\geq0$ be two integers. In the following let $D$ and $\Du$ be defined as in Section \eqref{SECSUBnullframework} for null hypersurfaces in Minkowski. 
\begin{enumerate}

\item For real numbers $u_0 < v_1 < v_2$ and $S_{u_0,v}$-tangential tensors $T$ on $\HH_{u_0,[v_1,v_2]}$, define
\begin{align*} 
\begin{aligned} 
\Vert T \Vert^2_{H^m_l(\HH_{u_0,[v_1,v_2]})} :=&  \int\limits_{v_1}^{v_2} \,\, \sum\limits_{0\leq j\leq l} \left\Vert D^j T \right\Vert^2_{H^m(S_{u_0,v})} dv, \\
\Vert T \Vert_{L^\infty_vH^m(S_{u_0,v})} :=&  \sup\limits_{v_1\leq v \leq v_2} \left\Vert T \right\Vert_{H^m(S_{u_0,v})}.
\end{aligned} 
\end{align*}
The function spaces are denoted by $H^m_l(\HH_{u_0,[v_1,v_2]}) := \{ F: \Vert F \Vert_{H^m_l(\HH_{u_0,[v_1,v_2]})} < \infty\}$, and $L^\infty_vH^m(S_{u_0,v}):= \{ F: \Vert F \Vert_{L^\infty_vH^m(S_{u_0,v})} < \infty\}$.

\item For real numbers $u_1 < u_2 < v_0$ and $S_{u,v_0}$-tangential tensors $T$ on $\HHb_{[u_1,u_2],v_0}$, define
\begin{align*} 
\begin{aligned} 
\Vert T \Vert^2_{H^m_l(\HHb_{[u_1,u_2],v_0})} :=  \int\limits_{u_1}^{u_2} \,\, \sum\limits_{0\leq j\leq l} \left\Vert \Du^j T \right\Vert^2_{H^m(S_{u,v_0})} du, \end{aligned} 
\end{align*}
and let $H^m_l(\HHb_{[u_1,u_2],v_0}) := \{ F: \Vert F \Vert_{H^m_l(\HHb_{[u_1,u_2],v_0})} < \infty\}$.

\end{enumerate}
\end{definition}

\ni The following calculus estimates are standard applied tacitly throughout this paper. 
\begin{lemma}[Calculus estimates] \label{LEMstandardTRACE} Let $u_0<v_1<v_2$ be real numbers. The following holds.

\begin{enumerate}

\item \textbf{Trace estimate.} Let $T$ be an $S_{u_0,v}$-tangent tensor on $\HH_{u_0,[v_1,v_2]}$. Then, for $v_1\leq v \leq v_2$,
\begin{align*} 
\begin{aligned} 
\Vert T \Vert_{L^\infty_vH^m(S_{u_0,v})} \leq& C_{m,u_0,v_1,v_2} \cdot \Vert T \Vert_{H^m_1(\HH_{u_0,[v_1,v_2]})},
\end{aligned} 
\end{align*}
where the constant $C_{m,u_0,v_1,v_2}>0$ depends on $m,u_0, v_1$ and $v_2$. 

\item \textbf{$L^\infty$-estimate.} For any $S_{u_0,v_0}$-tangent tensor $T$ on $S_{u_0,v_0}$ we have that
\begin{align*} 
\begin{aligned} 
\Vert T \Vert_{L^\infty(S_{u_0,v_0})} \leq& C_{u_0,v_0} \cdot \Vert T \Vert_{H^2(S_{u_0,v_0})},
\end{aligned} 
\end{align*}
where the constant $C_{u_0,v_0}>0$ depends on $u_0, v_0$. 

\item \textbf{Product estimate.} Let $m_1, m_2 \geq 2$ and $l_1, l_2 \geq 1$ be integers, and further let $T \in H^{m_1}_{l_1}(\HH_{u_0,[v_1,v_2]})$ and $T' \in H^{m_2}_{l_2}(\HH_{u_0,[v_1,v_2]})$ be two $S_{u_0,v}$-tangent tensors. Then it holds that for integers $0\leq m\leq \mathrm{min}(m_1,m_2)$ and $0\leq l \leq \mathrm{min}(l_1,l_2)$,
\begin{align*} 
\begin{aligned} 
\Vert T \cdot T' \Vert_{H^m_l(\HH_{u_0,[v_1,v_2]})} \leq& C \cdot \Vert T\Vert_{H^m_l(\HH_{u_0,[v_1,v_2]})} \cdot \Vert T' \Vert_{H^{m_2}_{l_2}(\HH_{u_0,[v_1,v_2]})} \\
&+ C \cdot \Vert T\Vert_{H^{m_1}_{l_1}(\HH_{u_0,[v_1,v_2]})} \cdot \Vert T' \Vert_{H^{m}_{l}(\HH_{u_0,[v_1,v_2]})},
\end{aligned} 
\end{align*}
where the constant $C>0$ depends on $m,m_1,m_2,l,l_1$ and $l_2$.
\end{enumerate}
\end{lemma}

\ni We are now in position to define the norm of sphere data, see Definition \ref{DEFspheredata2}.
\begin{definition}[Norm for sphere data] \label{DEFnormFirstOrderDATA} Let $x_{u,v}$ be sphere data. Define 
\begin{align*} 
\begin{aligned} 
\Vert x_{u,v} \Vert_{\XX(S_{u,v})} :=& \Vert \Om \Vert_{H^{6}(S_{u,v})} +\Vert \gd \Vert_{H^{6}(S_{u,v})} + \Vert \Om\trchi \Vert_{H^{6}(S_{u,v})} + \Vert \chih \Vert_{H^{6}(S_{u,v})}\\
& + \Vert \Om\trchib \Vert_{H^{4}(S_{u,v})} + \Vert \chibh \Vert_{H^{4}(S_{u,v})} + \Vert \eta \Vert_{H^{5}(S_{u,v})} \\
&+ \Vert \om \Vert_{H^{6}(S_{u,v})}+ \Vert D\om \Vert_{H^{6}(S_{u,v})}+\Vert \omb \Vert_{H^{4}(S_{u,v})} + \Vert \Du\omb \Vert_{H^{2}(S_{u,v})} \\
&+ \Vert \a \Vert_{H^{6}(S_{u,v})} +\Vert \ab \Vert_{H^{2}(S_{u,v})},
\end{aligned} 
\end{align*}
where the norms are with respect to the round metric $\ga= (v-u)^2 \gac$ on $S_{u,v}$. Let
\begin{align*} 
\begin{aligned} 
\XX(S_{u,v}) := \{ x_{u,v} : \Vert x_{u,v} \Vert_{\XX(S_{u,v})} < \infty\}.
\end{aligned}
\end{align*}
\end{definition}
\ni We remark that the charges $\mathbf{E}, \mathbf{P}, \mathbf{L}, \mathbf{G}$ (introduced in Definition \ref{DEFchargesEPLG}) are well-defined for sphere data $x\in \XX(S)$, with 
\begin{align*}
\begin{aligned}
\vert \mathbf{E} \vert + \vert \mathbf{P} \vert + \vert \mathbf{L}\vert + \vert\mathbf{G} \vert \les \Vert x-\mathfrak{m} \Vert_{\XX(S)}.
\end{aligned}
\end{align*}

\ni Next we define function spaces to control the solutions to the null structure equations constructed in this paper. First we recapitulate the notion of \emph{null data} from \cite{ACR1,ACR2}.  
\begin{definition}[Outgoing and ingoing null data] \label{DEFnulldata111} We define the following.
\begin{enumerate}
\item For real numbers $u_0 < v_1 <v_2$, \emph{outgoing null data} on $\HH_{u_0,[v_1,v_2]}$ is given by a tuple of $S_{u_0,v}$-tangent tensors 
\begin{align} 
\begin{aligned} 
x=(\Om,\gd,\Om\trchi, \chih, \Om\trchib, \chibh, \eta, \om, D\om, \omb, \Du\omb, \a, \ab),
\end{aligned} \label{EQcollectionofTensors11122}
\end{align}
such that $x_{u_0,v} := x \vert_{S_{u_0,v}}$ is sphere data on each $S_{u_0,v} \subset \HH_{u_0,[v_1,v_2]}$.

\item For real numbers $u_1 < u_2 <v_0$, \emph{ingoing null data} on $\HHb_{[u_1,u_2],v_0}$ is given by a tuple of $S_{u,v_0}$-tangent tensors 
\begin{align} 
\begin{aligned} 
\underline{x}=(\Om,\gd,\Om\trchi, \chih, \Om\trchib, \chibh, \etab, \om, D\om, \omb, \Du\omb, \a, \ab),
\end{aligned} \label{EQcollectionofTensors111223}
\end{align}
such that $x_{u,v_0} := x \vert_{S_{u,v_0}}$ is sphere data on each $S_{u,v_0} \subset \HHb_{[u_1,u_2],v_0}$.
\end{enumerate}

\end{definition}

\ni The reference outgoing and ingoing null data of Minkowski are denoted by $\mathfrak{m}$ and $\underline{\mathfrak{m}}$, respectively; see also \eqref{EQrefMinkowskiSphereData}. %

In the following we define norms for null data. We will need two different types of norms. On the one hand, the norms $\XX(\HH), \XX(\HHb)$ and $\XX^+(\HHb)$ had already been defined in \cite{ACR1,ACR2} and appear in the precise statements the perturbative and bifurcate codimension-$10$ null gluing, see Theorems \ref{THMcharGluing10d} and \ref{THMbifurcateNULLcite} below. In addition, in this paper we construct high-frequency solutions to the null structure equations which are $\varep$-small in the high-frequency norm $\XX^{\mathrm{h.f.}}$.
\begin{definition}[Norms for null data] \label{DEFnormHH} Introduce the following.
\begin{enumerate}
\item Let $x$ be outgoing null data on $\HH:=\HH_{u_0,[v_1,v_2]}$. Define
\begin{align*} 
\begin{aligned}
\Vert x \Vert_{\XX(\HH)} :=& \Vert \Om \Vert_{H^6_3(\HH)} +\Vert \gd \Vert_{H^6_3(\HH)}+ \Vert \Om\trchi \Vert_{H^6_3(\HH)}+\Vert \chih \Vert_{H^6_2(\HH)}\\
& + \Vert \Om\trchib \Vert_{H^4_2(\HH)}+ \Vert \chibh \Vert_{H^4_3(\HH)}+ \Vert \eta \Vert_{H^5_2(\HH)}\\
&+\Vert \om \Vert_{H^6_2(\HH)} +\Vert D\om \Vert_{H^6_1(\HH)}+ \Vert \omb \Vert_{H^4_3(\HH)}+ \Vert \Du\omb \Vert_{H^2_3(\HH)}\\
&+  \Vert \a \Vert_{H^{6}_1(\HH)} +\Vert \ab \Vert_{H^{2}_3(\HH)}.
\end{aligned} 
\end{align*}
and moreover, define the \emph{high-frequency norm}
\begin{align*} 
\begin{aligned}
\Vert x \Vert_{\XX^{\mathrm{h.f.}}(\HH)} :=& \Vert \Om \Vert_{H^6_3(\HH)} +\Vert \om \Vert_{H^6_2(\HH)} +\Vert D\om \Vert_{H^6_1(\HH)} \\
&+\Vert \gd \Vert_{L^\infty_v H^6(S_v)}+ \Vert \Om\trchi \Vert_{L^\infty_v H^6(S_v)}+\Vert \chih \Vert_{L^\infty_v H^6(S_v)}\\
& + \Vert \Om\trchib \Vert_{L^\infty_v H^4(S_v)}+ \Vert \chibh \Vert_{L^\infty_v H^4(S_v)}+ \Vert \eta \Vert_{L^\infty_v H^5(S_v)}\\
&+ \Vert \omb \Vert_{L^\infty_v H^4(S_v)}+ \Vert \Du\omb \Vert_{L^\infty_v H^2(S_v)} 
+\Vert \ab \Vert_{L^\infty_v H^2(S_v)}.
\end{aligned} 
\end{align*}
\item Let $\underline{x}$ be ingoing null data on $\HHb:=\HHb_{[u_0,u_1],v_0}$. Define 
\begin{align*}
\begin{aligned}
\Vert \underline{x} \Vert_{\XX(\HHb)} :=& \Vert \Om \Vert_{H^6_3(\HHb)} +\Vert \gd \Vert_{H^6_3(\HHb)}+ \Vert \Om\trchib \Vert_{H^6_3(\HHb)}+\Vert \chibh \Vert_{H^6_2(\HHb)}\\
& + \Vert \Om\trchi \Vert_{H^4_2(\HHb)}+ \Vert \chih \Vert_{H^4_3(\HHb)}+ \Vert \etab \Vert_{H^5_2(\HHb)}\\
&+\Vert \omb \Vert_{H^6_2(\HHb)} +\Vert \Du\omb \Vert_{H^6_1(\HHb)}+ \Vert \om \Vert_{H^4_3(\HHb)}+ \Vert D\om \Vert_{H^2_3(\HHb)}\\
&+  \Vert \ab \Vert_{H^{6}_1(\HHb)} +\Vert \a \Vert_{H^{2}_3(\HHb)}.
\end{aligned}
\end{align*}
and define moreover the higher-regularity norm (used in the statement of Theorem \ref{THMcharGluing10d})
\begin{align*} 
\begin{aligned} 
\Vert \underline{x} \Vert_{\XX^+\lrpar{\HHb}} :=& \Vert \Om \Vert_{H^{12}_9(\HHb)} +\Vert \gd \Vert_{H^{12}_9(\HHb)}+ \Vert \Om\trchib \Vert_{H^{12}_9(\HHb)}+\Vert \chibh \Vert_{H^{12}_8(\HHb)}\\
& + \Vert \Om\trchi \Vert_{H^{10}_8(\HHb)}+ \Vert \chih \Vert_{H^{10}_9(\HHb)}+ \Vert \etab \Vert_{H^{11}_8(\HHb)}\\
&+\Vert \omb \Vert_{H^{12}_8(\HHb)} +\Vert \Du\omb \Vert_{H^{12}_7(\HHb)}+ \Vert \om \Vert_{H^{10}_9(\HHb)}+ \Vert D\om \Vert_{H^8_9(\HHb)}\\
&+  \Vert \ab \Vert_{H^{12}_7(\HHb)} +\Vert \a \Vert_{H^{8}_9(\HHb)}.
\end{aligned} 
\end{align*}
\end{enumerate}
\end{definition}
\ni We remark that by Sobolev embedding (see Lemma \ref{LEMstandardTRACE} above) it clearly holds for $x\in \XX(\HH)$ that 
\begin{align}
\begin{aligned}
\Vert x \Vert_{\XX^{\mathrm{h.f.}}(\HH)} \les \Vert x \Vert_{\XX(\HH)}.
\end{aligned}\label{EQembeddingsobolevnulldata}
\end{align}

\subsection{Codimension-$10$ null gluing of \cite{ACR1,ACR2}} \label{SECcodim10gluingREF} 

\ni The main result of \cite{ACR2} concerning perturbative codimension-$10$ null gluing can be stated as follows.
\begin{theorem}[Perturbative codimension-$10$ null gluing of \cite{ACR1,ACR2}, version 2] \label{THMcharGluing10d} Let $\de>0$ be a real number. Consider sphere data ${x}_{1}$ on ${S}_{1}$ contained in ingoing null data ${\underline{x}}$ on ${\HHb}_{[-\de,\de],1}$ (with ${S}_{0,1}={S}_1$) satisfying the null structure equations, and consider sphere data $x_{2}$ on $S_{2}$. Assume that for some real number $\varep>0$, 
\begin{align} 
\begin{aligned} 
\Vert {\underline{x}}-\mathfrak{m} \Vert_{\mathcal{X}^+(\tilde{\HHb}_{[-\de,\de],1}) } + \Vert x_{2}-\mathfrak{m} \Vert_{\mathcal{X}(S_{2})} \leq \varep.
\end{aligned} \label{EQsmallnessAssumptionMainTheoremMAIN222}
\end{align}
\begin{enumerate}
\item \textbf{Existence.} There is a real number $\varep_0>0$ such that for all $0<\varep<\varep_0$ sufficiently small, then there is a solution $x$ to the null structure equations on $\HH_{[1,2]}$ such that
\begin{align*}
\begin{aligned}
x\vert_{S_1} = x_1',
\end{aligned}
\end{align*}
where $x_{1}'$ is the sphere data on a perturbed sphere $S_{1}' \subset {\HHb}_{[-\de,\de],1}$ near ${S}_{1}$ in ${\HHb}_{[-\de,\de],1}$, and such that
\begin{align*}
\begin{aligned}
x \vert_{S_{2}} = x_{2} \text{ up to } (\Ef,\Pf,\Lf,\Gf), 
\end{aligned}
\end{align*}
that is, if \emph{a posteriori} it holds that $\lrpar{\mathbf{E},\mathbf{P}, \mathbf{L}, \mathbf{G}}(x \vert_{S_{2}}) =\lrpar{\mathbf{E},\mathbf{P}, \mathbf{L}, \mathbf{G}}\lrpar{x_{2}}$
then $x \vert_{S_{2}} = x_{2}$.
\item \textbf{Estimates for the solution.} The following bounds hold, 
\begin{align} 
\begin{aligned} 
\Vert x -\mathfrak{m} \Vert_{\XX(\HH_{[1,2]})}+ \Vert x_{1}' -x_{1} \Vert_{\XX(S_{1})}\les \Vert {\underline{x}}-\mathfrak{m} \Vert_{\mathcal{X}^+(\tilde{\HHb}_{[-\de,\de],1}) } + \Vert x_{2}-\mathfrak{m} \Vert_{\mathcal{X}(S_{2})},
\end{aligned} \label{EQUniversalBound}
\end{align}
Furthermore, we have the perturbation estimate
\begin{align} 
\begin{aligned} 
\left\vert \lrpar{\mathbf{E},\mathbf{P},\mathbf{L},\mathbf{G}}\lrpar{x_{1}'} - \lrpar{\mathbf{E},\mathbf{P},\mathbf{L},\mathbf{G}}\lrpar{x_1} \right\vert \les \varep \lrpar{\Vert {\underline{x}}-\mathfrak{m} \Vert_{\mathcal{X}^+(\tilde{\HHb}_{[-\de,\de],1}) } + \Vert x_{2}-\mathfrak{m} \Vert_{\mathcal{X}(S_{2})} }.
\end{aligned} \label{EQChargeEstimatesMainTheorem0}
\end{align}
and the almost conservation law
\begin{align} 
\begin{aligned}
\left\vert \lrpar{\mathbf{E},\mathbf{P},\mathbf{L},\mathbf{G}}\lrpar{x \vert_{S_{2}}} - \lrpar{\mathbf{E},\mathbf{P},\mathbf{L},\mathbf{G}}\lrpar{x \vert_{S_{1}}}\right\vert \les \varep \lrpar{\Vert {\underline{x}}-\mathfrak{m} \Vert_{\mathcal{X}^+(\tilde{\HHb}_{[-\de,\de],1}) } + \Vert x_{2}-\mathfrak{m} \Vert_{\mathcal{X}(S_{2})} }.
\end{aligned} \label{EQChargeEstimatesMainTheorem}
\end{align}
\item \textbf{Difference estimate.} Given a pair of sphere data $({x}_{1})_{(1)}$ and $({x}_{1})_{(2)}$ on ${S}_{1}$ contained in the pair of ingoing null data ${\underline{x}}_{(1)}$ and ${\underline{x}}_{(2)}$ on ${\HHb}_{[-\de,\de],1}$ (with ${S}_{0,1}={S}_1$) satisfying the null structure equations, and given a pair of sphere data $(x_{2})_{(1)}$ and $(x_{2})_{(2)}$ on $S_{2}$, such that \eqref{EQsmallnessAssumptionMainTheoremMAIN222} is satisfied for $0<\varep<\varep_0$, consider the respectively constructed solutions $x_{(1)}$ and $x_{(2)}$ on $\HH_{[1,2]}$. It holds that
\begin{align}
\begin{aligned}
&\Vert x_{(1)} -x_{(2)} \Vert_{\XX(\HH_{0,[1,2]})}+ \Vert (x_{1})'_{(1)} -(x_{1})'_{(2)} \Vert_{\XX(S_{1})}\\
\les& \Vert {\underline{x}}_{(1)}-{\underline{x}}_{(2)} \Vert_{\mathcal{X}^+(\tilde{\HHb}_{[-\de,\de],1}) } + \Vert (x_{2})_{(1)}-(x_2)_{(2)} \Vert_{\mathcal{X}(S_{2})},
\end{aligned}\label{EQDifferenceEstimateIFT}
\end{align}
and moreover,
\begin{align*}
\begin{aligned}
&\left\vert \lrpar{\mathbf{E},\mathbf{P},\mathbf{L},\mathbf{G}}((x_{1})'_{(1)}) - \lrpar{\mathbf{E},\mathbf{P},\mathbf{L},\mathbf{G}}((x_{1})'_{(2)}) \right\vert \\
\les& \varep\lrpar{\Vert {\underline{x}}_{(1)}-{\underline{x}}_{(2)} \Vert_{\mathcal{X}^+(\tilde{\HHb}_{[-\de,\de],1}) } + \Vert (x_{2})_{(1)}-(x_2)_{(2)} \Vert_{\mathcal{X}(S_{2})}}.
\end{aligned}
\end{align*}
\end{enumerate}
\end{theorem}

\ni \emph{Remarks on Theorem \ref{THMcharGluing10d}.}
\begin{enumerate}
\item In \cite{ACR1,ACR2} the codimension-$10$ null gluing result is stated with perturbations applied to the sphere $S_2$ which is assumed to lie in an ingoing null hypersurface ${\HHb}_{[-\de,\de],2}$. The version stated above in Theorem \ref{THMcharGluing10d} in which the perturbations are applied to the sphere $S_1$ is proved by straight-forward adaption of the proof in \cite{ACR1,ACR2}.
\item The difference estimates of Theorem \ref{THMcharGluing10d} follow directly from the implicit function theorem that is used in \cite{ACR1,ACR2} to prove the codimension-$10$ null gluing.
\end{enumerate}

\ni As mentioned in Section \ref{SECSUBnullgluingsetup}, one can also consider null gluing along a bifurcate null hypersurface emanating from a spacelike $2$-sphere. The corresponding codimension-$10$ result proved in \cite{ACR1,ACR2} is the following.
\begin{theorem}[Bifurcate codimension-$10$ null gluing of \cite{ACR1,ACR2}, version 2]\label{THMbifurcateNULLcite} Let $\varep>0$ be a real number. Consider sphere data $x_1$ and $x_2$ on spheres $S_1$ and $S_2$, respectively, such that
\begin{align*}
\begin{aligned}
\Vert x_1- \mathfrak{m} \Vert_{\XX(S_1)} + \Vert x_2- \mathfrak{m} \Vert_{\XX(S_2)} \leq \varep.
\end{aligned}
\end{align*}
There exists a universal $\varep_0>0$ such that for $0<\varep<\varep_0$, there exists a solution $x$ to the null structure equations along the bifurcate null hypersurface $\HH\cup\HHb$ such that
\begin{align*}
\begin{aligned}
x\vert_{S_1} = x_1 \text{ on } S_1, \text{ and } \,\, x \vert_{S_{2}} = x_{2} \text{ up to } (\Ef,\Pf,\Lf,\Gf) \text{ on } S_2,
\end{aligned}
\end{align*}
that is, if \emph{a posteriori} it holds that $\lrpar{\mathbf{E},\mathbf{P}, \mathbf{L}, \mathbf{G}}(x \vert_{S_{2}}) =\lrpar{\mathbf{E},\mathbf{P}, \mathbf{L}, \mathbf{G}}\lrpar{x_{2}}$ then $x \vert_{S_{2}} = x_{2}$. Moreover, the following estimate holds,
\begin{align*}
\begin{aligned}
\Vert x-\mathfrak{m} \Vert_{\XX(\HH)} + \Vert x-\mathfrak{m} \Vert_{\XX(\HHb)} \les \Vert x_1 -\mathfrak{m} \Vert_{\XX(S_1)} + \Vert x_2 - \mathfrak{m} \Vert_{\XX(S_2)}.
\end{aligned}
\end{align*}
Analogous difference estimates to \eqref{EQDifferenceEstimateIFT} hold. Moreover, it is possible to glue higher-order sphere data without any additional obstructions.
\end{theorem}

\subsection{Precise statement of main theorem} \label{SECmainTHM}

The following is the main theorem concerning \emph{obstruction-free null gluing} of this paper.
\begin{theorem}[Main theorem: Obstruction-free null gluing, version 2] \label{THMmain1} 
Let $\de>0$ be a real number. Consider sphere data ${x}_{1}$ on ${S}_{1}$ contained in ingoing null data ${\underline{x}}$ on ${\HHb}_{[-\de,\de],1}$ (with ${S}_{0,1}={S}_1$) satisfying the null structure equations, and consider sphere data $x_{2}$ on $S_{2}$. Assume that for some real number $\varep>0$, 
\begin{align} 
\begin{aligned} 
\Vert {\underline{x}}-\mathfrak{m} \Vert_{\mathcal{X}^+(\tilde{\HHb}_{[-\de,\de],1}) } + \Vert x_{2}-\mathfrak{m} \Vert_{\mathcal{X}(S_{2})} \leq \varep.
\end{aligned} \label{EQsmallnessMain1001}
\end{align}
Using Definition \ref{DEFchargesEPLG} to calculate the charges $(\mathbf{E},\mathbf{P},\mathbf{L},\mathbf{G})$ of the sphere data $x_1$ and $x_2$, consider
\begin{align*}
\begin{aligned}
\lrpar{\triangle \mathbf{E},\triangle \mathbf{P},\triangle \mathbf{L},\triangle \mathbf{G}} := \lrpar{\mathbf{E},\mathbf{P},\mathbf{L},\mathbf{G}}(x_2) - \lrpar{\mathbf{E},\mathbf{P},\mathbf{L},\mathbf{G}}(x_1),
\end{aligned}
\end{align*}
Assume that for three real numbers $C_1,C_2,C_3>0$,
\begin{subequations}
\begin{align}
\triangle \mathbf{E} =& C_1\, \varep, \label{EQsmallnessMain21} \\ 
\vert \triangle \mathbf{L} \vert \leq& C_2 \varep^2,  \label{EQsmallnessMain1002} \\
\varep^{-1} (\triangle \mathbf{E}) >& C_3 \lrpar{\varep^{-1} \vert \triangle \mathbf{P}\vert + \varep^{-1} \vert \triangle \mathbf{G}\vert}. \label{EQsmallnessMain2}
\end{align}
\end{subequations}
There are real numbers $\tilde{\varep}>0$ and $\tilde{C}_3>0$ such that if $0<\varep<\tilde{\varep}$ and $C_3>\tilde{C}_3$, then there is a solution $x$ to the null structure equations along the outgoing null hypersurface $\HH$ (leading from a perturbed sphere $S_1'\subset {\HHb}_{[-\de,\de],1}$ to $S_2$) such that 
\begin{align*}
\begin{aligned}
x\vert_{S'_1} = x'_1, \,\, x\vert_{S_2} = x_2.
\end{aligned}
\end{align*}
Moreover, it holds that $x \in \XX^{\mathrm{h.f.}}(\HH)$ with
\begin{align}
\begin{aligned}
\Vert x-\mathfrak{m} \Vert_{\XX^{\mathrm{h.f.}}(\HH)} \les \Vert {\underline{x}}-\mathfrak{m} \Vert_{\mathcal{X}^+(\tilde{\HHb}_{[-\de,\de],1}) } + \Vert x_{2}-\mathfrak{m} \Vert_{\mathcal{X}(S_{2})}.
\end{aligned}\label{EQsolBoundMainTHM}
\end{align}
\end{theorem}

\ni \emph{Remarks on Theorem \ref{THMmain1}.}
\begin{enumerate}
\item By rescaling, the smallness assumption \eqref{EQsmallnessMain1002} for a constant $C_2>0$ is consistent with the decay rates of \emph{asymptotic flatness} under the assumption that the angular momentum $\mathbf{L}_{\mathrm{ADM}}$ is well-defined; see also Definition \ref{DEFasymptoticFlatness}. We remark that the proof of Theorem \ref{THMmain1} goes through also for larger $\triangle \mathbf{L}$, say, for $\vert \triangle \mathbf{L} \vert \leq \varep^{-1/16} \varep^2$. This is generalization is, however, not pursued in this paper.

\item The bound \eqref{EQsmallnessMain1001} implies an upper bound on $\triangle \mathbf{E}$. Thus the assumption \eqref{EQsmallnessMain21} should be seen as a lower bound on $\triangle \mathbf{E}$. Similarly to \eqref{EQsmallnessMain1002}, the assumption \eqref{EQsmallnessMain21} can be significantly weakened by a straight-forward refinement our construction. However, this goes beyond the scope of this paper, and is thus postponed to a future work.
\item The constructed solution $x$ is large in the sense that $\vert\chih\vert\sim\varep^{1/2}$ and $\vert\a\vert\sim \varep^{-1/2}$ along $\HH$. Nevertheless, the smallness bound \eqref{EQsolBoundMainTHM} in the norm $\XX^{\mathrm{h.f.}}(\HH)$ is sufficient to apply the well-posedness results of Luk-Rodnianski \cite{LukChar,LukRod1} on the characteristic initial value problem for the Einstein equations.
\item Theorem \ref{THMmain1} admits a straight-forward analogue for \emph{bifurcate null gluing}, see Theorem \ref{THMbifurcateFREE} below.
\item The regularity of the constructed solution is limited only by the assumed regularity of the sphere data. In particular, as in our previous 
\cite{ACR1,ACR2}, the results in Theorem \ref{THMmain1} and the subsequent null and spacelike gluing statements below can be upgraded 
to the $C^k$ gluing for any $k\ge 2$. We note again however that even though the $C^k$ norms of the sphere data are prescribed to be small,
the resulting solution will not be small in the same norm. 
\end{enumerate}

\ni The proof of Theorem \ref{THMmain1}, given in Section \ref{SECgluingEP}, is based on the codimension-$10$ null gluing of \cite{ACR1,ACR2} together with the following novel result proved in Section \ref{SECproofW}.
\begin{theorem}[Null gluing of $\mathbf{E},\mathbf{P},\mathbf{L},\mathbf{G}$] \label{THMmain0} Let $\varep>0$ be a real number. Let $x$ be a given solution to the null structure equations along $\HH$ such that
\begin{align}
\begin{aligned}
\Vert x - \mfm \Vert_{\XX(\HH)} \leq \varep,
\end{aligned}\label{EQsmallnessCONDEPLGnullgluing1}
\end{align}
where $\mfm$ denotes the reference Minkowski data along $\HH$, and let $(\gd_x,\Om_x)$ denote its characteristic seed along $\HH$. Let moreover 
$$(\triangle\mathbf{E}_0,\triangle\mathbf{P}_0,\triangle\mathbf{L}_0,\triangle\mathbf{G}_0) \in \RRR^{10}$$ 
be a vector such that for real numbers $C_1,C_2,C_3>0$, 
\begin{subequations}
\begin{align}
\triangle \mathbf{E}_0 =& C_1\, \varep, \label{EQsmallnessMain021} \\ 
\vert \triangle \mathbf{L}_0 \vert \leq& C_2 \varep^2, \label{EQsmallnessMain021LSMALL} \\
\varep^{-1} (\triangle \mathbf{E}_0) >& C_3 \lrpar{\varep^{-1} \vert \triangle \mathbf{P}_0\vert + \varep^{-1} \vert \triangle \mathbf{G}_0\vert}. \label{EQsmallnessMain02}
\end{align}
\end{subequations}
There exist real numbers $\tilde{\varep}>0$ and $\tilde{C}_3>0$ such that if $0<\varep<\tilde{\varep}$ and $C_3>\tilde{C}_3$, then the following holds.
\begin{enumerate}
\item \textbf{Existence.} There is a smooth $2$-tensorfield $W$ along $\HH$, such that the characteristic seed along $\HH$
\begin{align}
\begin{aligned}
\tilde{\gd} := \gd_x + \varphi W, \,\, \Om := \Om_x,
\end{aligned}\label{EQansatzcharseedstatementmainthm}
\end{align}
where $\varphi$ is a standard cut-off function supported in $\HH$, see \eqref{EQcutoffFCTdef}, leads to a solution $x_{+W}$ of the null structure equations which satisfies
\begin{align}
\begin{aligned}
x_{+W}\vert_{S_1} = x\vert_{S_1}, \,\, (\triangle\mathbf{E},\triangle\mathbf{P},\triangle\mathbf{L},\triangle\mathbf{G})(x_{+W})=(\triangle\mathbf{E}_0,\triangle\mathbf{P}_0,\triangle\mathbf{L}_0,\triangle\mathbf{G}_0).
\end{aligned}\label{EQchargeADJUSTED}
\end{align}

\item \textbf{Estimates for the solution.} It holds that 
\begin{align}
\begin{aligned}
&\Vert x_{+W} -x \Vert_{\XX^{\mathrm{h.f.}}(\HH)} + \Vert x_{+W}- x\Vert_{\XX(S_2)} \\
\les& \, \varep \lrpar{\sqrt{\varep^{-1}\triangle\Ef_0} + \sqrt{\varep^{-1}\vert\triangle\Pf_0\vert} + \sqrt{\varep^{-1}\vert\triangle\Gf_0\vert}} \\
& + \varep^{5/4} \sqrt{\varep^{-2}\vert\triangle\Lf_0\vert} + \varep^{1/4} \Vert x - \mathfrak{m}\Vert_{\XX(\HH)}.
\end{aligned}\label{EQdiffEstimate1mainthm2statement}
\end{align}

\item \textbf{Difference estimates.} Given two solutions $x$ and $x'$ to the null structure equations on $\HH$ satisfying \eqref{EQsmallnessCONDEPLGnullgluing1} for sufficiently small $\varep>0$, and given two vectors 
\begin{align}
\begin{aligned}
(\triangle\mathbf{E}_0,\triangle\mathbf{P}_0,\triangle\mathbf{L}_0,\triangle\mathbf{G}_0), \,\, (\triangle\mathbf{E}_0',\triangle\mathbf{P}_0',\triangle\mathbf{L}_0',\triangle\mathbf{G}_0') \in \RRR^{10},
\end{aligned}\label{EQtwovectors1200099}
\end{align}
both satisfying \eqref{EQsmallnessMain021}-\eqref{EQsmallnessMain02}, consider the respectively constructed solutions $W$ and $W'$ to 
\begin{align*}
\begin{aligned}
(\triangle\mathbf{E},\triangle\mathbf{P},\triangle\mathbf{L},\triangle\mathbf{G})(x_{+W})=&(\triangle\mathbf{E}_0,\triangle\mathbf{P}_0,\triangle\mathbf{L}_0,\triangle\mathbf{G}_0), \\
(\triangle\mathbf{E},\triangle\mathbf{P},\triangle\mathbf{L},\triangle\mathbf{G})(x'_{+W'})=&(\triangle\mathbf{E}'_0,\triangle\mathbf{P}'_0,\triangle\mathbf{L}'_0,\triangle\mathbf{G}'_0).
\end{aligned}
\end{align*}
Then it holds that
\begin{align}
\begin{aligned}
&\Vert (x_{+W} - x) - (x'_{+W'} - x')\Vert_{\XX(S_2)}+\Vert (x_{+W} - x) - (x'_{+W'} - x')\Vert_{\XX^{\mathrm{h.f.}}(\HH)} \\
\les& \, \varep \lrpar{\varep^{-1}\vert \triangle\Ef_0-\triangle\Ef_0'\vert + \varep^{-1}\vert\triangle\Pf_0-\triangle\Pf_0'\vert + \varep^{-1}\vert\triangle\Gf_0-\triangle\Gf_0'\vert} \\
&+\varep^{5/4}\lrpar{ \varep^{-2}\vert\triangle\Lf_0-\triangle\Lf_0'\vert} + \varep^{1/4} \Vert x - x' \Vert_{\XX(\HH)},
\end{aligned}\label{EQdiffEstimate2mainthm2statement}
\end{align}
where the constant in \eqref{EQdiffEstimate2mainthm2statement} depends on $\mathrm{min}(C_1,C_1')$ calculated from the two vectors \eqref{EQtwovectors1200099}.
\end{enumerate}
\end{theorem}

\ni \emph{Remarks on Theorem \ref{THMmain0}.}
\begin{itemize}
\item From \eqref{EQsmallnessCONDEPLGnullgluing1} and \eqref{EQdiffEstimate1mainthm2statement} we deduce in particular that $x_{+W}$ is close to Minkowski in the high-frequency norm $\XX^{\mathrm{h.f.}}$,
\begin{align*}
\begin{aligned}
\Vert x_{+W} -\mathfrak{m} \Vert_{\XX^{\mathrm{h.f.}}(\HH)}
\les \varep.
\end{aligned}
\end{align*}
\item Recall from Theorem \ref{THMcharGluing10d} that in the codimension-$10$ null gluing of \cite{ACR1,ACR2}, the charges $\mathbf{E},\mathbf{P},\mathbf{L},\mathbf{G}$ are transported (almost conserved) up to the error of order $\OO(\varep^2)$ along $\HH$ (see \eqref{EQChargeEstimatesMainTheorem}) and are thus not glued. In Theorem \ref{THMmain0}, the charges $\mathbf{E},\mathbf{P}, \mathbf{L}$ are adjusted by size $\varep$ along $\HH$, and $\mathbf{L}$ by size $\varep^2$. This is achieved by choosing the $2$-tensorfield $W$ in \eqref{EQansatzcharseedstatementmainthm} to be \emph{large} and \emph{high-frequency}. The constructed solution $x_{+W}$ to the null structure equations has $\vert\chih\vert\sim\varep^{1/2}$ and $\vert\a\vert\sim \varep^{-1/2}$ along $\HH$. In particular, in this gluing procedure the size ($\varep^{1/2}$) of the characteristic seed $\Om\chih$ is much larger than the size ($\varep$) of the characteristic seeds used in the codimension-$10$ null gluing.

\end{itemize}

\ni The proof of Theorem \ref{THMmain0} is split into two parts.
\begin{enumerate}
\item In Section \ref{SECconstructionSolution} we make a large, high-frequency ansatz for $W$, and construct and control from the resulting characteristic seed $\tilde{\gd}:=\gd_x+\varphi W$ and $\Om:= \Om_x$ (see \eqref{EQansatzcharseedstatementmainthm}) the corresponding solution to the null structure equations along $\HH$. Moreover, we prove estimates for the constructed high-frequency solutions. 
\item In Section \ref{SECproofW} we show that the high-frequency ansatz for $W$ allows to adjust the charges $\mathbf{E},\mathbf{P},\mathbf{L},\mathbf{G}$ at $S_2$, and prove the estimates \eqref{EQdiffEstimate1mainthm2statement} and \eqref{EQdiffEstimate2mainthm2statement}.
\end{enumerate}

\subsection{Bifurcate obstruction-free null gluing} \label{SECbifurcateObstructionFree}
\ni The main result of this paper, Theorem \ref{THMmain1}, is phrased and proved within the setup of perturbative null gluing, that is, null gluing along one outgoing null hypersurface. In complete analogy, it is possible to combine the codimension-$10$ \emph{bifurcate} null gluing of \cite{ACR1,ACR2} along a bifurcate null hypersurface $\HH\cup\HHb$ emanating from a spacelike $2$-sphere (see Theorem \ref{THMbifurcateNULLcite} above) with the novel high-frequency approach of this paper applied along $\HH$. The result is the following; see also Figure \ref{FIGbifurcateLATER} below.

\begin{theorem}[Bifurcate obstruction-free null gluing]  \label{THMbifurcateFREE} Consider sphere data $x_1$ and $x_2$ on sphere $S_1$ and $S_2$ such that for a real number $\varep>0$,
\begin{align*}
\begin{aligned}
\Vert x_1 - \mathfrak{m} \Vert_{\XX(S_1)} + \Vert x_2 - \mathfrak{m} \Vert_{\XX(S_2)} \leq \varep.
\end{aligned}
\end{align*}
Consider the vector
\begin{align*}
\begin{aligned}
\lrpar{\triangle \mathbf{E},\triangle \mathbf{P},\triangle \mathbf{L},\triangle \mathbf{G}} := \lrpar{\mathbf{E},\mathbf{P},\mathbf{L},\mathbf{G}}(x_2) - \lrpar{\mathbf{E},\mathbf{P},\mathbf{L},\mathbf{G}}(x_1),
\end{aligned}
\end{align*}
and assume that for three real numbers $C_1,C_2,C_3>0$,
\begin{subequations}
\begin{align}
\triangle \mathbf{E} =& C_1\, \varep, \label{EQsmallnessMain21BIF} \\ 
\vert \triangle \mathbf{L} \vert \leq& C_2 \varep^2,  \label{EQsmallnessMain1002BIF} \\
\varep^{-1} (\triangle \mathbf{E}) >& C_3 \lrpar{\varep^{-1} \vert \triangle \mathbf{P}\vert + \varep^{-1} \vert \triangle \mathbf{G}+\Pf(x_1)\vert}. \label{EQsmallnessMain2BIF}
\end{align}
\end{subequations}
There are real numbers $\tilde{\varep}>0$ and $\tilde{C}_3>0$ such that if $0<\varep<\tilde{\varep}$ and $C_3>\tilde{C}_3$, then there exists a solution $x$ to the null structure equations along the null hypersurface $\HH\cup\HHb$ such that 
\begin{align*}
\begin{aligned}
x\vert_{S_1} = x_1, \,\, x\vert_{S_2} = x_2.
\end{aligned}
\end{align*}
Moreover, it holds that $x\vert_{\HH} \in \XX^{\mathrm{h.f.}}(\HH)$ and $x\vert_{\HHb} \in \XX(\HHb)$ with
\begin{align*}
\begin{aligned}
\Vert x-\mathfrak{m} \Vert_{\XX^{\mathrm{h.f.}}(\HH)}+ \Vert x-\mathfrak{m} \Vert_{\XX(\HHb)} \les \Vert x_1-\mathfrak{m} \Vert_{\mathcal{X}(S_1) } + \Vert x_{2}-\mathfrak{m} \Vert_{\mathcal{X}(S_{2})}.
\end{aligned}
\end{align*}

\end{theorem}

\begin{figure}[H]
	\begin{center}
		\includegraphics[width=9.5cm]{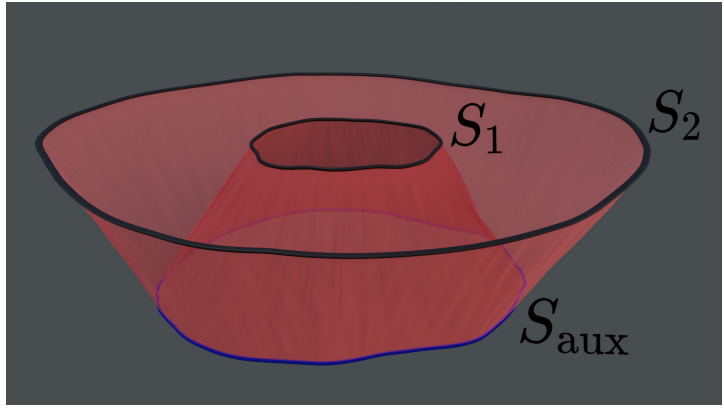} 
		\vspace{0.4cm} 
		\caption{Illustration of obstruction-free null gluing along a bifurcate null hypersurface (red) emanating from an auxiliary sphere $S_{\mathrm{aux}}$ (blue).} \label{FIGbifurcateLATER}
	\end{center}
\end{figure}

\ni \emph{Remarks on Theorem \ref{THMbifurcateFREE}.}
\begin{enumerate}
\item The proof of Theorem \ref{THMbifurcateFREE} is based, analogous to Theorem \ref{THMmain1}, on an iteration scheme between the bifurcate codimension-$10$ null gluing proved in \cite{ACR1,ACR2}, see Theorem \ref{THMbifurcateNULLcite}, and the high-frequency null gluing of $\Ef,\Pf,\Lf,\Gf$ of Theorem \ref{THMmain0}.
\item The difference in the conditions \eqref{EQsmallnessMain2} and \eqref{EQsmallnessMain2BIF} is explained as follows. The charge $\Gf$ on $S_{\mathrm{aux}}$ is determined from the charges $\mathbf{G}$ and $\mathbf{P}$ on $S_1$ as follows,
\begin{align}
\begin{aligned}
\mathbf{G}(x_{\mathrm{aux}}) = \mathbf{G}(x_1) - \mathbf{P}(x_1) + \OO(\varep^2).
\end{aligned}\label{EQchargeGchange}
\end{align}
This stems from the fact that $\mathbf{G}$ is \emph{not} linearly conserved \emph{along} $\HHb$ (but only along $\HH$). Namely, it is shown in \cite{ACR1} (see Lemma 6.3) that at the linear level, $\Gf +2 r \Pf$ is conserved along $\HHb$. Note however that the charges $\Ef,\Pf,\Lf$ \emph{are} linearly conserved along both $\HH$ and $\HHb$. Together with the setup that at the linear level, $r\vert_{S_1} = 1$ and $r\vert_{S_{\mathrm{aux}}}=3/2$, the above leads to \eqref{EQchargeGchange}. Consequently, at $S_{\mathrm{aux}}$ we thus calculate that the difference
\begin{align*}
\begin{aligned}
\Gf(x_2)-\Gf(x_{\mathrm{aux}}) = \Gf(x_2)-\Gf(x_1) + \mathbf{P}(x_1) + \OO(\varep^2),
\end{aligned}
\end{align*}
where the first three terms are precisely the quantity which appears in \eqref{EQsmallnessMain2BIF}.

\item Similar to the \emph{higher-order} bifurcate null gluing proved in \cite{ACR1,ACR2}, Theorem \ref{THMbifurcateFREE} generalizes to a \emph{higher-order} bifurcate null gluing result. We omit details and refer to Remark (3) after the bound \eqref{EQestimateDIFFx0x1} for the higher-order control of high-frequency solutions of the null structure equations.
\end{enumerate}

\subsection{Application to spacelike gluing} Before stating the main result of this section, see Corollary \ref{CORspacelike}, we shortly recall the setup and notions of spacelike initial data; for more details we refer to Section \ref{SECproofSPACELIKE} and also \cite{ACR3}.

Spacelike initial data for the Einstein equations consists of a triple $(\Sigma,g,k)$ where $\Sigma$ is a $3$-manifold, $g$ a Riemannian metric on $\Si$, and $k$ a symmetric $2$-tensor on $\Si$ satisfying the \emph{spacelike constraint equations}
\begin{align*}
\begin{aligned}
\mathrm{R}_{\mathrm{scal}}(g) = \vert k \vert_g^2 - (\tr_gk)^2, \,\, \mathrm{div}_g k = \mathrm{d}(\tr_g k),
\end{aligned}
\end{align*}
where $\mathrm{R}_{\mathrm{scal}}(g)$ denotes the scalar curvature of $g$, $\mathrm{d}$ the exterior derivative on $\Si$, and
\begin{align*}
\begin{aligned}
\vert k \vert_g^2 := g^{ab}g^{cd} k_{ac}k_{bd}, \,\, \tr_g k := g^{ij}k_{ij}, \,\, (\mathrm{div}_g k)_c := g^{ab}\nab_a k_{bc},
\end{aligned}
\end{align*}
where $\nab$ denotes the covariant derivative with respect to $g$.

The trivial spacelike initial data is the reference Minkowski $(\Si,g,k)=(\RRR^3,e,0)$ where $e$ denotes the Euclidean metric. This is the spacelike initial data induced on the hypersurface $\{t=0\}$ in standard coordinates $(t,x^1,x^2,x^3)$ on Minkowski spacetime $(\RRR^4,\mathbf{m}:= \mathrm{diag}(-1,1,1,1))$.  

Reference spacelike initial data for the $1$-dimensional Schwarzschild family parametrized by the mass $M\geq0$ is defined on $\RRR^3 \setminus \overline{B(0,2M)}$ in spherical coordinates $(r,\th^1,\th^2)$ by
\begin{align}
\begin{aligned}
g^M = \lrpar{1-\frac{2M}{r}}^{-1} dr^2 + r^2 ((d\th^1)^2 + (\sin\th^1)^2 (d\th^2)^2), \,\, k^M =0.
\end{aligned}\label{EQdefSSmassMdata}
\end{align}
This is the spacelike initial data induced on the hypersurface $\{t=0\}$ in Schwarzschild coordinates $(t,r,\th^1,\th^2)$ with respect to which the Schwarzschild spacetime metric of mass $M\geq0$ takes the form
\begin{align*}
\begin{aligned}
\mathbf{g}^M = -  \lrpar{1-\frac{2M}{r}} dt^2+ \lrpar{1-\frac{2M}{r}}^{-1} dr^2 + r^2 ((d\th^1)^2 + (\sin\th^1)^2 (d\th^2)^2).
\end{aligned}
\end{align*}

We also consider reference spacelike initial data for the exterior region of the \emph{$10$-dimensional} family of Kerr spacetimes. These spacelike initial data sets are induced on the spacelike hypersurfaces $\{t'=0\}$ in coordinates $(t',x'^1,x'^2,x'^3)$ which are related to the standard Boyer-Lindquist coordinates by Poincar\'e transformations, and can be parametrized through their $10$ ADM parameters $\Ef_{\mathrm{ADM}},\Pf_{\mathrm{ADM}},\Lf_{\mathrm{ADM}},\mathbf{C}_{\mathrm{ADM}}$. For the explicit construction and details we refer to \cite{ACR3}.

In this section we consider \emph{asymptotically flat} spacelike initial data defined as follows.
\begin{definition}[Asymptotic flatness] \label{DEFasymptoticFlatness} Spacelike initial data $(\Si,g,k)$ is \emph{asymptotically flat} if there is a compact set $K \subset \Sigma$ such that $\Sigma \setminus K$ is diffeomorphic to the complement of the closed unit ball in $\RRR^3$ and $(g,k)$ have the following asymptotics in this chart as $\vert x \vert\to\infty$, 
\begin{align}
\begin{aligned}
g_{ij}(x)-e_{ij} = \OO\lrpar{\vert x \vert^{-1}}, \,\, k_{ij}(x) = \OO\lrpar{\vert x \vert^{-2}}.
\end{aligned}\label{EQasymptoticflatness}
\end{align}
\end{definition}

\ni The above mentioned reference spacelike initial data for Minkowski, Schwarzschild and Kerr are asymptotically flat in the sense of Definition \ref{DEFasymptoticFlatness}.

From \eqref{EQasymptoticflatness} it follows that the ADM energy $\mathbf{E}_{\mathrm{ADM}}$ and linear momentum $\mathbf{P}_{\mathrm{ADM}}$, defined for $i=1,2,3$ by
\begin{align}
\begin{aligned}
\mathbf{E}_{\mathrm{ADM}} :=&  \lim\limits_{r \to\infty} \frac{1}{16\pi}\int\limits_{S_r} \sum\limits_{j=1,2,3} \lrpar{\pr_j g_{jl}-\pr_l g_{jj}} N^l d\mu_\gd, \\
\lrpar{\mathbf{P}_{\mathrm{ADM}}}^i :=& \lim\limits_{r \to\infty} \frac{1}{8\pi} \int\limits_{S_{r}}  \lrpar{k_{il}-\tr k \, g_{il}} N^l d\mu_\gd,
\end{aligned}\label{EQDEFADMEP}
\end{align}
are well-defined and finite. Here $N$ denotes the outward unit normal to $S_r$ and $d\mu_\gd$ the induced volume element on $S_r$.

In the following we assume to work with asymptotically flat spacelike initial data which has moreover well-defined and finite angular momentum $\mathbf{L}_{\mathrm{ADM}}$ and center-of-mass $\mathbf{C}_{\mathrm{ADM}}$, defined for $i=1,2,3$ by
\begin{align}
\begin{aligned}
\lrpar{\mathbf{L}_{\mathrm{ADM}}}^i :=&  \lim\limits_{r\to\infty} \frac{1}{8\pi}\int\limits_{S_{r}} (k_{jl}-\tr k \, g_{jl}) \lrpar{Y_{(i)}}^j N^l d\mu_\gd, \\
\lrpar{\mathbf{C}_{\mathrm{ADM}}}^i :=& \lim\limits_{r\to\infty} \int\limits_{S_r} \lrpar{ x^i \sum\limits_{j=1,2,3} \lrpar{\pr_j g_{jl}- \pr_l g_{jj}}N^l - \sum\limits_{j=1,2,3} \lrpar{g_{ji} N^j - g_{jj} N^i} }d\mu_\gd,
\end{aligned}\label{EQDEFADMLC}
\end{align}
where $Y_{(i)}$, $i=1,2,3$, are the rotation fields defined by $(Y_{(i)})_j := \in_{ilj} x^l$, where $x^l$ are the Cartesian coordinate functions on $\RRR^3$ and $\in_{ilj}$ is the volume form.

We are now in position to state the spacelike gluing corollary of our main result. Its proof is given in Section \ref{SECproofSPACELIKE}, based on the bifurcate obstruction-free null gluing of Theorem \ref{THMbifurcateFREE}.
\begin{corollary}[Obstruction-free spacelike gluing, version 2] \label{CORspacelike} Let $(\Si,g,k)$ be asymptotically flat spacelike initial data with well-defined, finite ADM parameters
\begin{align*}
\begin{aligned}
(\mathbf{E}^{(\Si,g,k)}_{\mathrm{ADM}},\mathbf{P}^{(\Si,g,k)}_{\mathrm{ADM}},\mathbf{L}^{(\Si,g,k)}_{\mathrm{ADM}},\mathbf{C}^{(\Si,g,k)}_{\mathrm{ADM}}),
\end{aligned}
\end{align*}
and consider Kerr initial data $(g^{\mathrm{Kerr}}, k^{\mathrm{Kerr}})$ of ADM parameters $(\mathbf{E}^{\mathrm{Kerr}}_{\mathrm{ADM}},\mathbf{P}^{\mathrm{Kerr}}_{\mathrm{ADM}},\mathbf{L}^{\mathrm{Kerr}}_{\mathrm{ADM}},\mathbf{C}^{\mathrm{Kerr}}_{\mathrm{ADM}})$. Assume that the ADM parameter differences
\begin{align*}
\begin{aligned}
\triangle \Ef_{\mathrm{ADM}} :=  \Ef_{\mathrm{ADM}}^{\mathrm{Kerr}} -\Ef^{(\Si,g,k)}_{\mathrm{ADM}}, \,\, \triangle \Pf_{\mathrm{ADM}} :=  \Pf_{\mathrm{ADM}}^{\mathrm{Kerr}} -\Pf^{(\Si,g,k)}_{\mathrm{ADM}}, \,\, \triangle \Lf_{\mathrm{ADM}} :=  \Lf_{\mathrm{ADM}}^{\mathrm{Kerr}} -\Lf^{(\Si,g,k)}_{\mathrm{ADM}},
\end{aligned}
\end{align*}
satisfy, for real numbers $C_1,C_2,C_3>0$, 
\begin{align*}
\begin{aligned}
\triangle \Ef_{\mathrm{ADM}} = C_1 >0, \,\,
\vert \triangle \Lf_{\mathrm{ADM}}\vert \leq C_2,\,\,
\triangle \Ef_{\mathrm{ADM}} > C_3 \cdot \vert \triangle \Pf_{\mathrm{ADM}}\vert.
\end{aligned}
\end{align*}
There exist real numbers $\tilde{R}>0$ and $\tilde{C}_3>0$ such that if $C_3>\tilde{C}_3$, then for all real numbers $R>\tilde{R}$, $(\Si,g,k)$ can be glued, far out in the asymptotic region, across the annulus $A_{[R,2R]}$ bounded by the coordinate spheres $S_R$ and $S_{2R}$, to $(g^{\mathrm{Kerr}}, k^{\mathrm{Kerr}})$.

In particular, there is a real number $M_0>0$ such that for any $M>M_0$, $(\Si,g,k)$ can be glued, far out in the asymptotically flat region, to the reference Schwarzschild spacelike initial data $(g^M,k^M)$ defined in \eqref{EQdefSSmassMdata}.
\end{corollary}

\ni \emph{Remarks on Corollary \ref{CORspacelike}.}
\begin{enumerate}
\item In Corollary \ref{CORspacelike} we work with initial data that is asymptotically flat in the sense of Definition \ref{DEFasymptoticFlatness}, and assume that the angular momentum $\mathbf{L}_{\mathrm{ADM}}$ and center-of-mass $\mathbf{C}_{\mathrm{ADM}}$ are well-defined. In contrast, the spacelike gluing construction of Corvino-Schoen \cite{Corvino,CorvinoSchoen} requires that the spacelike initial data is asymptotically flat and satisfies in addition the so-called \emph{Regge-Teitelbaum conditions} (parity conditions on $g$ and $k$ near spacelike infinity). The Regge-Teitelbaum conditions imply in particular that $\mathbf{L}_{\mathrm{ADM}}$ and $\mathbf{C}_{\mathrm{ADM}}$ are well-defined and finite, but are further used in \cite{Corvino,CorvinoSchoen} to reduce (by a density argument) the problem to the case of spacelike initial data which is conformally flat near spacelike infinity.
\end{enumerate}

\section{Proof of the main theorem}\label{SECgluingEP} \ni In this section we prove Theorem \ref{THMmain1}. The idea is to set up an iteration scheme which combines the codimension-$10$ null gluing of Theorem \ref{THMcharGluing10d} with the high-frequency null gluing of $(\Ef,\Pf,\Lf,\Gf)$ in Theorem \ref{THMmain0}. The iteration scheme is defined as follows.\\

\ni \underline{\textbf{Definition of Step $0$.}}

\begin{enumerate}
\item[(0.a)] Apply the codimension-$10$ null gluing of Theorem \ref{THMcharGluing10d} to construct a solution $x_{(0)}$ to the null structure equations along $\HH$ such that
\begin{align}
\begin{aligned}
x_{(0)}\vert_{S_2} = x_2 \text{ up to } (\Ef,\Pf,\Lf,\Gf), \,\, \text{ and } \,\, x_{(0)} \vert_{S_1} = (x_1)'_{(0)},
\end{aligned}\label{EQcodim10conditionstep0}
\end{align}
where $(x_1)'_{(0)}$ is induced sphere data on the perturbed sphere $(S_1)'_{(0)}\subset \HHb_1$ near $S_1$. 

\item[(0.b)] Apply the null gluing of $(\Ef,\Pf,\Lf,\Gf)$ of Theorem \ref{THMmain0} to $x_{(0)}$ to construct $W_{(0)}$ such that
\begin{align}
\begin{aligned}
(\mathbf{E}, \mathbf{P}, \mathbf{L}, \mathbf{G})\lrpar{{x_{(0)}}_{+W_{(0)}} \Big\vert_{S_2}} = (\mathbf{E}, \mathbf{P},\mathbf{L}, \mathbf{G})(x_2).
\end{aligned}\label{EQiterationEPLGconditionSTEP0}
\end{align}

\end{enumerate}

\ni \underline{\textbf{Definition of Step $i$, for $i\geq1$.}}
\begin{enumerate}
\item[(i.a)] Apply the codimension-$10$ null gluing of Theorem \ref{THMcharGluing10d} to construct a solution $x_{(i)}$ to the null structure equations along $\HH$ such that
\begin{align}
\begin{aligned}
x_{(i)}\vert_{S_2} = x_2 - \lrpar{{x_{(i-1)}}_{+W_{(i-1)}}-x_{(i-1)}}\Big\vert_{S_2} \text{ up to } (\Ef,\Pf,\Lf,\Gf), \,\, \text{ and } \,\, x_{(i)} \vert_{S_1} = (x_1)'_{(i)},
\end{aligned}\label{EQbdrycondith}
\end{align}
where $(x_1)'_{(i)}$ is induced sphere data on the perturbed sphere $(S_1)'_{(i)}\subset \HHb_1$ near $S_1$.

\item[(i.b)] Apply the null gluing of $(\Ef,\Pf,\Lf,\Gf)$ of Theorem \ref{THMmain0} to construct $W_{(i)}$ such that
\begin{align}
\begin{aligned}
(\mathbf{E}, \mathbf{P}, \mathbf{L}, \mathbf{G})\lrpar{{x_{(i)}}_{+W_{(i)}}\Big\vert_{S_2}} = (\mathbf{E}, \mathbf{P},\mathbf{L}, \mathbf{G})(x_2).
\end{aligned}\label{EQiterationEPLGconditionSTEPi}
\end{align}

\end{enumerate}

\ni We claim that for $\varep>0$ sufficiently small and $C_3>0$ sufficiently large, the iteration scheme is well-defined, uniformly bounded,
\begin{align}
\begin{aligned}
\Vert x_{(i)} -\mathfrak{m} \Vert_{\XX(\HH)}+\left\Vert {x_{(i)}}_{+W_{(i)}} -\mathfrak{m} \right\Vert_{\XX^{\mathrm{h.f.}}(\HH)} \les \varep \,\,\, \text{ for all } i\geq0,
\end{aligned}\label{EQuniformbound}
\end{align}
and that the sequence $(x_{(i)})_{i\geq0}$ converges. Indeed, these claims follow in a straight-forward, standard fashion once we show that for $\varep>0$ sufficiently small, the iteration scheme is a contraction in the sense that
\begin{align}
\begin{aligned}
&\left\Vert \lrpar{{x_{(i)}}_{+W_{(i)}}- {x_{(i)}}} - \lrpar{{x_{(i-1)}}_{+W_{(i-1)}}- {x_{(i-1)}}} \right\Vert_{\XX(S_2)} \\
\leq& \half \left\Vert \lrpar{{x_{(i-1)}}_{+W_{(i-1)}}- {x_{(i-1)}}} - \lrpar{{x_{(i-2)}}_{+W_{(i-2)}}- {x_{(i-2)}}} \right\Vert_{\XX(S_2)}.
\end{aligned}\label{EQiterationcontract}
\end{align}

\ni In the following we only discuss estimates for Step $0$ (\emph{i.e.} the base case) and the difference estimates for Step $i$ (\emph{i.e.} the induction step proving \eqref{EQiterationcontract}). We omit details and focus on the essential contraction property \eqref{EQiterationcontract}. The conclusion of the proof of Theorem \ref{THMmain1} is at the end of this section.\\

\ni \textbf{Estimates for step $0$.} By assumption \eqref{EQsmallnessMain1001} it holds that
\begin{align*}
\begin{aligned}
\Vert {\underline{x}}-\mathfrak{m} \Vert_{\mathcal{X}^+(\tilde{\HHb}_{[-\de,\de],1}) } + \Vert x_{2}-\mathfrak{m} \Vert_{\mathcal{X}(S_{2})} \leq \varep,
\end{aligned}
\end{align*}
so that we can indeed apply Theorem \ref{THMcharGluing10d} to construct $x_{(0)}$ solving \eqref{EQcodim10conditionstep0} and satisfying the bounds
\begin{align}
\begin{aligned}
\Vert x_{(0)} - \mathfrak{m} \Vert_{\XX(\HH)} \les \Vert {\underline{x}}-\mathfrak{m} \Vert_{\mathcal{X}^+(\tilde{\HHb}_{[-\de,\de],1}) } + \Vert x_2 - \mathfrak{m} \Vert_{\XX(S_2)} \les \varep,
\end{aligned}\label{EQx0estimate}
\end{align}
and
\begin{align}
\begin{aligned}
\vert (\Ef,\Pf,\Lf,\Gf)((x_1)'_{(0)}) - (\Ef,\Pf,\Lf,\Gf)(x_1) \vert \les& \varep\lrpar{ \Vert {\underline{x}}-\mathfrak{m} \Vert_{\mathcal{X}^+(\tilde{\HHb}_{[-\de,\de],1}) } + \Vert x_{2}-\mathfrak{m} \Vert_{\mathcal{X}(S_{2})}} \les \varep^2.
\end{aligned}\label{EQchargeDiffestimstep0}
\end{align}
We observe that by the assumptions \eqref{EQsmallnessMain021}, \eqref{EQsmallnessMain021LSMALL}, \eqref{EQsmallnessMain02}, and \eqref{EQcodim10conditionstep0}, \eqref{EQiterationEPLGconditionSTEP0} and \eqref{EQchargeDiffestimstep0}, the vector
\begin{align*}
\begin{aligned}
\lrpar{\triangle \Ef_{(0)},\triangle \Pf_{(0)},\triangle \Lf_{(0)},\triangle \Gf_{(0)}}:=& (\Ef,\Pf,\Lf,\Gf)\lrpar{{x_{(0)}}_{+W_{(0)}}\big\vert_{S_2}} - (\Ef,\Pf,\Lf,\Gf)\lrpar{{x_{(0)}}_{+W_{(0)}}\big\vert_{S_1}}\\
=& (\Ef,\Pf,\Lf,\Gf)(x_2)-(\Ef,\Pf,\Lf,\Gf)((x_1)'_{(0)}) \\
=& \lrpar{\triangle \mathbf{E},\triangle \mathbf{P},\triangle \mathbf{L},\triangle \mathbf{G}} + \OO(\varep^2),
\end{aligned}
\end{align*}
satisfies for $\varep>0$ sufficiently small the conditions \eqref{EQsmallnessMain021}, \eqref{EQsmallnessMain021LSMALL}, \eqref{EQsmallnessMain02} for the application of Theorem \ref{THMmain0} with slightly differing constants $C_1' > C_1/2$, $C_2' < 2C_2 $ and $C_3> C_3' > C_3/2$, that is,
\begin{align}
\begin{aligned}
\triangle \mathbf{E}_{(0)} =& C_1' \, \varep,  \\ 
\vert \triangle \mathbf{L}_{(0)} \vert \leq& C_2' \varep^2, \\
\varep^{-1} (\triangle \mathbf{E}_0) >& C_3' \lrpar{\varep^{-1} \vert \triangle \mathbf{P}_0\vert + \varep^{-1} \vert \triangle \mathbf{G}_0\vert}.
\end{aligned}\label{EQsmallnessMain1002VARIED}
\end{align}
Thus for $C_3>0$ sufficiently large and $\varep>0$ sufficiently small, we can indeed apply Theorem \ref{THMmain0} to $x_{(0)}$ with \eqref{EQiterationEPLGconditionSTEP0} and the constructed ${x_{(0)}}_{+W_{(0)}}$ is bounded by 
\begin{align*}
\begin{aligned}
&\Vert {x_{(0)}}_{+W_{(0)}} -x_{(0)} \Vert_{\XX^{\mathrm{h.f.}}(\HH)} + \Vert {x_{(0)}}_{+W_{(0)}} -x_{(0)}\Vert_{\XX(S_2)} \\
\les& \, \varep \lrpar{\sqrt{\varep^{-1}\triangle\Ef_{(0)}} + \sqrt{\varep^{-1}\vert\triangle\Pf_{(0)}\vert} + \sqrt{\varep^{-1}\vert\triangle\Gf_{(0)}\vert}} + \varep^{5/4} \sqrt{\varep^{-2}\vert\triangle\Lf_{(0)}\vert} + \varep^{1/4} \Vert x_{(0)} - \mathfrak{m}\Vert_{\XX(\HH)} \\
\les& \varep.
\end{aligned}
\end{align*}

\ni \textbf{Estimates for the $i^{\mathrm{th}}$ step, for $i\geq2$.} On the one hand, from the codimension-$10$ gluing of Theorem \ref{THMcharGluing10d} we get the estimate
\begin{align}
\begin{aligned}
&\Vert x_{(i)}-x_{(i-1)} \Vert_{\XX(\HH)} \\
\les& \Vert \underbrace{\underline{x}-\underline{x}}_{=0}\Vert_{\XX^+(\HHb)} \\
&+  
\left\Vert \lrpar{x_2 - \lrpar{{x_{(i-1)}}_{+W_{(i-1)}}- {x_{(i-1)}}}} - \lrpar{x_2 - \lrpar{{x_{(i-2)}}_{+W_{(i-2)}}- {x_{(i-2)}}} } \right\Vert_{\XX(S_2)} \\
\les& \left\Vert \lrpar{{x_{(i-1)}}_{+W_{(i-1)}}- {x_{(i-1)}}} - \lrpar{{x_{(i-2)}}_{+W_{(i-2)}}- {x_{(i-2)}}} \right\Vert_{\XX(S_2)},
\end{aligned}\label{EQiterationXestimate1}
\end{align}
as well as the charge estimate
\begin{align}
\begin{aligned}
&\vert (\Ef,\Pf,\Lf,\Gf)((x_1)'_{(i)}) - (\Ef,\Pf,\Lf,\Gf)((x_1)'_{(i-1)}) \vert \\
\les& \varep \left\Vert \lrpar{{x_{(i-1)}}_{+W_{(i-1)}}- {x_{(i-1)}}} - \lrpar{{x_{(i-2)}}_{+W_{(i-2)}}- {x_{(i-2)}}} \right\Vert_{\XX(S_2)}.
\end{aligned}\label{EQiterationXestimate1CHARGE14}
\end{align}

\ni On the other hand, from the $\Ef,\Pf,\Lf,\Gf$ null gluing of Theorem \ref{THMmain0} with the charge condition \eqref{EQiterationEPLGconditionSTEPi} on $S_2$, we get the estimate
\begin{align}
\begin{aligned}
&\left\Vert \lrpar{{x_{(i)}}_{+W_{(i)}}- {x_{(i)}}} - \lrpar{{x_{(i-1)}}_{+W_{(i-1)}}- {x_{(i-1)}}} \right\Vert_{\XX(S_2)} \\
\les&_{C_1^{-1}} \varep \lrpar{\varep^{-1}\vert \triangle\Ef_{(i)}-\triangle\Ef_{(i-1)}\vert + \varep^{-1}\vert\triangle\Pf_{(i)}-\triangle\Pf_{(i-1)}\vert + \varep^{-1}\vert\triangle\Gf_{(i)}-\triangle\Gf_{(i-1)}\vert} \\
&+\varep^{5/4}\lrpar{ \varep^{-2}\vert\triangle\Lf_{(i)}-\triangle\Lf_{(i-1)}\vert} + \varep^{1/4} \Vert x_{(i)} - x_{(i-1)}\Vert_{\XX(\HH)}, 
\end{aligned}\label{EQstepipart2estim1}
\end{align} 
where, for $i\geq0$, using the condition \eqref{EQiterationEPLGconditionSTEPi} on $S_2$ and the condition \eqref{EQbdrycondith} on $S_1$, and employing that $W_{(i)}$ vanishes near $S_1$ for each $i\geq0$,
\begin{align}
\begin{aligned}
(\triangle\Ef_{(i)},\triangle\Pf_{(i)},\triangle\Lf_{(i)},\triangle\Gf_{(i)}) :=& (\Ef,\Pf,\Lf,\Gf)\lrpar{{x_{(i)}}_{+W_{(i)}}\big\vert_{S_2}} - (\Ef,\Pf,\Lf,\Gf)\lrpar{{x_{(i)}}_{+W_{(i)}}\big\vert_{S_1}}\\
=& (\Ef,\Pf,\Lf,\Gf)(x_2) - (\Ef,\Pf,\Lf,\Gf)\lrpar{x_{(i)}\big\vert_{S_1}} \\
=& (\Ef,\Pf,\Lf,\Gf)(x_2) - (\Ef,\Pf,\Lf,\Gf)\lrpar{(x_1)'_{(i)}}.
\end{aligned}\label{EQexpansionProblemsetupproofmainthm99099099}
\end{align}
By \eqref{EQiterationXestimate1CHARGE14} we can estimate the difference
\begin{align}
\begin{aligned}
&\left\vert (\triangle\Ef_{(i)},\triangle\Pf_{(i)},\triangle\Lf_{(i)},\triangle\Gf_{(i)}) - (\triangle\Ef_{(i-1)},\triangle\Pf_{(i-1)},\triangle\Lf_{(i-1)},\triangle\Gf_{(i-1)})\right\vert \\
=&\left\vert (\Ef,\Pf,\Lf,\Gf)\lrpar{(x_1)'_{(i)}} - (\Ef,\Pf,\Lf,\Gf)\lrpar{(x_1)'_{(i-1)}} \right\vert \\
\les& \varep \left\Vert \lrpar{{x_{(i-1)}}_{+W_{(i-1)}}- {x_{(i-1)}}} - \lrpar{{x_{(i-2)}}_{+W_{(i-2)}}- {x_{(i-2)}}} \right\Vert_{\XX(S_2)}.
\end{aligned}\label{EQchargeprescribdiffestim}
\end{align}
Plugging \eqref{EQchargeprescribdiffestim} and \eqref{EQiterationXestimate1} into \eqref{EQstepipart2estim1} yields
\begin{align*}
\begin{aligned}
&\left\Vert \lrpar{{x_{(i)}}_{+W_{(i)}}- {x_{(i)}}} - \lrpar{{x_{(i-1)}}_{+W_{(i-1)}}- {x_{(i-1)}}} \right\Vert_{\XX(S_2)} \\
\les&_{C_1^{-1}} \varep^{1/4} \left\Vert \lrpar{{x_{(i-1)}}_{+W_{(i-1)}}- {x_{(i-1)}}} - \lrpar{{x_{(i-2)}}_{+W_{(i-2)}}- {x_{(i-2)}}} \right\Vert_{\XX(S_2)}.
\end{aligned}
\end{align*} 
From the above we conclude that for $\varep>0$ sufficiently small,
\begin{align*}
\begin{aligned}
&\left\Vert \lrpar{{x_{(i)}}_{+W_{(i)}}- {x_{(i)}}} - \lrpar{{x_{(i-1)}}_{+W_{(i-1)}}- {x_{(i-1)}}} \right\Vert_{\XX(S_2)} \\
\leq& \half \left\Vert \lrpar{{x_{(i-1)}}_{+W_{(i-1)}}- {x_{(i-1)}}} - \lrpar{{x_{(i-2)}}_{+W_{(i-2)}}- {x_{(i-2)}}} \right\Vert_{\XX(S_2)}.
\end{aligned}
\end{align*}
This finishes the proof of \eqref{EQiterationcontract}.\\

\ni \textbf{Analysis of the limit of the iteration scheme and conclusion of Theorem \ref{THMmain1}.} First let us rephrase what we proved above. For sphere data $y\in \XX(S_2)$ consider the mapping
\begin{align*}
\begin{aligned}
y\mapsto \FF(y):= (x_{+W}-x) \vert_{S_2},
\end{aligned}
\end{align*}
where $x$ and $W$ are respectively defined by applications of Theorems \ref{THMcharGluing10d} and \ref{THMmain0} with the conditions
\begin{align}
\begin{aligned}
&x\vert_{S_2} = x_2 - y \text{ up to } (\Ef,\Pf,\Lf,\Gf), \,\, \text{ and } \,\, x \vert_{S_1} = x_1', \\
&(\mathbf{E}, \mathbf{P}, \mathbf{L}, \mathbf{G})\lrpar{{x}_{+W}\Big\vert_{S_2}} = (\mathbf{E}, \mathbf{P},\mathbf{L}, \mathbf{G})(x_2),
\end{aligned}\label{EQdefCauchysequence2}
\end{align}
where $x_1'$ is induced sphere data on the perturbed sphere $S_1'\subset \HHb_1$ near $S_1$ constructed in the application of Theorem \ref{THMcharGluing10d}.

The above estimates show that for $\varep>0$ sufficiently small and $C_3>0$ sufficiently large, the mapping $y\mapsto \FF(y)$ is well-defined on $\Vert y \Vert_{\XX(S_2)} \les \varep$ and is a \emph{contraction}. In particular, the sequence $(y_{(i)})_{i\geq-1} \subset \XX(S_2)$ of sphere data on $S_2$ defined by
\begin{align}
\begin{aligned}
y_{(-1)} := 0, \,\,
y_{(i)} := \FF(y_{(i-1)}) \,\,\, \text{ for } i\geq 0,
\end{aligned}\label{EQdefCauchysequence}
\end{align}
is well-defined and converges to a fixed point $y_{(\infty)}\in \XX(S_2)$ of $\FF$ of the form
\begin{align}
\begin{aligned}
y_{(\infty)} = {x_{(\infty)}}_{+W_{(\infty)}} -x_{(\infty)},
\end{aligned}\label{EQformFixedPoint}
\end{align}
where $x_{(\infty)}$ and $W_{(\infty)}$ satisfy the properties
\begin{align}
\begin{aligned}
&x_{(\infty)}\vert_{S_2} = x_2 - y_{(\infty)} \text{ up to } (\Ef,\Pf,\Lf,\Gf), \,\, \text{ and } \,\, x_{(\infty)} \vert_{S_1} = (x_1)'_{(\infty)}, \\
&(\mathbf{E}, \mathbf{P}, \mathbf{L}, \mathbf{G})\lrpar{{x_{(\infty)}}_{+W_{(\infty)}}\Big\vert_{S_2}} = (\mathbf{E}, \mathbf{P},\mathbf{L}, \mathbf{G})(x_2),
\end{aligned}\label{EQbdrycondith14}
\end{align}
where $(x_1)'_{(\infty)}$ is induced sphere data on the perturbed sphere $(S_1)'_{(\infty)}\subset \HHb_1$ near $S_1$ constructed by the application of Theorem \ref{THMcharGluing10d}.

Combining \eqref{EQformFixedPoint} with the first of \eqref{EQbdrycondith14}, we get that
\begin{align*}
\begin{aligned}
&x_{(\infty)}\vert_{S_2} = x_2 - \lrpar{{x_{(\infty)}}_{+W_{(\infty)}}-x_{(\infty)}}\Big\vert_{S_2} \text{ up to } (\Ef,\Pf,\Lf,\Gf) \\
\Leftrightarrow\,\,\,\,& {x_{(\infty)}}_{+W_{(\infty)}} \Big\vert_{S_2} = x_2 \text{ up to } (\Ef,\Pf,\Lf,\Gf),
\end{aligned}
\end{align*}
which, with the second of \eqref{EQbdrycondith14}, implies that
\begin{align*}
\begin{aligned}
{x_{(\infty)}}_{+W_{(\infty)}} \Big\vert_{S_2} = x_2.
\end{aligned}
\end{align*}
Moreover, applying the estimates of Theorems \ref{THMcharGluing10d} and \ref{THMmain0} as above, it follows in a standard fashion that
\begin{align*}
\begin{aligned}
\Vert {x_{(\infty)}}_{+W_{(\infty)}}-\mathfrak{m} \Vert_{\XX^{\mathrm{h.f.}}(\HH)} \les \Vert {\underline{x}}-\mathfrak{m} \Vert_{\mathcal{X}^+(\tilde{\HHb}_{[-\de,\de],1}) } + \Vert x_{2}-\mathfrak{m} \Vert_{\mathcal{X}(S_{2})}.
\end{aligned}
\end{align*}
This shows that ${x_{(\infty)}}_{+W_{(\infty)}}$ is the claimed solution to obstruction-free null gluing, and finishes the proof of Theorem \ref{THMmain1}.

\section{Proof of obstruction-free spacelike gluing} \label{SECproofSPACELIKE}
\ni In this section we prove Corollary \ref{CORspacelike}. Before turning to the proof, we recapitulate some preliminaries concerning spacelike initial data. The material is an adaption of the presentation and analysis in \cite{ACR3} (Sections 6 and 7) to the asymptotically flat spacelike initial data treated in this paper.\\

\ni \textbf{(1) Local ADM integrals.} Let $(\Si,g,k)$ be asymptotically flat spacelike initial data. In the chart near spacelike infinity we define for radii $r>0$ large the following \emph{local ADM integrals} on $S_r$, for $i=1,2,3$,
\begin{align} 
\begin{aligned}  
{\mathbf{E}}^{\mathrm{loc}}_{\mathrm{ADM}}(r,g,k) :=& -\frac{1}{8\pi} \int\limits_{S_{r}} \lrpar{\RRRic - \half R_{\mathrm{scal}} \, g}(X ,N) d\mu_\gd, \\
\lrpar{\mathbf{P}^{\mathrm{loc}}_{\mathrm{ADM}}}^i(r,g,k) :=& \frac{1}{8\pi} \int\limits_{S_{r}}  \lrpar{k_{il}-\tr k \, g_{il}} N^l d\mu_\gd, \\
\lrpar{\mathbf{L}^{\mathrm{loc}}_{\mathrm{ADM}}}^i(r,g,k) :=& \frac{1}{8\pi} \int\limits_{S_{r}} (k_{jl}-\tr k \, g_{jl}) \lrpar{Y_{(i)}}^j N^l d\mu_\gd, \\
\lrpar{{\mathbf{C}}^{\mathrm{loc}}_{\mathrm{ADM}}}^i(r,g,k):=& \frac{1}{16\pi} \int\limits_{S_{r}} \lrpar{\RRRic - \half R_{\mathrm{scal}} \, g }(Z^{(i)},N) d\mu_\gd,
\end{aligned} \label{EQlocalisedADMcharges}
\end{align}
where the vectorfields $X$ and $Z^{(i)}$, $i=1,2,3$, are defined by
\begin{align*}
\begin{aligned}
X := x^i \pr_i, \,\, Z^{(i)} := \lrpar{\vert x \vert^2 \de^{ij}-2x^i x^j} \pr_j,
\end{aligned}
\end{align*}
and $N$ and $Y_{(i)}$, $i=1,2,3$ are defined after \eqref{EQDEFADMEP} and \eqref{EQDEFADMLC}.

Given \eqref{EQlocalisedADMcharges} and the definition of the ADM parameters $({\mathbf{E}}_{\mathrm{ADM}},{\mathbf{P}}_{\mathrm{ADM}},{\mathbf{L}}_{\mathrm{ADM}},{\mathbf{C}}_{\mathrm{ADM}})$ in \eqref{EQDEFADMEP} and \eqref{EQDEFADMLC}, it is well-known (see \cite{ACR3} and references therein) that for asymptotically flat spacelike initial data, as $r\to \infty$,
\begin{align}
\begin{aligned}
{\mathbf{E}}^{\mathrm{loc}}_{\mathrm{ADM}}(r,g,k) = {\mathbf{E}}_{\mathrm{ADM}}(g,k) + \OO(r^{-1}), \,\, {\mathbf{P}}^{\mathrm{loc}}_{\mathrm{ADM}}(r,g,k) = {\mathbf{P}}_{\mathrm{ADM}}(g,k) + \OO(r^{-1}),
\end{aligned}\label{EQconvergenceADMlocalEP}
\end{align}
and moreover, if ${\mathbf{L}}_{\mathrm{ADM}}(g,k)$ and ${\mathbf{C}}_{\mathrm{ADM}}(g,k)$ are well-defined and finite, then, as $r\to\infty$,
\begin{align}
\begin{aligned}
{\mathbf{L}}^{\mathrm{loc}}_{\mathrm{ADM}}(r,g,k) = {\mathbf{L}}_{\mathrm{ADM}}(g,k) + \smallO(1), \,\, {\mathbf{C}}^{\mathrm{loc}}_{\mathrm{ADM}}(r,g,k) = {\mathbf{C}}_{\mathrm{ADM}}(g,k) + \smallO(1),
\end{aligned}\label{EQconvergenceADMlocalLC}
\end{align}
where $\smallO(1)$ denotes terms that go to zero as $r\to \infty$.\\

\ni \textbf{(2) Scaling of spacelike initial data.} It is well-known that spacelike initial data can be rescaled as follows. Let $(\Si,g,k)$ be asymptotically flat spacelike initial data, and let $(x^1,x^2,x^3)$ be coordinates near spacelike infinity. The scaling of $(g,k)$ is defined in two steps.  
\begin{enumerate}
\item For a given real number $R\geq1$, define new coordinates $(y^1,y^2, y^3)$ by 
\begin{align} 
\begin{aligned} 
\Psi_R(y^1,y^2,y^3) := (R\cdot y^1, R \cdot y^2, R\cdot y^3) = (x^1, x^2, x^3).
\end{aligned} \label{EQdefcoordinatechangeSPACELIKE8889}
\end{align}

\item Define $({}^{(R)} g,{}^{(R)} k)$ by
\begin{align} 
\begin{aligned} 
{}^{(R)} g := R^{-2} g, \,\,{}^{(R)}k := R^{-1}\, k.
\end{aligned} \label{EQdefconformalchangeSPACELIKE8889}
\end{align}
\noindent By construction, $({}^{(R)} g,{}^{(R)} k)$ solve the spacelike constraint equations.
\end{enumerate}
\ni We remark that by \eqref{EQdefcoordinatechangeSPACELIKE8889} and \eqref{EQdefconformalchangeSPACELIKE8889}, for all integers $l\geq0$, we have the relations
\begin{align} 
\begin{aligned} 
\pr_y^{l} \lrpar{{}^{(R)}g_{ij}}= R^l \lrpar{\pr_x^l g_{ij}} \circ \Psi_R, \,\, \pr_y^{l} \lrpar{{}^{(R)}k_{ij}} = R^{l+1} \lrpar{\pr_x^l k_{ij}} \circ \Psi_R,
\end{aligned} \label{EQscalingderivatives}
\end{align}
where we denote
\begin{align*} 
\begin{aligned} 
{}^{(R)} g_{ij} := {}^{(R)} g(\pr_{y^i},\pr_{y^j}), \,\, {}^{(R)} k_{ij} := {}^{(R)} k(\pr_{y^i},\pr_{y^j}).
\end{aligned} 
\end{align*}
Moreover, it is straight-forward to prove that the local integrals scale as follows, for reals $r>0$,
\begin{align} 
\begin{aligned} 
\mathbf{E}_{\mathrm{ADM}}^{\mathrm{loc}}\lrpar{r,{}^{(R)}g,{}^{(R)}k} =& R^{-1}  \mathbf{E}_{\mathrm{ADM}}^{\mathrm{loc}}\lrpar{R r,g,k}, &
\mathbf{P}_{\mathrm{ADM}}^{\mathrm{loc}}\lrpar{r,{}^{(R)}g,{}^{(R)}k} =& R^{-1}  \mathbf{P}_{\mathrm{ADM}}^{\mathrm{loc}}\lrpar{R r,g,k}, \\
\mathbf{L}_{\mathrm{ADM}}^{\mathrm{loc}}\lrpar{r,{}^{(R)}g,{}^{(R)}k} =& R^{-2}  \mathbf{L}_{\mathrm{ADM}}^{\mathrm{loc}}\lrpar{R r,g,k}, &
\mathbf{C}_{\mathrm{ADM}}^{\mathrm{loc}}\lrpar{r,{}^{(R)}g,{}^{(R)}k} =& R^{-2}  \mathbf{C}_{\mathrm{ADM}}^{\mathrm{loc}}\lrpar{R r,g,k}.
\end{aligned} \label{EQLEMscalingADMlocal}
\end{align}

\ni\textbf{(3) Norms for spacelike initial data.} \ni We now turn to the introduction of local norms for spacelike initial data. For ease of presentation we use $C^k$-spaces.

\begin{definition}[$C^k$-norms for tenors] Let $K\subset \RRR^3$ denote a compact set with smooth boundary, and let $T$ be a $j$-tensor on $K$. Define, for integers $k\geq0$,
\begin{align*} 
\begin{aligned} 
\Vert T \Vert_{C^k(K)} := \sum\limits_{1\leq i_1, \cdots i_j \leq 3}\sum\limits_{0\leq \vert \a \vert \leq k} \Vert \pr^\a T_{i_1\cdots i_j} \Vert_{L^{\infty}(K)},
\end{aligned} 
\end{align*}
where $\a=(\a_1,\a_2,\a_3) \in \mathbb{N}^3$, $\pr^\a = \pr_1^{\a_1}\pr_2^{\a_2}\pr_3^{\a_3}$, and $T_{i_1 \cdots i_l}$ denotes the Cartesian coordinate components of $T$. Let $C^k(K)$ be the space of $k$-times continuously differentiable tensors $T$ on $K$ with $\Vert T  \Vert_{C^k(K)} < \infty$.
\end{definition}

\ni \textbf{Notation.} For two real numbers $0<r_1<r_2$, define the coordinate annulus $A_{[r_1,r_2]}$ by
\begin{align*} 
\begin{aligned} 
A_{[r_1,r_2]} := \left\{ x \in \RRR^3: r_1 \leq \vert x \vert \leq r_2 \right\}.
\end{aligned} 
\end{align*}

\begin{definition}[Local norm for spacelike initial data] \label{DEFnormSpacelikeInitialData} Let $0<r_1<r_2$ be two reals, and let $k\geq1$ be an integer. Given spacelike initial data $(g,k)$ on $A_{[r_1,r_2]}$, we define
\begin{align*} 
\begin{aligned} 
\left\Vert \lrpar{g,k} \right\Vert_{C^{k}(A_{[r_1,r_2]}) \times C^{k-1}(A_{[r_1,r_2]})} :=
\Vert g \Vert_{C^{k}(A_{[r_1,r_2]})} + \Vert k \Vert_{C^{k-1}(A_{[r_1,r_2]})}.
\end{aligned} 
\end{align*}
\end{definition}

\ni \textbf{Notation.} In the following we assume that the metric $g$ is $k_0$-times and the second fundamental form $k$ is $k_0$-times continuously differentiable, where the universal integer $k_0 \geq 8$ is determined by the condition that if, for $\varep>0$ sufficiently small,
\begin{align*}
\begin{aligned}
\Vert (g-e,k) \Vert_{C^{k_0}(A_{[r_1,r_2]}) \times C^{k_0-1}(A_{[r_1,r_2]})} \leq \varep,
\end{aligned}
\end{align*}
then the corresponding sphere data on spheres $S_r$, $r_1\leq r \leq r_2$ (see Section 7 in \cite{ACR3} for the construction and the precise gauge choices) lies in $\XX(S_r)$, with
\begin{align}
\begin{aligned}
\Vert x - \mathfrak{m} \Vert_{\XX(S_r)} \les \Vert (g-e,k) \Vert_{C^{k_0}(A_{[r_1,r_2]}) \times C^{k_0-1}(A_{[r_1,r_2]})}.
\end{aligned}\label{EQboundednessSpheredata}
\end{align}

\ni By scaling and the definition of asymptotic flatness, we have the following estimates for rescaled spacelike initial data. Its proof is omitted, see \eqref{EQscalingderivatives} and also Lemma 6.12 in \cite{ACR3}.
\begin{lemma}[Smallness of rescaled spacelike initial data] \label{LEMspacelikeRescaling} Let $(\Si,g,k)$ be asymptotically flat spacelike initial data. Let $(x^1,x^2,x^3)$ be coordinates near spacelike infinity. For real numbers $R\geq1$ sufficiently large, the rescaled spacelike initial data
$({}^{(R)}{g}_{ij}, {}^{(R)}{k}_{ij})$ is well-defined on ${A}_{[1/2,7/2]}$ and
\begin{align} 
\begin{aligned} 
\left\Vert \lrpar{{}^{(R)}{g} -e,  {}^{(R)}{k}} \right\Vert_{C^{k_0}({A}_{[1/2,7/2]})\times C^{k_0-1}({A}_{[1/2,7/2]})} 
= \OO(R^{-1}),
\end{aligned} \label{EQestimateLEMspacelikeRescaling}
\end{align}
where $e_{ij}=\de_{ij}$ denotes the Euclidean metric in coordinates $(x^1,x^2,x^3)$.
\end{lemma}

\ni\textbf{(4) Comparison of local ADM integrals to charge integrals $(\Ef,\Pf,\Lf,\Gf)$.} Consider spacelike initial data $(g,k)$ on the annulus $A_{[1/2,7/2]}$ such that for a real number $\varep>0$,
\begin{align}
\begin{aligned}
\Vert (g-e,k) \Vert_{C^{k_0}(A_{[r_1,r_2]}) \times C^{k_0-1}(A_{[r_1,r_2]})} \leq \varep.
\end{aligned}\label{EQsmallness1}
\end{align}
By \eqref{EQboundednessSpheredata}, the associated sphere data $x_{r}$ on $S_r$ for each $r_1\leq r \leq r_2$ (see Section 7 in \cite{ACR3}) satisfies 
\begin{align}
\begin{aligned}
\Vert x_r - \mathfrak{m} \Vert_{\XX(S_r)} \les \varep.
\end{aligned}\label{EQsmallness2}
\end{align}
The analysis in \cite{ACR3} shows that under the smallness assumptions \eqref{EQsmallness1} and \eqref{EQsmallness2} we can relate
\begin{align}
\begin{aligned}
{\mathbf{E}}^{\mathrm{loc}}_{\mathrm{ADM}}(r,g,k) =& \Ef(x_r) + \OO(\varep^2), & ({\mathbf{P}}^{\mathrm{loc}}_{\mathrm{ADM}})^i(r,g,k) =& \Pf^{m_i}(x_r) + \OO(\varep^2), \\
({\mathbf{L}}^{\mathrm{loc}}_{\mathrm{ADM}})^i(r,g,k) =& \Lf^{m_i}(x_r) + \OO(\varep^2), & ({\mathbf{C}}^{\mathrm{loc}}_{\mathrm{ADM}})^i(r,g,k) =& \Gf^{m_i}(x_r) + r \Pf^{m_i}(x_r) + \OO(\varep^2),
\end{aligned}\label{EQchargecomparisonSR}
\end{align}
where $i=1,2,3$ and $(m_1,m_2,m_3)=(1,-1,0)$.\\

\ni We are now in position to prove Corollary \ref{CORspacelike}. The idea is to first rescale from the asymptotic region to small data, and then apply the bifurcate obstruction-free gluing; see also Figure \ref{FIGbifurcateLATERspacelike} below. The important points are to relate the rescaled small charges to the ADM parameters, and to verify that the assumptions of bifurcate obstruction-free null gluing are satisfied.

First, for $R>0$ large (to be determined), rescale the sphere $S_R$ with sphere data $x_R$ in the given spacelike initial data set $(\Si,g,k)$ to $S_1$. We denote the rescaled sphere data on $S_1$ by ${}^{(R)}x_1$, and by asymptotic flatness, it holds in particular that 
\begin{align*}
\begin{aligned}
\Vert {}^{(R)}x_1 - \mathfrak{m} \Vert_{\XX(S_1)} = \OO(R^{-1}).
\end{aligned}
\end{align*}
By the above, it holds for $R>0$ large that
\begin{align*}
\begin{aligned}
\Ef({}^{(R)}x_1)= \Ef^{\mathrm{loc}}_{\mathrm{ADM}}(1,{}^{(R)}g,{}^{(R)}k) + \OO(R^{-2})
= \frac{1}{R} \Ef^{\mathrm{loc}}_{\mathrm{ADM}}(R,g,k) +\OO(R^{-2})
= \frac{1}{R} \Ef_{\mathrm{ADM}}^{(\Si,g,k)} +\OO(R^{-2}).
\end{aligned}
\end{align*} 
Similarly, we can calculate, for $m=-1,0,1$ and $(i_{-1},i_0,i_1)=(2,3,1)$,
\begin{align*}
\begin{aligned}
\Pf^m({}^{(R)}x_1) =& \frac{1}{R} (\Pf_{\mathrm{ADM}}^{(\Si,g,k)})^{i_m} +\OO(R^{-2}), \,\, \Lf^m({}^{(R)}x_1)= \frac{1}{R^2} (\Lf_{\mathrm{ADM}}^{(\Si,g,k)})^{i_m} +\OO(R^{-2}),\\
\Gf^m({}^{(R)}x_1) =& \frac{1}{R^2}(\mathbf{C}_{\mathrm{ADM}}^{(\Si,g,k)})^{i_m} - \frac{1}{R} (\Pf_{\mathrm{ADM}}^{(\Si,g,k)})^{i_m} +\OO(R^{-2}).
\end{aligned}
\end{align*}

Second, rescale the sphere $S_{2R}$ with sphere data $x_{2R}^{\mathrm{Kerr}}$ in Kerr reference initial data to the sphere $S_2$ with sphere data ${}^{(R)}x_2^{\mathrm{Kerr}}$. Similarly to the above, for $R>0$ large it holds that
\begin{align*}
\begin{aligned}
\Ef({}^{(R)}x_2^{\mathrm{Kerr}}) =& \frac{1}{R} \Ef_{\mathrm{ADM}}^{\mathrm{Kerr}} +\OO(R^{-2}), & \Pf^m({}^{(R)}x_2^{\mathrm{Kerr}}) =& \frac{1}{R} (\Pf_{\mathrm{ADM}}^{\mathrm{Kerr}})^{i_m} +\OO(R^{-2}),\\
\Lf^m({}^{(R)}x_2^{\mathrm{Kerr}}) =& \frac{1}{R^2} (\Lf_{\mathrm{ADM}}^{\mathrm{Kerr}})^{i_m} +\OO(R^{-2}), & \Gf^m({}^{(R)}x_2^{\mathrm{Kerr}}) =& \frac{1}{R^2} (\mathbf{C}_{\mathrm{ADM}}^{\mathrm{Kerr}})^{i_m}-\frac{2}{R} (\Pf_{\mathrm{ADM}}^{\mathrm{Kerr}})^{i_m} +\OO(R^{-2}).
\end{aligned}
\end{align*}

We are now in position to apply the bifurcate obstruction-free null gluing, Theorem \ref{THMbifurcateFREE}, to the spheres $S_1$ and $S_2$ with respective sphere data ${}^{(R)}x_1$ and ${}^{(R)}x_1^{\mathrm{Kerr}}$. We have to check that there are constants $C_1, C_2>0$ such that for $R>0$ sufficiently large,
\begin{align}
\begin{aligned}
\triangle \Ef =& C_1 \varep, \\
\triangle \Lf =& C_2 \varep^2, \\
\triangle \Ef >& C_3 \lrpar{\vert \triangle \Pf \vert + \left\vert \triangle \Gf + \Pf({}^{(R)}x_1) \right\vert},
\end{aligned}\label{EQconditionsTBS}
\end{align}
for a sufficiently large constant $C_3>0$, and we denoted $\varep = R^{-1}$.\\

\ni \textbf{First of \eqref{EQconditionsTBS}.} By the above it holds that
\begin{align*}
\begin{aligned}
\triangle \Ef := \Ef({}^{(R)}x_2^{\mathrm{Kerr}})-\Ef({}^{(R)}x_1)= \frac{1}{R} \lrpar{\Ef_{\mathrm{ADM}}^{\mathrm{Kerr}}-\Ef_{\mathrm{ADM}}^{(\Si,g,k)}}+\OO(R^{-2}).
\end{aligned}
\end{align*}
Hence for $R>0$ sufficiently large, the first of \eqref{EQconditionsTBS} holds, with a constant
\begin{align*}
\begin{aligned}
C_1 > \half \lrpar{\Ef_{\mathrm{ADM}}^{\mathrm{Kerr}}-\Ef_{\mathrm{ADM}}^{(\Si,g,k)}}=\half \triangle \Ef_{\mathrm{ADM}}>0.
\end{aligned}
\end{align*}

\ni\textbf{Second of \eqref{EQconditionsTBS}.} By the above it holds that
\begin{align*}
\begin{aligned}
\triangle \Lf = \OO(R^{-2}),
\end{aligned}
\end{align*}
so that the second of \eqref{EQconditionsTBS} is satisfied.\\

\ni \textbf{Third of \eqref{EQconditionsTBS}.} On the one hand, from the above we have that
\begin{align*}
\begin{aligned}
\triangle \Ef = \frac{1}{R}\triangle\Ef_{\mathrm{ADM}}+\OO(R^{-2}), \,\, \vert \triangle \Pf \vert = \frac{1}{R} \vert \triangle\Pf_{\mathrm{ADM}}\vert+\OO(R^{-2}).
\end{aligned}
\end{align*}
On the other hand, we calculate that
\begin{align*}
\begin{aligned}
\left\vert \triangle \Gf + \Pf({}^{(R)}x_1) \right\vert =& \left\vert \Gf({}^{(R)}x_2^{\mathrm{Kerr}}) -\Gf({}^{(R)}x_1) + \Pf({}^{(R)}x_1) \right\vert \\
=& \left\vert \lrpar{\frac{1}{R^2} \mathbf{C}_{\mathrm{ADM}}^{\mathrm{Kerr}}-\frac{2}{R} \Pf_{\mathrm{ADM}}^{\mathrm{Kerr}}} -\lrpar{\frac{1}{R^2}\mathbf{C}_{\mathrm{ADM}}^{(\Si,g,k)} - \frac{1}{R} \Pf_{\mathrm{ADM}}^{(\Si,g,k)}} +\frac{1}{R} \Pf_{\mathrm{ADM}}^{(\Si,g,k)} \right\vert +\OO(R^{-2}) \\
=& \left\vert \frac{1}{R^2}\lrpar{\mathbf{C}_{\mathrm{ADM}}^{(\Si,g,k)}-\mathbf{C}_{\mathrm{ADM}}^{\mathrm{Kerr}}} -\frac{2}{R}\lrpar{ \Pf_{\mathrm{ADM}}^{\mathrm{Kerr}} - \Pf_{\mathrm{ADM}}^{(\Si,g,k)}} \right\vert +\OO(R^{-2}) \\
=& \frac{2}{R} \left\vert \triangle\Pf_{\mathrm{ADM}} \right\vert + \OO(R^{-2}).
\end{aligned}
\end{align*}
Recall that by assumption it holds that for a constant $\tilde{C}_3>0$,
\begin{align}
\begin{aligned}
\triangle \Ef_{\mathrm{ADM}}>0, \,\, \triangle \Ef_{\mathrm{ADM}} > \tilde{C}_3 \left\vert \Pf_{\mathrm{ADM}}\right\vert.
\end{aligned}\label{EQassumptionrecall9909}
\end{align}
Plugging the above expressions into the third of \eqref{EQconditionsTBS} and multiplying with $R>0$, we can rewrite the condition as
\begin{align*}
\begin{aligned}
\triangle \Ef_{\mathrm{ADM}} + \OO(R^{-1}) > C_3 \lrpar{ 3\left\vert \triangle\Pf_{\mathrm{ADM}} \right\vert +\OO(R^{-1})},
\end{aligned}
\end{align*}
which holds true for $\tilde{C}_3>0$ in \eqref{EQassumptionrecall9909} sufficiently large and $R>0$ sufficiently large (depending on $\triangle \Ef_{\mathrm{ADM}}$). This shows that the third of \eqref{EQconditionsTBS} is satisfied. 

Thus we can apply the bifurcate obstruction-free null gluing, Theorem \ref{THMbifurcateFREE}, to get a solution $x$ to the null structure equations along $\HH$ and $\HHb$ which agrees with ${}^{(R)}x_1$ on $S_1$ and with ${}^{(R)}x_2^{\mathrm{Kerr}}$ on $S_2$. As remarked after Theorem \ref{THMmain1}, the constructed solution (and the resulting solution to the full null structure equations) is sufficiently regular to apply local existence results for the Einstein equations and solve forward for a spacetime $(\MM,\g)$. In this spacetime, it is straight-forward to pick a spacelike hypersurface $\Si'$ connecting the spheres $S_1$ and $S_2$ such that $\Si'$ regularly extends the initial data hypersurfaces on which $(g,k)$ and the Kerr reference spacelike initial data are posed; for details we refer to Section 4.4 in \cite{ACR3} where a similar spacelike hypersurface $\Si'$ is constructed.

We conclude the proof by scaling the spacelike initial data back out by applying the inverse scaling transformation that was used to reduce to the small data setup. This finishes the proof of Corollary \ref{CORspacelike}.

\begin{figure}[H]
	\begin{center}
		\includegraphics[width=9.5cm]{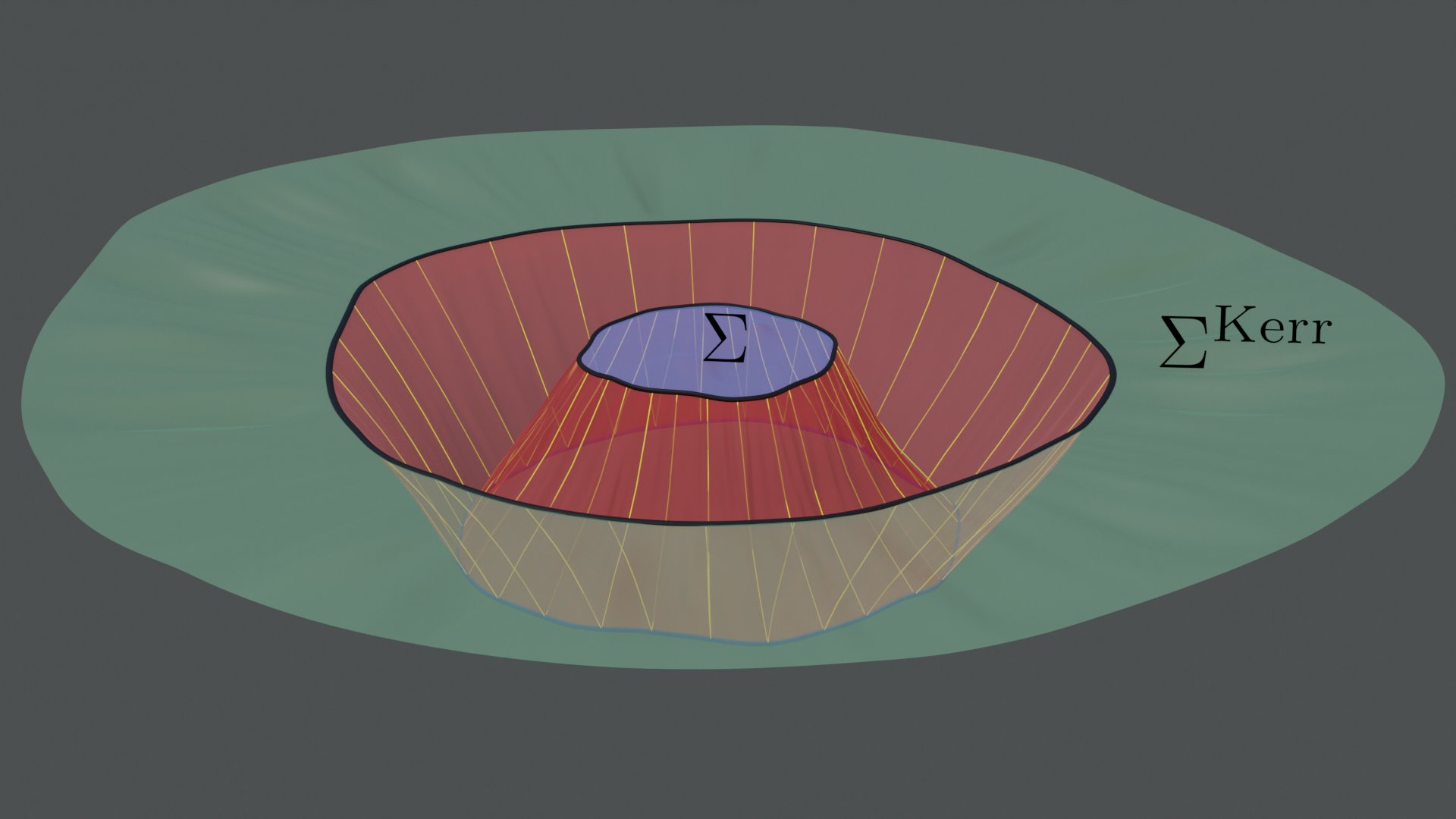} 
		\vspace{0.4cm} 
		\caption{Spacelike gluing via bifurcate null gluing.} \label{FIGbifurcateLATERspacelike}
	\end{center}
\end{figure}

\newpage

\section{Construction of high-frequency solutions}\label{SECconstructionSolution} \ni In this section we make an ansatz for a class of \emph{high-frequency characteristic seeds along an outgoing null hypersurface $\HH$} and subsequently construct, in explicit fashion, the corresponding high-frequency solutions to the null structure equations along $\HH$.

The construction starts from the following two base ``inputs". 
\begin{enumerate}
\item A (low-frequency) solution $x$ to the null structure equations along $\HH$, satisfying for some real number $\varep>0$,
\begin{align}
\Vert x-\mathfrak{m} \Vert_{\XX(\HH)} \leq \varep. \label{EQsmallnesscondBaseCONSTRUCTION}
\end{align}
Denote the metric components, Ricci components and null curvature components of $x$ by an $x$-subscript. For example, the \emph{null lapse} and \emph{induced metric on $S_v$} of $x$ are denoted by $\Om_x$ and $\gd_x$, respectively. By Definition \ref{DEFcharSEEDH}, the \emph{characteristic seed along $\HH$} of $x$ is given by $(\Om_x, \mathrm{conf}(\gd_x))$, and the \emph{characteristic seed on $S_1$} of $x$ is given by
\begin{align}
(\gd_x,\trchi_x,\trchib_x,\omb_x,(\Du\omb)_x,\eta_x,\chibh_x,\ab_x).	\label{EQcharSEEDgiven}
\end{align}

\item An $S_v$-tangential $2$-covariant tensor $W$ along $\HH$ of high-frequency in $v$. The precise class of $W$'s to be considered is introduced in Section \ref{SECAnsatzEstimates} below. In particular, $W$ is a linear combination of finitely many (fixed) tensor spherical harmonics, with the coefficients being high-frequency functions in $v$ and having large amplitude.
\end{enumerate}

\ni Given the above input, we consider the new \emph{\emph{high-frequency characteristic seed}} 
\begin{align}
\begin{aligned}
(\Om,\mathrm{conf}(\tilde{\gd})) \text{ along } \HH,
\end{aligned}\label{EQhighfrequcharseed1introsection}
\end{align}
where $\Om$ and $\tilde{\gd}$ are defined by
\begin{align}
\Om :=& \Om_x, \label{EQomegadef} \\
\tilde{\gd} :=& \, \gd_x + \varphi W, \label{EQrep1}
\end{align}
where $\varphi$ is a smooth cut-off function near $v\in \{1,2\}$ to ensure that $\varphi W$ is supported only inside of $1<v<2$. The goal of this section is to construct from \eqref{EQhighfrequcharseed1introsection} and \eqref{EQcharSEEDgiven} a high-frequency solution $x_{+W}$ to the null structure equations along $\HH$, and derive bounds in appropriate function spaces; see Section \ref{SECpreliminariesHIGHfrequency} below for an overview of the estimates and further remarks. 

As outlined in Section \ref{SUBSECcharseed}, we construct the solution $x_{+W}$ by integrating the null transport equations of the null structure equations in hierarchical order; the explicit order is \eqref{EQntePHIestimatesection01} for $\phi$, \eqref{EQnulleta} for $\eta$, \eqref{EQtrchib} for $\Om\trchib$, \eqref{EQchibh} for $\Om\chibh$, \eqref{EQombnulltransport} for $\omb$, \eqref{EQabNTE} for $\ab$, \eqref{EQnteFORduomb} for $\Du\omb$. Along the way we show how the remaining Ricci coefficients and null curvature components are directly determined through the null structure equations.

The rest of this section is structured as follows. 
\begin{itemize}
\item In Section \ref{SECAnsatzEstimates} we give the ansatz for the tensor $W$, and define the corresponding tensor norm. 
\item In Section \ref{SECpreliminariesHIGHfrequency} we give an overview of the estimates of this section and make further remarks on the control of high-frequency terms.
\item We derive preliminary control estimates in Section \ref{SECcontrol2}. 
\item In Sections \ref{SECphiEstimatesPrecise} to \ref{SECestimateDuomb} we construct and control the solution of the null transport equations along $\HH$ by hierarchical integration of null transport equations.

\item In Section \ref{SECestimateDiffest2} we derive \emph{difference estimates} for constructed solutions. These are employed in Section \ref{SECdiffEstimates2W} to prove \eqref{EQdiffEstimate2mainthm2statement} of Theorem \ref{THMmain0}, which is applied in the iteration scheme in the proof of the main theorem (see Section \ref{SECgluingEP}). 
\end{itemize}
As mentioned above, in Sections \ref{SECphiEstimatesPrecise}, \ref{SECestimatesOMTRCHI} and \ref{SECcontrolETA} the quantities $\phi, \Om\trchi$ and $\eta$ are analyzed \emph{beyond} the lowest order $\varep$. This is important for Section \ref{SECgluingL} where we study the transport equation for the charge $\mathbf{L}$ at the order of $\varep^2$.

\subsection{Ansatz for the tensor $W$} \label{SECAnsatzEstimates} 

\ni In this section we introduce our precise ansatz for the tensor $W$ which is used to define the \emph{high-frequency} characteristic seed in \eqref{EQrep1} from which, together with \eqref{EQcharSEEDgiven}, the null structure equations are solved in a hierarchical fashion in Sections \ref{SECphiEstimatesPrecise} to \ref{SECestimateDuomb}.

We define $W(v,\th^1,\th^2)$ along $\HH$ to be the $S_v$-tangential $\gac$-tracefree symmetric $2$-covariant tensorfield on $\HH$ given by
\begin{subequations}
\begin{align}
\begin{aligned}
W := -\varep^{3/2} \cos\lrpar{\frac{v}{\varep}} \mfW_0 - \varep^{5/4} \lrpar{\cos\lrpar{\frac{v}{\sqrt{\varep}}} \mfW_1 -\sin\lrpar{\frac{v}{\sqrt{\varep}}} \mfW_2},
\end{aligned}\label{EQWansatz}
\end{align}
where
\begin{align}
\begin{aligned}
\mfW_0 := f_{30} \psi^{30} +  f_{22} \psi^{22} + \sum\limits_{m=-1}^1 f_{2m} \psi^{2m}, \,\,
\mfW_1 := \tilde{f}_{30} \psi^{30}, \,\,
\mfW_2 := \sum\limits_{m=-1}^1\tilde{f}_{2m}  \phi^{2m}.
\end{aligned}\label{EQWansatzDetail}
\end{align}
\end{subequations}
\ni Here $\psi^{lm}$ and $\phi^{lm}$ are the ($v$-independent) electric and magnetic tensor spherical harmonics, respectively, defined in Appendix \ref{APPconstructionW}, and $f_{30}, f_{22}, f_{2m},\tilde{f}_{30}, \tilde{f}_{2m}, $ for $m=-1,0,1$ are smooth $v$-dependent (low-frequency) coefficient functions given by
\begin{subequations}
\begin{align}
f_{30} :=& \half c_{30}, \,\, f_{2m} := \frac{1}{\lrpar{Y^{30}\cdot Y^{2m}}^{(1m)}} \lrpar{c_{2m} \xi_1 + c'_{2m} \xi_2}, \label{EQchoiceofCoeffFunction1} \\
f_{22} :=& c_{22} \xi_4,\label{EQchoiceofCoeffFunction2} \\
\tilde{f}_{30} :=& \tilde{c}_{30}, \,\, \tilde{f}_{2m} := \frac{1}{\lrpar{Y^{30}\cdot Y^{2m}}^{(1m)}} \tilde{c}_{2m} \xi_3,\label{EQchoiceofCoeffFunction3}
\end{align}
\end{subequations}
where $c_{30},c_{2m},c'_{2m},c_{22},\tilde{c}_{30},\tilde{c}_{2m}$ are freely choosable constants, $\lrpar{Y^{30}\cdot Y^{2m}}^{(1m)}$, $m=-1,0,1$, are Fourier coefficients, see Appendix \ref{APPconstructionW}, and $\xi_1,\xi_2,\xi_3,\xi_4$ are fixed smooth (low-frequency) universal functions satisfying $6$ universal, linearly independent integral conditions (explicitly stated in \eqref{EQintegralXcond1}, \eqref{EQintegralXcond2} and \eqref{EQintegralXcond3}).\\

\ni \emph{Remarks on the ansatz for $W$.}
\begin{enumerate} 

\item For $\varep>0$ sufficiently small, $\tilde g$ in \eqref{EQrep1} is a positive-definite metric on the spheres $S_v$, hence the characteristic seed $(\Om,\mathrm{conf}(\tilde \gd))$ is well-defined.

\item In Section \ref{SECproofW} we analyze the transport of the charges $\mathbf{E},\mathbf{P},\mathbf{L},\mathbf{G}$ (see Definition \ref{DEFchargesEPLG}) of the solution $x_{+W}$ along $\HH$, and prove that the constants $c_{30},c_{2m},c'_{2m},c_{22},\tilde{c}_{30},\tilde{c}_{2m}$ can be used to adjust the charges on $S_2$ (subject to some inequalities).

\item Our choice of the ansatz for $W$ above is specific. It is however clear from our analysis that the problem of high-frequency null gluing of the charges $\Ef,\Pf,\Lf,\Gf$ admits much flexibility in the choice of $W$, in fact both in its high frequency $v$ and the angular dependences, see 
Section \ref{SECproofW}. We do not pursue this direction here.
\end{enumerate}

\ni Based on \eqref{EQWansatz}, \eqref{EQWansatzDetail}, \eqref{EQchoiceofCoeffFunction1}, \eqref{EQchoiceofCoeffFunction2}, \eqref{EQchoiceofCoeffFunction3}, we make the following definition.

\begin{definition}[Norm of $W$] \label{DEFdefinitionWnorm}
We define
\begin{align}
\begin{aligned}
\Vert \mfW_0 \Vert := \vert c_{30} \vert + \sum\limits_{m=-1}^1( \vert c_{2m} \vert +\vert c'_{2m} \vert ) + \vert c_{22} \vert, \,\,
\Vert \mfW_1 \Vert := \vert \tilde{c}_{30} \vert, \,\,
\Vert \mfW_2 \Vert := \sum\limits_{m=-1}^1 \vert \tilde{c}_{2m} \vert,
\end{aligned}\label{EQdefSubnorms}
\end{align}
and let moreover
\begin{align}
\begin{aligned}
\Vert W \Vert := \Vert \mfW_0 \Vert + \Vert \mfW_1 \Vert + \Vert \mfW_2 \Vert.
\end{aligned}\label{EQdefNormW}
\end{align}
\end{definition}

\ni \textbf{Notation.} In the next sections, Ricci coefficients and null curvature components with tilde are calculated with respect to the metric $\tilde{\gd}$ and null lapse $\Om$ along $\HH$. 

\subsection{Overview of estimates and discussion}  \label{SECpreliminariesHIGHfrequency}

In the following Sections \ref{SECcontrol2} to \ref{SECestimateDuomb} we prove the following bound,
\begin{align}
\begin{aligned}
\Vert x_{+W} - x \Vert_{\XX(S_2)} + \Vert x_{+W} - x \Vert_{\XX^{\mathrm{h.f.}}(\HH)} \les \varep \Vert \mfW_0 \Vert + \varep^{5/4} \Vert W \Vert,
\end{aligned}\label{EQestimateDIFFx0x1}
\end{align}

\ni\textbf{Notation.} In \eqref{EQestimateDIFFx0x1}, as in the rest of the paper, we tacitly suppressed the dependence of the constant in the inequality on $\Vert \mfW_0\Vert$ and $\Vert \mfW_1\Vert+\Vert\mfW_2\Vert$. In fact, the right-hand side of \eqref{EQestimateDIFFx0x1} is given by
\begin{align}
\begin{aligned}
C_{\Vert \mfW_0\Vert}\varep \Vert \mfW_0 \Vert + C_{\Vert \mfW_1\Vert+\Vert \mfW_2\Vert}\varep^{5/4} \Vert W \Vert,
\end{aligned}\label{EQrealequationestimates}
\end{align}
with constants $C_{\Vert \mfW_0\Vert}>0$ and $C_{\Vert \mfW_1\Vert+\Vert \mfW_2\Vert}>0$ depending on $\Vert \mfW_0\Vert$ and $\Vert \mfW_1\Vert+\Vert \mfW_2\Vert$, respectively. In the setting of this paper, see Remark \ref{REMestimatesW} with \eqref{EQsmallnessw0estimatew1w2}, it holds that 
\begin{align}
\begin{aligned}
\Vert \mfW_0 \Vert \les 1, \,\, \Vert \mfW_1\Vert+\Vert \mfW_2\Vert \les 1 + \sqrt{C_2},
\end{aligned}\label{EQsizespaper}
\end{align}
where $C_2>0$ is the (possibly large) constant appearing in \eqref{EQsmallnessMain1002}. However, the constant $C_{\Vert \mfW_1\Vert+\Vert \mfW_2\Vert}$ always appears in terms with extra leeway in $\varep$ compared to terms with $C_{\Vert \mfW_0\Vert}$; for example, in \eqref{EQrealequationestimates} it comes with an extra $\varep^{1/4}$, which can be used to control the constant in front of $\Vert W \Vert$ in \eqref{EQrealequationestimates}. To not carry around the constants through the paper, we therefore write the estimates as in \eqref{EQestimateDIFFx0x1} and just remember the above. \\ 

\ni \emph{Remarks on the bound \eqref{EQestimateDIFFx0x1}.}

\begin{enumerate}
\item The bound \eqref{EQestimateDIFFx0x1} shows that despite the high-frequency nature of the characteristic seed, the resulting sphere data $x_{+W}\vert_{S_2}$ on $S_2$ (which includes, in particular, $\Om\chih$ and $\a$ evaluated at $S_2$) is small in $\XX(S_2)$. This is a consequence of $\varphi W$ being supported inside $1<v<2$. 
\item By \eqref{EQsmallnesscondBaseCONSTRUCTION} and the trivial estimate
\begin{align*}
\begin{aligned}
\Vert x_{+W} - \mathfrak{m} \Vert_{\XX^{\mathrm{h.f.}}(\HH)} \les \Vert x_{+W} - x \Vert_{\XX^{\mathrm{h.f.}}(\HH)} + \Vert x - \mathfrak{m} \Vert_{\XX(\HH)},
\end{aligned}
\end{align*}
the bound \eqref{EQestimateDIFFx0x1} shows that the constructed solution $x_{+W}$ to the null structure equations along $\HH$ is close to Minkowski in the space $\XX^{\mathrm{h.f.}}(\HH)$ which is sufficiently regular to apply local existence results for the characteristic Cauchy problem for the Einstein equations; see also Remark (3) after Theorem \ref{THMmain1}.)
\item The estimate \eqref{EQestimateDIFFx0x1} can be generalized to bound also higher transversal derivatives (such as $\Du\ab$) along $\HH$ and on $S_2$. Indeed, it is straight-forward to check that commutation of the null structure equations with $\Du$ does not lead to extra $D$-derivatives on the right-hand sides of the null transport equations of transversal derivatives, and thus there is no additional largeness coming from the high-frequency characteristic seed. Hence the control of the null transport equations and their nonlinearities are similar as for \eqref{EQestimateDIFFx0x1}.
\end{enumerate}

\ni We now turn to a preliminary discussion of the proof of \eqref{EQestimateDIFFx0x1}. By the null transport equations \eqref{EQfirstvariation001} and \eqref{EQtransportEQLnullstructurenonlinear},
\begin{align*}
\begin{aligned}
\Om\chi=\half D\gd, \,\,  \Om \ab  = -D\chih +\Om \vert \chih \vert^2 \gd + \om \chih,
\end{aligned}
\end{align*}
the high-frequency nature of the characteristic seed \eqref{EQhighfrequcharseed1introsection} causes $\Om\chih$ and $\a$ to be large along $\HH$,
\begin{align*}
\begin{aligned}
\vert \Om\chih \vert \sim \varep^{1/2}, \,\, \vert \a \vert \sim \varep^{-1/2}.
\end{aligned}
\end{align*}
As $\Om\chih$ lies at the core of the null transport equations of the null structure equations and their integration, it may seem surprising to nevertheless be able to bound $x$ in $\XX(S_2)$ and $\XX^{\mathrm{h.f.}}(\HH)$ to be of small size $\varep>0$. The key to proving \eqref{EQestimateDIFFx0x1} is to use that \emph{high-frequency terms improve when integrated}.

To illustrate this feature in a simple setting, let $f$ be a smooth (low-frequency) function along $\HH$ and consider first the following ODE for a scalar quantity $Q$,
\begin{align}
\begin{aligned}
DQ = f \varep^{1/2} \sin\lrpar{\frac{v}{\sqrt{\varep}}} \text{ along } \HH,
\end{aligned}\label{EQODEsimpletoy}
\end{align}
where $\varphi$ is a cut-off functions on $\HH$ vanishing at $v=1$ and $v=2$. We underline that the right-hand side source term of \eqref{EQODEsimpletoy} is large and high-frequency, and represents the way $\Om\chih$ appears linearly (for example, on the right-hand side of the null transport equation for $\eta$) or multiplied with a low-frequency term (for example, the second line of \eqref{EQomchihcompleteformula}) in null transport equations. When integrating \eqref{EQODEsimpletoy} in $v$, we observe the \emph{high frequency improvement} under integration, namely
\begin{align*}
\begin{aligned}
Q(v) = Q(1) + \varep^{1/2} \int\limits_1^v f(v') \sin\lrpar{\frac{v'}{\sqrt{\varep}}} dv' = \OO(\varep),
\end{aligned}
\end{align*}
where we integrated by parts and estimated as follows,
\begin{align*}
\begin{aligned}
\int\limits_1^v f(v') \sin\lrpar{\frac{v'}{\sqrt{\varep}}} dv' =&\sqrt{\varep} \left[ f \lrpar{-\cos\lrpar{\frac{v'}{\sqrt{\varep}}}} \right]_1^v + \sqrt{\varep}\int\limits_1^v (Df) \cos\lrpar{\frac{v'}{\sqrt{\varep}}} dv' \\
\les& \sqrt{\varep} \lrpar{\Vert f \Vert_{L^\infty(\HH)} + \Vert Df \Vert_{L^\infty(\HH)} }.
\end{aligned}
\end{align*}
For the null transport equations where the high-frequency terms appear quadratically, for example, 
\begin{itemize}
\item the Raychaudhuri equation at order $\varep$ through $\vert \Om\chih\vert^2_\gd$, or
\item the null transport equation for $\eta$ at order $\varep^2$ through $\Divd_{\gd}(\Om\chih)$, 
\end{itemize}
we observe that in the corresponding model problem, the right-hand side in the analog of \eqref{EQODEsimpletoy} is given by, for example,
\begin{align}
\begin{aligned}
\lrpar{\varep^{1/2} \sin^2\lrpar{\frac{v}{\sqrt{\varep}}} }^2 = \frac{\varep}{2} -  \frac{\varep}{2} \cos\lrpar{\frac{2v}{\sqrt{\varep}}},
\end{aligned}\label{EQsinesquaredformula}
\end{align}
where the first term on the right-hand side yields a (low-frequency) contribution of size $\varep/2$ (this is the key to adjust the charges $\Ef,\Pf,\Lf,\Gf$), and the second term improves by integration to be of smaller size $\varep^{3/2}$. With view on the ansatz \eqref{EQWansatz} for $W$ where both $\sin$ and $\cos$ appear, we note that we can similarly express
\begin{align*}
\begin{aligned}
\lrpar{\varep^{1/2} \cos^2\lrpar{\frac{v}{\sqrt{\varep}}} }^2 = \frac{\varep}{2}+  \frac{\varep}{2} \cos\lrpar{\frac{2v}{\sqrt{\varep}}}.
\end{aligned}
\end{align*}
Concerning cross-over terms in the quadratic nonlinearities, we note that the terms
\begin{align*}
\begin{aligned}
\sin\lrpar{\frac{v}{\sqrt{\varep}}}\cos\lrpar{\frac{v}{\sqrt{\varep}}}, \,\, \sin\lrpar{\frac{v}{{\varep}}}\sin\lrpar{\frac{v}{\sqrt{\varep}}}, \,\, \sin\lrpar{\frac{v}{{\varep}}}\cos\lrpar{\frac{v}{\sqrt{\varep}}},
\end{aligned}
\end{align*}
are high-frequency and yield high-frequency integration improvements of $\varep^{1/2}$.

In addition to the control \eqref{EQestimateDIFFx0x1} we prove in Section \ref{SECestimateDiffest2} the following \emph{difference estimates} which form the basis of the estimates in the iteration scheme in the proof of the main theorem, see Sections \ref{SECdiffEstimates2W} and \ref{SECgluingEP}. For given solutions $x$ and $x'$ to the null structure equations (both assumed to satisfy \eqref{EQsmallnesscondBaseCONSTRUCTION} for $\varep>0$ sufficiently small) and given tensors $W$ and $W'$ of the form \eqref{EQWansatz}-\eqref{EQchoiceofCoeffFunction3}, it holds that
\begin{align}
\begin{aligned}
\Vert (x_{+W} - x)-(x'_{+W'}-x') \Vert_{\XX(S_2)} \les \varep \Vert x - x' \Vert_{\XX(\HH)} +\varep \Vert \mfW_0-\mfW_0' \Vert + \varep^{5/4} \Vert W-W' \Vert,
\end{aligned}\label{EQtoProveIterationEstimateForX}
\end{align}
where, similarly to \eqref{EQestimateDIFFx0x1}, we tacitly suppressed the dependence of the constant. The proof of \eqref{EQtoProveIterationEstimateForX} is a generalization of the proof of \eqref{EQestimateDIFFx0x1}, based on the same idea of high-frequency improvement after integration, and standard product estimates.

\subsection{Preliminary estimates for the characteristic seed} \label{SECcontrol2} \ni In this section we calculate and estimate the quantity
\begin{align}
\begin{aligned}
e:= \vert \Om\chih\vert^2_\gd		
\end{aligned}\label{EQedef}
\end{align}
of the solution $x_{+W}$ to be constructed, which plays a crucial part in the Raychaudhuri equation \eqref{EQcompactRAY}. It is possible to calculate $e$ from the characteristic seed due to the conformal invariance (see \eqref{EQshearINVARIANCE}),
\begin{align*} 
\begin{aligned}
e:= \vert \Om\chih \vert^2_{\gd} = \vert \Om\widehat{\tilde{\chi}} \vert^2_{\tilde{\gd}}=: \tilde{e},
\end{aligned}
\end{align*}
where $\tilde{e}$ is directly calculated from $(\Om,\tilde{\gd})$ along $\HH$. 

In the following we calculate and estimate $\tilde{e}$, see \eqref{EQomchihcompleteformula} below, and prove that in $L^\infty_vH^6(S_v)$,
\begin{align}
\begin{aligned}
\int\limits_1^v \lrpar{\vert \Om\widehat{\tilde{\chi}} \vert^2_{\tilde{\gd}}- \vert (\Om\widehat{\tilde{\chi}} )_x\vert^2_{\gd_x}} = \varep \int\limits_1^v \frac{\varphi^2}{8r^4} \lrpar{\vert \mfW_0\vert^2_\gac}+ \varep^{3/2} \int\limits_1^v \frac{\varphi^2 }{8r^4}\lrpar{ \vert \mfW_1\vert^2_\gac+\vert \mfW_2\vert^2_\gac} + \OO(\varep^{7/4}\Vert W \Vert).
\end{aligned}\label{EQomchih2expr}
\end{align}

\begin{remark} \label{REMdiffWEIGHTS}
In the following sections we use that analogous versions of \eqref{EQomchih2expr} also hold for different $r$-weights, for example,
\begin{align*}
	\begin{aligned}
		\int\limits_1^v r \lrpar{\vert \Om\widehat{\tilde{\chi}} \vert^2_{\tilde{\gd}}- \vert (\Om\widehat{\tilde{\chi}} )_x\vert^2_{\gd_x}} = \varep \int\limits_1^v \frac{\varphi^2}{8r^3} \lrpar{\vert \mfW_0\vert^2_\gac}+ \varep^{3/2} \int\limits_1^v \frac{\varphi^2 }{8r^3}\lrpar{ \vert \mfW_1\vert^2_\gac+\vert \mfW_2\vert^2_\gac} + \OO(\varep^{7/4}\Vert W \Vert).
	\end{aligned}
\end{align*}
The proof is identical to \eqref{EQomchih2expr} and thus omitted.
\end{remark}

First, from \eqref{EQrep1} and \eqref{EQWansatz} we directly calculate that in $L^\infty_v H^6(S_v)$,
\begin{align*}
\begin{aligned}
D\tilde{\gd} = D\gd_x + \varphi \lrpar{\varep^{1/2} \sin\lrpar{\frac{v}{\varep}} \mfW_0 + \varep^{3/4} \sin\lrpar{\frac{v}{\sqrt{\varep}}} \mfW_1 + \varep^{3/4} \cos\lrpar{\frac{v}{\sqrt{\varep}}} \mfW_2} + \OO(\varep^{5/4}\Vert W \Vert).
\end{aligned}
\end{align*}
Taking the trace yields that in $L^\infty_v H^6(S_v)$,
\begin{align*}
\begin{aligned}
\tr_{\tilde{\gd}}(D\tilde{\gd}) = \tr_{{\gd}_x}(D{\gd}_x) + \OO(\varep^{5/4}\Vert W \Vert),
\end{aligned}
\end{align*}
where we used \eqref{EQsmallnesscondBaseCONSTRUCTION} and that by \eqref{EQsmallnesscondBaseCONSTRUCTION} in $L^\infty_v H^6(S_v)$,
\begin{align}
\begin{aligned}
\tilde{\gd}^{AB} = \gd_x^{AB} + \OO(\varep^{5/4}\Vert W \Vert).
\end{aligned}\label{EQinverseMetricExpansion}
\end{align}
From the above we calculate that in $L^\infty_v H^6(S_v)$,
\begin{align}
\begin{aligned}
\Om\widehat{\tilde{\chi}} =& \half \lrpar{D\tilde{\gd}-\half \tr_{\tilde{\gd}}(D\tilde{\gd})\tilde{\gd}}\\
=& \frac{\varphi}{2} \lrpar{\varep^{1/2} \sin\lrpar{\frac{v}{\varep}} \mfW_0 + \varep^{3/4} \sin\lrpar{\frac{v}{\sqrt{\varep}}} \mfW_1 + \varep^{3/4} \cos\lrpar{\frac{v}{\sqrt{\varep}}} \mfW_2} \\
&+ (\Om\widehat{\chi})_x+\OO(\varep^{5/4}\Vert W \Vert).
\end{aligned}\label{EQOmchihtildeexpr}
\end{align}
From \eqref{EQOmchihtildeexpr} we get, using \eqref{EQinverseMetricExpansion}, that in $L^\infty_vH^6(S_v)$,
\begin{align}
\begin{aligned}
\vert \Om\widehat{\tilde{\chi}} \vert^2_{\tilde{\gd}}=& \frac{\varphi^2 \varep}{4r^4} \lrpar{\sin^2\lrpar{\frac{v}{\varep}}\vert \mfW_0\vert^2_\gac}+ \frac{\varphi^2 \varep^{3/2}}{4r^4}\lrpar{ \sin^2\lrpar{\frac{v}{\varep^{1/2}}}\vert \mfW_1\vert^2_\gac+\cos^2\lrpar{\frac{v}{\varep^{1/2}}}\vert \mfW_2\vert^2_\gac} \\
&+ (\Om\chih)_x\cdot \varphi \lrpar{\varep^{1/2} \sin\lrpar{\frac{v}{\varep}} \mfW_0 + \varep^{3/4} \sin\lrpar{\frac{v}{\varep^{1/2}}} \mfW_1 + \varep^{3/4} \cos\lrpar{\frac{v}{\varep^{1/2}}} \mfW_2} \\
&+ \frac{\varphi^2}{2} \lrpar{\varep^{1/2} \sin\lrpar{\frac{v}{\varep}} \mfW_0} \cdot \lrpar{\varep^{3/4} \sin\lrpar{\frac{v}{\varep^{1/2}}} \mfW_1 + \varep^{3/4} \cos\lrpar{\frac{v}{\varep^{1/2}}} \mfW_2} \\
&+ \frac{\varphi^2}{2} \lrpar{\varep^{3/4} \sin\lrpar{\frac{v}{\varep^{1/2}}} \mfW_1}\cdot\lrpar{\varep^{3/4} \cos\lrpar{\frac{v}{\varep^{1/2}}} \mfW_2} + \OO(\varep^{7/4}\Vert W \Vert) + \vert (\Om\chih)_x\vert^2_{\gd_x}.
\end{aligned}\label{EQomchihcompleteformula}
\end{align}
Integrating \eqref{EQomchihcompleteformula} and applying the remarks of Section \ref{SECpreliminariesHIGHfrequency} leads to \eqref{EQomchih2expr}.

\subsection{Construction of $\gd$} \label{SECphiEstimatesPrecise} Having calculated and estimated $e$ in the previous section (see \eqref{EQedef}), we are now in position to construct and estimate the metric $\gd$ along $\HH$. Recall from \eqref{EQconfDecompintro} that $\gd$ is given through its conformal decomposition,
\begin{align*}
\begin{aligned}
\gd = \phi^2 \gd_c,
\end{aligned}
\end{align*}
where $\gd_c$ is defined as the unique representative of $\mathrm{conf}(\gd)$ along $\HH$ such that $\det \gd_c=\det\gac$, which in our setting is explicitly given by
\begin{align}
\begin{aligned}
\gd_c = \sqrt{\det\gac} \sqrt{\det\tilde{\gd}}^{-1} \tilde{\gd},
\end{aligned}\label{EQgdcDEF}
\end{align}
and the conformal factor $\phi$ is determined by the following linear second-order ODE (stemming from combining \eqref{EQcompactRAYpart1} and \eqref{EQfirstvariation001part2}),
\begin{align}
\begin{aligned}
D^2 \phi - 2 \om D\phi +\half \vert \Om\chih \vert^2 \phi =0.
\end{aligned}\label{EQntePHIestimatesection01}
\end{align}
The coefficients in \eqref{EQntePHIestimatesection01} are calculated from the characteristic seed along $\HH$, and the initial values $(\phi\vert_{S_1}, D\phi \vert_{S_1})$ for \eqref{EQntePHIestimatesection01} are given by the characteristic seed on $S_1$ through \eqref{EQconfDecompintro} and \eqref{EQcompactRAYpart1},
\begin{align*}
\begin{aligned}
\gd = \phi^2\gd_c, \,\, D\phi = \frac{\Om\trchi \phi}{2} \text{ on } S_1.
\end{aligned}
\end{align*}
In the following we first derive estimates for $\phi$, and then prove bounds for $\gd_c$ and $\gd$.\\

\ni \textbf{Estimates for $\phi$.} We claim that $\phi$ is bounded by
\begin{align}
\begin{aligned}
\Vert \phi-\phi_x \Vert_{H^6_2(\HH)} \les \varep \Vert \mfW_0 \Vert + \varep^{5/4} \Vert W \Vert.
\end{aligned}\label{EQphidiffest1}
\end{align}
and moreover, $\phi$ satisfies the following representation formula in $H^6_2(\HH)$,
\begin{align}
\begin{aligned}
\phi-\phi_x= -\varep \int\limits_1^v\int\limits_1^{v'} \frac{\varphi^2}{16r^3} \vert \mfW_0\vert_{\gac}^2 - \varep^{3/2}\int\limits_1^v\int\limits_1^{v'} \frac{\varphi^2}{16r^3} \lrpar{\vert \mfW_1\vert_{\gac}^2+\vert \mfW_2\vert_{\gac}^2}+\OO\lrpar{\varep^{7/4} \Vert W \Vert}.
\end{aligned}\label{EQrepformulaphi}
\end{align}
The representation formula \eqref{EQrepformulaphi} plays an important role in Section \ref{SECgluingL} where it is used that $\phi-\phi_x$ is low-frequency up to an error $\OO(\varep^{7/4}\Vert W \Vert)$.

We first rewrite the ODE \eqref{EQntePHIestimatesection01} in the form \eqref{EQcompactRAY}, that is,
\begin{align*}
\begin{aligned}
D\lrpar{\Om^{-2}D\phi} = -\frac{\phi}{2\Om^2} \vert \Om\chih \vert^2_{\gd}.
\end{aligned}
\end{align*}
To prove \eqref{EQphidiffest1} and \eqref{EQrepformulaphi} consider the ODE for the difference $\phi-\phi_x$,
\begin{align*}
\begin{aligned}
D\lrpar{\Om^{-2} D\lrpar{\phi-\phi_x}} =& -\frac{1}{2\Om^2} \lrpar{\phi \vert\Om\chih\vert^2_\gd - \phi_x \vert(\Om\chih)_x\vert^2_{\gd_x}} \\
=& -\frac{1}{2\Om^2} \lrpar{(\phi-\phi_x) \vert\Om\chih\vert^2_\gd + \phi_x \lrpar{\vert\Om\chih\vert^2_\gd-\vert(\Om\chih)_x\vert^2_{\gd_x}}},
\end{aligned}
\end{align*}
and integrate in $v$ twice to get
\begin{align}
\begin{aligned}
\lrpar{\phi-\phi_x}\vert_{S_v} = - \int\limits_1^v \Om^2 \int\limits_1^{v'} \frac{1}{2\Om^2} \lrpar{(\phi-\phi_x) \vert\Om\chih\vert^2_\gd + \phi_x \lrpar{\vert\Om\chih\vert^2_\gd-\vert(\Om\chih)_x\vert^2_{\gd_x}}} dv''dv',
\end{aligned}\label{EQrepFormulaDiffPHI}
\end{align}
where we used that by construction, $\phi=\phi_x$ and $D\phi=D\phi_x$ on $S_1$.

Using \eqref{EQsmallnesscondBaseCONSTRUCTION}, \eqref{EQomchihcompleteformula} and  \eqref{EQrepFormulaDiffPHI}, a straight-forward bootstrapping argument proves \eqref{EQphidiffest1}, i.e.
\begin{align*}
\begin{aligned}
\Vert \phi-\phi_x \Vert_{H^6_2(\HH)} \les \varep \Vert W \Vert.
\end{aligned}
\end{align*}
Having \eqref{EQphidiffest1}, we can apply a version of \eqref{EQomchih2expr} (see Remark \ref{REMdiffWEIGHTS}) to \eqref{EQrepFormulaDiffPHI} to get that in $H^6_1(\HH)$,
\begin{align*}
\begin{aligned}
\lrpar{\phi-\phi_x}\vert_{S_v} =& - \int\limits_1^v \Om^2 \int\limits_1^{v'} \frac{1}{2\Om^2} \lrpar{(\phi-\phi_x) \vert\Om\chih\vert^2_\gd + \phi_x \lrpar{\vert\Om\chih\vert^2_\gd-\vert(\Om\chih)_x\vert^2_{\gd_x}}} dv''dv' \\
=&- \int\limits_1^v \int\limits_1^{v'} \frac{r}{2} \lrpar{\vert\Om\chih\vert^2_\gd-\vert(\Om\chih)_x\vert^2_{\gd_x}} dv''dv' + \OO(\varep^{7/4}\Vert W \Vert)\\
=& -\varep \int\limits_1^v\int\limits_1^{v'} \frac{\varphi^2}{16r^3} \vert \mfW_0\vert_{\gac}^2 - \varep^{3/2}\int\limits_1^v\int\limits_1^{v'} \frac{\varphi^2}{16r^3} \lrpar{\vert \mfW_1\vert_{\gac}^2+\vert \mfW_2\vert_{\gac}^2}+\OO\lrpar{\varep^{7/4} \Vert W \Vert}.
\end{aligned}
\end{align*}
This finishes the proof of \eqref{EQrepformulaphi}.\\

\ni \textbf{Estimate for $\gd-\gd_x$.} Next we show that
\begin{align}
\Vert \gd-\gd_x \Vert_{L^\infty_vH^6(S_v)} \les \varep \Vert \mfW_0 \Vert + \varep^{5/4} \Vert W \Vert. \label{EQestimationdifferencegd}
\end{align}
Indeed, by \eqref{EQconfDecompintro} and \eqref{EQgdcDEF} it holds that 
\begin{align}
\begin{aligned}
\gd = \phi^2 \gd_c = \phi^2 \sqrt{\det\gac}\sqrt{\det\tilde{\gd}}^{-1} \tilde{\gd},
\end{aligned}\label{EQconformalrelation2}
\end{align}
so that
\begin{align}
\begin{aligned}
\gd-\gd_x =& \phi^2 \sqrt{\det\gac}\sqrt{\det\tilde{\gd}}^{-1} \tilde{\gd} - \gd_x \\
=& \phi^2 \sqrt{\det\gac}\sqrt{\det\tilde{\gd}}^{-1} (\tilde{\gd}-\gd_x) + \gd_x \lrpar{\phi^2 \sqrt{\det\gac}\sqrt{\det\tilde{\gd}}^{-1}-1}.
\end{aligned}\label{EQdiffgdgdx1009}
\end{align}
On the one hand, by \eqref{EQrep1}, in $L^\infty_vH^6(S_v)$,
\begin{align}
\begin{aligned}
\tilde{\gd}-\gd_x = \varphi W = \OO(\varep^{5/4}\Vert W \Vert).
\end{aligned}\label{EQrecallDifferencecontrolsection0091}
\end{align}
On the other hand, in $L^\infty_vH^6(S_v)$,
\begin{align}
\begin{aligned}
&\phi^2 \sqrt{\det\gac}\sqrt{\det\tilde{\gd}}^{-1}-1 \\
=& \phi^2 \sqrt{\det\gac}\sqrt{\det\tilde{\gd}}^{-1}- \phi_x^2 \sqrt{\det\gac}\sqrt{\det\gd_x}^{-1}\\
=& \sqrt{\det\gac} \lrpar{(\phi^2-\phi_x^2)\sqrt{\det\tgd}^{-1}+\phi_x^2 \lrpar{\sqrt{\det\tgd}^{-1}-\sqrt{\det\gd_x}^{-1}}} \\
=& \sqrt{\det\gac} \lrpar{(2\phi_x +(\phi-\phi_x))(\phi-\phi_x)\sqrt{\det\tgd}^{-1}+\phi_x^2 \lrpar{\sqrt{\det\tgd}^{-1}-\sqrt{\det\gd_x}^{-1}}}\\
=&\frac{2}{r}(\phi-\phi_x) + \OO(\varep^{7/4} \Vert W \Vert)\\
=& \frac{2}{r} \lrpar{-\varep \int\limits_1^v\int\limits_1^{v'} \frac{\varphi^2}{16r^3} \vert \mfW_0\vert_{\gac}^2 - \varep^{3/2}\int\limits_1^v\int\limits_1^{v'} \frac{\varphi^2}{16r^3} \lrpar{\vert \mfW_1\vert_{\gac}^2+\vert \mfW_2\vert_{\gac}^2}}+\OO\lrpar{\varep^{7/4} \Vert W \Vert}.
\end{aligned}\label{EQexpansionConformalFactor0091}
\end{align}
where we used that, as $x$ is a solution to the null structure equations,
\begin{align*}
\begin{aligned}
\phi_x^2 \sqrt{\det\gac}\sqrt{\det\gd_x}^{-1}=1,
\end{aligned}
\end{align*}
and moreover, as $\tr_{r^2\gac}W=0$, in $L^\infty_v H^6(S_v)$,
\begin{align}
\begin{aligned}
\sqrt{\det\tgd}^{-1}-\sqrt{\det\gd_x}^{-1} = \OO(\varep^{9/4} \Vert W \Vert).
\end{aligned}\label{EQdetdiffestimate}
\end{align}
Plugging \eqref{EQrecallDifferencecontrolsection0091} and \eqref{EQexpansionConformalFactor0091} into \eqref{EQdiffgdgdx1009} finishes the proof of \eqref{EQestimationdifferencegd}.\\

\ni \textbf{Estimate for $\gd_c-(\gd_c)_x$.} In the following we show that
\begin{align}
	\begin{aligned}
		\Vert \gd_c-(\gd_c)_x \Vert_{L^\infty_v H^6(S_v)} \les \varep^{5/4} \Vert W \Vert.
	\end{aligned}\label{EQgdcdiffest1}
\end{align}
Indeed, by \eqref{EQgdcDEF} we have that
\begin{align*}
\begin{aligned}
\gd_c-(\gd_c)_x =& \sqrt{\det\gac} \sqrt{\det\tilde{\gd}}^{-1} \tilde{\gd} - \sqrt{\det\gac} \sqrt{\det{\gd_x}}^{-1} {\gd_x}\\ 
=& \sqrt{\det\gac}\lrpar{ \lrpar{\sqrt{\det\tilde{\gd} }^{-1}-\sqrt{\det\gd_x}}\tilde{\gd}+\sqrt{\det\gd_x}\lrpar{\tilde{\gd}-\gd_x} }.
\end{aligned}
\end{align*}
From here, applying \eqref{EQsmallnesscondBaseCONSTRUCTION} and \eqref{EQrep1} leads to \eqref{EQgdcdiffest1}.\\

\ni \textbf{Estimate for the Gauss curvature $K$.} The Gauss curvature $K$ of $\gd$ can be explicitly expressed in terms of (derivatives up to order $2$ of) the metric components $\gd$, and can thus by \eqref{EQestimationdifferencegd} be estimated as follows,
\begin{align}
\begin{aligned}
\Vert K - K_x \Vert_{L^\infty_vH^4(S_v)} \les \varep\Vert \mfW_0 \Vert + \varep^{5/4} \Vert W \Vert.
\end{aligned}\label{EQestimateDiffKcurvature}
\end{align}

\subsection{Control of $\Om\tr\chi$ and $\Om\chih$} \label{SECestimatesOMTRCHI} By the first variation equation \eqref{EQfirstvariation001}, $\Om\chi$ can be directly calculated from the metric $\gd$. In this section we first bound $\Om\trchi$ and then $\Om\chih$.\\

\ni \textbf{Estimates for $\Om\trchi$.} In the following we show that
\begin{align}
\begin{aligned}
\Vert \Om\trchi-(\Om\trchi)_x \Vert_{H^6_1(\HH)} \les \varep \Vert \mfW_0 \Vert + \varep^{5/4} \Vert W \Vert.
\end{aligned}\label{EQomtrchidiffsimpleestimate}
\end{align}
We prove \eqref{EQomtrchidiffsimpleestimate} by studying the transport equation for $\frac{\phi \Om\trchi}{\Om^2}$ and applying the previous estimates \eqref{EQphidiffest1}, \eqref{EQrepformulaphi} for $\phi$. Indeed, by \eqref{EQcompactRAY} we have that
\begin{align}
\begin{aligned}
D\lrpar{\frac{\phi\Om\trchi}{\Om^2}}=&-\frac{\phi}{\Om^2} \vert \Om\chih \vert^2,
\end{aligned}\label{EQomtrchinullconstrainteq0011}
\end{align}
so that, using that $\Om=\Om_x$ by \eqref{EQomegadef},
\begin{align}
\begin{aligned}
D\lrpar{\frac{\phi\Om\trchi}{\Om^2}- \lrpar{\frac{\phi\Om\trchi}{\Om^2}}_x}=&-\frac{1}{\Om^2} \lrpar{\phi \vert \Om\chih \vert^2-\phi_x \vert (\Om\chih)_x \vert^2} \\
=& -\frac{1}{\Om^2} \lrpar{(\phi-\phi_x) \vert \Om\chih \vert^2+ \phi_x\lrpar{\vert \Om\chih \vert^2-  \vert (\Om\chih)_x \vert^2}},
\end{aligned}\label{EQNTEomtrchiphiom21}
\end{align}
Integrating \eqref{EQNTEomtrchiphiom21} in $v$ and applying \eqref{EQomchih2expr} (see also Remark \ref{REMdiffWEIGHTS}), \eqref{EQomchihcompleteformula} and \eqref{EQphidiffest1}, we get that in $L^\infty_vH^6(S_v)$,
\begin{align}
\begin{aligned}
\frac{\phi\Om\trchi}{\Om^2}- \lrpar{\frac{\phi\Om\trchi}{\Om^2}}_x =& - \int\limits_1^v \frac{1}{\Om^2} \lrpar{(\phi-\phi_x) \vert \Om\chih \vert^2+ \phi_x\lrpar{\vert \Om\chih \vert^2-  \vert (\Om\chih)_x \vert^2}}\\
=& -\int\limits_1^v r \lrpar{\vert \Om\chih \vert^2-  \vert (\Om\chih)_x \vert^2} dv' + \OO(\varep^2\Vert W \Vert)\\
=& - \varep \int\limits_1^v \frac{\varphi^2}{8r^3} \vert \mfW_0 \vert^2_\gac - \varep^{3/2} \int\limits_1^v \frac{\varphi^2}{8r^3} \lrpar{\vert \mfW_1\vert^2_\gac + \vert \mfW_2\vert^2_\gac} dv' + \OO(\varep^{7/4}\Vert W \Vert),
\end{aligned}\label{EQrepformulaRaychauduri}
\end{align}
from which we conclude, in particular, that
\begin{align*}
\begin{aligned}
\left\Vert \frac{\phi\Om\trchi}{\Om^2}- \lrpar{\frac{\phi\Om\trchi}{\Om^2}}_x \right\Vert_{L^\infty_v H^6(S_v)} \les \varep \Vert \mfW_0 \Vert + \varep^{5/4} \Vert W \Vert.
\end{aligned}
\end{align*}

\ni From the above it is straight-forward to conclude \eqref{EQomtrchidiffsimpleestimate} by writing 
\begin{align*}
\begin{aligned}
\frac{\phi\Om\trchi}{\Om^2}- \lrpar{\frac{\phi\Om\trchi}{\Om^2}}_x =& \frac{1}{\Om^2} \lrpar{\phi\Om\trchi-\phi_x(\Om\trchi)_x} \\
=& \frac{1}{\Om^2} \lrpar{(\phi-\phi_x) (\Om\trchi)_x + \phi \lrpar{\Om\trchi-(\Om\trchi)_x}},
\end{aligned}
\end{align*}
and isolating the term $\Om\trchi-(\Om\trchi)_x$. Here we divide by $\phi$ and use the estimates \eqref{EQphidiffest1}, \eqref{EQrepformulaphi}. \\

\ni\textbf{Estimate for $\Om\chih$.} In the following we prove that in $L^\infty_vH^6(S_v)$,
\begin{align}
\begin{aligned}
\Om\chih - (\Om\chih)_x =& \frac{\varphi}{2} \lrpar{\varep^{1/2} \sin\lrpar{\frac{v}{\varep}} \mfW_0 + \varep^{3/4} \sin\lrpar{\frac{v}{\varep^{1/2}}} \mfW_1 + \varep^{3/4} \cos\lrpar{\frac{v}{\varep^{1/2}}} \mfW_2} \\
& + \OO(\varep^{5/4}\Vert W \Vert).
\end{aligned}\label{EQexpansionomchih}
\end{align}
Indeed, using \eqref{EQconformalrelation2} we have that
\begin{align}
\begin{aligned}
\Om\chih= \half \widehat{D\gd}= \phi^2 \sqrt{\det\gac}\sqrt{\det\tilde{\gd}}^{-1} \Om\widehat{\tilde{\chi}}.
\end{aligned}\label{EQfirstexprOmchihexprdiff}
\end{align}
Using \eqref{EQOmchihtildeexpr}, \eqref{EQphidiffest1} and that by \eqref{EQsmallnesscondBaseCONSTRUCTION} and \eqref{EQrep1}, with $\tr_{r^2\gac} W =0$,
\begin{align}
\begin{aligned}
\left\Vert \sqrt{\det\tilde{\gd}}-\sqrt{\det\gd_x}\right\Vert_{L^\infty_vH^6(S_v)} \les \varep^{9/4} \Vert W \Vert,
\end{aligned}\label{EQdeterminantEstimatecloseness0009}
\end{align}
we get from \eqref{EQfirstexprOmchihexprdiff} that
\begin{align*}
\begin{aligned}
\Om\chih =& \frac{\varphi}{2} \lrpar{\varep^{1/2} \sin\lrpar{\frac{v}{\varep}} \mfW_0 + \varep^{3/4} \sin\lrpar{\frac{v}{\varep^{1/2}}} \mfW_1 + \varep^{3/4} \cos\lrpar{\frac{v}{\varep^{1/2}}} \mfW_2} + (\Om\chih)_x \\
&+ \OO(\varep^{5/4}\Vert W \Vert).
\end{aligned}
\end{align*}
This finishes the proof of \eqref{EQexpansionomchih}. 

Moreover, we note that on $S_2$ it holds by construction that (see \eqref{EQomegadef}, \eqref{EQrep1} and recall that $W$ is compactly supported in $\HH$)
\begin{align}
\Om=\Om_x, \,\, \om=\om_x, \,\, \tilde{\gd}=\gd_x, \,\, D\tilde{\gd}= D\gd_x  \text{ on } S_2, 
\end{align}
so that, using \eqref{EQfirstexprOmchihexprdiff} and \eqref{EQphidiffest1}, in $H^6(S_2)$,
\begin{align}
\begin{aligned}
\Om\chih-(\Om\chih)_x = \half \lrpar{\phi^2-\phi^2_x} \sqrt{\det\gac}\sqrt{\det\gd_x}^{-1}\widehat{D\gd_x}
= \OO_{x,\mfW_0}(\varep^2 \Vert\mfW_0\Vert)+ \OO(\varep^{9/4} \Vert W\Vert).
\end{aligned}\label{EQomchihS2identity}
\end{align}

\subsection{Control of $\a$ on $S_2$} By \eqref{EQtransportEQLnullstructurenonlinear} we calculate $\a$ along $\HH$ as follows,
\begin{align*}
\begin{aligned}
\a = \frac{1}{\Om}\lrpar{-D\chih + \om \chih + \Om \vert \chih \vert^2 \gd}= \frac{1}{\Om^2} \lrpar{-D\lrpar{\Om\chih} + 2\om (\Om\chih)+  \vert \Om\chih \vert^2 \gd}.
\end{aligned}
\end{align*}
Using that by \eqref{EQomegadef}, \eqref{EQrep1} and the compact support of $\varphi W$ in $\HH$ we have that in a neighbourhood of $S_2$ on $\HH$ it holds that
\begin{align}
\Om=\Om_x, \,\, \om=\om_x, \,\, \tilde{\gd}=\gd_x, \,\, D\tilde{\gd}= D\gd_x, \,\, D^2\tilde{\gd}= D^2\gd_x, \label{EQequalitys2}
\end{align}
so that on $S_2$, in $H^6(S_2)$,
\begin{align}
\begin{aligned}
\a -\a_x =& \frac{1}{\Om^2} \lrpar{\lrpar{-D\lrpar{\Om\chih} + 2\om (\Om\chih)+  \vert \Om\chih \vert^2 \gd}-\lrpar{-D\lrpar{\Om\chih}_x + 2\om (\Om\chih)_x+  \vert (\Om\chih)_x \vert^2_{\gd_x} \gd_x}}\\
=& \frac{1}{\Om^2} \Big(-D\lrpar{\Om\chih-(\Om\chih)_x}+2\om \lrpar{\Om\chih-(\Om\chih)_x} +\vert (\Om\chih)_x \vert^2_{\gd_x} (\gd-\gd_x) \Big),
\end{aligned}\label{EQdiffformulaaax}
\end{align}
where we used that by the conformal invariance of $e$ and \eqref{EQequalitys2}, on $S_2$,
\begin{align*}
\begin{aligned}
\vert \Om\chih \vert_\gd^2 = \vert \Om\widehat{\tilde{\chi}} \vert^2_{\tilde{\gd}} = \vert (\Om\chih)_x \vert^2_{\gd_x}. 
\end{aligned}
\end{align*}
By \eqref{EQequalitys2} we can write (see also \eqref{EQomchihS2identity}), near $S_2$ on $\HH$,
\begin{align*}
\begin{aligned}
\Om\chih-(\Om\chih)_x = \half \lrpar{\phi^2-\phi^2_x} \sqrt{\det\gac}\sqrt{\det\gd_x}^{-1}\widehat{D\gd_x}
\end{aligned}
\end{align*}
so that in particular, on $S_2$, by \eqref{EQphidiffest1}, in $H^6(S_2)$,
\begin{align}
\begin{aligned}
D\lrpar{\Om\chih-(\Om\chih)_x} =& \half D\lrpar{\phi^2-\phi^2_x} \sqrt{\det\gac}\sqrt{\det\gd_x}^{-1}\widehat{D\gd_x} \\
&+ \half \lrpar{\phi^2-\phi^2_x} D\lrpar{\sqrt{\det\gac}\sqrt{\det\gd_x}^{-1}\widehat{D\gd_x}} \\
=& \OO(\varep^2 \Vert \mfW_0\Vert) + \OO(\varep^{9/4} \Vert W \Vert).
\end{aligned}\label{EQdiffdomchih}
\end{align}

\ni Plugging \eqref{EQdiffdomchih} into \eqref{EQdiffformulaaax} and using \eqref{EQphidiffest1}, \eqref{EQestimationdifferencegd} and \eqref{EQomchihS2identity} leads to, in $H^6(S_2)$,
\begin{align}
\begin{aligned}
\a -\a_x= \OO(\varep^2 \Vert\mfW_0\Vert) + \OO(\varep^{9/4} \Vert W\Vert).
\end{aligned}\label{EQalphaS2estimate}
\end{align}

\subsection{Construction of $\eta$} \label{SECcontrolETA} 
\ni By combining \eqref{EQtransportEQLnullstructurenonlinear}, \eqref{EQcompactRAYpart1} and \eqref{EQGaussCodazzi}, the following null transport equation for $\eta$ holds,
\begin{align}
\begin{aligned}
D\lrpar{\phi^2\eta} =& \phi^2 \lrpar{\Divd \lrpar{\Om\chih} -\half \di\lrpar{\Om\trchi} +2\Om\trchi\di\log\Om}.
\end{aligned}\label{EQnulleta}
\end{align}
In this section we prove that 
\begin{align}
\begin{aligned}
\Vert\eta-\eta_x\Vert_{L^\infty_vH^5(S_v)} \les \varep \Vert \mfW_0 \Vert + \varep^{5/4} \Vert W \Vert,
\end{aligned}\label{EQetaNORMestimatediff}
\end{align}
and moreover, the following representation formula holds in $L^\infty_vH^5(S_v)$, 
\begin{align}
\begin{aligned}
\eta-\eta_x=&-\frac{\varep}{2r^2} \di \lrpar{\int\limits_1^v r^2 \lrpar{-\frac{1}{r} \int\limits_1^{v'} \frac{\varphi^2}{8r^3} \lrpar{\vert \mfW_0\vert^2_\gac} dv''+\frac{2}{r^2} \int\limits_1^{v'}  \int\limits_1^{v''} \frac{\varphi^2}{16r^4} \lrpar{\vert \mfW_0\vert^2_\gac} dv'''dv''}dv'}\\
&  + \OO(\varep^{5/4}\Vert W \Vert).
\end{aligned}\label{EQdiffestimateeta} 
\end{align}
To prove \eqref{EQetaNORMestimatediff} and \eqref{EQdiffestimateeta}, consider the null transport equation for the difference, i.e.
\begin{align*}
\begin{aligned}
D\lrpar{\phi^2 \eta - \phi^2_x \eta_x} =& \phi^2 \lrpar{\Divd_\gd \lrpar{\Om\chih} -\half \di\lrpar{\Om\trchi} +2\Om\trchi\di\log\Om} \\
&- \phi_x^2 \lrpar{\Divd_{\gd_x} \lrpar{\Om\chih}_x -\half \di\lrpar{\Om\trchi}_x +2(\Om\trchi)_x\di\log\Om} \\
=& \phi^2 \Divd_\gd (\Om\chih) - \phi_x^2 \Divd_{\gd_x} (\Om\chih)_x -\half \lrpar{\phi^2 \di(\Om\trchi) - \phi^2_x \di(\Om\trchi)_x} \\
&+ 2 \di\log\Om \cdot\lrpar{\phi^2\Om\trchi-\phi^2_x (\Om\trchi)_x},
\end{aligned}
\end{align*}
which, integrated along $v$, and writing
\begin{align*}
\begin{aligned}
\phi^2 \eta - \phi^2_x \eta_x = (\phi^2 - \phi^2_x)\eta_x + \phi^2_x (\eta-\eta_x) = (\phi^2 - \phi^2_x)\eta_x + \lrpar{\phi^2_x -r^2} (\eta-\eta_x)+ r^2 (\eta-\eta_x),
\end{aligned}
\end{align*}
yields the representation formula
\begin{align}
\begin{aligned}
r^2 (\eta-\eta_x)=&- (\phi^2 - \phi^2_x)\eta_x - (\phi^2_x-r^2) (\eta-\eta_x) + \underbrace{\int\limits_1^v \lrpar{\phi^2 \Divd_\gd (\Om\chih) - \phi_x^2 \Divd_{\gd_x} (\Om\chih)_x}dv'}_{:=\II_1} \\
&-\half \underbrace{\int\limits_1^v \lrpar{\phi^2 \di(\Om\trchi) - \phi^2_x \di(\Om\trchi)_x} dv'}_{:=\II_2} + 2 \underbrace{\int\limits_1^v \di\log\Om\cdot \lrpar{\phi^2\Om\trchi-\phi^2_x (\Om\trchi)_x}}_{:=\II_3}.
\end{aligned}\label{EQrepformeta}
\end{align}
\textbf{Control of $\II_1$.} We can rewrite $\II_1$ as follows,
\begin{align*}
\begin{aligned}
\II_1 =& \int\limits_1^v \lrpar{ \lrpar{\phi^2-\phi^2_x} \Divd_\gd (\Om\chih) + \phi^2_x \lrpar{\Divd_{\gd_x}\lrpar{\Om\chih-(\Om\chih)_x} + \lrpar{\Divd_\gd(\Om\chih)-\Divd_{\gd_x}(\Om\chih) }}}.
\end{aligned}
\end{align*}
On the one hand, from \eqref{EQsmallnesscondBaseCONSTRUCTION}, \eqref{EQexpansionomchih} and \eqref{EQphidiffest1} we directly have that in $L^\infty_vH^5(S_v)$,
\begin{align*}
\begin{aligned}
\lrpar{\phi^2-\phi^2_x} \Divd_\gd (\Om\chih) = \OO(\varep^{5/4} \Vert W \Vert), \,\, \Divd_\gd(\Om\chih)-\Divd_{\gd_x}(\Om\chih) = \OO(\varep^{5/4} \Vert W \Vert).
\end{aligned}
\end{align*}
On the other hand, by \eqref{EQexpansionomchih} and the high-frequency improvement when integrating, in $L^\infty_vH^5(S_v)$,
\begin{align*}
\begin{aligned}
\int\limits_1^v \phi^2_x \Divd_{\gd_x}\lrpar{\Om\chih-(\Om\chih)_x} = \OO(\varep^{5/4} \Vert W \Vert).
\end{aligned}
\end{align*}
We conclude from the above that in $L^\infty_vH^5(S_v)$,
\begin{align*}
\begin{aligned}
\II_1 = \OO(\varep^{5/4} \Vert W \Vert).
\end{aligned}
\end{align*}

\ni \textbf{Control of $\II_2$.} We have that 
\begin{align*}
\begin{aligned}
\II_2 =& \int\limits_1^v \lrpar{ \lrpar{\phi^2-\phi_x^2} \di(\Om\trchi) + \phi^2_x \di\lrpar{\Om\trchi-(\Om\trchi)_x}} dv'.
\end{aligned}
\end{align*}
On the one hand, by \eqref{EQsmallnesscondBaseCONSTRUCTION}, \eqref{EQphidiffest1} and \eqref{EQomtrchidiffsimpleestimate}, we have that in $L^\infty_vH^5(S_v)$,
\begin{align*}
\begin{aligned}
\lrpar{\phi^2-\phi_x^2} \di(\Om\trchi) = \OO(\varep^2\Vert W \Vert).
\end{aligned}
\end{align*}
On the other hand, by \eqref{EQsmallnesscondBaseCONSTRUCTION} and \eqref{EQomtrchidiffsimpleestimate}, in $L^\infty_vH^5(S_v)$,
\begin{align*}
\begin{aligned}
&\int\limits_1^v \phi^2_x \di\lrpar{\Om\trchi-(\Om\trchi)_x} dv' \\
=& \varep \cdot \di \lrpar{\int\limits_1^v r^2 \lrpar{-\frac{1}{r} \int\limits_1^{v'} \frac{\varphi^2}{8r^3} \lrpar{\vert \mfW_0\vert^2_\gac} dv''+\frac{2}{r^2} \int\limits_1^{v'}  \int\limits_1^{v''} \frac{\varphi^2}{16r^4} \lrpar{\vert \mfW_0\vert^2_\gac} dv'''dv''}dv'} +\OO(\varep^{5/4} \Vert W \Vert).
\end{aligned}
\end{align*}
From the above we conclude that in $L^\infty_vH^5(S_v)$,
\begin{align*}
\begin{aligned}
\II_2 =&\varep \cdot \di \lrpar{\int\limits_1^v r^2 \lrpar{-\frac{1}{r} \int\limits_1^{v'} \frac{\varphi^2}{8r^3} \lrpar{\vert \mfW_0\vert^2_\gac} dv''+\frac{2}{r^2} \int\limits_1^{v'}  \int\limits_1^{v''} \frac{\varphi^2}{16r^4} \lrpar{\vert \mfW_0\vert^2_\gac} dv'''dv''}dv'} +\OO(\varep^{5/4} \Vert W \Vert).
\end{aligned}
\end{align*}

\ni \textbf{Control of $\II_3$.} Similar to the estimates above it follows that in $L^\infty_v H^5(S_v)$,
\begin{align*}
\begin{aligned}
\II_3 = \OO(\varep^2\Vert W \Vert).
\end{aligned}
\end{align*}

\ni Plugging the above control of $\II_1,\II_2$ and $\II_3$ into \eqref{EQrepformeta}, and using \eqref{EQsmallnesscondBaseCONSTRUCTION} and \eqref{EQphidiffest1}, yields that in $L^\infty_vH^5(S_v)$,
\begin{align*}
\begin{aligned}
r^2 (\eta-\eta_x)=& - (\phi^2_x-r^2) (\eta-\eta_x) + \OO(\varep^{5/4}\Vert W \Vert) \\
&-\half \varep \cdot \di \lrpar{\int\limits_1^v r^2 \lrpar{-\frac{1}{r} \int\limits_1^{v'} \frac{\varphi^2}{8r^3} \lrpar{\vert \mfW_0\vert^2_\gac} dv''+\frac{2}{r^2} \int\limits_1^{v'}  \int\limits_1^{v''} \frac{\varphi^2}{16r^4} \lrpar{\vert \mfW_0\vert^2_\gac} dv'''dv''}dv'}.
\end{aligned}
\end{align*}
A simple bootstrapping argument based on the above representation formula shows \eqref{EQetaNORMestimatediff}, i.e.
\begin{align*}
\begin{aligned}
\Vert \eta-\eta_x \Vert_{L^\infty_vH^5(S_v)} \les \varep \Vert \mfW_0 \Vert + \varep^{5/4} \Vert W \Vert,
\end{aligned}
\end{align*}
and subsequently \eqref{EQdiffestimateeta}, that is, in $L^\infty_vH^5(S_v)$,
\begin{align*}
\begin{aligned}
\eta-\eta_x=&-\frac{\varep}{2r^2} \di \lrpar{\int\limits_1^v r^2 \lrpar{-\frac{1}{r} \int\limits_1^{v'} \frac{\varphi^2}{8r^3} \lrpar{\vert \mfW_0\vert^2_\gac} dv''+\frac{2}{r^2} \int\limits_1^{v'}  \int\limits_1^{v''} \frac{\varphi^2}{16r^4} \lrpar{\vert \mfW_0\vert^2_\gac} dv'''dv''}dv'}\\
&  + \OO(\varep^{5/4}\Vert W \Vert).
\end{aligned}
\end{align*}
\ni \textbf{Estimates for $\beta$ and $\si$.} At this point, the null curvature component $\beta$ is determined from \eqref{EQGaussCodazzi}, that is,
\begin{align*}
\begin{aligned}
-\beta= \Divd \chih -\half \di \trchi + \chih \cdot (\eta-\di\log\Om) -\half \trchi (\eta-\di\log\Om), 
\end{aligned}
\end{align*}
which implies, by \eqref{EQsmallnesscondBaseCONSTRUCTION}, \eqref{EQestimationdifferencegd}, \eqref{EQomtrchidiffsimpleestimate}, \eqref{EQexpansionomchih} and \eqref{EQdiffestimateeta}, that in $L^\infty_vH^5(S_v)$,
\begin{align*}
\begin{aligned}
\beta-\beta_x =& -\Divd_{r^2\gac}\lrpar{\frac{\varphi}{2} \lrpar{\varep^{1/2} \sin\lrpar{\frac{v}{\varep}} \mfW_0 + \varep^{3/4} \sin\lrpar{\frac{v}{\varep^{1/2}}} \mfW_1 + \varep^{3/4} \cos\lrpar{\frac{v}{\varep^{1/2}}} \mfW_2}}  \\
&+ \OO(\varep \Vert \mfW_0 \Vert) + \OO(\varep^{5/4} \Vert W \Vert).
\end{aligned}
\end{align*}
\subsection{Construction of $\chib$} \label{SECestimatechib} First consider $\Om\trchib$. By combining \eqref{EQtransporttrchitrchib1}, \eqref{EQcompactRAYpart1} and \eqref{EQGaussEQ}, the following null transport equation for $\Om\trchib$ holds,
\begin{align}
D\lrpar{\phi^2 \Om\trchib} =& -2\phi^2\Om^2 \Divd \lrpar{\eta-2\di\log\Om} + 2\phi^2\Om^2 \vert \eta - 2\di\log\Om\vert^2 - 2\phi^2\Om^2 K \label{EQtrchib}.
\end{align}
We claim that
\begin{align}
\begin{aligned}
\Vert \Om\trchib-(\Om\trchib)_x \Vert_{L^\infty_vH^4(S_v)} \les \varep \Vert \mfW_0 \Vert + \varep^{5/4} \Vert W \Vert.
\end{aligned}\label{EQestimateomtrchibdiff}
\end{align}
Indeed, the right-hand side of the null transport equation for the difference $(\phi^2\Om\trchib- \phi_x^2(\Om\trchib)_x)$ (see \eqref{EQtrchib}) is straight-forward to estimate in $L^\infty_vH^4(S_v)$ using \eqref{EQsmallnesscondBaseCONSTRUCTION}, \eqref{EQestimationdifferencegd}, \eqref{EQestimateDiffKcurvature}, \eqref{EQomtrchidiffsimpleestimate} and \eqref{EQetaNORMestimatediff}. Integrating in $v$ then yields \eqref{EQestimateomtrchibdiff}.

Second consider $\Om\chibh$. By combining \eqref{EQchihequations1}, \eqref{EQcompactRAYpart1} and \eqref{EQcompactRAYpart1}, the following null transport equation holds for $\Om\chibh$,
\begin{align}
\begin{aligned}
D\lrpar{\frac{\Om\chibh}{\phi}} - \lrpar{\Om\chih,\frac{\Om\chibh}{\phi}} \gd =& \frac{\Om^2}{\phi} \lrpar{\Nd \widehat{\otimes} \lrpar{2\di\log\Om-\eta}} \\
&+ \frac{\Om^2}{\phi} \lrpar{\lrpar{2\di\log\Om-\eta}\widehat{\otimes} \lrpar{2\di\log\Om-\eta} - \half \trchib \chih}.
\end{aligned}\label{EQchibh}
\end{align}
In the following we prove that
\begin{align}
\begin{aligned}
\Vert \Om\chibh-(\Om\chibh)_x \Vert_{L^\infty_v H^4(S_v)} \les \varep \Vert \mfW_0 \Vert + \varep^{5/4} \Vert W \Vert.
\end{aligned}\label{EQestimateDIFFomchibh}
\end{align}
Similarly to proofs of the above estimates, \eqref{EQestimateDIFFomchibh} is derived by integrating the null transport equation satisfied by the difference $$\frac{\Om\chibh}{\phi}-\frac{(\Om\chibh)_x}{\phi_x}$$
and applying \eqref{EQsmallnesscondBaseCONSTRUCTION},\eqref{EQphidiffest1}, \eqref{EQestimationdifferencegd}, \eqref{EQetaNORMestimatediff}, \eqref{EQestimateomtrchibdiff}. We omit the detailed proof and observe only that the largest term on the right-hand side of the differences of \eqref{EQchibh}, that is,
\begin{align*}
\begin{aligned}
\lrpar{- \frac{1}{2\phi} \Om\trchib \Om\chih}-\lrpar{- \frac{1}{2\phi_x} (\Om\trchib)_x (\Om\chih)_x},
\end{aligned}
\end{align*}
sees by \eqref{EQOmchihtildeexpr} a high-frequency improvement when integrating, that is, at lowest-order, in $L^\infty_vH^4(S_v)$,
\begin{align*}
\begin{aligned}
\int\limits_1^v -\frac{1}{2r} \lrpar{-\frac{2}{r}} \lrpar{\Om\chih-(\Om\chih)_x} dv' = \OO(\varep^{5/4}\Vert W \Vert).
\end{aligned}
\end{align*}

\ni\textbf{Estimates for $\beb$, $\rho$ and $\sigma$.} The null curvature components $\rho$, $\beb$ and $\si$ are determined by \eqref{EQGaussEQ}, \eqref{EQGaussCodazzi} and \eqref{EQcurlEquations}, that is,
\begin{align*}
\begin{aligned}
-\rho =& K + \frac{1}{4}\trchi\trchib - \half (\chih,\chibh), \,\, \beb = \Divd \chibh - \half \di \trchib -\chibh \cdot \zeta +\half \trchib \zeta, \\
\si =& - \half \chih \wedge \chibh - \Curld \eta,
\end{aligned}
\end{align*}
which implies by the previous estimates that
\begin{align*}
\begin{aligned}
\Vert \rho-\rho_x \Vert_{L^\infty_vH^4(S_v)} + \Vert \beb-\beb_x \Vert_{L^\infty_vH^3(S_v)} + \Vert \si-\si_x \Vert_{L^\infty_vH^4(S_v)} \les \varep\Vert \mfW_0 \Vert + \varep^{5/4} \Vert W \Vert.
\end{aligned}
\end{align*}

\subsection{Construction of $\omb$}\label{SECestimateomb} By combining \eqref{EQDUOMU1}, \eqref{EQGaussCodazzi}, \eqref{EQGaussEQ}, the following null transport equation for $\omb$ holds
\begin{align}
D\omb =& \Om^2 \lrpar{4\lrpar{\eta,\di\log\Om}-3\vert\eta\vert^2 + K + \frac{1}{4}\trchi\trchib-\half (\chih,\chibh)}, \label{EQombnulltransport}
\end{align}
By using \eqref{EQsmallnesscondBaseCONSTRUCTION}, \eqref{EQestimateDiffKcurvature}, \eqref{EQomtrchidiffsimpleestimate}, \eqref{EQetaNORMestimatediff}, \eqref{EQestimateomtrchibdiff} and \eqref{EQestimateDIFFomchibh}, it follows, similarly to the above estimates, that
\begin{align}
\begin{aligned}
\Vert \omb-\omb_x \Vert_{L^\infty_vH^4(S_v)} \les \varep \Vert \mfW_0 \Vert + \varep^{5/4} \Vert W \Vert.
\end{aligned}\label{EQestimateombdiff}
\end{align}
\subsection{Construction of $\ab$} \label{SECestimateab} By combining \eqref{EQnullBianchiEquations}, \eqref{EQcompactRAYpart1}, \eqref{EQGaussCodazzi}, \eqref{EQGaussEQ} and \eqref{EQcurlEquations}, the following null transport equation holds for $\ab$,
\begin{align}
\begin{aligned}
&D\lrpar{\frac{\Om^2\ab}{\phi}} - \lrpar{\Om\chih, \frac{\Om^2\ab}{\phi}} \gd \\
=& - \Om \Nd\widehat{\otimes} \lrpar{\Divd\chibh-\half\di\trchib-\chibh \cdot(\eta-\di\log\Om) +\half\trchib (\eta-\di\log\Om)} \\
&- \Om (9\di\log\Om-5\eta) \widehat{\otimes} \lrpar{\Divd \chibh - \half \di\trchib - \chibh \cdot (\eta-\di\log\Om) + \half \trchib (\eta-\di\log\Om)}\\
&+ 3\Om\chibh \lrpar{K+\frac{1}{4}\trchi\trchib -\half (\chih,\chibh)} + 3 \Om {}^\ast \chibh (\Curld \eta + \half \chih \wedge \chibh).
\end{aligned}\label{EQabNTE}
\end{align}
By using \eqref{EQsmallnesscondBaseCONSTRUCTION}, \eqref{EQphidiffest1}, \eqref{EQestimateDiffKcurvature}, \eqref{EQomtrchidiffsimpleestimate}, \eqref{EQetaNORMestimatediff}, \eqref{EQestimateomtrchibdiff}, \eqref{EQestimateDIFFomchibh}, it follows, similarly to the above estimates, that
\begin{align}
\begin{aligned}
\Vert \ab-\ab_x \Vert_{L^\infty_vH^2(S_v)} \les \varep \Vert \mfW_0 \Vert + \varep^{5/4} \Vert W \Vert.
\end{aligned}\label{EQestimateABdiff}
\end{align}

\subsection{Construction of $\Du \omb$} \label{SECestimateDuomb} By combining \eqref{EQDUOMU1}, \eqref{EQGaussCodazzi}, \eqref{EQGaussEQ}, the following null transport equation for $\Du\omb$ holds,
\begin{align}
\begin{aligned}
&D\Du\omb +12 \Om^2 (\eta-\di\log\Om,\di\omb) - 2\Om^2\omb\cdot(\eta,-3\eta+4\di\log\Om)  \\
&- \lrpar{2\Om^2 \omb-\frac{3}{2}\Om^3\trchib}\lrpar{K+ \frac{1}{4} \trchi\trchib-\half (\chih,\chibh)} -4\Om^3\chib(\eta,\di\log\Om) - \frac{\Om^3}{2} (\chih, \ab)\\
&-\Om^3 \lrpar{\Divd \chibh-\half \di\trchib-\chibh \cdot \lrpar{\eta- \di \log \Om}+ \half \trchib \lrpar{\eta-\di\log\Om},7\eta-3\di\log\Om} \\
&-\Om^3 \Divd\lrpar{\Divd \chibh-\half \di\trchib-\chibh \cdot \lrpar{\eta- \di \log \Om}+ \half \trchib \lrpar{\eta-\di\log\Om}} \\
=&0.
\end{aligned}\label{EQnteFORduomb}
\end{align}
By using \eqref{EQsmallnesscondBaseCONSTRUCTION}, \eqref{EQestimateDiffKcurvature}, \eqref{EQomtrchidiffsimpleestimate}, \eqref{EQetaNORMestimatediff}, \eqref{EQestimateomtrchibdiff}, \eqref{EQestimateDIFFomchibh} and \eqref{EQestimateABdiff}, it follows, similarly to the above estimates, that
\begin{align*}
\begin{aligned}
\Vert \Du\omb-(\Du\omb)_x \Vert_{L^\infty_vH^2(S_v)} \les \varep \Vert \mfW_0 \Vert + \varep^{5/4} \Vert W \Vert.
\end{aligned}
\end{align*}
\subsection{Difference estimates} \label{SECestimateDiffest2}

\ni In this section we prove \eqref{EQtoProveIterationEstimateForX}, that is, for given solutions $x$ and $x'$ to the null structure equations (both assumed to satisfy \eqref{EQsmallnesscondBaseCONSTRUCTION} for $\varep>0$ sufficiently small) and given tensors $W$ and $W'$ of the form \eqref{EQWansatz}-\eqref{EQchoiceofCoeffFunction3}, we show that
\begin{align}
\begin{aligned}
&\Vert \lrpar{x_{+W} - x}-\lrpar{x'_{+W'}-x'} \Vert_{\XX(S_2)} + \Vert \lrpar{x_{+W} - x}-\lrpar{x'_{+W'}-x'} \Vert_{\XX^{\mathrm{h.f.}}(\HH)}  \\
\les& \varep \Vert \mfW_0-\mfW_0' \Vert + \varep^{5/4} \Vert W-W' \Vert + \varep^{1/2} \Vert x - x' \Vert_{\XX(\HH)}.
\end{aligned}\label{EQtoProveIterationEstimateForXLATER}
\end{align}
The proof is analogous to the proof of \eqref{EQestimateDIFFx0x1}. The estimates are derived exactly the same hierarchical order as for the proof of \eqref{EQestimateDIFFx0x1} in the previous sections, similarly applying the high-frequency improvement in integration, and suitably estimating higher order terms by product estimates. In the following we give some details on the lowest-order control of the fundamental integrals
\begin{align*}
\begin{aligned}
\mathfrak{D}_1 :=& \int\limits_1^v \lrpar{\lrpar{\Om\chih - (\Om\chih)_x}-\lrpar{(\Om\chih)' -  (\Om\chih)'_{x'}}}dv', \\
\mathfrak{D}_2 :=& \int\limits_1^v \lrpar{\lrpar{\vert\Om\chih\vert^2 - \vert (\Om\chih)_x\vert^2}-\lrpar{\vert(\Om\chih)'\vert^2 - \vert (\Om\chih)'_{x'}\vert^2}}dv',
\end{aligned}
\end{align*}
to point out the use of our high-frequency ansatz which is relevant for the proof of \eqref{EQtoProveIterationEstimateForXLATER}.\\

\ni \textbf{Bound for $\mathfrak{D}_1$.} In the following we estimate the lowest-order terms in the integrand of $\mathfrak{D}_1$. Recall from \eqref{EQexpansionomchih} that,
\begin{align*}
\begin{aligned}
\Om\chih - (\Om\chih)_x =& \frac{\varphi}{2} \varep^{1/2} \sin\lrpar{\frac{v}{\varep}} \mfW_0 + \OO(\varep^{3/4}).
\end{aligned}
\end{align*}
Hence we can write, considering in this discussion only the lowest-order terms in the integrand,
\begin{align*}
\begin{aligned}
\mathfrak{D}_1 :=& \int\limits_1^v \lrpar{\lrpar{\Om\chih - (\Om\chih)_x}-\lrpar{(\Om\chih)' -  (\Om\chih)'_{x'}}}dv' \\
=&\int\limits_1^v \lrpar{\lrpar{\frac{\varphi}{2} \varep^{1/2} \sin\lrpar{\frac{v}{\varep}} \mfW_0}-\lrpar{\frac{\varphi}{2} \varep^{1/2} \sin\lrpar{\frac{v}{\varep}} \mfW_0'}}dv' \\
=& \int\limits_1^v \frac{\varphi}{2} \varep^{1/2} \sin\lrpar{\frac{v}{\varep}} (\mfW_0- \mfW_0') dv' \\
\les& \varep^{3/2} \Vert \mfW_0-\mfW_0' \Vert,
\end{aligned}
\end{align*}
where we used high-frequency improvement upon integration, and that, as the ansatz for $\mfW_0$ is in a finite-dimensional space (it is built from finitely many tensor spherical harmonics), any norm for the tensor $\mfW_0$ is equivalent to the norm for $\mfW_0$ which appears on the right-hand side above (see Definition \ref{DEFdefinitionWnorm}). \\

\ni\textbf{Bound for $\mathfrak{D}_2$.} In the following we estimate the lowest-order terms in the integrand of $\mathfrak{D}_2$. Recall from \eqref{EQomchihcompleteformula} that 
\begin{align*}
\begin{aligned}
\vert \Om\widehat{{\chi}} \vert^2_{{\gd}} -  \vert (\Om\chih)_x\vert^2_{\gd_x} = \frac{\varphi^2 \varep}{4r^4} \sin^2\lrpar{\frac{v}{\varep}}\vert \mfW_0\vert^2_\gac+ \OO(\varep^{5/4}\Vert W \Vert).
\end{aligned}
\end{align*}
Hence we can write, considering in this discussion only the lowest-order terms in the integrand,
\begin{align*}
\begin{aligned}
\mathfrak{D}_2 :=& \int\limits_1^v \lrpar{\lrpar{\vert\Om\chih\vert^2 - \vert (\Om\chih)_x\vert^2}-\lrpar{\vert(\Om\chih)'\vert^2 - \vert (\Om\chih)'_{x'}\vert^2}}dv' \\
=& \int\limits_1^v \frac{\varphi^2 \varep}{4r^4} \sin^2\lrpar{\frac{v}{\varep}}\lrpar{\vert \mfW_0\vert^2_\gac- \vert \mfW_0'\vert^2_\gac} dv'\\
\les& \varep \Vert \mfW_0 - \mfW_0' \Vert,
\end{aligned}
\end{align*}
where we recall that $\sin^2(v/\varep)$ has a low-frequency component (see \eqref{EQsinesquaredformula}), and argued as above.

This finishes our discussion of the proof of \eqref{EQtoProveIterationEstimateForX}.

\section{Null gluing of $\mathbf{E},\mathbf{P},\mathbf{L},\mathbf{G}$} \label{SECproofW} 

\ni In this section we prove Theorem \ref{THMmain0}. The strategy of the proof is as follows. Given a (low-frequency) solution $x \in \XX(\HH)$ to the null structure equations along $\HH$, and satisfying
\begin{align}
\begin{aligned}
\Vert x-\mathfrak{m} \Vert_{\XX(\HH)} \leq \varep,
\end{aligned}\label{EQsmallnessCONDbackgroundEPLGnullGluingsectionback}
\end{align}
we consider the corresponding high-frequency solution $x_{+W}$ constructed in Section \ref{SECconstructionSolution} \emph{for a tensor $W$ to be determined}. In Sections \ref{SECgluingE} to \ref{SECgluingL} we study the null transport of the charges $\mathbf{E},\mathbf{P},\mathbf{L},\mathbf{G}$ and prove the following decompositions,
\begin{align}
\begin{aligned}
\triangle \mathbf{E}(x_{+W}) =& \mathbf{E}(W) + \OO(\varep^{5/4}), &
\triangle \mathbf{P}(x_{+W})=& \mathbf{P}(W)+ \OO(\varep^{5/4}), \\
\triangle \mathbf{G}(x_{+W}) =& \mathbf{G}(W) + \OO(\varep^{5/4}), &
\triangle \mathbf{L}(x_{+W}) =& \mathbf{L}(W) + \OO_{x}(\varep^2)+\OO(\varep^{9/4}),
\end{aligned}\label{EQoverviewEPLGdecomps}
\end{align}
where $\Ef(W), \Pf(W), \Gf(W)$ are $W$-dependent terms of size $\varep \Vert W \Vert$, and $\Lf(W)$ is a $W$-dependent term of size $\varep^2 \Vert W \Vert$. In \eqref{EQoverviewEPLGdecomps}, $\OO_{x}(\varep^2)$ denotes error terms of size $\varep^2$ which depend only on $x$.

In Section \ref{SEC1EPLGgluing} we study the mapping $W\mapsto (\Ef(W),\Pf(W),\Lf(W),\Gf(W))$, showing that it is a bijection onto a specific set of vectors -- this leads, in particular, to the assumption \eqref{EQsmallnessMain02} in Theorem \ref{THMmain0} -- and deriving quantitative bounds.

In Section \ref{SECdegreeargument1} we prove the existence of null gluing of $\Ef,\Pf,\Lf,\Gf$, that is, part (1) of Theorem \ref{THMmain0}, by a classical degree argument based on the above surjective of the mapping $W\mapsto(\Ef(W),\Pf(W),\Lf(W),\Gf(W))$ together with the expansions \eqref{EQoverviewEPLGdecomps}.

In Section \ref{SECdiffEstimates1W} we derive quantitative bounds for the solutions $W$ to the null gluing of $\Ef,\Pf,\Lf,\Gf$, proving part (2) of Theorem \ref{THMmain0}. Despite the use of the degree argument (which \emph{a priori} does not give uniqueness), this is nevertheless possible due to the explicit nature of the class of tensors $W$. Namely, one can essentially read off the $W$-dependent quantities $\Ef(W),\Pf(W),\Lf(W),\Gf(W)$ from the coefficients of $W$.

In Section \ref{SECdiffEstimates2W} we prove \emph{difference estimates}, that is, part (3) of Theorem \ref{THMmain0}. These bounds are only a slight generalization of the estimates proved in Sections \ref{SECgluingE}-\ref{SEC1EPLGgluing} and \ref{SECdiffEstimates1W}, but important for the iteration scheme in the proof of the main theorem of this paper, see Section \ref{SECgluingEP}. \\

\ni \textbf{Notation.} In the following we use schematic notation for higher-order in $\varep$ terms to be estimated by product estimates (see, for example, \eqref{EQDmuschematic1}) where we tacitly suppress angular derivatives and denote by $x_{+W}-\mathfrak{m}$ the following components
\begin{align}
\begin{aligned}
x_{+W}-\mathfrak{m} = (\Om,\gd, \Om\trchi, \Om\trchib, \chibh, \eta, \om, D\om, \omb, \Du\omb, \ab)(x_{+W}).
\end{aligned}\label{EQschematicnotation}
\end{align}
We emphasize that in this schematic notation, $x_{+W}-\mathfrak{m}$ does not contain $\Om\chih$ and $\a$.

\subsection{Transport estimates for $\mathbf{E}$ and $\mathbf{P}$}\label{SECgluingE} Recall the definition of $\Ef$ and $\Pf$ from Definition \ref{DEFchargesEPLG}, that is, for $m=-1,0,1$,
\begin{align*}
\begin{aligned}
\mathbf{E} := \mathfrak{m}^{(0)}, \,\, \mathbf{P}^m := \mathfrak{m}^{(1m)},
\end{aligned}
\end{align*}
where the scalar function $\mathfrak{m}$ is defined as
\begin{align*}
\begin{aligned}
\mathfrak{m} := \phi^3 \lrpar{K+\frac{1}{4}\trchi\trchib} -\phi \Divd \lrpar{\frac{\phi^3}{2\Om^2} \lrpar{\di(\Om\trchi)+\Om\trchi\lrpar{\eta-2\di\log\Om}}}.
\end{aligned}
\end{align*}
In this section we prove that for $\varep>0$ sufficiently small,
\begin{align}
\begin{aligned}
\triangle \mathbf{E}(x_{+W})= \mathbf{E}(W) + \OO(\varep^{5/4}), \,\, \triangle \mathbf{P}^m(x_{+W})= \mathbf{P}(W)+ \OO(\varep^{5/4}),
\end{aligned}\label{EQDEFEWPW}
\end{align}
with
\begin{align*}
\begin{aligned}
\mathbf{E}(W):= \varep \int\limits_1^2 \frac{1}{16r^2}\varphi^2 \lrpar{\vert \mfW_0 \vert^2_{\gac}}^{(0)}, \,\,
\mathbf{P}(W):= -\varep \int\limits_1^2 \frac{1}{16r^2}\varphi^2 \lrpar{\vert \mfW_0 \vert^2_{\gac}}^{(1m)}.
\end{aligned}
\end{align*}

\ni Indeed, as shown in Appendix \ref{AppDerivation} (see \eqref{EQnulltransportMU}), $\mathfrak{m}$ satisfies the following null transport equation,
\begin{align}
\begin{aligned}
D\mathfrak{m}=& \frac{r^2}{2} \lrpar{\Ldo+1}\lrpar{\vert \Om\chih\vert^2} - \lrpar{\frac{r^2}{2}+\frac{\Om\trchib\phi^3}{4\Om^2}} \vert\Om\chih\vert^2 - \half\lrpar{r^2\Ldo-\frac{\phi^4}{\Om^2}\Ld} \vert\Om\chih\vert^2\\
& + \lrpar{\phi^3-\frac{\phi^4\Om\trchi}{2\Om^2}} \Divd\Divd \lrpar{\Om\chih} - \frac{\phi}{2} \Nd \lrpar{\frac{\Om\trchi \phi^3}{\Om^2}}\cdot \Divd\lrpar{\Om\chih} \\
&- \lrpar{ \frac{\phi^3}{2}-\frac{\Om\trchi\phi^4}{4\Om^2}} \Ld \lrpar{\Om\trchi} + \frac{\Om\trchi\phi}{4} \Nd\lrpar{\frac{\phi^3}{\Om^2}}\cdot \Nd\lrpar{\Om\trchi}\\
&-\lrpar{\frac{\Om\trchi\phi^3}{2}-\frac{(\Om\trchi)^2\phi^4}{4\Om^2}} \Divd\lrpar{\eta-2\di\log\Om} + \frac{\Om\trchi\phi}{4}\Nd\lrpar{\frac{\phi^3\trchi}{\Om}}\cdot \lrpar{\eta-2\di\log\Om} \\
& + \frac{\phi}{2} \Nd \lrpar{\frac{\phi^3}{\Om^2}}\cdot \Nd \lrpar{\vert\Om\chih\vert^2} +\frac{\Om\trchi \phi^3}{2} \vert \eta-2\di\log\Om\vert^2 + 2\phi \Divd \lrpar{\Om\chih \cdot \mfb} \\
&+\frac{\phi}{2}\Divd\lrpar{\frac{\phi^3}{\Om^2} (\eta-2\di\log\Om) \vert\Om\chih\vert^2 },
\end{aligned} \label{EQnulltransportMUEstimationForm}
\end{align}
Using the notation \eqref{EQschematicnotation}, we can express \eqref{EQnulltransportMUEstimationForm} as
\begin{align}
\begin{aligned}
D\mathfrak{m} =& \frac{r^2}{2} \lrpar{\Ldo+1}\lrpar{\vert \Om\chih\vert^2} + (x_{+W}-\mathfrak{m}) \cdot \Om\chih +  (x_{+W}-\mathfrak{m}) \cdot \vert \Om\chih\vert_\gd^2 \\
&+ (x_{+W}-\mathfrak{m})\cdot (x_{+W}-\mathfrak{m}).
\end{aligned}\label{EQDmuschematic1}
\end{align}
By \eqref{EQestimateDIFFx0x1}, \eqref{EQOmchihtildeexpr} and \eqref{EQomchihcompleteformula}, we get that
\begin{align}
\begin{aligned}
\triangle \mathfrak{m} =& \varep \int\limits_1^2 \frac{1}{16r^2}\varphi^2 (\Ldo+1)\lrpar{\vert \mfW_0 \vert^2_{\gac}} + \OO(\varep^{5/4}).
\end{aligned}\label{EQintegrationMU}
\end{align}
\ni The proof of \eqref{EQDEFEWPW} follows by projecting \eqref{EQintegrationMU} onto the modes $l=0$ and $l=1$ (recall that projection commutes with differentiation and integration in $v$), and using that for any scalar function $f$ on the sphere, for $m=-1,0,1$,
\begin{align*}
\begin{aligned}
((\Ldo+1)f)^{(1m)} = - f^{(1m)}.
\end{aligned}
\end{align*}

\subsection{Transport estimate for $\mathbf{G}$}\label{SECgluingG} Recall the definition of $\Gf$ from Definition \ref{DEFchargesEPLG}, that is, for $m=-1,0,1$,
\begin{align*}
\begin{aligned}
\Gf^m := \lrpar{\frac{\phi^3}{2\Om^2} \lrpar{\di(\Om\trchi)+\Om\trchi\lrpar{\eta-2\di\log\Om}}}^{(1m)}_E.
\end{aligned}
\end{align*}

\ni In this section we prove that
\begin{align}
\begin{aligned}
\triangle \mathbf{G}^m(x_{+W}) = \mathbf{G}^m(W) + \OO(\varep^{5/4}),
\end{aligned}\label{EQdefGW}
\end{align}
with
\begin{align*}
\begin{aligned}
\mathbf{G}^m(W) := \varep \int\limits_1^2 \frac{\sqrt{2}}{16r} \varphi^2 \lrpar{\vert \mfW_0 \vert_\gac^2}^{(1m)}.
\end{aligned}
\end{align*}
Indeed, $\Gf$ satisfies the following null transport equation (see \eqref{EQBequation1}),
\begin{align}
\begin{aligned}
D\mathbf{G}^m = D\mfb^{(1m)}_E
= \lrpar{ \frac{\phi^3}{2\Om^2} \lrpar{-(\eta-2\di\log\Om) \vert \Om\chih\vert^2 - \di \lrpar{\vert \Om\chih\vert^2} +\Om\trchi \Divd\lrpar{\Om\chih}}}^{(1m)}_E.
\end{aligned}\label{EQBequation1appliedTransportEstimateG}
\end{align}
On the one hand, by \eqref{EQestimateDIFFx0x1} and \eqref{EQomchihcompleteformula},
\begin{align}
\begin{aligned}
(\eta-2\di\log\Om) \vert \Om\chih\vert^2 = (x_{+W}-\mathfrak{m}) \cdot \vert\Om\chih\vert^2 = \OO(\varep^2).
\end{aligned}\label{EQcontrolG1}
\end{align}
On the other hand, by \eqref{EQestimateDIFFx0x1} and \eqref{EQOmchihtildeexpr}, 
\begin{align}
\begin{aligned}
\lrpar{ \frac{\phi^3}{2\Om^2}\Om\trchi \Divd\lrpar{\Om\chih}}^{(1m)}_E = (x_{+W}-\mathfrak{m})\cdot \Om\chih = \OO(\varep^{3/2}).
\end{aligned}\label{EQcontrolG2}
\end{align}
where we used that, by Fourier theory (see Appendix \ref{APPconstructionW}), for all $\gac$-tracefree symmetric $2$-tensors $T$ on the sphere,
\begin{align*}
\begin{aligned}
(\Divdo T)^{[1]} =0.
\end{aligned}
\end{align*}
We conclude \eqref{EQdefGW} by plugging \eqref{EQcontrolG1} and \eqref{EQcontrolG2} into \eqref{EQBequation1appliedTransportEstimateG}, integrating in $v$, and using that for any scalar function $f$ on the sphere, for $m=-1,0,1$,
\begin{align*}
\begin{aligned}
(-\di f)_E^{(1m)}= \sqrt{2} f^{(1m)}.
\end{aligned}
\end{align*}
\subsection{Transport estimate for $\mathbf{L}$}\label{SECgluingL} Recall the definition of $\Lf$ from Definition \ref{DEFchargesEPLG}, that is, for $m=-1,0,1$,
\begin{align*}
\begin{aligned}
\Lf^m := \lrpar{\frac{\phi^3}{2\Om^2} \lrpar{\di(\Om\trchi)+\Om\trchi\lrpar{\eta-2\di\log\Om}}}^{(1m)}_H.
\end{aligned}
\end{align*}
In this section we prove that the following expansion holds,
\begin{align}
\begin{aligned}
\triangle \mathbf{L}^m(x_{+W}) =& \mathbf{L}^m(W) + \OO^{\Lf^m}_{x}+\OO(\varep^{9/4}),
\end{aligned}\label{EQdeftriangleLW00}
\end{align}
with
\begin{align}
\begin{aligned}
\mathbf{L}^m(W):= \int\limits_1^2 \frac{\varep^2\varphi^2}{r^2}\lrpar{-\half  \lrpar{\Nd_A\lrpar{\mfW_1^{AB} (\mfW_2)_{B\cdot} }}^{(1m)}_H + \frac{1}{4}\lrpar{ \Nd_\cdot \lrpar{\mfW_1 \cdot\mfW_2}}^{(1m)}_H} + \OO^{\Lf^m}_{x,\mfW_0}.
\end{aligned}\label{EQdeftriangleLW}
\end{align}
where 
\begin{itemize}
\item $\OO^{\Lf^m}_{x}$ consists of terms that depend only on $x$, and is of size $\varep^2$,
\item $\OO^{\Lf^m}_{x,\mfW_0}$ consists of terms that depend only on $\mfW_0$ and (possibly) $x$, and is of size $\varep^2$.
\end{itemize}
The precise expressions for $\OO^{\Lf^m}_{x}$ and $\OO^{\Lf^m}_{x,\mfW_0}$ are given in \eqref{EQpreciseerrorL}.

Indeed, $\Lf$ satisfies the following null transport equation (see \eqref{EQLGtransportEquations} in Appendix \ref{AppDerivation}),
\begin{align}
\begin{aligned}
D\mathbf{L}^m =&\lrpar{ \frac{\phi^3}{2\Om^2} \lrpar{-(\eta-2\di\log\Om) \vert \Om\chih\vert^2 - \di \lrpar{\vert \Om\chih\vert^2} +\Om\trchi \Divd\lrpar{\Om\chih}}}^{(1m)}_H \\
=& \underbrace{\lrpar{-\frac{\phi^3}{2\Om^2}(\eta-2\di\log\Om) \vert \Om\chih\vert^2+ \di\lrpar{\frac{\phi^3}{2\Om^2}} \vert \Om\chih\vert^2}^{(1m)}_H}_{:= \TT_1} + \underbrace{\lrpar{ \frac{\phi^3\Om\trchi}{2\Om^2} \Divd\lrpar{\Om\chih}}^{(1m)}_H}_{:=\TT_2},
\end{aligned}\label{EQBequation1appliedTransportEstimateL}
\end{align}
where used that for any scalar function $f$, the projection of $\di f$ onto $H^{(1m)}$ vanishes.

In the following we control the terms $\TT_1$ and $\TT_2$ in full detail \emph{at the order $\varep^2$}.\\

\ni\underline{\textbf{Control of $\TT_1$.}} In the following we prove that 
\begin{align}
\begin{aligned}
\TT_1 := \lrpar{-\frac{\phi^3}{2\Om^2}(\eta-2\di\log\Om) \vert \Om\chih\vert^2+ \di\lrpar{\frac{\phi^3}{2\Om^2}} \vert \Om\chih\vert^2}^{(1m)}_H = \OO^{\TT_1}_{x,\mfW_0} + \OO(\varep^{9/4}),
\end{aligned}\label{EQLLLcontrol1}
\end{align}
where $\OO^{\TT_1}_{x,\mfW_0}$ is of size $\varep^2$ and its precise definition is given in \eqref{EQstructureOOTT1}.

To prove \eqref{EQLLLcontrol1}, recall first \eqref{EQomegadef} and the expansions for $x_{+W}$ proved in \eqref{EQrepformulaphi} and \eqref{EQdiffestimateeta},
\begin{align*}
\begin{aligned}
\Om =& \Om_x, &
\phi=& \varep[\phi] + \varep^{3/2}[\phi] + \OO(\varep^{7/4}) \, \text{ in } L^\infty_vH^6(S_v), &
\eta =& \varep[\eta] + \OO(\varep^{5/4}) \, \text{ in } L^\infty_vH^5(S_v),
\end{aligned}
\end{align*}
where the terms $\varep[\phi]$, $\varep^{3/2}[\phi]$ and $\varep[\eta]$ are explicitly given by
\begin{align*}
\begin{aligned}
\varep[\phi] :=& \phi_x -\varep \int\limits_1^v\int\limits_1^{v'} \frac{\varphi^2}{16r^3} \vert \mfW_0\vert_{\gac}^2,\qquad
\varep^{3/2}[\phi] := - \varep^{3/2}\int\limits_1^v\int\limits_1^{v'} \frac{\varphi^2}{16r^3} \lrpar{\vert \mfW_1\vert_{\gac}^2+\vert \mfW_2\vert_{\gac}^2},\\
\varep[\eta] :=&\eta_x -\frac{\varep}{2r^2} \di \lrpar{\int\limits_1^v r^2 \lrpar{-\frac{1}{r} \int\limits_1^{v'} \frac{\varphi^2}{8r^3} \lrpar{\vert \mfW_0\vert^2_\gac} dv''+\frac{2}{r^2} \int\limits_1^{v'}  \int\limits_1^{v''} \frac{\varphi^2}{16r^4} \lrpar{\vert \mfW_0\vert^2_\gac} dv'''dv''}dv'}.
\end{aligned}
\end{align*}
Moreover, recall from \eqref{EQomchihcompleteformula} that in $L^\infty_v H^6(S_v)$,
\begin{align*}
\begin{aligned}
\vert \Om\chih \vert^2_\gd = \varep\left[\vert \Om\chih \vert^2_\gd\right]+ \OO(\varep^{5/4}).
\end{aligned}
\end{align*}
with the $\mfW_0$-dependent term $\varep\left[\vert \Om\chih \vert^2_\gd\right]$ being given by
\begin{align}
\begin{aligned}
\varep\left[\vert \Om\chih \vert^2_\gd\right] := \frac{\varphi^2 \varep}{4r^4} \lrpar{\sin^2\lrpar{\frac{v}{\varep}}\vert \mfW_0\vert^2_\gac}.
\end{aligned}\label{EQexpansionchihsquaredsimplified}
\end{align}
Using these expansions, we get that
\begin{align*}
\begin{aligned}
\frac{\phi^3}{2\Om^2}(\eta-2\di\log\Om)\cdot \vert \Om\chih\vert^2 =& \frac{r^3}{2} \lrpar{\varep[\eta]-2\di(\log\Om_x)} \cdot \varep\left[\vert \Om\chih \vert^2_\gd\right] + \OO(\varep^{9/4}).
\end{aligned}
\end{align*}
and moreover,
\begin{align*}
\begin{aligned}
\di\lrpar{\frac{\phi^3}{2\Om^2}} \cdot \vert \Om\chih\vert^2 = \lrpar{\frac{3r^2 \di(\varep[\phi])}{2\Om^2}-\frac{r^3\di(\Om_x)}{\Om^3}} \cdot \varep\left[\vert \Om\chih \vert^2_\gd\right] + \OO(\varep^{9/4}).
\end{aligned}
\end{align*}
Combining the above yields that
\begin{align*}
\begin{aligned}
\TT_1 = \OO^{\TT_1}_{x,\mfW_0}+ \OO(\varep^{9/4}),
\end{aligned}
\end{align*}
with the term $\OO^{\TT_1}_{x,\mfW_0}$ being given by
\begin{align}
\begin{aligned}
\OO^{\TT_1}_{x,\mfW_0} := \lrpar{-\frac{r^3}{2} \lrpar{\varep[\eta]-2\di\log\Om_x} \cdot \varep\left[\vert \Om\chih \vert^2_\gd\right]+ \lrpar{\frac{3r^2 \di(\varep[\phi])}{2\Om^2}-\frac{r^3\di(\Om_x)}{\Om^3}} \cdot \varep\left[\vert \Om\chih \vert^2_\gd\right]}^{(1m)}_H
\end{aligned}\label{EQstructureOOTT1}
\end{align}
which depends only on $x$ and $\mfW_0$ by the expansions \eqref{EQrepformulaphi} and \eqref{EQdiffestimateeta} (see \eqref{EQexpansionchihsquaredsimplified} above). This finishes the proof of the control \eqref{EQLLLcontrol1} of $\TT_1$. \\

\ni\underline{\textbf{Control of $\TT_2$.}} The remainder of this section is concerned with the control of the term $\TT_2$ which turns out to be essential for our gluing of the charge $\Lf$ at the order of $\varep^2$. In the following we prove that
\begin{align}
\begin{aligned}
\int\limits_1^2 \TT_2 =& \varep^2\int\limits_1^2 \frac{\varphi^2}{r^2}\lrpar{-\half  \lrpar{\Nd_A\lrpar{\mfW_1^{AB} (\mfW_2)_{B\cdot} }}^{(1m)}_H + \frac{1}{4}\lrpar{ \Nd_\cdot \lrpar{\mfW_1}_{AD}\lrpar{\mfW_2}^{AD}}^{(1m)}_H}\\
&+ \OO^{\TT_2}_x + \OO^{\TT_2}_{x,\mfW_0}+\OO(\varep^{9/4}),
\end{aligned}\label{EQbigTT2statementclaim}
\end{align}
where the terms $\OO^{\TT_2}_{x}$ and $\OO^{\TT_2}_{x,\mfW_0}$ are of size $\varep^2$, see \eqref{EQpreciseerrorTT2} for their definition.

Indeed, first, by the fact that the divergence operator is conformal (see, for example, (2.2.1e) in \cite{ChrKl}), that is, for any $2$-covariant tracefree symmetric tensor $T$ and scalar function $f>0$ on $S_v$,
\begin{align*}
\begin{aligned}
\Divd_{f \gd} T = f^{-1} \Divd_{\gd} T,
\end{aligned}
\end{align*}
we can express, using \eqref{EQconformalrelation2} and \eqref{EQfirstexprOmchihexprdiff},\begin{align*}
\begin{aligned}
\Divd_\gd\lrpar{\Om\chih} = \tilde{f}^{-1} \tilde{\Nd}\tilde{f} \cdot \Om\widehat{\tilde{\chi}} + \Divd_\tgd \lrpar{\Om\widehat{\tilde{\chi}}},
\end{aligned}
\end{align*}
where we introduced the notation
\begin{align}
\begin{aligned}
\tilde{f}:= \phi^2 \sqrt{\det\gac}\sqrt{\det\tgd}^{-1}.
\end{aligned}\label{EQdeftildef}
\end{align}
Applying the above, we can rewrite $\TT_2$ as follows, 
\begin{align*}
\begin{aligned}
\TT_2 :=& \lrpar{\frac{\phi^3\Om\trchi}{2\Om^2}\Divd\lrpar{\Om\chih}}^{(1m)}_H \\
=& \underbrace{\lrpar{\lrpar{\lrpar{\frac{\phi^3\Om\trchi}{2\Om^2}-\frac{\phi_x^3(\Om\trchi)_x}{2\Om_x^2}}+\lrpar{\lrpar{\frac{\phi^3\Om\trchi}{2\Om^2}}_x-r^2}}\lrpar{\tilde{f}^{-1} \tilde{\Nd}\tilde{f} \cdot \Om\widehat{\tilde{\chi}}}}^{(1m)}_H}_{:=\TT_{2.1}} \\
& +\underbrace{\lrpar{r^2 \lrpar{\tilde{f}^{-1} \tilde{\Nd}\tilde{f} \cdot \Om\widehat{\tilde{\chi}}}}^{(1m)}_H}_{:=\TT_{2.2}}\\
&+\underbrace{\lrpar{\lrpar{\lrpar{\frac{\phi^3\Om\trchi}{2\Om^2}-\frac{\phi_x^3(\Om\trchi)_x}{2\Om_x^2}}+\lrpar{\frac{\phi_x^3(\Om\trchi)_x}{2\Om_x^2}-r^2}}\Divd_\tgd \lrpar{\Om\widehat{\tilde{\chi}}}}^{(1m)}_H}_{:=\TT_{2.3}} \\
&+ \underbrace{\lrpar{r^2 \Divd_\tgd \lrpar{\Om\widehat{\tilde{\chi}}}}^{(1m)}_H}_{:=\TT_{2.4}}.
\end{aligned}
\end{align*}

\ni{\textbf{Control of $\TT_{2.1}$.}} In the following we show that $\TT_{2.1}$ is directly a higher-order term, 
\begin{align}
\begin{aligned}
\TT_{2.1} = \OO(\varep^{9/4}).
\end{aligned}\label{EQterm21islowerorder}
\end{align}
Indeed, first, using that by \eqref{EQphidiffest1} and \eqref{EQrepformulaRaychauduri} in $L^\infty_vH^6(S_v)$,
\begin{align*}
\begin{aligned}
\phi-\phi_x = \OO(\varep), \,\, \frac{\phi\Om\trchi}{\Om^2}-\lrpar{\frac{\phi\Om\trchi}{\Om^2}}_x = \OO(\varep),
\end{aligned}
\end{align*}
we get that in $L^\infty_vH^6(S_v)$,
\begin{align}
\begin{aligned}
\frac{\phi^3\Om\trchi}{\Om^2}- \frac{\phi^3_x(\Om\trchi)_x}{\Om_x^2} =& \lrpar{\frac{\phi\Om\trchi}{\Om^2}}_x\lrpar{\phi^2-\phi_x^2} + \phi^2 \lrpar{\frac{\phi\Om\trchi}{\Om^2}- \lrpar{\frac{\phi\Om\trchi}{\Om^2}}_x \,}\\
=& \lrpar{\frac{\phi\Om\trchi}{\Om^2}}_x\lrpar{\phi-\phi_x}\lrpar{\phi+\phi_x} + \phi^2 \lrpar{\frac{\phi\Om\trchi}{\Om^2}-\lrpar{\frac{\phi\Om\trchi}{\Om^2}}_x \,} \\
=& \OO(\varep).
\end{aligned}\label{EQestimatediffsquares009}
\end{align}
Second, it clearly holds by \eqref{EQsmallnessCONDbackgroundEPLGnullGluingsectionback} that in $L^\infty_vH^6(S_v)$,
\begin{align}
\begin{aligned}
\lrpar{\frac{\phi^3\Om\trchi}{2\Om^2}}_x-r^2 = \OO(\varep).
\end{aligned}\label{EQestimatediffsquares0092}
\end{align}
Third, using that by \eqref{EQexpansionomchih} we have $\Om\widehat{\tilde{\chi}} = \OO(\varep^{1/2})$ in $L^\infty_vH^6(S_v)$, and  that by \eqref{EQexpansionConformalFactor0091}, in $L^\infty_vH^6(S_v)$ (see \eqref{EQdeftildef} for the definition of $\tilde{f}$)
\begin{align*}
\begin{aligned}
\tilde{f} = 1 +\OO(\varep) \text{ in } L^\infty_vH^6(S_v),  \,\, \tilde{f}^{-1} \tilde{\Nd}\tilde{f} = \OO(\varep) \text{ in } L^\infty_vH^5(S_v), 
\end{aligned}
\end{align*}
it follows that in $L^\infty_vH^5(S_v)$,
\begin{align}
\begin{aligned}
\tilde{f}^{-1} \tilde{\Nd}\tilde{f} \cdot \Om\widehat{\tilde{\chi}} = \OO(\varep^{3/2}).
\end{aligned}\label{EQproductestimate0093}
\end{align}
Combining the estimates \eqref{EQestimatediffsquares009}, \eqref{EQestimatediffsquares0092}, \eqref{EQproductestimate0093} in the definition of $\TT_{2.1}$, i.e.
\begin{align*}
\begin{aligned}
\TT_{2.1} := \lrpar{\lrpar{\lrpar{\frac{\phi^3\Om\trchi}{2\Om^2}-\frac{\phi_x^3(\Om\trchi)_x}{2\Om_x^2}}+\lrpar{\lrpar{\frac{\phi^3\Om\trchi}{2\Om^2}}_x-r^2}}\lrpar{\tilde{f}^{-1} \tilde{\Nd}\tilde{f} \cdot \Om\widehat{\tilde{\chi}}}}^{(1m)}_H,
\end{aligned}
\end{align*}
finishes the proof of \eqref{EQterm21islowerorder}. \\

\ni{\textbf{Control of $\TT_{2.2}$.}} In the following we prove that
\begin{align}
\begin{aligned}
\int\limits_1^2 \TT_{2.2} = \OO^{\TT_{2.2}}_{x,\mfW_0} + \OO(\varep^{9/4}),
\end{aligned}\label{EQTT22estimate}
\end{align}
with
\begin{align*}
\begin{aligned}
\OO^{\TT_{2.2}}_{x,\mfW_0} := -2\varep \int\limits_1^2 r \di \lrpar{ \int\limits_1^{v}\int\limits_1^{v'} \frac{\varphi^2}{16r^3} \vert \mfW_0\vert_{\gac}^2dv''dv'} \cdot (\Om\widehat{\chi})_x dv,
\end{aligned}
\end{align*}
where we note in particular that $\OO^{\TT_{2.2}}_{x,\mfW_0}=\OO_{x,\mfW_0}(\varep^2)$ by \eqref{EQsmallnessCONDbackgroundEPLGnullGluingsectionback}. 

Indeed, from \eqref{EQOmchihtildeexpr} and \eqref{EQexpansionConformalFactor0091} we get that in $L^\infty_vH^5(S_v)$,
\begin{align*}
\begin{aligned}
r^2 \tilde{f}^{-1}\tilde{\Nd}\tilde{f} \cdot \Om\widehat{\tilde{\chi}} =& 2r \di \lrpar{-\varep \int\limits_1^v\int\limits_1^{v'} \frac{\varphi^2}{16r^3} \vert \mfW_0\vert_{\gac}^2 - \varep^{3/2}\int\limits_1^v\int\limits_1^{v'} \frac{\varphi^2}{16r^3} \lrpar{\vert \mfW_1\vert_{\gac}^2+\vert \mfW_2\vert_{\gac}^2}} \\
&\qquad \times \lrpar{\frac{\varphi}{2} \lrpar{\varep^{1/2} \sin\lrpar{\frac{v}{\varep}} \mfW_0 + \varep^{3/4} \sin\lrpar{\frac{v}{\sqrt{\varep}}} \mfW_1 + \varep^{3/4} \cos\lrpar{\frac{v}{\sqrt{\varep}}} \mfW_2} + (\Om\widehat{\chi})_x} \\
&+ \OO(\varep^{9/4}),
\end{aligned}
\end{align*}
where we note that the terms in the first factor (i.e. inside the $\di$-derivative) are low-frequency. Thus, when integrating in $v$ we see the high-frequency improvement due to the high-frequency terms in the second factor, and arrive at
\begin{align*}
\begin{aligned}
\int\limits_1^v r^2 \tilde{f}\tilde{\Nd}\tilde{f} \cdot \Om\widehat{\tilde{\chi}} = \int\limits_1^v 2r \di \lrpar{-\varep \int\limits_1^{v'}\int\limits_1^{v''} \frac{\varphi^2}{16r^3} \vert \mfW_0\vert_{\gac}^2dv'''dv''} \cdot (\Om\widehat{\chi})_x dv' + \OO(\varep^{9/4}),
\end{aligned}
\end{align*}
Projecting the above onto the vectorfields $H^{(1m)}$ for $m=-1,0,1$ finishes the proof of \eqref{EQTT22estimate}.\\

\ni{\textbf{Control of $\TT_{2.3}$.}} In the following we prove that
\begin{align}
\begin{aligned}
\int\limits_1^2 \TT_{2.3} = \OO^{\TT_{2.3}}_x+\OO^{\TT_{2.3}}_{x,\mfW_0} + \OO(\varep^{9/4}),
\end{aligned}\label{EQTT23estimate}
\end{align}
with
\begin{align}
\begin{aligned}
\OO^{\TT_{2.3}}_x :=&\lrpar{ \int\limits_1^2  \lrpar{\lrpar{\frac{\phi^3\Om\trchi}{2\Om^2}}_x-r^2} \Divd_{\gd_x}(\Om\widehat{\chi})_x}^{(1m)}_H, \\
\OO^{\TT_{2.3}}_{x,\mfW_0}:=&\lrpar{- \varep \int\limits_1^2 \lrpar{4r \int\limits_1^v\int\limits_1^{v'} \frac{\varphi^2}{16r^3} \vert \mfW_0\vert_{\gac}^2 + r^2 \int\limits_1^v \frac{\varphi^2}{8r^3} \vert \mfW_0 \vert^2_\gac}\Divd_{\gd_x}(\Om\widehat{\chi})_x}^{(1m)}_H,
\end{aligned}\label{EQdeferrortermtt23}
\end{align}
where we note that by \eqref{EQsmallnessCONDbackgroundEPLGnullGluingsectionback},
\begin{align*}
\begin{aligned}
\OO^{\TT_{2.3}}_x = \OO(\varep^2), \,\, \OO^{\TT_{2.3}}_{x,\mfW_0}=\OO(\varep^2).
\end{aligned}
\end{align*}
To prove \eqref{EQTT22estimate}, recall first the definition of $\TT_{2.3}$,
\begin{align}
\begin{aligned}
\TT_{2.3} := \lrpar{\lrpar{\lrpar{\frac{\phi^3\Om\trchi}{2\Om^2}-\frac{\phi_x^3(\Om\trchi)_x}{2\Om_x^2}}+\lrpar{\frac{\phi_x^3(\Om\trchi)_x}{2\Om_x^2}-r^2}}\Divd_\tgd \lrpar{\Om\widehat{\tilde{\chi}}}}^{(1m)}_H.
\end{aligned}\label{EQrecalldefT23}
\end{align}
On the one hand, we can rewrite, using the definition of the divergence operator, in $L^\infty_vH^5(S_v)$,
\begin{align}
\begin{aligned}
\Divd_\tgd \lrpar{\Om\widehat{\tilde{\chi}}}=& \Divd_{\gd_x} \lrpar{\Om\widehat{\tilde{\chi}}} + \lrpar{\Divd_\tgd \lrpar{\Om\widehat{\tilde{\chi}}}- \Divd_{\gd_x}\lrpar{\Om\widehat{\tilde{\chi}}}} \\
=& \Divd_{\gd_x} \lrpar{\Om\widehat{\tilde{\chi}}} + (\tgd-\gd_x) \prd \lrpar{\Om\widehat{\tilde{\chi}}} + \prd(\tgd-\gd_x) \lrpar{\Om\widehat{\tilde{\chi}}} \\
=& \Divd_{\gd_x} \lrpar{\Om\widehat{\tilde{\chi}}} + \OO(\varep^{5/4})\\
=&\frac{\varphi}{2} \Divd_{\gd_x} \lrpar{ \varep^{1/2} \sin\lrpar{\frac{v}{\varep}} \mfW_0 + \varep^{3/4} \sin\lrpar{\frac{v}{\sqrt{\varep}}} \mfW_1 + \varep^{3/4} \cos\lrpar{\frac{v}{\sqrt{\varep}}} \mfW_2} \\
&+ \Divd_{\gd_x}(\Om\widehat{\chi})_x + \OO(\varep^{5/4}),
\end{aligned}\label{EQestimateingred10010}
\end{align}
where we applied \eqref{EQrep1} and \eqref{EQOmchihtildeexpr}. 

On the other hand, by \eqref{EQphidiffest1} and \eqref{EQrepformulaRaychauduri}, in $L^\infty_vH^6(S_v)$,
\begin{align}
\begin{aligned}
&\frac{\phi^3\Om\trchi}{\Om^2}- \lrpar{\frac{\phi^3\Om\trchi}{\Om^2}}_x \\
=& \lrpar{\frac{\phi\Om\trchi}{\Om^2}}_x\lrpar{\phi^2-\phi_x^2} + \phi^2 \lrpar{\frac{\phi\Om\trchi}{\Om^2}- \lrpar{\frac{\phi\Om\trchi}{\Om^2}}_x \,}\\
=& \lrpar{\frac{\phi\Om\trchi}{\Om^2}}_x\lrpar{\phi-\phi_x}\lrpar{\phi+\phi_x} + \phi^2 \lrpar{\frac{\phi\Om\trchi}{\Om^2}-\lrpar{\frac{\phi\Om\trchi}{\Om^2}}_x \,} \\
=& 4r \lrpar{-\varep \int\limits_1^v\int\limits_1^{v'} \frac{\varphi^2}{16r^3} \vert \mfW_0\vert_{\gac}^2 - \varep^{3/2}\int\limits_1^v\int\limits_1^{v'} \frac{\varphi^2}{16r^3} \lrpar{\vert \mfW_1\vert_{\gac}^2+\vert \mfW_2\vert_{\gac}^2}+\OO\lrpar{\varep^{7/4} }} \\
&+ r^2 \lrpar{- \varep \int\limits_1^v \frac{\varphi^2}{8r^3} \vert \mfW_0 \vert^2_\gac - \varep^{3/2} \int\limits_1^v \frac{\varphi^2}{8r^3} \lrpar{\vert \mfW_1\vert^2_\gac + \vert \mfW_2\vert^2_\gac} dv' + \OO(\varep^{7/4})} + \OO(\varep^2),
\end{aligned}\label{EQestimateingred20010}
\end{align}
where we emphasize that the terms of size $\varep$ and $\varep^{3/2}$ in \eqref{EQestimateingred20010} are low-frequency.

Combining \eqref{EQestimateingred10010} and \eqref{EQestimateingred20010} in \eqref{EQrecalldefT23} yields, in $L^\infty_vH^5(S_v)$,
\begin{align*}
\begin{aligned}
&\TT_{2.3} \\
=& \lrpar{\lrpar{\frac{\phi^3\Om\trchi}{\Om^2}}_x -r^2} \lrpar{\frac{\varphi}{2} \Divd_{\gd_x} \lrpar{ \varep^{1/2} \sin\lrpar{\frac{v}{\varep}} \mfW_0 + \varep^{3/4} \sin\lrpar{\frac{v}{\sqrt{\varep}}} \mfW_1 + \varep^{3/4} \cos\lrpar{\frac{v}{\sqrt{\varep}}} \mfW_2} } \\
&+ \lrpar{\lrpar{\frac{\phi^3\Om\trchi}{2\Om^2}}_x-r^2} \Divd_{\gd_x}(\Om\widehat{\chi})_x + \OO(\varep^{9/4})\\
&- 4r \lrpar{\varep \int\limits_1^v\int\limits_1^{v'} \frac{\varphi^2}{16r^3} \vert \mfW_0\vert_{\gac}^2 + \varep^{3/2}\int\limits_1^v\int\limits_1^{v'} \frac{\varphi^2}{16r^3} \lrpar{\vert \mfW_1\vert_{\gac}^2+\vert \mfW_2\vert_{\gac}^2}} \\
&\qquad\,\,\, \times \lrpar{\frac{\varphi}{2} \Divd_{\gd_x} \lrpar{ \varep^{1/2} \sin\lrpar{\frac{v}{\varep}} \mfW_0 + \varep^{3/4} \sin\lrpar{\frac{v}{\sqrt{\varep}}} \mfW_1 + \varep^{3/4} \cos\lrpar{\frac{v}{\sqrt{\varep}}} \mfW_2} + \Divd_{\gd_x}(\Om\widehat{\chi})_x} \\
&- r^2 \lrpar{ \varep \int\limits_1^v \frac{\varphi^2}{8r^3} \vert \mfW_0 \vert^2_\gac + \varep^{3/2} \int\limits_1^v \frac{\varphi^2}{8r^3} \lrpar{\vert \mfW_1\vert^2_\gac + \vert \mfW_2\vert^2_\gac} dv' } \\
&\qquad\,\,\, \times \lrpar{ \frac{\varphi}{2} \Divd_{\gd_x} \lrpar{ \varep^{1/2} \sin\lrpar{\frac{v}{\varep}} \mfW_0 + \varep^{3/4} \sin\lrpar{\frac{v}{\sqrt{\varep}}} \mfW_1 + \varep^{3/4} \cos\lrpar{\frac{v}{\sqrt{\varep}}} \mfW_2}+ \Divd_{\gd_x}(\Om\widehat{\chi})_x}.
\end{aligned}
\end{align*}
Integrating the above in $v$ and using the high-frequency improvement, we arrive at
\begin{align*}
\begin{aligned}
\int\limits_1^2 \TT_{2.3} =& \int\limits_1^2  \lrpar{\lrpar{\frac{\phi^3\Om\trchi}{2\Om^2}}_x-r^2} \Divd_{\gd_x}(\Om\widehat{\chi})_x\\
&- \varep \int\limits_1^2 \lrpar{4r \int\limits_1^v\int\limits_1^{v'} \frac{\varphi^2}{16r^3} \vert \mfW_0\vert_{\gac}^2 + r^2 \int\limits_1^v \frac{\varphi^2}{8r^3} \vert \mfW_0 \vert^2_\gac}\Divd_{\gd_x}(\Om\widehat{\chi})_x  + \OO(\varep^{9/4}).
\end{aligned}
\end{align*}
Projecting the above onto the vectorfields $H^{(1m)}$ for $m=-1,0,1$ finishes the proof of \eqref{EQTT22estimate}. \\

\ni{\textbf{Control of $\TT_{2.4}$.}} The rest of this section is concerned with the analysis of the term $\TT_{2.4}$, that is,
\begin{align}
\begin{aligned}
\TT_{2.4}:= \lrpar{r^2 \Divd_\tgd \lrpar{\Om\widehat{\tilde{\chi}}}}^{(1m)}_H.
\end{aligned}\label{EQrecalldefT24}
\end{align}
This term lies at the heart of our gluing of the charge $\mathbf{L}$ at the order $\varep^2$ in Section \ref{SEC1EPLGgluing}, and thus requires a precise analysis. Specifically, we show in the following that
\begin{align}
\begin{aligned}
\int\limits_1^2 \TT_{2.4} =& \varep^2\int\limits_1^2 \frac{\varphi^2}{r^2}\lrpar{-\half  \lrpar{\Nd_A\lrpar{\mfW_1^{AB} (\mfW_2)_{B\cdot} }}^{(1m)}_H + \frac{1}{4}\lrpar{ \Nd_\cdot \lrpar{\mfW_1}_{AD}\lrpar{\mfW_2}^{AD}}^{(1m)}_H}\\
&+ \OO^{\TT_{2.4}}_{x,\mfW_0}+\OO(\varep^{9/4}),
\end{aligned}\label{EQLLLcontrol4}
\end{align}
where term $\OO^{\TT_{2.4}}_{x,\mfW_0}$ in \eqref{EQLLLcontrol4} is of size $\varep^2$, see its precise structure in \eqref{EQpreciseerrorTT24}. To prove \eqref{EQLLLcontrol4} we first discuss in an interlude the divergence-operator $\Divd_{\tilde{g}}$, and then turn to the estimate \eqref{EQLLLcontrol4} for $\TT_4$. \\

\ni\emph{\underline{Interlude:} The divergence operator $\Divd_{\tgd}$.} In the following we prove that for any $2$-covariant symmetric tensor $M_{AB} \in H^6(S_v)$ it holds that, in $H^5(S_v)$, with $\overset{\circ}{M}{}^{AD} := (r^2\gac)^{AE}(r^2\gac)^{DF}M_{EF}$ and $\overset{\circ}{M}{}^A_{\,\,\,\,\,\,E} := (r^2\gac)^{AB}M_{BE}$,
\begin{align}
\begin{aligned}
\lrpar{\Divd_{\tilde{\gd}} M}_C =& \lrpar{\Divd_{r^2\gac}M}_{C} - \varphi \overset{\circ}{\Nd}_A\lrpar{{W}_{AB}\overset{\circ}{M}{}^{BC}} + \half \varphi \overset{\circ}{\Nd}_C W_{DA} \overset{\circ}{M}{}^{AD} \\
&+ \overset{\circ}{\Nd}_A\lrpar{\lrpar{\gd_x^{AB}-(r^2\gac)^{AB}}M_{BC}} + \lrpar{(\Ga_x)^{A}_{AE}-\overset{\circ}{\Ga}{}^A_{AE}} \overset{\circ}{M}{}^E_{\,\,\,\,\,\,C}\\
&+ \lrpar{\OO^{\PP_3}_x}_{AC}^E \overset{\circ}{M}{}^A_{\,\,\,\,\,\,E} + \OO\lrpar{\varep^{2} \Vert M \Vert_{H^6(S_v)}},
\end{aligned}\label{EQdivOperatorExpansion}
\end{align}
where the term $\OO^{\PP_3}_x$ is of size $\varep$, see \eqref{EQdefOOPP3} for its definition.

Indeed, by definition of the divergence operator we have that, with $\tilde{M}^A_{\,\,\,\,\,\,C}:= \tgd^{AB}M_{BC}$, 
\begin{align}
\begin{aligned}
\lrpar{\Divd_{\tilde{\gd}} M}_C =& \tgd^{AB}\tilde{\Nd}_A M_{BC} \\
=& \tilde{\Nd}_A \tilde{M}^A_{\,\,\,\,\,\,C}\\
=& \pr_A\lrpar{\tilde{M}^A_{\,\,\,\,\,\,C}} + \tilde{\Ga}^A_{AE} \tilde{M}^E_{\,\,\,\,\,\,C} - \tilde{\Ga}^E_{CA} \tilde{M}^A_{\,\,\,\,\,\,E}\\
=& \underbrace{\overset{\circ}{\Nd}_A \tilde{M}^A_{\,\,\,\,\,\,C}}_{:=\PP_1} + \underbrace{\lrpar{\tilde{\Ga}^A_{AE}-\overset{\circ}{\Ga}{}^A_{AE}} \tilde{M}^E_{\,\,\,\,\,\,C}}_{:=\PP_2} -\underbrace{\lrpar{\tilde{\Ga}^E_{AC}-{\overset{\circ}{\Ga}}{}^E_{AC}}\tilde{M}^A_{\,\,\,\,\,\,E}}_{:=\PP_3}.
\end{aligned}\label{EQdivanalysis1}
\end{align}
where $\tilde{\Nd}$ and $\tilde{\Gamma}$ denote the covariant derivative and Christoffel symbols with respect to $\tilde{\gd}$, and $\overset{\circ}{\Ga}$ and $\overset{\circ}{\Nd}$ denote the Christoffel symbols and the covariant derivative with respect to $r^2\gac$. In the following we analyze $\PP_1, \PP_2$ and $\PP_3$.

First, consider the term $\PP_1$ on the right-hand side of \eqref{EQdivanalysis1}. We have that
\begin{align}
\begin{aligned}
\PP_1 := \overset{\circ}{\Nd}_A\tilde{M}^A_{\,\,\,\,\,\,C} =& \overset{\circ}{\Nd}_A\overset{\circ}{M}{}^A_{\,\,\,\,\,\,C} + \overset{\circ}{\Nd}_A\lrpar{\tilde{M}^A_{\,\,\,\,\,\,C} - \overset{\circ}{M}{}^A_{\,\,\,\,\,\,C}} \\
=&\overset{\circ}{\Nd}_A\overset{\circ}{M}{}^A_{\,\,\,\,\,\,C} + \overset{\circ}{\Nd}_A\lrpar{\lrpar{\tgd^{AB}-(r^2\gac)^{AB}}M_{BC}} \\
=& \lrpar{\Divd_{r^2\gac}M}_{C} + \overset{\circ}{\Nd}_A\lrpar{\lrpar{\lrpar{\tgd^{AB}-\gd_x^{AB}}+\lrpar{\gd_x^{AB}-(r^2\gac)^{AB}}}M_{BC}}.
\end{aligned}\label{EQexpansionPP1}
\end{align}
By \eqref{EQrep1} and \eqref{EQsmallnessCONDbackgroundEPLGnullGluingsectionback} it holds that in $L^\infty_vH^6(S_v)$,
\begin{align}
\begin{aligned}
\tgd^{AB}-\gd_x^{AB} =& -\gd_x^{AC}\gd_x^{BD}(\varphi W)_{CD} + \OO(\varep^{9/4})\\
=&-(r^2\gac)^{AC}(r^2\gac)^{BD}(\varphi W)_{CD} + \OO(\varep^{9/4}),
\end{aligned}\label{EQinvMetricdiff00101}
\end{align}
so that from \eqref{EQexpansionPP1} we get that in $L^\infty_vH^5(S_v)$,
\begin{align}
\begin{aligned}
\PP_1 =& \lrpar{\Divd_{r^2\gac}M}_{C} - \varphi \overset{\circ}{\Nd}_A\lrpar{\overset{\circ}{W}{}^{AB}M_{BC}} \\
&+ \overset{\circ}{\Nd}_A\lrpar{\lrpar{\gd_x^{AB}-(r^2\gac)^{AB}}M_{BC}} + \OO\lrpar{\varep^{9/4} \Vert M \Vert_{H^6(S_v)}},
\end{aligned}\label{EQPP1expansionfinal}
\end{align}
where $\overset{\circ}{W}{}^{AB}:= (r^2\gac)^{AC}(r^2\gac)^{BD}(\varphi W)_{CD}$.

Second, consider the term $\PP_{2}$ on the right-hand side of \eqref{EQdivanalysis1}, that is,
\begin{align*}
\begin{aligned}
\PP_2:= \lrpar{\tilde{\Ga}^A_{AE}-\overset{\circ}{\Ga}{}^A_{AE}} \tilde{M}^E_{\,\,\,\,\,\,C}.
\end{aligned}
\end{align*}
From \eqref{EQdeterminantEstimatecloseness0009} and the classical differential geometry identity 
$${\Ga}^A_{AE} = \sqrt{\det \gd}^{-1}\pr_E \lrpar{\sqrt{\det\gd}},$$ 
it follows that in $L^\infty_vH^6(S_v)$,
\begin{align*}
\begin{aligned}
\tilde{\Ga}^A_{AE}-\overset{\circ}{\Ga}{}^A_{AE} =& \lrpar{\tilde{\Ga}^A_{AE}-(\Ga_x)^{A}_{AE}}+\lrpar{(\Ga_x)^{A}_{AE}-\overset{\circ}{\Ga}{}^A_{AE}} = \OO(\varep^{9/4}) + \lrpar{(\Ga_x)^{A}_{AE}-\overset{\circ}{\Ga}{}^A_{AE}},
\end{aligned}
\end{align*}
so that (using \eqref{EQsmallnessCONDbackgroundEPLGnullGluingsectionback} and \eqref{EQinvMetricdiff00101}) in $L^\infty_vH^6(S_v)$,
\begin{align}
\begin{aligned}
\PP_2 := \lrpar{\tilde{\Ga}^A_{AE}-\overset{\circ}{\Ga}{}^A_{AE}} \tilde{M}^E_{\,\,\,\,\,\,C} =& \lrpar{(\Ga_x)^{A}_{AE}-\overset{\circ}{\Ga}{}^A_{AE}} \tilde{M}^E_{\,\,\,\,\,\,C} + \OO\lrpar{\varep^{9/4} \Vert M\Vert_{H^6(S_v)}} \\
=& \lrpar{(\Ga_x)^{A}_{AE}-\overset{\circ}{\Ga}{}^A_{AE}} \overset{\circ}{M}{}^E_{\,\,\,\,\,\,C} + \OO\lrpar{\varep^{2} \Vert M\Vert_{H^6(S_v)}}.
\end{aligned}\label{EQPP2expansionfinal}
\end{align} 
\ni Third, consider the term $\PP_3$ on the right-hand side of \eqref{EQdivanalysis1}, that is,
\begin{align}
\begin{aligned}
\PP_3 := \lrpar{\tilde{\Ga}^E_{AC}-{\overset{\circ}{\Ga}}{}^E_{AC}}\tilde{M}^A_{\,\,\,\,\,\,E}.
\end{aligned}\label{EQdefpp3seclate}
\end{align}
On the one hand, as above we can rewrite
\begin{align}
\begin{aligned}
\tilde{M}^A_{\,\,\,\,\,\,E} = \overset{\circ}{M}{}^A_{\,\,\,\,\,\,E} + \lrpar{\lrpar{\tgd^{AB}-\gd_x^{AB}}+\lrpar{\gd_x^{AB}-(r^2\gac)^{AB}}}M_{BC}.
\end{aligned}\label{EQexppp3term}
\end{align}

Using \eqref{EQrep1} we can express
\begin{align*}
\begin{aligned}
\tilde{\Ga}^E_{CA} =& \half \tgd^{ED}\lrpar{\pr_C\tgd_{DA}+\pr_A\tgd_{DC}-\pr_D\tgd_{CA}} \\
=& \half \lrpar{\lrpar{\tgd^{ED}-\gd_x^{ED}}+\lrpar{\gd_x^{ED}-(r^2\gac)^{ED}}+(r^2\gac)^{ED}} \\
&\qquad \times \lrpar{\varphi \lrpar{\pr_C W_{DA}+\pr_A W_{DC}-\pr_D W_{CA}} + \lrpar{\pr_C(r^2\gac)_{DA}+\pr_A(r^2\gac)_{DC}-\pr_D(r^2\gac)_{CA}}} \\
&+ \half \lrpar{\lrpar{\tgd^{ED}-\gd_x^{ED}}+\lrpar{\gd_x^{ED}-(r^2\gac)^{ED}}+(r^2\gac)^{ED}} \\
&\qquad \times \lrpar{\pr_C(\gd_x-r^2\gac)_{DA}+\pr_A(\gd_x-r^2\gac)_{DC}-\pr_D(\gd_x-r^2\gac)_{CA}},
\end{aligned}
\end{align*}
so that by \eqref{EQsmallnessCONDbackgroundEPLGnullGluingsectionback}, \eqref{EQWansatz} and \eqref{EQinvMetricdiff00101},
\begin{align}
\begin{aligned}
\tilde{\Ga}^E_{CA}-\overset{\circ}{\Ga}{}^E_{CA} =& \half \lrpar{\tgd^{ED}-\gd_x^{ED}} \lrpar{\pr_C(r^2\gac)_{DA}+\pr_A(r^2\gac)_{DC}-\pr_D(r^2\gac)_{CA}} \\
&+ \half (r^2\gac)^{ED} \varphi \lrpar{\pr_C W_{DA}+\pr_A W_{DC}-\pr_D W_{CA}} +\OO_{\PP_3} + \OO(\varep^{9/4}) \\
=& - (r^2\gac)^{EG}\varphi W_{GN} \overset{\circ}{\Ga}{}^N_{AC} + \half (r^2\gac)^{ED} \varphi \lrpar{\pr_C W_{DA}+\pr_A W_{DC}-\pr_D W_{CA}} \\
&+\lrpar{\OO^{\PP_3}_x}_{AC}^E + \OO(\varep^{9/4}),
\end{aligned}\label{EQexpansiondiffChrSym001012}
\end{align}
with the term $\lrpar{\OO^{\PP_3}_x}_{AC}^E$ being given by
\begin{align}
\begin{aligned}
&\lrpar{\OO^{\PP_3}_x}_{AC}^E\\
 :=& \half \lrpar{\gd_x^{ED}-(r^2\gac)^{ED}} \lrpar{\pr_C(r^2\gac)_{DA}+\pr_A(r^2\gac)_{DC}-\pr_D(r^2\gac)_{CA}} \\
&+\half \lrpar{\gd_x^{ED}-(r^2\gac)^{ED}} \lrpar{\pr_C(\gd_x-r^2\gac)_{DA}+\pr_A(\gd_x-r^2\gac)_{DC}-\pr_D(\gd_x-r^2\gac)_{CA}} \\
&+ \half (r^2\gac)^{ED} \lrpar{\pr_C(\gd_x-r^2\gac)_{DA}+\pr_A(\gd_x-r^2\gac)_{DC}-\pr_D(\gd_x-r^2\gac)_{CA}},
\end{aligned}\label{EQdefOOPP3}
\end{align}
where we note that $\lrpar{\OO^{\PP_3}_x}_{AC}^E= \OO(\varep)$ and depends only on $x$.

Moreover, we observe that the first line on the right-hand side of \eqref{EQexpansiondiffChrSym001012} satisfies the property that 
\begin{align}
\begin{aligned}
&\lrpar{- (r^2\gac)^{EG}\varphi W_{GN} \overset{\circ}{\Ga}{}^N_{AC} + \half (r^2\gac)^{ED} \varphi \lrpar{\pr_C W_{DA}+\pr_A W_{DC}-\pr_D W_{CA}}} \overset{\circ}{M}{}^A_{\,\,\,\,\,\,E} \\
=& - \varphi W_{DB} \overset{\circ}{\Ga}{}^B_{AC} \overset{\circ}{M}{}^{AD} + \half \varphi \pr_C W_{DA}\overset{\circ}{M}{}^{AD} \\
=& \half \varphi \lrpar{\pr_C W_{DA} - \overset{\circ}{\Ga}{}^B_{CA} W_{DB}-\overset{\circ}{\Ga}{}^B_{CD} W_{AB} } \overset{\circ}{M}{}^{AD}\\
=& \half \varphi \overset{\circ}{\Nd}_C W_{DA} \overset{\circ}{M}{}^{AD},
\end{aligned}\label{EQcancellation}
\end{align}
where we denote $\overset{\circ}{M}{}^{AD} := (r^2\gac)^{AB}(r^2\gac)^{DC} M_{BC}$ and used that $\overset{\circ}{M}{}^{AD}=\overset{\circ}{M}{}^{DA}$.

Plugging \eqref{EQexppp3term} and \eqref{EQexpansiondiffChrSym001012} into \eqref{EQdefpp3seclate} and using \eqref{EQcancellation}, we thus arrive at 
\begin{align}
\begin{aligned}
\PP_3 =&\half \varphi \overset{\circ}{\Nd} W_{DA} \overset{\circ}{M}{}^{AD} + \lrpar{\OO^{\PP_3}_x}_{AC}^E \overset{\circ}{M}{}^A_{\,\,\,\,\,\,E} + \OO(\varep^2 \Vert M \Vert_{H^6(S_v)}).
\end{aligned}\label{EQestimationpp3final}
\end{align}
Plugging the above estimates \eqref{EQPP1expansionfinal}, \eqref{EQPP2expansionfinal} and \eqref{EQestimationpp3final} into \eqref{EQdivanalysis1} leads to, in $H^5(S_v)$,
\begin{align*}
\begin{aligned}
\lrpar{\Divd_{\tilde{\gd}} M}_C =& \lrpar{\Divd_{r^2\gac}M}_{C} - \varphi \overset{\circ}{\Nd}_A\lrpar{\overset{\circ}{W}{}^{AB}M_{BC}} + \half \varphi \overset{\circ}{\Nd}_C W_{DA} \overset{\circ}{M}{}^{AD} \\
&+ \overset{\circ}{\Nd}_A\lrpar{\lrpar{\gd_x^{AB}-(r^2\gac)^{AB}}M_{BC}} + \lrpar{(\Ga_x)^{A}_{AE}-\overset{\circ}{\Ga}{}^A_{AE}} \overset{\circ}{M}{}^E_{\,\,\,\,\,\,C}+ \lrpar{\OO^{\PP_3}_x}_{AC}^E \overset{\circ}{M}{}^A_{\,\,\,\,\,\,E} \\
&+ \OO\lrpar{\varep^{2} \Vert M \Vert_{H^6(S_v)}}.
\end{aligned}
\end{align*}
and thus finishes the proof of \eqref{EQdivOperatorExpansion}. We now return to the proof of \eqref{EQLLLcontrol4}. \\

\ni \underline{\emph{Proof of the estimate \eqref{EQLLLcontrol4} for $\TT_{2.4}$.}} Recalling the definition of $\TT_{2.4}$ in \eqref{EQrecalldefT24}, and applying \eqref{EQOmchihtildeexpr} and \eqref{EQdivOperatorExpansion}, we get that 
\begin{align*}
\begin{aligned}
\TT_{2.4} :=& r^2\lrpar{\Divd_{\tilde{\gd}} \lrpar{\Om\widehat{\tilde{\chi}} \, }}^{(1m)}_H\\
=& r^2\underbrace{\lrpar{\Divd_{r^2\gac} \lrpar{\Om\widehat{\tilde{\chi}} \, }}^{(1m)}_H}_{=0} - r^2\varphi \underbrace{\lrpar{ \overset{\circ}{\Nd}_A \lrpar{ W_{AB} \overset{\circ}{\lrpar{\Om\widehat{\tilde{\chi}}}}{}^{B\cdot}}}^{(1m)}_H}_{:= \TT_{2.4.1}} + \frac{r^2\varphi}{2} \underbrace{\lrpar{ \overset{\circ}{\Nd}_{\cdot} W_{DA} \overset{\circ}{\lrpar{\Om\widehat{\tilde{\chi}}}}{}^{AD}}^{(1m)}_H}_{:=\TT_{2.4.2}} \\
& +\underbrace{\lrpar{ \overset{\circ}{\Nd}_A\lrpar{\lrpar{\gd_x^{AB}-(r^2\gac)^{AB}}\lrpar{\Om\widehat{\tilde{\chi}}}_{B\cdot}} + \lrpar{(\Ga_x)^{A}_{AE}-\overset{\circ}{\Ga}{}^A_{AE}} (r^2\gac)^{ED}(\Om\widehat{\tilde{\chi}})_{D\cdot} }^{(1m)}_H}_{:=\TT_{2.4.3}}\\
&+ \underbrace{\lrpar{\lrpar{\OO^{\PP_3}_x}_{A\cdot}^E (r^2\gac)^{AD}(\Om\widehat{\tilde{\chi}})_{DE} }^{(1m)}_H}_{:= \TT_{2.4.4}} + \OO(\varep^{9/4}),
\end{aligned}
\end{align*}
where we used that for any $2$-covariant symmetric tensor $T$, 
\begin{align*}
\begin{aligned}
(\Divd_{r^2\gac} V)^{(1m)}_H =0, \,\, \text{ for } m=-1,0,1;
\end{aligned}
\end{align*}
here we refer to the Fourier theory in Appendix \ref{APPconstructionW} for the definition of the projection onto the vectorfield $H^{(1m)}$, and recall, in particular, that $(\di f)_H=0$ for any scalar function $f$. 

We immediately observe that by \eqref{EQOmchihtildeexpr}, due to similar high-frequency improvement as before (note that $\OO_{\PP_3}$ is low-frequency, see the expression \eqref{EQdefOOPP3}), it holds that
\begin{align}
\begin{aligned}
\int\limits_1^2 (\TT_{2.4.3} + \TT_{2.4.4}) dv= \OO^{\TT_{2.4.3}}_x+ \OO^{\TT_{2.4.4}}_x + \OO(\varep^{9/4}),
\end{aligned}\label{EQprelimcontrolTT24}
\end{align}
where
\begin{align*}
\begin{aligned}
\OO^{\TT_{2.4.3}}_x:=& \int\limits_1^2 \lrpar{ \overset{\circ}{\Nd}_A\lrpar{\lrpar{\gd_x^{AB}-(r^2\gac)^{AB}}\lrpar{(\Om\widehat{{\chi}})_x}_{B\cdot}} + \lrpar{(\Ga_x)^{A}_{AE}-\overset{\circ}{\Ga}{}^A_{AE}} (r^2\gac)^{ED}((\Om\widehat{{\chi}})_x)_{D\cdot} }^{(1m)}_H,\\
\OO^{\TT_{2.4.4}}_x:=& \int\limits_1^2 \lrpar{\lrpar{\OO^{\PP_3}_x}_{A\cdot}^E (r^2\gac)^{AD}((\Om\widehat{{\chi}})_x)_{DE} }^{(1m)}_H.
\end{aligned}
\end{align*}
Thus, with view onto proving \eqref{EQLLLcontrol4}, it remains only to discuss the terms $\TT_{2.4.1}$ and $\TT_{2.4.2}$.  \\

\ni \textbf{Notation.} In the following control of $\TT_{2.4.1}$ and $\TT_{2.4.2}$, all differential operators on the sphere and raising of indices are with respect to the round metric $r^2\gac$. \\ 

\ni \emph{Control of $\TT_{2.4.1}$.} In the following we prove that
\begin{align}
\begin{aligned}
\int\limits_1^2 \TT_{2.4.1} = - \int\limits_{1}^2 \half \varep^2 \lrpar{\Nd_A\lrpar{\mfW_1^{AB} (\mfW_2)_{B\cdot} }}^{(1m)}_H + \OO(\varep^{9/4}),
\end{aligned}\label{EQclaimTT241integralproperty}
\end{align}
where we recall that $\mfW_1$ and $\mfW_2$ appear in the high-frequency ansatz \eqref{EQWansatz} for $W$, that is,
\begin{align*}
\begin{aligned}
W := -\varep^{3/2} \cos\lrpar{\frac{v}{\varep}} \mfW_0 - \varep^{5/4} \lrpar{\cos\lrpar{\frac{v}{\sqrt{\varep}}} \mfW_1 - \sin\lrpar{\frac{v}{\sqrt{\varep}}} \mfW_2}.
\end{aligned}
\end{align*}
To prove \eqref{EQclaimTT241integralproperty}, we first we note that for any symmetric $2$-covariant tensor $M$ on the sphere,
\begin{align}
\begin{aligned}
\lrpar{\Nd_A \lrpar{ M^{AB} M_{B\cdot}}}^{(1m)}_H := \int (H^{1m})^C \Nd_A\lrpar{M^{AB} M_{BC}}=  - \int \Nd_AH^{1m}_C M^{AB} M_{B}^{\,\,\,\,C} = 0,
\end{aligned}\label{EQvanishingP1tensoranalysis}
\end{align}
where we used the symmetry $M^{AB} M_{B}^{\,\,\,\,C}=M^{CB} M_{B}^{\,\,\,\,A}$ and that $\Nd_A H^{(1m)}_C=-\Nd_C H^{(1m)}_A$ because $H^{(1m)}$ is a Killing vectorfield of the round sphere for $m=-1,0,1$; see Appendix \ref{APPconstructionW}. 

Expanding $\TT_{2.4.1}$ with the formulas \eqref{EQWansatz} for $W$ and \eqref{EQOmchihtildeexpr} for $\Om\widehat{\tilde{\chi}}$, and applying \eqref{EQvanishingP1tensoranalysis}, yields
\begin{align}
\begin{aligned}
\TT_{2.4.1} :=& \lrpar{ \Nd_A \lrpar{ W_{AB} \lrpar{\Om\widehat{\tilde{\chi}}}^{B\cdot}}}^{(1m)}_H \\
=& \lrpar{\TT_{2.4.1.a}}^{(1m)}_H + \lrpar{\TT_{2.4.1.b}}^{(1m)}_H + \lrpar{\TT_{2.4.1.c}}^{(1m)}_H +\OO(\varep^{9/4}),
\end{aligned}\label{EQP1expansion1}
\end{align}
where the terms $\TT_{2.4.1.a}$, $\TT_{2.4.1.b}$ and $\TT_{2.4.1.c}$ are
\begin{align*}
\begin{aligned}
\TT_{2.4.1.a} =& \Nd_A\lrpar{ \varep^{5/4} \lrpar{-\cos\lrpar{\frac{v}{\sqrt{\varep}}}\mfW_1 + \sin\lrpar{\frac{v}{\sqrt{\varep}}} \mfW_2 }^{AB}\lrpar{\frac{\varphi \varep^{1/2}}{2}\sin\lrpar{\frac{v}{\varep}}\mfW_0}_{B\cdot}  },\\
\TT_{2.4.1.b} =& \Nd_A\lrpar{ \lrpar{- \varep^{3/2}\cos\lrpar{\frac{v}{\varep}}\mfW_0}^{AB}\lrpar{\varep^{3/4}\lrpar{\sin\lrpar{\frac{v}{\sqrt{\varep}}}\mfW_1 + \cos\lrpar{\frac{v}{\sqrt{\varep}}} \mfW_2 } }_{B\cdot}}, \\
\TT_{2.4.1.c} =& \half \varep^2 \Nd_A\lrpar{ \sin^2\lrpar{\frac{v}{\sqrt{\varep}}} \mfW_2^{AB} (\mfW_1)_{B\cdot} - \cos^2\lrpar{\frac{v}{\sqrt{\varep}}} \mfW_1^{AB} (\mfW_2)_{B\cdot} }.
\end{aligned}
\end{align*}
First, using that $\cos(\frac{v}{\varep^{1/2}})\cdot \sin(\frac{v}{\varep})$ and $\sin(\frac{v}{\varep^{1/2}})\cdot \cos(\frac{v}{\varep})$ are high-frequency functions (see Section \ref{SECpreliminariesHIGHfrequency}), it holds that
\begin{align*}
\begin{aligned}
-\int\limits_1^2 r^2 \varphi \TT_{2.4.1.a} = \OO(\varep^{9/4}),
\end{aligned}
\end{align*}
Second, we can directly estimate $\TT_{2.4.1.b}$ to be a higher-order term,
\begin{align*}
\begin{aligned}
\TT_{2.4.1.b} = \OO(\varep^{9/4}).
\end{aligned}
\end{align*}
Third, using that $\cos^2\lrpar{\frac{v}{\sqrt{\varep}}} = 1- \sin^2\lrpar{\frac{v}{\sqrt{\varep}}}$, it holds that 
\begin{align*}
\begin{aligned}
\lrpar{\TT_{2.4.1.c}}^{(1m)}_H =& \half \varep^2 \lrpar{ \Nd_A\lrpar{ \sin^2\lrpar{\frac{v}{\sqrt{\varep}}} \mfW_2^{AB} (\mfW_1)_{B\cdot} - \cos^2\lrpar{\frac{v}{\sqrt{\varep}}} \mfW_1^{AB} (\mfW_2)_{B\cdot} }}^{(1m)}_H\\
=& \half \varep^2 \lrpar{\Nd_A\lrpar{ \sin^2\lrpar{\frac{v}{\sqrt{\varep}}} (\mfW_2^{AB} (\mfW_1)_{B\cdot} + \mfW_1^{AB} (\mfW_2)_{B\cdot}) -  \mfW_1^{AB} (\mfW_2)_{B\cdot} }}^{(1m)}_H \\
=& -\half \varep^2 \lrpar{\Nd_A\lrpar{\mfW_1^{AB} (\mfW_2)_{B\cdot} }}^{(1m)}_H,
\end{aligned}
\end{align*}
where we used the symmetry $\mfW_2^{AB} (\mfW_1)_{B}^{\,\,\,\,\,C} + \mfW_1^{AB} (\mfW_2)_{B}^{\,\,\,\,\,C}=\mfW_2^{CB} (\mfW_1)_{B}^{\,\,\,\,\,A} + \mfW_1^{CB} (\mfW_2)_{B}^{\,\,\,\,\,A}$, and argued as in \eqref{EQvanishingP1tensoranalysis}. Integrating \eqref{EQP1expansion1} in $v$ and applying the above finishes the proof of \eqref{EQclaimTT241integralproperty}.\\ 

\ni \emph{Control of $\TT_{2.4.2}$.} In the following we prove that
\begin{align}
\begin{aligned}
\int\limits_1^2 \TT_{2.4.2} = -\int\limits_1^2 \half \varep^{2}\lrpar{ \Nd_\cdot \lrpar{\mfW_1}_{AD}\lrpar{\mfW_2}^{AD}}^{(1m)}_H + \OO(\varep^{9/4}).
\end{aligned}\label{EQpropertyTT242}
\end{align}

\ni Indeed, first we note that for any symmetric $2$-covariant tensor $M$,
\begin{align}
\begin{aligned}
\lrpar{\Nd_\cdot M_{AD} M^{AD}}^{(1m)}_H = \lrpar{\di \lrpar{\vert M \vert^2_{r^2\gac}}}^{(1m)}_H =0,
\end{aligned}\label{EQdinoHpartMM}
\end{align}
Expanding $\TT_{2.4.2}$ with the formulas \eqref{EQWansatz} for $W$ and \eqref{EQOmchihtildeexpr} for $\Om\widehat{\tilde{\chi}}$, and applying \eqref{EQdinoHpartMM}, yields
\begin{align}
\begin{aligned}
\TT_{2.4.2} :=& \lrpar{ \Nd_{\cdot} W_{DA} \lrpar{\Om\widehat{\tilde{\chi}}}^{AD}}^{(1m)}_H\\
=&  \lrpar{ \TT_{2.4.2.a}}^{(1m)}_H+ \lrpar{\TT_{2.4.2.b}}^{(1m)}_H + \lrpar{\TT_{2.4.2.c}}^{(1m)}_H + \OO(\varep^{9/4}),
\end{aligned}\label{EQP2expansion2}
\end{align}
where the terms $\TT_{2.4.2.a}$, $\TT_{2.4.2.b}$, $\TT_{2.4.2.c}$ are
\begin{align*}
\begin{aligned}
\TT_{2.4.2.a} :=& \Nd_C  \lrpar{\varep^{5/4}\lrpar{-\cos\lrpar{\frac{v}{\sqrt{\varep}}}\mfW_1 +\sin\lrpar{\frac{v}{\sqrt{\varep}}}\mfW_2}}_{AD}\lrpar{\frac{\varphi\varep^{1/2}}{2}\sin\lrpar{\frac{v}{\varep}} \mfW_0}^{AD}, \\
\TT_{2.4.2.b} :=& \Nd_C \lrpar{-\varep^{3/2} \cos\lrpar{\frac{v}{\varep}}\mfW_0}_{AD} \lrpar{\varep^{3/4}\lrpar{\sin\lrpar{\frac{v}{\sqrt{\varep}}}\mfW_1 +\cos\lrpar{\frac{v}{\sqrt{\varep}}}\mfW_2}}^{AD}, \\
\TT_{2.4.2.c} :=& \half \varep^{2}\lrpar{-\cos^2\lrpar{\frac{v}{\sqrt{\varep}}}  \Nd_C  \lrpar{\mfW_1}_{AD}\lrpar{\mfW_2}^{AD}+\sin^2\lrpar{\frac{v}{\sqrt{\varep}}} \Nd_C  \lrpar{\mfW_2}_{AD}\lrpar{\mfW_1}^{AD}}.
\end{aligned}
\end{align*}
First, using that $\cos(\frac{v}{\varep^{1/2}})\cdot \sin(\frac{v}{\varep})$ and $\sin(\frac{v}{\varep^{1/2}})\cdot \cos(\frac{v}{\varep})$ are high-frequency functions (see Section \ref{SECpreliminariesHIGHfrequency}), it holds that
\begin{align*}
\begin{aligned}
\int\limits_1^2 r^2 \varphi \TT_{2.4.2.a} = \OO(\varep^{9/4}),
\end{aligned}
\end{align*}
Second, we directly have that $\TT_{2.4.2.b}$ is a higher-order term,
\begin{align*}
\begin{aligned}
\TT_{2.4.2.b} = \OO(\varep^{9/4}).
\end{aligned}
\end{align*}
Third, it holds that 
\begin{align*}
\begin{aligned}
\lrpar{\TT_{2.4.2.c}}^{(1m)}_H :=& \half \varep^{2}\lrpar{-\cos^2\lrpar{\frac{v}{\varep^{1/2}}}  \Nd_C  \lrpar{\mfW_1}_{AD}\lrpar{\mfW_2}^{AD}+\sin^2\lrpar{\frac{v}{\varep^{1/2}}} \Nd_C  \lrpar{\mfW_2}_{AD}\lrpar{\mfW_1}^{AD}}^{(1m)}_H\\
=& - \half \varep^{2}\lrpar{ \Nd_\cdot \lrpar{\mfW_1}_{AD}\lrpar{\mfW_2}^{AD}}^{(1m)}_H,
\end{aligned}
\end{align*}
where we integrated by parts the second term. Integrating \eqref{EQP2expansion2} in $v$ and applying the above finishes the proof of \eqref{EQpropertyTT242}.

To conclude this section, combining \eqref{EQprelimcontrolTT24}, \eqref{EQclaimTT241integralproperty} and \eqref{EQpropertyTT242} finishes the proof of the estimate \eqref{EQLLLcontrol4} for $\TT_{2.4}$ with
\begin{align}
\begin{aligned}
\OO^{\TT_{2.4}}_{x} = \OO^{\TT_{2.4.3}}_{x}+\OO^{\TT_{2.4.4}}_{x}.
\end{aligned}\label{EQpreciseerrorTT24}
\end{align}
Combined with the control estimates \eqref{EQterm21islowerorder} for $\TT_{2.1}$, \eqref{EQTT22estimate} for $\TT_{2.2}$, \eqref{EQTT23estimate} for $\TT_{2.3}$, this proves the control \eqref{EQbigTT2statementclaim} for $\TT_2$, with $\OO^{\TT_2}_x$ and $\OO^{\TT_2}_{x,\mfW_0}$ being given by (see \eqref{EQTT22estimate}, \eqref{EQdeferrortermtt23}, \eqref{EQprelimcontrolTT24})
\begin{align}
\begin{aligned}
\OO^{\TT_2}_x =& \OO^{\TT_{2.3}}_{x} + \OO^{\TT_{2.4.3}}_{x}+\OO^{\TT_{2.4.4}}_{x},\\
\OO^{\TT_2}_{x,\mfW_0} =& \OO^{\TT_{2.2}}_{x,\mfW_0} +\OO^{\TT_{2.3}}_{x,\mfW_0}.
\end{aligned}\label{EQpreciseerrorTT2}
\end{align}
Moreover, combining in \eqref{EQBequation1appliedTransportEstimateL} the control \eqref{EQbigTT2statementclaim} of $\TT_2$ with the estimate \eqref{EQLLLcontrol1} for $\TT_1$, we conclude the proof of the bound \eqref{EQdeftriangleLW00} for $\triangle \Lf(x_{+W})$ with \eqref{EQdeftriangleLW}. The precise definition of $\OO^{\Lf^m}_{x}$ and $\OO^{\Lf^m}_{x,\mfW_0}$ is given by (see \eqref{EQpreciseerrorTT2})
\begin{align}
\begin{aligned}
\OO^{\Lf^m}_{x}:= \OO^{\TT_2}_{x}, \,\,
\OO^{\Lf^m}_{x,\mfW_0} := \int\limits_1^2 \OO^{\TT_1}_{x,\mfW_0}dv+\OO^{\TT_2}_{x,\mfW_0}.
\end{aligned}\label{EQpreciseerrorL}
\end{align}

\subsection{Preliminaries for null gluing of $\Ef,\Pf,\Lf,\Gf$} \label{SEC1EPLGgluing} 

In Sections \ref{SECdegreeargument1} and \ref{SECdiffEstimates1W} below we prove the existence of and estimates for the null gluing of $\Ef,\Pf,\Lf,\Gf$ of Theorem \ref{THMmain0}. That is, for a given solution $x$ to the null structure equations satisfying $\Vert x- \mathfrak{m} \Vert_{\XX(\HH)} \leq \varep$ for small $\varep>0$, and given charge difference difference vector $(\triangle\Ef_0,\triangle\Pf_0,\triangle\Lf_0,\triangle\Gf_0)\in \RRR^{10}$ (subject to further conditions), there exists a tensor $W$ such that, in terms of the decompositions of the previous sections,
\begin{align}
\begin{aligned}
\Ef(W)+\OO(\varep^{5/4}) =& \triangle\Ef_0, & \Pf(W)+\OO(\varep^{5/4}) =& \triangle\Pf_0, \\
\Gf(W)+\OO(\varep^{5/4}) =& \triangle\Gf_0, & \Lf(W)+\OO_x^{\Lf}+\OO(\varep^{9/4}) =& \triangle\Lf_0,
\end{aligned}\label{EQtobesolved101past99}
\end{align}
and $W$ satisfies appropriate bounds in terms of $x, \Ef_0,\Pf_0,\Lf_0,\Gf_0$; see \eqref{EQdiffEstimate1mainthm2statement}.

In preparation for proving the above, we show in this section that for $\varep>0$ sufficiently small, the mapping
\begin{align}
\begin{aligned}
W\mapsto (\Ef(W),\Pf(W),\Lf(W),\Gf(W)) \text{ is a bijection }
\end{aligned}\label{EQbijectionstatementWproblem}
\end{align}
onto the set of vectors $({\triangle \EE}_0,{\triangle \PP}_0,{\triangle \LL}_0,{\triangle\GG}_0)\in \RRR^{10}$ satisfying, for a sufficiently large real number $\CC_3>0$, and real numbers $\CC_1, \CC_2>0$,
\begin{align}
\begin{aligned}
\triangle \EE_0 =& \CC_1\varep, \\
\vert \triangle \LL_0 \vert \leq& \CC_2 \varep^2, \\
\varep^{-1} \triangle \EE_0 >& \CC_3 \lrpar{\varep^{-1} \vert \triangle \PP_0\vert + \varep^{-1} \vert \triangle \GG_0\vert },
\end{aligned}\label{EQrecallassumptionsEPLGtopordergluing}
\end{align}
and derive the following bounds for $W$,
\begin{align}
\begin{aligned}
\Vert \mfW_0 \Vert \les& \sqrt{\varep^{-1} \mathbf{E}(W)} + \sqrt{\varep^{-1} \vert\mathbf{P}(W)\vert} + \sqrt{\varep^{-1} \vert\mathbf{G}(W)\vert}, \\
\Vert \mfW_1 \Vert + \Vert \mfW_2 \Vert \les& \sqrt{\varep^{-1} \mathbf{E}(W)}+ \sqrt{\varep^{-1} \mathbf{P}(W)}+ \sqrt{\varep^{-1} \vert \mathbf{G}(W)\vert} + \sqrt{\varep^{-2} \vert \mathbf{L}(W)\vert}.
\end{aligned}\label{EQboundsWproblem}
\end{align}
We emphasize here that $(\triangle \EE_0,\triangle \PP_0,\triangle \LL_0,\triangle\GG_0)$ are not charge differences but constitutes simply a vector in $\RRR^{10}$.

Subsequently, \eqref{EQbijectionstatementWproblem} is used in Section \ref{SECdegreeargument1} in a degree argument to derive the existence of solutions to \eqref{EQtobesolved101past99}, and \eqref{EQboundsWproblem} is applied in Section \ref{SECdiffEstimates1W} together with bounds for the error terms in \eqref{EQtobesolved101past99} to derive estimates for solutions to \eqref{EQtobesolved101past99}.

We turn to the proof of \eqref{EQbijectionstatementWproblem} and \eqref{EQboundsWproblem}. For given vector $(\triangle\EE_0,\triangle\PP_0,\triangle\LL_0,\triangle\GG_0)$ satisfying \eqref{EQrecallassumptionsEPLGtopordergluing}, in the following we determine $W$ such that 
\begin{align}
\begin{aligned}
(\Ef(W),\Pf(W),\Lf(W),\Gf(W))=(\triangle\EE_0,\triangle\PP_0,\triangle\LL_0,\triangle\GG_0)
\end{aligned}\label{EQtoporderproblemtobesolved}
\end{align}
and prove the estimates \eqref{EQboundsWproblem}. 

Recall from Section \ref{SECAnsatzEstimates} the precise ansatz for $W$. The $2$-tensorfield $W(v,\th^1,\th^2)$ along $\HH$ is defined by
\begin{subequations}
\begin{align}
\begin{aligned}
W := -\varep^{3/2} \cos\lrpar{\frac{v}{\varep}} \mfW_0 - \varep^{5/4} \lrpar{\cos\lrpar{\frac{v}{\sqrt{\varep}}} \mfW_1 - \sin\lrpar{\frac{v}{\sqrt{\varep}}} \mfW_2},
\end{aligned}\label{EQWansatzrecallgluing}
\end{align}
with
\begin{align}
\begin{aligned}
\mfW_0 := f_{30} \psi^{30} +  f_{22} \psi^{22} + \sum\limits_{m=-1}^1 f_{2m} \psi^{2m}, \,\,
\mfW_1 := \tilde{f}_{30} \psi^{30}, \,\,
\mfW_2 := \sum\limits_{m=-1}^1\tilde{f}_{2m}  \phi^{2m},
\end{aligned}\label{EQWansatzDetailrecallgluing}
\end{align}
\end{subequations}
where $\psi^{lm}$ and $\phi^{lm}$ are the ($v$-independent) electric and magnetic tensor spherical harmonics, respectively (see Appendix \ref{APPconstructionW}), and the $v$-dependent scalar functions $f_{30}, f_{22}, f_{2m},\tilde{f}_{30}, \tilde{f}_{2m}$, for $m=-1,0,1$, are defined by
\begin{subequations}
\begin{align}
f_{30} :=& \half c_{30}, \,\, f_{2m} := \frac{1}{\lrpar{Y^{30}\cdot Y^{2m}}^{(1m)}} \lrpar{c_{2m} \xi_1 + c'_{2m} \xi_2}, \label{EQchoiceofCoeffFunction1recallgluing} \\
f_{22} :=& c_{22} \xi_4,\label{EQchoiceofCoeffFunction2recallgluing} \\
\tilde{f}_{30} :=& \tilde{c}_{30}, \,\, \tilde{f}_{2m} := \frac{1}{\lrpar{Y^{30}\cdot Y^{2m}}^{(1m)}} \tilde{c}_{2m} \xi_3,\label{EQchoiceofCoeffFunction3recallgluing}
\end{align}
\end{subequations}
where
\begin{itemize}
\item $\lrpar{Y^{30}\cdot Y^{2m}}^{(1m)}$, $m=-1,0,1$, are (non-vanishing) Fourier coefficients, see Appendix \ref{APPconstructionW},
\item $\xi_1,\xi_2,\xi_3,\xi_4$ are fixed, smooth (low-frequency) universal $v$-dependent scalar functions satisfying $6$ universal, linearly independent integral conditions (to be stated in \eqref{EQintegralXcond1}, \eqref{EQintegralXcond2} and \eqref{EQintegralXcond3} below),
\item $c_{30},c_{2m},c'_{2m},c_{22},\tilde{c}_{30},\tilde{c}_{2m}$ are freely choosable constants.
\end{itemize}
\ni In the following we solve \eqref{EQtoporderproblemtobesolved} by suitably choosing the constants $c_{30},c_{2m},c'_{2m},c_{22},\tilde{c}_{30},\tilde{c}_{2m}$. This can be done in an explicit manner, which allows for a straight-forward control of $\Vert W \Vert$ (the norm is set up in Definition \eqref{DEFdefinitionWnorm}) to prove \eqref{EQboundsWproblem}. \\

\ni \textbf{Solving prescribed $\Pf(W)$ and $\Gf(W)$.} In the following we solve
\begin{align}
\begin{aligned}
\mathbf{P}(W)=\triangle \PP_0, \,\, \mathbf{G}(W)=\triangle \GG_0.
\end{aligned}\label{EQPWGWgluingtopordercond}
\end{align}
by suitably choosing the constants $c_{30}$, $c_{2m}$ and $c'_{2m}$, for $m=-1,0,1$. Recall from \eqref{EQDEFEWPW} and \eqref{EQdefGW} the definitions
\begin{align}
\begin{aligned}
\mathbf{P}(W) := -\varep \int\limits_1^2 \frac{1}{16r^2} \varphi^2 \lrpar{\vert \mfW_0 \vert_\gac^2}^{(1m)},\,\, \mathbf{G}(W) := \varep \int\limits_1^2 \frac{\sqrt{2}}{16r} \varphi^2 \lrpar{\vert \mfW_0 \vert_\gac^2}^{(1m)}.
\end{aligned}\label{EQrecallDEFsPandGofW}
\end{align}
By the concrete ansatz \eqref{EQWansatzDetailrecallgluing} for $\mfW_0$ it is possible to explicitly calculate (see Appendix \ref{SECFouriervertWvert2}) that for $m=-1,0,1$,
\begin{align}
\begin{aligned}
\lrpar{\vert \mfW_0 \vert_\gac^2}^{(1m)} = 2\frac{\sqrt{5}}{3} \underbrace{\lrpar{Y^{30}\cdot Y^{2m}}^{(1m)}}_{\neq 0} f_{30}f_{2m}.
\end{aligned}\label{EQFourierIdentityL1}
\end{align}
Using \eqref{EQrecallDEFsPandGofW} and \eqref{EQFourierIdentityL1}, \eqref{EQPWGWgluingtopordercond} becomes the integral conditions
\begin{align}
\begin{aligned}
-2\frac{\sqrt{5}}{3}\lrpar{Y^{30}\cdot Y^{2m}}^{(1m)} \int\limits_1^2 \frac{1}{16r^2} \varphi^2 f_{30}f_{2m} =& \varep^{-1} \triangle \PP^m_0, \\ 
2\frac{\sqrt{5}}{3}\lrpar{Y^{30}\cdot Y^{2m}}^{(1m)} \int\limits_1^2 \frac{\sqrt{2}}{16r} \varphi^2 f_{30}f_{2m} =& \varep^{-1} \triangle \GG^m_0.
\end{aligned}\label{EQintegralcondEPgluingtoporder}
\end{align}
Recalling that by \eqref{EQchoiceofCoeffFunction1recallgluing}, for $m=-1,0,1$,
\begin{align}
\begin{aligned}
f_{30} := \half c_{30}, \,\, f_{2m} := \frac{1}{\lrpar{Y^{30}\cdot Y^{2m}}^{(1m)}} \lrpar{c_{2m} \xi_1 + c'_{2m} \xi_2},
\end{aligned}\label{EQexplicitf2mformPG}
\end{align}
and stipulating that $\xi_1(v)$ and $\xi_2(v)$ satisfy the universal, linearly independent integral conditions 
\begin{align}
\begin{aligned}
\frac{\sqrt{5}}{3}\int\limits_1^2 \frac{1}{16r^2} \varphi^2 \xi_1 = 1,\,\,  \int\limits_1^2 \frac{\sqrt{2}}{16r} \varphi^2 \xi_1 =0, \,\,\int\limits_1^2 \frac{1}{16r^2} \varphi^2 \xi_2 = 0,\,\, \frac{\sqrt{5}}{3} \int\limits_1^2 \frac{\sqrt{2}}{16r} \varphi^2 \xi_2 =1,
\end{aligned}\label{EQintegralXcond1}
\end{align}
to solve \eqref{EQintegralcondEPgluingtoporder} it is sufficient to pick the constants $c_{30}$, $c_{2m}$ and $c'_{2m}$, for $m=-1,0,1$, as follows,
\begin{align}
\begin{aligned}
c_{30} :=& \sqrt{d_0\varep^{-1}\triangle\EE_0+ \varep^{-1}\vert\triangle\PP_0\vert + \varep^{-1}\vert\triangle\GG_0\vert}, \\
c_{2m} :=& \frac{-\varep^{-1}\triangle \PP^m_0}{\sqrt{d_0\varep^{-1}\triangle\EE_0+\varep^{-1}\vert\triangle\PP_0\vert + \varep^{-1}\vert\triangle\GG_0\vert}}, \\
c'_{2m} :=& \frac{\varep^{-1}\triangle \GG^m_0}{\sqrt{d_0\varep^{-1}\triangle\EE_0 +\varep^{-1}\vert\triangle\PP_0\vert + \varep^{-1}\vert\triangle\GG_0\vert}},
\end{aligned}\label{EQchoiceConstantsPG}
\end{align}
where $d_0>0$ is a universal constant to be determined below in the solving of $\Ef(W)=\triangle \EE_0$.

From the explicit expressions in \eqref{EQchoiceConstantsPG} we directly get the estimate
\begin{align}
\begin{aligned}
\vert c_{30} \vert+ \sum\limits_{m=-1}^1 ( \vert c_{2m} \vert + \vert c'_{2m} \vert ) \les \sqrt{d_0\varep^{-1} \triangle\EE_0}+ \sqrt{\varep^{-1} \vert \triangle\PP_0\vert} + \sqrt{\varep^{-1} \vert \triangle\GG_0\vert}.
\end{aligned}\label{EQestimateCoeffconstantsPG}
\end{align}
 
\ni \textbf{Solving prescribed $\mathbf{E}(W)$.}
\ni In the following we solve
\begin{align}
\begin{aligned}
\mathbf{E}(W)=\triangle \EE_0,
\end{aligned}\label{EQEWgluingtopordercond}
\end{align}
by suitably choosing the constant $c_{22}$. Recall from \eqref{EQDEFEWPW} the definition
\begin{align}
\begin{aligned}
\mathbf{E}(W) := \varep \int\limits_1^2 \frac{1}{16r^2}\varphi^2 \lrpar{\vert \mfW_0 \vert^2_{\gac}}^{(0)}.
\end{aligned}\label{EQdefErecallMatching}
\end{align}
\ni By the concrete ansatz \eqref{EQWansatzDetailrecallgluing} for $\mfW_0$ it is possible to explicitly calculate (see Appendix \ref{SECFouriervertWvert2}) that
\begin{align}
\begin{aligned}
 \lrpar{\vert \mfW_0 \vert^2_{\gac}}^{(0)} = \frac{1}{\sqrt{4\pi}}\lrpar{ f_{30}^2 + \sum\limits_{m=-1}^1 f_{2m}^2 + f_{22}^2}.
\end{aligned}\label{EQexprEgluingfourier}
\end{align}
Using \eqref{EQdefErecallMatching} and \eqref{EQexprEgluingfourier}, \eqref{EQEWgluingtopordercond} turns into the integral condition
\begin{align}
\begin{aligned}
\varep^{-1} \triangle\Ef_0 =& \int\limits_1^2 \frac{1}{\sqrt{4\pi}16r^2}\varphi^2 \lrpar{ f_{30}^2 + \sum\limits_{m=-1}^1 f_{2m}^2 + f_{22}^2}\\
=& \int\limits_1^2 \frac{1}{\sqrt{4\pi}16r^2}\varphi^2 \lrpar{ f_{30}^2 + \sum\limits_{m=-1}^1 f_{2m}^2}+ \int\limits_1^2 \frac{1}{\sqrt{4\pi}16r^2}\varphi^2 c_{22}^2 (\xi_4)^2,
\end{aligned}\label{EQenergyadjustment1}
\end{align}
where used \eqref{EQchoiceofCoeffFunction2recallgluing}. We note that the first term on the right-hand side is determined from solving \eqref{EQPWGWgluingtopordercond} for $\Pf(W)$ and $\Gf(W)$ above, see \eqref{EQexplicitf2mformPG} and \eqref{EQchoiceConstantsPG}. We moreover remark that the second term is \emph{positive-definite}, which puts constraints on the solvability of \eqref{EQenergyadjustment1}. 

Using \eqref{EQestimateCoeffconstantsPG}, the first term on the right-hand side of \eqref{EQenergyadjustment1} is bounded for $d_0>0$ sufficiently small (depending only on universal quantities) by 
\begin{align}
\begin{aligned}
\int\limits_1^2 \frac{1}{\sqrt{4\pi}16r^2}\varphi^2 \lrpar{ f_{30}^2 + \sum\limits_{m=-1}^1 f_{2m}^2} \les& d_0 \varep^{-1} \triangle \EE_0 + \varep^{-1} \vert \triangle \PP_0 \vert + \varep^{-1} \vert \triangle \GG_0 \vert \\
\leq& \half \varep^{-1} \triangle \EE_0 +\tilde{C} \lrpar{\varep^{-1} \vert \triangle \PP_0 \vert + \varep^{-1} \vert \triangle \GG_0 \vert},
\end{aligned}\label{EQenergyadjustment2}
\end{align}
where $\tilde{C}>0$ is a universal constant. 

By assumption \eqref{EQrecallassumptionsEPLGtopordergluing} it holds that for a real number $\CC_3>0$,
\begin{align*}
\begin{aligned}
\varep^{-1} \triangle \EE_0 > C_3 \lrpar{\varep^{-1} \vert \triangle \PP_0 \vert + \varep^{-1} \vert \triangle \GG_0 \vert}.
\end{aligned}
\end{align*}
Plugging this into the right-hand side of \eqref{EQenergyadjustment2}, we get that for $\CC_3>0$ sufficiently large (it suffices to have $\CC_3>\tilde{C}\frac{10}{4}$),
\begin{align*}
\begin{aligned}
\int\limits_1^2 \frac{1}{\sqrt{4\pi}16r^2}\varphi^2 \lrpar{ f_{30}^2 + \sum\limits_{m=-1}^1 f_{2m}^2} <& \frac{9}{10}\varep^{-1}\triangle\EE_0,
\end{aligned}
\end{align*}
so that
\begin{align}
\begin{aligned}
\varep^{-1} \triangle \EE_0 - \int\limits_1^2 \frac{1}{\sqrt{4\pi}16r^2}\varphi^2 \lrpar{ f_{30}^2 + \sum\limits_{m=-1}^1 f_{2m}^2} >\frac{1}{10} \varep^{-1} \triangle \EE_0 >0.
\end{aligned}\label{EQlowerboundc22}
\end{align}
Having the control \eqref{EQlowerboundc22}, we are now in position to solve \eqref{EQenergyadjustment1} by suitably choosing the constant $c_{22}$. It is straight-forward to check that a solution to \eqref{EQenergyadjustment1} is given by picking
\begin{align}
\begin{aligned}
c_{22}:= \sqrt{\varep^{-1} \triangle \EE_0 - \int\limits_1^2 \frac{1}{\sqrt{4\pi}16r^2}\varphi^2 \lrpar{ f_{30}^2 + \sum\limits_{m=-1}^1 f_{2m}^2}},
\end{aligned}\label{EQcoefficientdefEadjustment}
\end{align}
where the square-root is well-defined due to \eqref{EQlowerboundc22}, and stipulating that the universal scalar function $\xi_4$ satisfies the integral condition
 \begin{align}
\begin{aligned}
\int\limits_1^2 \frac{1}{\sqrt{4\pi}16r^2}\varphi^2 (\xi_4)^2 =1.
\end{aligned}\label{EQintegralXcond2}
\end{align}
By \eqref{EQexplicitf2mformPG}, \eqref{EQestimateCoeffconstantsPG}, \eqref{EQlowerboundc22}, \eqref{EQcoefficientdefEadjustment}, we can bound
\begin{align}
\begin{aligned}
\sqrt{\frac{1}{10}\varep^{-1}\triangle\EE_0}<\vert c_{22} \vert \les \sqrt{\varep^{-1}\triangle\EE_0} +  \sqrt{\varep^{-1}\vert\triangle\PP_0\vert} +  \sqrt{\varep^{-1}\vert \triangle\GG_0\vert}
\end{aligned}\label{EQestimatec22choice}
\end{align}
We note that by the above choices \eqref{EQchoiceConstantsPG} and \eqref{EQestimatec22choice}, the full tensor $\mfW_0$ is determined. In the following we determine the tensors $\mfW_1$ and $\mfW_2$.\\

\ni \textbf{Solving prescribed $\mathbf{L}(W)$.}
\ni In the following we solve
\begin{align}
\begin{aligned}
\mathbf{L}(W)=\triangle \LL_0,
\end{aligned}\label{EQLWtopordergluingprob}
\end{align}
by suitably choosing the constants $\tilde{c}_{30}$ and $\tilde{c}_{2m}$. Recall from \eqref{EQdeftriangleLW} the definition
\begin{align}
\begin{aligned}
\mathbf{L}^m(W):=& \int\limits_1^2 \frac{\varep^2\varphi^2}{r^2}\lrpar{-\half  \lrpar{\Nd_A\lrpar{\mfW_1^{AB} (\mfW_2)_{B\cdot} }}^{(1m)}_H + \frac{1}{4}\lrpar{ \Nd_\cdot \lrpar{\mfW_1 \cdot\mfW_2}}^{(1m)}_H} + \OO_{x,\mfW_0}^{\Lf^m},
\end{aligned}\label{EQdefLmWrecallgluing}
\end{align}
where we note that $\OO_{x,\mfW_0}^{\Lf^m}$ depends on $x$ and $\mfW_0$ and is thus determined from \eqref{EQPWGWgluingtopordercond} and \eqref{EQEWgluingtopordercond} above (see \eqref{EQWansatzDetailrecallgluing}, \eqref{EQchoiceofCoeffFunction1recallgluing} \eqref{EQchoiceofCoeffFunction2recallgluing}, \eqref{EQchoiceConstantsPG}, \eqref{EQcoefficientdefEadjustment}). However, $\OO_{x,\mfW_0}^{\Lf^m}$ does \emph{not} depend on $\mfW_1$ and $\mfW_2$; in other words, $\OO_{x,\mfW_0}^{\Lf^m}$ does not involve the constants $\tilde{c}_{30}$ and $\tilde{c}_{2m}$.

By the concrete ansatz \eqref{EQWansatzDetailrecallgluing} for $\mfW_1$ and $\mfW_2$ it is possible to explicitly calculate (see Appendix \ref{SECFouriervertWvert2}) that for $m=-1,0,1$,
\begin{align}
\begin{aligned}
-\half  \lrpar{\Nd_A\lrpar{\mfW_1^{AB} (\mfW_2)_{B\cdot} }}^{(1m)}_H + \frac{1}{4}\lrpar{ \Nd_\cdot \lrpar{\mfW_1 \cdot\mfW_2}}^{(1m)}_H = -\frac{\sqrt{5}}{3\sqrt{2}} \underbrace{\lrpar{Y^{30}Y^{2{m}}}^{(1 m)}}_{\neq 0} \tilde{f}_{30}\tilde{f}_{2{m}}.
\end{aligned}\label{EQFourierIdentityL}
\end{align}
By \eqref{EQdefLmWrecallgluing} and \eqref{EQFourierIdentityL}, \eqref{EQLWtopordergluingprob} turns into the following integral conditions for $m=-1,0,1$,
\begin{align}
\begin{aligned}
-\frac{\sqrt{5}}{3\sqrt{2}}\lrpar{Y^{30}Y^{2{m}}}^{(1 m)} \int\limits_1^2 \frac{\varphi^2}{r^2} \tilde{f}_{30}\tilde{f}_{2{m}} = \varep^{-2} \lrpar{\triangle\LL^{m}_0 - \OO_{\Lf^m}}.
\end{aligned}\label{EQintegralcondLW1}
\end{align}
Recalling from \eqref{EQchoiceofCoeffFunction3recallgluing} that 
\begin{align}
\begin{aligned}
\tilde{f}_{30} := \tilde{c}_{30}, \,\, \tilde{f}_{2m} := \frac{1}{\lrpar{Y^{30}\cdot Y^{2m}}^{(1m)}} \tilde{c}_{2m} \xi_3,
\end{aligned}\label{EQexplicitf2mformL}
\end{align}
and stipulating that the universal scalar function $\xi_3(v)$ satisfies the integral condition 
\begin{align}
\begin{aligned}
-\frac{\sqrt{5}}{3\sqrt{2}} \int\limits_1^2 \frac{\varphi^2}{r^2} \xi_3 = 1,
\end{aligned}\label{EQintegralXcond3}
\end{align}
to solve \eqref{EQintegralcondLW1} it suffices to choose the constants $\tilde{c}_{30}$ and $\tilde{c}_{2m}$ as
\begin{align}
\begin{aligned}
\tilde{c}_{30} :=& \sqrt{\varep^{-1}\triangle \EE_0+ \varep^{-2}\vert\triangle\LL_0\vert + \varep^{-2}\vert\OO_{x,\mfW_0}^{\Lf^m}\vert}, \\
\tilde{c}_{2m}:=& \frac{\varep^{-2} (\triangle\LL^m_0- \OO_{x,\mfW_0}^{\Lf^m})}{\sqrt{\varep^{-1}\triangle \EE_0+ \varep^{-2}\vert\triangle\LL_0\vert + \varep^{-2}\vert\OO_{x,\mfW_0}^{\Lf^m}\vert}}.
\end{aligned}\label{EQchoiceConstantsL}
\end{align}
We note that this fully determines the tensors $\mfW_1$ and $\mfW_2$.

From the explicit expressions in \eqref{EQchoiceConstantsL} we directly get
\begin{align}
\begin{aligned}
\vert \tilde{c}_{30} \vert + \sum\limits_{m=-1}^1 \vert \tilde{c}_{2m}\vert  \les& \sqrt{\varep^{-1} \triangle\EE_0}+ \sqrt{\varep^{-2} \vert \triangle\LL_0\vert} + \sum\limits_{m=-1}^1\sqrt{\varep^{-2} \vert \OO_{x,\mfW_0}^{\Lf^m}\vert}.
\end{aligned}\label{EQestimateCoeffconstantsLchoice}
\end{align}

\ni To summarize the above, in \eqref{EQchoiceConstantsPG}, \eqref{EQcoefficientdefEadjustment} and \eqref{EQchoiceConstantsL} we determined the constants 
\begin{align*}
\begin{aligned}
c_{30},c_{2m},c'_{2m},c_{22},\tilde{c}_{30},\tilde{c}_{2m} \, \text{ for } m=-1,0,1,
\end{aligned}
\end{align*}
such that the corresponding tensor $W$, given through the ansatz \eqref{EQWansatzrecallgluing}-\eqref{EQchoiceofCoeffFunction3recallgluing}, satisfies
\begin{align*}
\begin{aligned}
(\mathbf{E},\mathbf{P},\mathbf{L},\mathbf{G})(W) = (\triangle \EE_0,\triangle \PP_0,\triangle \LL_0,\triangle\GG_0).
\end{aligned}
\end{align*}
It remains to prove the estimates \eqref{EQboundsWproblem} for the constructed tensor $W$.\\

\ni \textbf{Estimates for $W$.} 
\ni By \eqref{EQestimateCoeffconstantsPG}, \eqref{EQestimatec22choice} and \eqref{EQestimateCoeffconstantsLchoice}, the solution $W$ constructed above to
\begin{align*}
\begin{aligned}
\lrpar{\mathbf{E},\mathbf{P},\mathbf{L},\mathbf{G}}(W) =  \lrpar{\triangle\EE_0,\triangle\PP_0,\triangle\LL_0,\triangle\LL_0}. 
\end{aligned}
\end{align*}
is controlled by
\begin{align}
\begin{aligned}
\Vert \mfW_0 \Vert \les& \sqrt{\varep^{-1} \mathbf{E}(W)} + \sqrt{\varep^{-1} \vert\mathbf{P}(W)\vert} + \sqrt{\varep^{-1} \vert\mathbf{G}(W)\vert}, \\
\Vert \mfW_1 \Vert + \Vert \mfW_2 \Vert \les& \sqrt{\varep^{-1} \mathbf{E}(W)}+ \sqrt{\varep^{-2} \vert \mathbf{L}(W)\vert} + \sum\limits_{m=-1}^1\sqrt{\varep^{-2} \vert \OO_{x,\mfW_0}^{\Lf^m}\vert}.
\end{aligned}\label{EQWbounduniform}
\end{align}

\ni Furthermore, by inspecting the term $\OO^{\Lf^m}_{x,\mfW_0}$, see \eqref{EQpreciseerrorTT2} and \eqref{EQpreciseerrorL}, that is,
\begin{align*}
\begin{aligned}
\OO^{\Lf^m}_{x,\mfW_0} = \int\limits_1^2 \OO^{\TT_1}_{x,\mfW_0}dv+\OO^{\TT_{2.2}}_{x,\mfW_0} +\OO^{\TT_{2.3}}_{x,\mfW_0}, 
\end{aligned}
\end{align*}
with $\OO^{\TT_1}_{x,\mfW_0}, \OO^{\TT_{2.2}}_{x,\mfW_0}, \OO^{\TT_{2.3}}_{x,\mfW_0}$ given in \eqref{EQstructureOOTT1}, \eqref{EQTT22estimate}, \eqref{EQdeferrortermtt23}, and using that each term depends on $\mfW_0$, it is straight-forward to show with \eqref{EQWbounduniform} that
\begin{align}
\begin{aligned}
\sum\limits_{m=-1}^1\sqrt{\varep^{-2} \vert \OO_{x,\mfW_0}^{\Lf^m}\vert} \les \Vert \mfW_0 \Vert.
\end{aligned}\label{EQnonlinearityEstimateLxerrorterm}
\end{align}
Combining \eqref{EQnonlinearityEstimateLxerrorterm} with \eqref{EQWbounduniform} yields
\begin{align}
\begin{aligned}
\Vert \mfW_1 \Vert + \Vert \mfW_2 \Vert \les& \sqrt{\varep^{-1} \mathbf{E}(W)}+ \sqrt{\varep^{-1} \mathbf{P}(W)}+ \sqrt{\varep^{-1} \vert \mathbf{G}(W)\vert} + \sqrt{\varep^{-2} \vert \mathbf{L}(W)\vert}.
\end{aligned}\label{EQestimateCoeffconstantsLchoice2}
\end{align}
This finishes the proof of \eqref{EQbijectionstatementWproblem} and \eqref{EQboundsWproblem}.
\subsection{Proof of existence} \label{SECdegreeargument1} In this section we conclude the proof of part (1) of Theorem \ref{THMmain0} by a classical degree argument. First, we recall the following.

\begin{itemize}
\item The assumptions \eqref{EQsmallnessMain021}, \eqref{EQsmallnessMain021LSMALL} and \eqref{EQsmallnessMain02} on $(\triangle \mathbf{E}_0,\triangle \mathbf{P}_0,\triangle \mathbf{L}_0,\triangle \mathbf{G}_0)$, are that for real numbers $C_1,C_2,C_3>0$, 
\begin{subequations}
\begin{align}
\triangle \mathbf{E}_0 =& C_1\, \varep, \label{EQsmallnessMain021inProof} \\ 
\vert \triangle \mathbf{L}_0 \vert \leq& C_2 \varep^2, \label{EQsmallnessMain021LSMALLinProof} \\
\varep^{-1} (\triangle \mathbf{E}_0) >& C_3 \lrpar{\varep^{-1} \vert \triangle \mathbf{P}_0\vert + \varep^{-1} \vert \triangle \mathbf{G}_0\vert}. \label{EQsmallnessMain02inProof}
\end{align}
\end{subequations}

\item By \eqref{EQsmallnessMain1001} and the estimates of Sections \ref{SECgluingE}, \ref{SECgluingG} and \ref{SECgluingL}, we have that for $\varep>0$ sufficiently small,
\begin{align}
\begin{aligned}
\triangle\mathbf{E}(x_{+W}) =& \mathbf{E}(W) + \OO(\varep^{5/4}), &
\triangle\mathbf{P}(x_{+W}) =& \mathbf{P}(W) + \OO(\varep^{5/4}), \\
\triangle\mathbf{L}(x_{+W}) =& \mathbf{L}(W) + \OO_x(\varep^2)+\OO(\varep^{9/4}), &
\triangle\mathbf{G}(x_{+W}) =& \mathbf{G}(W) + \OO(\varep^{5/4}),
\end{aligned}\label{EQtransportEstimatesForDegree}
\end{align}
where the term $\OO_x(\varep^2)$ is bounded by
\begin{align}
\begin{aligned}
\OO_x(\varep^2) \leq C_0 \varep^2,
\end{aligned}\label{EQuniversalErrorBound}
\end{align}
where $C_0>0$ is a universal constant.
\end{itemize}
For tensors $W$ of the form \eqref{EQWansatz}-\eqref{EQchoiceofCoeffFunction3}, consider the error map
\begin{align}
\begin{aligned}
\FF(W):=& \lrpar{ \varep^{-1} \triangle\mathbf{E},\varep^{-1}\triangle\mathbf{P},\varep^{-2}\triangle\mathbf{L},\varep^{-1}\triangle\mathbf{G}}(x_{+W}) \\
&- \lrpar{ \varep^{-1} \triangle\mathbf{E}_0,\varep^{-1}\triangle\mathbf{P}_0,\varep^{-2}\triangle\mathbf{L}_0,\varep^{-1}\triangle\mathbf{G}_0}. 
\end{aligned}\label{EQdefF}
\end{align}
In the following we prove that there exists a $W$ such that 
\begin{align}
\FF(W) =0. \label{EQexistenceWzeroerror}
\end{align}
Plugging \eqref{EQtransportEstimatesForDegree} into \eqref{EQdefF} we get that
\begin{align*}
\begin{aligned}
\FF(W) =& \varep^{-1} \lrpar{\mathbf{E}(W)-\triangle\mathbf{E}_0 ,\mathbf{P}(W)-\triangle\mathbf{P}_0,\varep^{-1}\lrpar{\mathbf{L}(W)-\triangle\mathbf{L}_0+ \OO_x(\varep^2)},\mathbf{G}(W)-\triangle\mathbf{G}_0} \\
&+ \lrpar{ \OO(\varep^{1/4}), \OO(\varep^{1/4}), \OO(\varep^{1/4}), \OO(\varep^{1/4})}.
\end{aligned}
\end{align*}
For $t \in [0,1]$ define the homotopy 
\begin{align*}
\begin{aligned}
\FF(W,t):=& \varep^{-1}\lrpar{\mathbf{E}(W)-\triangle\mathbf{E}_0,\mathbf{P}(W)-\triangle\mathbf{P}_0,\varep^{-1}\lrpar{\mathbf{L}(W)-\triangle\mathbf{L}_0+ \OO_x(\varep^2)},\mathbf{G}(W)-\triangle\mathbf{G}_0} \\
&+ t\cdot \lrpar{ \OO(\varep^{1/4}), \OO(\varep^{1/4}), \OO(\varep^{1/4}), \OO(\varep^{1/4})},
\end{aligned}
\end{align*}
such that $\FF(W,1)=\FF(W)$. We make the following two claims.\\

\noindent \textbf{Claim $1$.} \emph{For $\varep>0$ sufficiently small and $C_3>0$ sufficiently large, the map $$W\mapsto(\Ef(W),\Pf(W),\Lf(W),\Gf(W))$$ is a bijection onto the ellipsoid $\EE$ centered at $(\triangle \mathbf{E}_0, \triangle \mathbf{P}_0,\triangle \mathbf{L}_0 - \OO_x(\varep^2),\triangle \mathbf{G}_0 )$ defined to be the set of vectors $(\triangle \mathbf{E},\triangle \mathbf{P},\triangle \mathbf{L},\triangle \mathbf{G}) \in \RRR^{10}$ such that}

\begin{align}
\begin{aligned}
&\lrpar{\varep^{-1}( \triangle\mathbf{E}- \triangle \mathbf{E}_0)}^2 + \lrpar{\varep^{-1}\vert \triangle\mathbf{P}- \triangle \mathbf{P}_0\vert}^2 \\
&+ \lrpar{\frac{C_1}{C_2+C_0} \varep^{-2} \vert \triangle \mathbf{L}- (\triangle \mathbf{L}_0-\OO_x(\varep^2)) \vert}^2 + \lrpar{\varep^{-1} \vert \triangle\mathbf{G}- \triangle \mathbf{G}_0\vert}^2 \leq \lrpar{\frac{1}{C_3}\varep^{-1}\triangle \mathbf{E}_0}^2.
\end{aligned}\label{EQdefEllipsoid}
\end{align}

\ni \textbf{Claim $2$.} \emph{Let $\EE' \subset \RRR^{10}$ denote the pre-image of $\EE$ under $\FF(W,0)$. Then for each $t\in[0,1]$ there is no $W\in \pr \EE'$ such that $\FF(W,t)=0$}. \\

\ni From the above two claims, a classical degree argument (see, for example, Lemma 4.1 in \cite{ACR3}) proves the existence of $W\in \EE'$ satisfying \eqref{EQexistenceWzeroerror}, and thus finishes the proof of part (1) of Theorem \ref{THMmain0}. It remains to prove the two claims. \\

\ni\emph{Proof of Claim $1$.} From Section \ref{SEC1EPLGgluing} we have that for $\varep>0$ sufficiently small and $C_3>0$ sufficiently large, the map 
\begin{align}
\begin{aligned}
W \mapsto \lrpar{\mathbf{E}(W),\mathbf{P}(W) ,\mathbf{L}(W),\mathbf{G}(W)}
\end{aligned}\label{EQWmapproofdegree}
\end{align}
is a bijection onto the set $\CC$ of $(\triangle \mathbf{E},\triangle \mathbf{P},\triangle \mathbf{L},\triangle \mathbf{G})\in \RRR^{10}$ satisfying
\begin{align*}
\begin{aligned}
\frac{1}{2}C_1 \varep \leq& \triangle \mathbf{E}, \\
\vert \triangle\mathbf{L} \vert \leq& 2(C_2+C_0) \varep^2, \\
\varep^{-1} (\triangle \mathbf{E}) >& \frac{1}{8}C_3 \lrpar{\varep^{-1} \vert \triangle \mathbf{P}\vert + \varep^{-1} \vert \triangle \mathbf{G}\vert}.
\end{aligned}
\end{align*}
Using \eqref{EQsmallnessMain021inProof}, \eqref{EQsmallnessMain021LSMALLinProof}, \eqref{EQsmallnessMain02inProof}, is straight-forward to verify that the ellipsoid $\EE$ (defined in \eqref{EQdefEllipsoid}) is contained in $\CC$ for $\varep>0$ sufficiently small and $C_3>0$ sufficiently large. This finishes the proof of Claim $1$.\\

\ni \emph{Proof of Claim $2$.} Assume for contradiction that $W\in \pr \EE'$ satisfies $\FF(W,t)=0$. By definition of $\FF(W,t)$, $\FF(W,t)=0$ implies that
\begin{align*}
\begin{aligned}
&\lrpar{\varep^{-1}( \triangle\mathbf{E}- \triangle \mathbf{E}_0)}^2 + \lrpar{\varep^{-1}( \triangle\mathbf{P}- \triangle \mathbf{P}_0)}^2 \\
&+ \lrpar{\frac{C_1}{C_2+C_0}\varep^{-2}\vert\triangle \mathbf{L}- (\triangle \mathbf{L}_0 -\OO_x(\varep^2)) \vert}^2 + \lrpar{\varep^{-1}( \triangle\mathbf{G}- \triangle \mathbf{G}_0)}^2 \leq t \cdot \OO(\varep^{1/2}).
\end{aligned}
\end{align*}
For $\varep>0$ sufficiently small, this implies that $W$ lies in the \emph{interior} of $\EE'$ (see the definition of $\EE$ in \eqref{EQdefEllipsoid}), a contradiction to the assumption $W\in\partial \EE'$. This finishes the proof of Claim $2$.

\subsection{Estimates} \label{SECdiffEstimates1W} In this section we prove the estimate \eqref{EQdiffEstimate1mainthm2statement} of Theorem \ref{THMmain0}, that is, we derive the following bound for the constructed $\Ef,\Pf,\Lf,\Gf$ null gluing solution,
\begin{align}
\begin{aligned}
&\Vert x_{+W} -x \Vert_{\XX^{\mathrm{h.f.}}(\HH)} + \Vert x_{+W}- x\Vert_{\XX(S_2)} 
\\
\les&\varep\lrpar{\sqrt{\varep^{-1}\triangle\Ef(x_{+W})} + \sqrt{\varep^{-1}\vert\triangle\Pf(x_{+W})\vert} + \sqrt{\varep^{-1}\vert\triangle\Gf(x_{+W})\vert}}\\
&+ \varep^{5/4} \sqrt{\varep^{-2}\vert\triangle\Lf(x_{+W})\vert }+ \varep^{1/4} \Vert x - \mathfrak{m}\Vert_{\XX(\HH)}.
\end{aligned}\label{EQcontrolPartTHMmain0}
\end{align}
The proof of \eqref{EQcontrolPartTHMmain0} follows by combining 
\begin{itemize}
\item the bounds for the high-frequency solution $x_{+W}$ proved in Section \ref{SECconstructionSolution},
\item the analysis of the null transport of the charges $\Ef,\Pf,\Lf,\Gf$ studied in Sections \ref{SECgluingE}, \ref{SECgluingG}, \ref{SECgluingL},
\item the estimates for the control of $W$ proved in Section \ref{SEC1EPLGgluing}.
\end{itemize}

\ni As first step in the proof of \eqref{EQcontrolPartTHMmain0}, we claim that combining the analysis of Sections \ref{SEC1EPLGgluing} and \ref{SECgluingE}, \ref{SECgluingG}, \ref{SECgluingL} leads to the following, 
\begin{align}
\begin{aligned}
\Vert \mfW_0 \Vert \les& \sqrt{\varep^{-1}\triangle \mathbf{E}(x_{+W})} +\sqrt{\varep^{-1}\vert\triangle \mathbf{P}(x_{+W}) \vert}+ \sqrt{\varep^{-1}\vert\triangle \mathbf{G}(x_{+W}) \vert}\\
&+\varep^{1/2} \sqrt{\varep^{-2}\vert\triangle \mathbf{L}(x_{+W}) \vert} + \varep^{-3/4} \lrpar{ \Vert x_{+W}-\mathfrak{m} \Vert_{\XX^{\mathrm{h.f.}}(\HH)}+ \Vert x-\mathfrak{m} \Vert_{\XX(\HH)}}, \\
\Vert \mfW_1 \Vert+ \Vert \mfW_2 \Vert\les& \sqrt{\varep^{-1}\triangle \mathbf{E}(x_{+W})} +\sqrt{\varep^{-1}\vert\triangle \mathbf{P}(x_{+W}) \vert}+ \sqrt{\varep^{-1}\vert\triangle \mathbf{G}(x_{+W}) \vert}\\
&+\sqrt{\varep^{-2}\vert\triangle \mathbf{L}(x_{+W}) \vert}  + \varep^{-1} \lrpar{ \Vert x_{+W}-\mathfrak{m} \Vert_{\XX^{\mathrm{h.f.}}(\HH)}+ \Vert x-\mathfrak{m} \Vert_{\XX(\HH)}}.
\end{aligned}\label{EQw0estimate}
\end{align}
Indeed, on the one hand, from \eqref{EQWbounduniform} and \eqref{EQestimateCoeffconstantsLchoice2} in Section \ref{SEC1EPLGgluing} we have that
\begin{align}
\begin{aligned}
\Vert \mfW_0 \Vert \les& \sqrt{\varep^{-1} \mathbf{E}(W)} + \sqrt{\varep^{-1} \vert\mathbf{P}(W)\vert} + \sqrt{\varep^{-1} \vert\mathbf{G}(W)\vert} \\
\Vert \mfW_1 \Vert + \Vert \mfW_2 \Vert \les& \sqrt{\varep^{-1} \mathbf{E}(W)}+ \sqrt{\varep^{-1} \vert\mathbf{P}(W)\vert}+ \sqrt{\varep^{-1} \vert \mathbf{G}(W)\vert} + \sqrt{\varep^{-2} \vert \mathbf{L}(W)\vert}.
\end{aligned}\label{EQnonlinWestimate1proof}
\end{align}
On the other hand, a straight-forward inspection of the higher-order terms in \eqref{EQDEFEWPW} and \eqref{EQdefGW} shows that
\begin{align}
\begin{aligned}
&\vert\triangle \Ef(x_{+W})-\Ef(W)\vert+\vert\triangle \Pf(x_{+W})-\Pf(W)\vert+\vert\triangle \Gf(x_{+W})-\Gf(W)\vert \\
\les& \varep^{3/2} \Vert \mfW_0\Vert^2+ \varep^{2} (\Vert \mfW_1\Vert^2+\Vert \mfW_2\Vert^2) + \varep^{-1/2} \lrpar{ \Vert x_{+W}-\mathfrak{m} \Vert_{\XX^{\mathrm{h.f.}}(\HH)}^2+ \Vert x-\mathfrak{m} \Vert_{\XX(\HH)}^2},
\end{aligned}\label{EQnonlinWestimate2proof}
\end{align}
and, similarly, inspecting the higher-order terms of \eqref{EQdeftriangleLW00} shows that 
\begin{align}
\begin{aligned}
\vert\triangle \Lf(x_{+W})-\Lf(W)\vert \les& \varep^2 \Vert \mfW_0 \Vert^2+ \varep^{5/2} (\Vert \mfW_1 \Vert^2+\Vert \mfW_2\Vert^2) \\
&+ \Vert x_{+W}-\mathfrak{m} \Vert_{\XX^{\mathrm{h.f.}}(\HH)}^2 + \Vert x- \mathfrak{m} \Vert_{\XX(\HH)}^2,
\end{aligned}\label{EQnonlinWestimate3proof}
\end{align}
where we estimated the quadratic term $\OO^{\Lf^m}_{x}$ in \eqref{EQdeftriangleLW00} (which depends only on $x$) by
\begin{align*}
\begin{aligned}
\vert\OO^{\Lf^m}_{x}\vert\les \Vert x-\mathfrak{m} \Vert^2_{\XX(\HH)}.
\end{aligned}
\end{align*} 
Plugging \eqref{EQnonlinWestimate3proof} and \eqref{EQnonlinWestimate2proof} into \eqref{EQnonlinWestimate1proof}, and arguing in a standard fashion (i.e. bringing terms to the left-hand side for $\varep>0$) finishes the proof of \eqref{EQw0estimate}.

As second step to prove \eqref{EQcontrolPartTHMmain0}, we recall \eqref{EQestimateDIFFx0x1} from Section \ref{SECconstructionSolution} and apply the above \eqref{EQw0estimate}, 
\begin{align*}
\begin{aligned}
&\Vert x_{+W} - x \Vert_{\XX(S_2)} + \Vert x_{+W} - x \Vert_{\XX^{\mathrm{h.f.}}(\HH)} \\
\les& \varep \Vert \mfW_0 \Vert + \varep^{5/4} \lrpar{\Vert \mfW_1 \Vert+\Vert \mfW_2 \Vert} \\
\les& \varep\lrpar{\sqrt{\varep^{-1}\triangle\Ef(x_{+W})} + \sqrt{\varep^{-1}\vert\triangle\Pf(x_{+W})\vert} + \sqrt{\varep^{-1}\vert\triangle\Gf(x_{+W})\vert}}+ \varep^{5/4}  \sqrt{\varep^{-2}\vert\triangle\Lf(x_{+W})\vert }\\
&+\varep^{1/4} \Vert x-\mathfrak{m} \Vert_{\XX(\HH)}+ \varep^{1/4} \Vert x_{+W}-\mathfrak{m} \Vert_{\XX^{\mathrm{h.f.}}(\HH)} \\
\les& \varep\lrpar{\sqrt{\varep^{-1}\triangle\Ef(x_{+W})} + \sqrt{\varep^{-1}\vert\triangle\Pf(x_{+W})\vert} + \sqrt{\varep^{-1}\vert\triangle\Gf(x_{+W})\vert}}+ \varep^{5/4}  \sqrt{\varep^{-2}\vert\triangle\Lf(x_{+W})\vert }\\
&+\varep^{1/4} \Vert x-\mathfrak{m} \Vert_{\XX(\HH)}+ \varep^{1/4} \Vert x_{+W}-x \Vert_{\XX^{\mathrm{h.f.}}(\HH)},
\end{aligned}
\end{align*}
where we wrote $x_{+W}-\mathfrak{m}=(x-\mathfrak{m})+(x_{+W}-x)$ and used \eqref{EQembeddingsobolevnulldata}. For $\varep>0$ sufficiently small, we can bring the last term to the left-hand side and thereby conclude the proof of \eqref{EQcontrolPartTHMmain0}.

\begin{remark}[Estimates for $W$] \label{REMestimatesW}Combining the estimates \eqref{EQcontrolPartTHMmain0} and \eqref{EQw0estimate}, we directly get that
\begin{align*}
\begin{aligned}
\Vert \mfW_0 \Vert \les& \sqrt{\varep^{-1}\triangle \mathbf{E}(x_{+W})} +\sqrt{\varep^{-1}\vert\triangle \mathbf{P}(x_{+W}) \vert}+ \sqrt{\varep^{-1}\vert\triangle \mathbf{G}(x_{+W}) \vert}\\
&+\varep^{1/2} \sqrt{\varep^{-2}\vert\triangle \mathbf{L}(x_{+W}) \vert} + \varep^{-3/4} \Vert x-\mathfrak{m} \Vert_{\XX(\HH)}, \\
\Vert \mfW_1 \Vert+ \Vert \mfW_2 \Vert\les& \sqrt{\varep^{-1}\triangle \mathbf{E}(x_{+W})} +\sqrt{\varep^{-1}\vert\triangle \mathbf{P}(x_{+W}) \vert}+ \sqrt{\varep^{-1}\vert\triangle \mathbf{G}(x_{+W}) \vert}\\
&+\sqrt{\varep^{-2}\vert\triangle \mathbf{L}(x_{+W}) \vert}  + \varep^{-1} \Vert x-\mathfrak{m} \Vert_{\XX(\HH)}.
\end{aligned}
\end{align*}
In particular, in view of the construction of obstruction-free null gluing in this paper, if for real numbers $\mathfrak{c}_2, \mathfrak{c}_3>0$ ($\mathfrak{c}_3>0$ large) the prescriptions $\triangle\mathbf{E}(x_{+W}), \triangle\mathbf{P}(x_{+W}),\triangle\mathbf{L}(x_{+W}),\triangle\mathbf{G}(x_{+W})$ satisfy 
\begin{align*}
\begin{aligned}
\vert \triangle\mathbf{E}(x_{+W}) \vert\les \varep, \,\, \vert\triangle\mathbf{L}(x_{+W})\vert \leq \mathfrak{c}_2 \varep^2, \,\, \triangle\mathbf{E}(x_{+W}) > \mathfrak{c}_3 \lrpar{\vert \triangle\mathbf{P}(x_{+W})\vert + \vert \triangle\mathbf{G}(x_{+W}) \vert},
\end{aligned}
\end{align*}
then by the above, with \eqref{EQsmallnessCONDEPLGnullgluing1}, for $\varep>0$ sufficiently small,
\begin{align}
\begin{aligned}
\Vert \mfW_0 \Vert \les 1+ \varep^{1/2} \sqrt{\mathfrak{c}_2} \les 1, \,\, \Vert \mfW_1 \Vert+ \Vert \mfW_2 \Vert\les 1 + \sqrt{\mathfrak{c}_2}.
\end{aligned}\label{EQsmallnessw0estimatew1w2}
\end{align}
\end{remark}

\subsection{Difference estimates}\label{SECdiffEstimates2W} In this section we prove the \emph{difference estimates} \eqref{EQdiffEstimate2mainthm2statement} of Theorem \ref{THMmain0}. That is, given two solutions $x$ and $x'$ to the null structure equations on $\HH$ satisfying 
\begin{align}
\begin{aligned}
\Vert x-\mathfrak{m} \Vert_{\XX(\HH)} \leq \varep, \,\, \Vert x'-\mathfrak{m} \Vert_{\XX(\HH)} \leq \varep,
\end{aligned}\label{EQdoublesmallnesscondition}
\end{align}
for some sufficiently small real number $\varep>0$, and given two vectors 
\begin{align*}
\begin{aligned}
(\triangle\mathbf{E},\triangle\mathbf{P},\triangle\mathbf{L},\triangle\mathbf{G}), \,\, (\triangle\mathbf{E}',\triangle\mathbf{P}',\triangle\mathbf{L}',\triangle\mathbf{G}') \in \RRR^{10},
\end{aligned}
\end{align*}
both satisfying the conditions \eqref{EQsmallnessMain021}-\eqref{EQsmallnessMain02}, we consider the solutions $W$ and $W'$ to (respectively constructed by applying part (1) of Theorem \ref{THMmain0})
\begin{align}
\begin{aligned}
(\triangle\mathbf{E},\triangle\mathbf{P},\triangle\mathbf{L},\triangle\mathbf{G})(x_{+W})=&(\triangle\mathbf{E}_0,\triangle\mathbf{P}_0,\triangle\mathbf{L}_0,\triangle\mathbf{G}_0), \\
(\triangle\mathbf{E},\triangle\mathbf{P},\triangle\mathbf{L},\triangle\mathbf{G})(x'_{+W'})=&(\triangle\mathbf{E}'_0,\triangle\mathbf{P}'_0,\triangle\mathbf{L}'_0,\triangle\mathbf{G}'_0),
\end{aligned}\label{EQnonlineardifferences1}
\end{align}
and prove that
\begin{align}
\begin{aligned}
&\Vert (x_{+W} - x) - (x'_{+W'} - x')\Vert_{\XX(S_2)} \\
\les& \, \varep \lrpar{\varep^{-1}\vert \triangle\Ef_0-\triangle\Ef_0'\vert + \varep^{-1}\vert\triangle\Pf_0-\triangle\Pf_0'\vert + \varep^{-1}\vert\triangle\Gf_0-\triangle\Gf_0'\vert} \\
&+\varep^{5/4} \lrpar{\varep^{-2}\vert\triangle\Lf_0-\triangle\Lf_0'\vert} +\varep^{1/4} \Vert x - x' \Vert_{\XX(\HH)},
\end{aligned}\label{EQdiffEstimate2mainthm2statementCITEPROOFSECTION}
\end{align}
where the constant in the inequality depends on $C_1>0$ and $C_1'>0$ of \eqref{EQsmallnessMain02}.

The proof of \eqref{EQdiffEstimate2mainthm2statementCITEPROOFSECTION} follows from direct combination of 
\begin{itemize}
\item the difference estimates proved for the construction of high-frequency solutions $x_{+W}$ in Section \ref{SECconstructionSolution}, 
\item difference estimates for two tensors $W$ and $W'$ analogous to the estimates of Section \ref{SEC1EPLGgluing}. These are proved in Section \ref{SECdiff00122} below, based on the explicit coefficient formulas \eqref{EQchoiceConstantsPG}, \eqref{EQcoefficientdefEadjustment} and \eqref{EQchoiceConstantsL}.
\item difference estimates for the null transport of $\Ef,\Pf,\Lf,\Gf$ along $\HH$. These estimates are discussed in Section \ref{SECdiff00121}, generalizing the analysis of the $\Ef,\Pf,\Lf,\Gf$ null transport in Sections \ref{SECgluingE}, \ref{SECgluingG} and \ref{SECgluingL}.
\end{itemize}
The proof of \eqref{EQdiffEstimate2mainthm2statementCITEPROOFSECTION} is given in Section \ref{SECdiff00123} below.

\subsubsection{Difference estimates for two tensors $W$ and $W'$.} \label{SECdiff00122}

\ni Let $x$ and $x'$ be two given solutions to the null structure equations along $\HH$ satisfying \eqref{EQdoublesmallnesscondition}, and consider two solutions $W$ and $W'$ constructed in Section \ref{SEC1EPLGgluing} to
\begin{align*}
\begin{aligned}
\lrpar{\mathbf{E},\mathbf{P},\mathbf{L},\mathbf{G}}(W) = \lrpar{\triangle\mathbf{E}_0,\triangle\mathbf{P}_0,\triangle\mathbf{L}_0,\triangle\mathbf{G}_0}, \,\, \lrpar{\mathbf{E},\mathbf{P},\mathbf{L},\mathbf{G}}(W') = \lrpar{\triangle\mathbf{E}'_0,\triangle\mathbf{P}'_0,\triangle\mathbf{L}'_0,\triangle\mathbf{G}'_0}.
\end{aligned}
\end{align*}
From the explicit formulas \eqref{EQchoiceConstantsPG}, \eqref{EQcoefficientdefEadjustment} and \eqref{EQchoiceConstantsL} it follows, denoting 
\begin{align}
\begin{aligned}
m:= \min\{\varep^{-1}\mathbf{E}(W'),\varep^{-1}\mathbf{E}(W)\}
\end{aligned}\label{EQdefm}
\end{align}
and using \eqref{EQrecallassumptionsEPLGtopordergluing}, that
\begin{align}
\begin{aligned}
\Vert \mfW'_0-\mfW_0 \Vert \les& \frac{1}{\sqrt{m}}\lrpar{{\varep^{-1}\vert \mathbf{E}(W')- \mathbf{E}(W) \vert} + {\varep^{-1}\vert  \mathbf{P}(W')-  \mathbf{P}(W) \vert}} \\
&+\frac{1}{\sqrt{m}}\lrpar{\varep^{-1}\vert  \mathbf{G}(W')-  \mathbf{G}(W) \vert},
\end{aligned}\label{EQdifferenceEstimateWlinearPART1}
\end{align}
and
\begin{align}
\begin{aligned}
\Vert \mfW'_1-\mfW_1 \Vert + \Vert \mfW'_2-\mfW_2 \Vert \les& \frac{1}{\sqrt{m}}\lrpar{{\varep^{-1}\vert  \mathbf{E}(W')-  \mathbf{E}(W) \vert} + {\varep^{-2}\vert  \mathbf{L}(W')-  \mathbf{L}(W) \vert}}  \\
&+ \frac{1}{\sqrt{m}}\lrpar{\sum\limits_{m=-1}^1 \varep^{-2}\vert \OO_{x',\mfW'_0}^{\Lf^m} -\OO_{x,\mfW_0}^{\Lf^m} \vert}.
\end{aligned}\label{EQdifferenceEstimateWlinearPART2precursor}
\end{align}
Similar to \eqref{EQnonlinearityEstimateLxerrorterm}, by inspecting $\OO_{x,\mfW_0}^{\Lf^m}$ we can estimate
\begin{align*}
\begin{aligned}
\varep^{-2}\vert \OO_{x',\mfW'_0}^{\Lf^m} -\OO_{x,\mfW_0}^{\Lf^m} \vert &\les \varep^{-1} \Vert x-x' \Vert_{\XX(\HH)} + \Vert \mfW_0-\mfW_0' \Vert,
\end{aligned}
\end{align*}
which, plugged into \eqref{EQdifferenceEstimateWlinearPART2precursor} and using \eqref{EQdifferenceEstimateWlinearPART1}, leads to
\begin{align}
\begin{aligned}
\Vert \mfW'_1-\mfW_1 \Vert + \Vert \mfW'_2-\mfW_2 \Vert \les& \frac{1}{\sqrt{m}}\lrpar{{\varep^{-1}\vert  \mathbf{E}(W')-  \mathbf{E}(W) \vert} + {\varep^{-2}\vert  \mathbf{L}(W')-  \mathbf{L}(W) \vert}}\\
&+\frac{1}{\sqrt{m}}\lrpar{{\varep^{-1}\vert  \mathbf{P}(W')- \mathbf{P}(W) \vert}+{\varep^{-1}\vert \mathbf{G}(W')- \mathbf{G}(W) \vert}}  \\
&+ \frac{1}{\sqrt{m}}\lrpar{\varep^{-1} \Vert x-x' \Vert_{\XX(\HH)}}.
\end{aligned}\label{EQdifferenceEstimateWlinearPART2}
\end{align}

\subsubsection{Difference estimates for the null transport of $\Ef,\Pf,\Lf,\Gf$.} \label{SECdiff00121}

\ni In the previous Sections \ref{SECgluingE}, \ref{SECgluingG} and \ref{SECgluingL}, we showed that given a solution $x$ to the null structure equations, satisfying for some sufficiently small real number $\varep>0$,
\begin{align*}
\begin{aligned}
\Vert x -\mathfrak{m} \Vert_{\XX(\HH)} \leq \varep,
\end{aligned}
\end{align*}
and given a tensor $W$ in the form of the ansatz \eqref{EQWansatz}-\eqref{EQchoiceofCoeffFunction3}, then the charges $\Ef,\Pf,\Lf,\Gf$ of the corresponding high-frequency solution $x_{+W}$ (constructed in Section \ref{SECconstructionSolution}) are transported along $\HH$ as follows (see also the more detailed \eqref{EQnonlinWestimate2proof} and \eqref{EQnonlinWestimate3proof})
\begin{align*}
\begin{aligned}
\triangle\mathbf{E}(x_{+W}) =& \mathbf{E}(W) + \OO(\varep^{5/4}), &
\triangle\mathbf{P}(x_{+W}) =& \mathbf{P}(W) + \OO(\varep^{5/4}), \\
\triangle\mathbf{G}(x_{+W}) =& \mathbf{G}(W) + \OO(\varep^{5/4}), & \triangle\mathbf{L}(x_{+W}) =& \mathbf{L}(W) + \OO_x^{\Lf}+\OO(\varep^{9/4}),
\end{aligned}
\end{align*}
where the term $\OO_x^{\Lf}$ depends only on $x$ and is of size $\varep^2$.

Inspecting the higher-order terms in the null transport equations for $\Ef,\Pf,\Lf,\Gf$ in Sections \ref{SECgluingE}, \ref{SECgluingG} and \ref{SECgluingL} and applying standard product estimates (similarly to \eqref{EQnonlinWestimate2proof} and \eqref{EQnonlinWestimate3proof}) yields the following difference estimates. Consider two solutions $x$ and $x'$ to the null structure equations along $\HH$, satisfying for some sufficiently small real number $\varep>0$,
\begin{align*}
\begin{aligned}
\Vert x -\mathfrak{m} \Vert_{\XX(\HH)} \leq \varep, \,\, \Vert x' -\mathfrak{m} \Vert_{\XX(\HH)} \leq \varep,
\end{aligned}
\end{align*}
and consider moreover two tensors $W$ and $W'$ in the form of the ansatz \eqref{EQWansatz}-\eqref{EQchoiceofCoeffFunction3}. Then we have the following difference estimates for the charges $\Ef,\Pf,\Gf$,
\begin{align}
\begin{aligned}
&\left\vert\lrpar{\triangle\mathbf{E}(x_{+W})- \mathbf{E}(W)} - \lrpar{\triangle\mathbf{E}(x'_{+W'})- \mathbf{E}(W')}\right\vert \\
&+\left\vert\lrpar{\triangle\mathbf{P}(x_{+W})- \mathbf{P}(W)} - \lrpar{\triangle\mathbf{P}(x'_{+W'})- \mathbf{P}(W')}\right\vert \\
&+\left\vert\lrpar{\triangle\mathbf{G}(x_{+W})- \mathbf{G}(W)} - \lrpar{\triangle\mathbf{G}(x'_{+W'})- \mathbf{G}(W')}\right\vert \\
\les& \varep^{1/2} \Vert x_{+W}-x'_{+W'} \Vert_{\XX^{\mathrm{h.f.}}(\HH)} + \varep^{3/2} \Vert \mfW_0-\mfW_0' \Vert \\
&+ \varep^{7/4} \lrpar{\Vert \mfW_1-\mfW_1' \Vert +\Vert \mfW_2-\mfW_2' \Vert } + \varep^{1/2} \Vert x-x'\Vert_{\XX(\HH)}.
\end{aligned}\label{EQdiffestimateEPLG1}
\end{align}
and for the charge $\Lf$,
\begin{align}
\begin{aligned}
&\left\vert \lrpar{\triangle\mathbf{L}(x_{+W})- \mathbf{L}(W)} - \lrpar{\triangle\mathbf{L}(x'_{+W'})- \mathbf{L}(W')}\right\vert \\
\les& \varep^{9/4} \Vert \mfW_0 - \mfW_0' \Vert + \varep^{9/4} \lrpar{\Vert \mfW_1 - \mfW_1' \Vert+\Vert \mfW_2 - \mfW_2' \Vert }\\
&+ \varep \lrpar{ \Vert x - x' \Vert_{\XX(\HH)} + \Vert x_{+W} - x'_{+W'} \Vert_{\XX^{\mathrm{h.f.}}(\HH)}},
\end{aligned}\label{EQdiffestimateEPLG2}
\end{align}
where we note that the ($x$- and $x'$-dependent, respectively) terms $\OO^{\Lf^m}_x$ and $\OO^{\Lf^m}_{x'}$ (defined in \eqref{EQpreciseerrorL}) were estimated as follows,
\begin{align}
\begin{aligned}
\left\vert \OO^{\Lf^m}_x - \OO^{\Lf^m}_{x'} \right\vert \les \varep \Vert x - x' \Vert_{\XX(\HH)}.
\end{aligned}\label{EQtoporderdiffestimate00112}
\end{align}

\subsubsection{Conclusion of difference estimates for the null gluing of $\Ef,\Pf,\Lf,\Gf$.} \label{SECdiff00123} In this section we conclude the proof of \eqref{EQdiffEstimate2mainthm2statementCITEPROOFSECTION}. The argument is analogous to the conclusion of \eqref{EQcontrolPartTHMmain0} in Section \ref{SECdiffEstimates1W}.

First, by \eqref{EQDEFEWPW} and \eqref{EQsmallnessMain02}, we have that for $\varep>0$ sufficiently small,
\begin{align}
\begin{aligned}
m :=& \min\{\varep^{-1}\mathbf{E}(W),\varep^{-1}\mathbf{E}(W')\} \\
\geq& \half \min\{\varep^{-1}\mathbf{E}(x_{+W}),\varep^{-1}\mathbf{E}(x'_{+W'})\} \\
=& \half \min\{C_1,C_1'\} \\
>& 0.
\end{aligned}\label{EQlowerboundWiterationestimates}
\end{align}
Second, combining the estimates \eqref{EQdifferenceEstimateWlinearPART1}, \eqref{EQdifferenceEstimateWlinearPART2} and \eqref{EQdiffestimateEPLG1}, \eqref{EQdiffestimateEPLG2} yields that, for $\varep>0$ sufficiently small (such that terms can be brought onto the left-hand side)
\begin{align}
\begin{aligned}
\Vert \mfW_0 - \mfW_0' \Vert \les& \, \varep^{-1}\vert \triangle\Ef(x_{+W})-\triangle\Ef_0(x'_{+W'})\vert + \varep^{-1}\vert\triangle\Pf(x_{+W})-\triangle\Pf(x'_{+W'})\vert \\
& + \varep^{-1}\vert\triangle\Gf(x_{+W})-\triangle\Gf(x'_{+W'})\vert +\varep^{3/4}\cdot \varep^{-2}\vert\triangle\Lf(x_{+W})-\triangle\Lf(x'_{+W'})\vert  \\
&+ \varep^{-1/4} \lrpar{\Vert x - x' \Vert_{\XX(\HH)}+\Vert x_{+W} - x'_{+W'} \Vert_{\XX^{\mathrm{h.f.}}(\HH)}},
\end{aligned}\label{EQW0fullnonlindffest13}
\end{align}
and
\begin{align}
\begin{aligned}
&\Vert \mfW_1 - \mfW_1' \Vert+\Vert \mfW_2- \mfW_2' \Vert \\
\les& \, \varep^{-1}\vert \triangle\Ef(x_{+W})-\triangle\Ef_0(x'_{+W'})\vert + \varep^{-1}\vert\triangle\Pf(x_{+W})-\triangle\Pf(x'_{+W'})\vert \\
& + \varep^{-1}\vert\triangle\Gf(x_{+W})-\triangle\Gf(x'_{+W'})\vert + \varep^{-2}\vert\triangle\Lf(x_{+W})-\triangle\Lf(x'_{+W'})\vert  \\
&+ \varep^{-1} \lrpar{\Vert x - x' \Vert_{\XX(\HH)}+\Vert x_{+W} - x'_{+W'} \Vert_{\XX^{\mathrm{h.f.}}(\HH)}},
\end{aligned}\label{EQW0fullnonlindffest134}
\end{align}
where the constants in the inequalities depend on $C_1>0$ and $C_1'>0$ of \eqref{EQsmallnessMain02}.

Third, applying \eqref{EQW0fullnonlindffest13} and \eqref{EQW0fullnonlindffest134} in the difference estimates for the construction of $x_{+W}$, see \eqref{EQtoProveIterationEstimateForX} and Section \ref{SECestimateDiffest2}, and writing
\begin{align*}
\begin{aligned}
x_{+W} - x'_{+W'} = \lrpar{\lrpar{x_{+W} - x} -\lrpar{x'_{+W'}-x'}} + \lrpar{x-x'},
\end{aligned}
\end{align*}
we get that
\begin{align*}
\begin{aligned}
&\Vert \lrpar{x_{+W} - x}-\lrpar{x'_{+W'}-x'} \Vert_{\XX(S_2)} + \Vert \lrpar{x_{+W} - x} -\lrpar{x'_{+W'}-x'} \Vert_{\XX^{\mathrm{h.f.}}(\HH)} \\
\les& \varep \Vert \mfW_0-\mfW_0' \Vert + \varep^{5/4} \Vert W-W' \Vert + \varep^{1/2} \Vert x - x' \Vert_{\XX(\HH)} \\
\les& \, \varep \lrpar{\varep^{-1}\vert \triangle\Ef_0-\triangle\Ef_0'\vert + \varep^{-1}\vert\triangle\Pf_0-\triangle\Pf_0'\vert + \varep^{-1}\vert\triangle\Gf_0-\triangle\Gf_0'\vert} \\
&+\varep^{5/4} \lrpar{\varep^{-2}\vert\triangle\Lf_0-\triangle\Lf_0'\vert} +\varep^{1/4} \Vert x - x' \Vert_{\XX(\HH)}+ \varep^{1/4}\Vert \lrpar{x_{+W} - x} -\lrpar{x'_{+W'}-x'} \Vert_{\XX^{\mathrm{h.f.}}(\HH)},
\end{aligned}
\end{align*}
where the constant in the inequality depends on $C_1>0$ and $C_1'>0$ of \eqref{EQsmallnessMain02}. For $\varep>0$ sufficiently small we can bring the last term to the left-hand side, which finishes the proof of \eqref{EQdiffEstimate2mainthm2statementCITEPROOFSECTION}.

\appendix

\section{Fourier theory} \label{APPconstructionW} 

\ni In this section all integrals and differential operators are with respect to $\gac$ on the unit sphere $\SSS^2$. 

\subsection{Tensor spherical harmonics} \label{SECdefTENSORspherical} For $l\geq0$, $m=-l,\dots,l$, we let $Y^{lm}$ denote the real-valued spherical harmonics. Moreover, we define (see, for example, \cite{Czimek1}),
\begin{align*}
\begin{aligned}
E^{lm} :=& \frac{1}{\sqrt{l(l+1)}} \DDd_1^\ast(Y^{lm},0), & H^{lm} :=& \frac{1}{\sqrt{l(l+1)}} \DDd_1^\ast(0,Y^{lm}), & &\text{ for } l\geq1, m=-l,\dots,l,\\
\psi^{lm} :=& \frac{1}{\sqrt{\half l(l+1)-1}} \DDd_2^\ast E^{lm}, & \phi^{lm} :=& \frac{1}{\sqrt{\half l(l+1)-1}} \DDd_2^\ast H^{lm} & &\text{ for } l\geq2, m=-l,\dots,l,
\end{aligned}
\end{align*}
where for functions $f_1,f_2$ and vectorfields $X$, with ${}^\ast X_A := \ind_{AB}X^B$,
\begin{align*}
\begin{aligned}
\DD_1^\ast(f_1,f_2):= -\di f_1 +{}^\ast \di f_2, \,\,(\DD_2^\ast X)_{AB} := -\half \lrpar{\Nd_AX_B+\Nd_BX_A - \Divd X \gac_{AB}}.
\end{aligned}
\end{align*}
Moreover, it holds that ${}^\ast E^{lm} = -H^{lm}, {}^\ast H^{lm} = E^{lm}$, and
\begin{align}
\begin{aligned}
E^{lm}\cdot E^{l'm'} = H^{lm}\cdot H^{l'm'}, \,\, E^{lm}\cdot H^{l'm'} = - H^{lm}\cdot E^{l'm'}, \,\, E^{lm}\cdot H^{lm} =0.
\end{aligned}\label{EQbasicsEM}
\end{align}
Tensor spherical harmonics satisfy moreover 
\begin{align}
\begin{aligned}
\Ld Y^{lm} =& -l(l+1)Y^{lm}, & \Ld E^{lm}=&(1-l(l+1)) E^{lm}, & \Ld H^{lm}=&(1-l(l+1)) H^{lm}, \\ 
\Ld \psi^{lm}=&(4-l(l+1)) \psi^{lm}, & \Ld \phi^{lm}=&(4-l(l+1)) \psi^{lm}.
\end{aligned}\label{EQeigenvalue}
\end{align}
For scalar functions $f$, vectorfields $X$ and $\gac$-tracefree symmetric $2$-tensors $V$, we define
\begin{align}
\begin{aligned}
f^{(lm)} :=& \int_S f \, Y^{lm} &&\text{ for } l\geq0, m=-l,\dots,l, \\
X_E^{(lm)} :=& \int_S X\cdot E^{lm}, \,\, X_H^{(lm)} := \int_S X\cdot H^{lm}, & &\text{ for } l\geq1, m=-l,\dots,l, \\
V_\psi^{(lm)} :=& \int_S V\cdot \psi^{lm}, \,\, V_\phi^{(lm)} := \int_S V\cdot \phi^{lm}, & &\text{ for } l\geq2, m=-l,\dots,l.
\end{aligned}\label{EQdefProjectionsFourier}
\end{align}

\subsection{Fourier relations} The following relations are used for calculations.

\begin{lemma}\label{LEMfourier} For $l\geq1,l'\geq1, m=-l,\dots,l, m'=-l',\dots,l'$, $\tilde l \geq0, \tilde m = -\tilde l, \dots, \tilde l$,
\begin{align*}
\begin{aligned}
\lrpar{E^{lm} \cdot E^{l'm'}}^{(\tilde{l}\tilde{m})} = \frac{l(l+1)+l'(l'+1)-\tilde{l}(\tilde{l}+1)}{2\sqrt{l(l+1)}\sqrt{l'(l'+1)}} \lrpar{ Y^{lm}Y^{l'm'} }^{\tilde{l}\tilde{m}}.
\end{aligned}
\end{align*}
Moreover, for $l\geq2,l'\geq2, m=-l,\dots,l, m'=-l',\dots,l'$ and $\tilde l \geq0, \tilde m = -\tilde l, \dots, \tilde l$,
\begin{align*}
\begin{aligned}
\lrpar{\psi^{lm} \cdot \psi^{l'm'}}^{(\tilde{l}\tilde{m})} =& \frac{(l(l+1)-1)+(l'(l'+1)-1) - \tilde l (\tilde l+1)}{2\sqrt{\half l(l+1)-1} \sqrt{\half l'(l'+1)-1}} \lrpar{ E^{l'm'} \cdot E^{lm}}^{(\tilde l \tilde m)} \\
& - \frac{\sqrt{l(l+1)} \sqrt{l'(l'+1)}}{2 \sqrt{\half l(l+1)-1} \sqrt{\half l'(l'+1)-1}}\lrpar{ Y^{l'm'}Y^{lm} }^{(\tilde l \tilde m)}.
\end{aligned}
\end{align*}
\end{lemma}
\begin{proof}[Proof of Lemma \ref{LEMfourier}] By \eqref{EQeigenvalue} it holds that
\begin{align*}
\begin{aligned}
-\tilde{l}(\tilde{l}+1) \int Y^{lm}Y^{l'm'} Y^{\tilde{l}\tilde{m}} =& \int Y^{lm}Y^{l'm'} \Ldo Y^{\tilde{l}\tilde{m}} \\
=& \int \Ldo \lrpar{ Y^{lm}Y^{l'm'}} Y^{\tilde{l}\tilde{m}} \\
=& \lrpar{-l(l+1)-l'(l'+1)} \int Y^{lm}Y^{l'm'} Y^{\tilde{l}\tilde{m}} +2 \int  \di Y^{lm}\cdot \di Y^{l'm'} Y^{\tilde{l}\tilde{m}},
\end{aligned}
\end{align*}
from which we conclude that
\begin{align*}
\begin{aligned}
\int  \di Y^{lm}\cdot \di Y^{l'm'} Y^{\tilde{l}\tilde{m}} = \frac{l(l+1)+l'(l'+1)-\tilde{l}(\tilde{l}+1)}{2} \int Y^{lm}Y^{l'm'} Y^{\tilde{l}\tilde{m}}.
\end{aligned}
\end{align*}
The first identity of Lemma \ref{LEMfourier} follows by definition of $E^{lm}$ in Section \ref{SECdefTENSORspherical}. The proof of the second identity is similar and omitted. \end{proof}

\ni The next lemma is proved similarly by integration by parts, its proof is omitted.
\begin{lemma} \label{LEMbilinearFourierID} For $l,l'\geq2$, $m=-l,\dots,l$, $m'=-l',\dots,l'$, and $\tilde m =-1,0,1$,
\begin{align}
\begin{aligned}
\int\psi^{lm}_{AB} \Nd_A \phi^{l'm'}_{BC}H^{1\tm}_C =& \frac{4-l'(l'+1)-l(l+1)+2}{2\sql} \int H^{1\tm}_C E^{lm}_B \phi^{l'm'}_{BC} \\
& +\frac{\sqll \sqlp}{2\sql}  \int Y^{lm} E^{l'm'} \cdot E^{1\tm}, \\
\int\phi^{lm}_{AB} \Nd_A \psi^{l'm'}_{BC}H^{1\tm}_C =& -\sqrt{\half l(l+1)-1} \int H^{lm}_B \Psi^{l'm'}_{BC} H^{1\tm}_C \\
&-\frac{4-l(l+1)+l'(l'+1)}{\sqrt{\half l'(l'+1)-1}} \int \phi^{lm}_{AB}E^{l'm'}_A H^{1\tm}_B.
\end{aligned}\label{EQintegrationByPartsLemmaTRI}
\end{align}
\end{lemma}
\ni The following integration-by-parts identity is used to calculate terms such as appearing on the right-hand side of \eqref{EQintegrationByPartsLemmaTRI}. Its proof is omitted.
\begin{lemma} \label{LEMibp1} Let $X,Y$ be two vectorfields and let $V= \DDd_2^\ast Z$ be a symmetric tracefree $2$-tensor. Then it holds that
\begin{align*}
\begin{aligned}
\int V^{AB}X_AY_B = \half \int \lrpar{Y\cdot Z \Divd X + X \cdot Z \Divd Y +X\cdot{}^\ast Z\Curld Y + Y\cdot{}^{\ast}Z \Curld X}.
\end{aligned}
\end{align*}
\end{lemma}

\subsection{Collection of Fourier coefficients} The next lemma follows from explicit calculation.\begin{lemma}[Products of spherical harmonics] \label{LEMCGcoeff} It holds that for $\tilde m=-1,0,1$,
\begin{align*}
\begin{aligned}
\lrpar{Y^{2m} Y^{2m'}}^{(1\tilde{m})} = 0 \,\,\,\,\, \text{ for } m,m'=-2,\dots,2.
\end{aligned}
\end{align*}
Moreover, the only non-vanishing components of 
\begin{align*}
\begin{aligned}
\lrpar{Y^{30}Y^{30}}^{(1\tilde{m})}, \,\, \lrpar{Y^{30}Y^{2m'}}^{(1\tilde{m})}  \,\,\,\,\, \text{ for } m'=-2,\dots,2, \, \tilde m=-1,0,1,
\end{aligned}
\end{align*}
are given by
\begin{align*}
\begin{aligned}
\lrpar{Y^{30}Y^{2-1}}^{(1-1)} = -\frac{\sqrt{3}}{\sqrt{35}} \sqrt{\frac{3}{4\pi}}, \,\,
\lrpar{ Y^{30}Y^{20}}^{(10)} = \frac{3}{\sqrt{35}} \sqrt{\frac{3}{4\pi}}, \,\,
\lrpar{ Y^{30}Y^{21}}^{11} = -\frac{\sqrt{3}}{\sqrt{35}} \sqrt{\frac{3}{4\pi}}.
\end{aligned}
\end{align*}
\end{lemma}

\subsection{Fourier theory for the gluing of $\mathbf{E},\mathbf{P}$ and $\mathbf{G}$} \label{SECFouriervertWvert2} 

By the ansatz \eqref{EQWansatzDetail} for $\mfW_0$ it follows on the one hand that for $\tilde m =-1,0,1$,
\begin{align}
\begin{aligned}
\lrpar{\vert \mfW_0 \vert^2_\gac}^{(1\tilde m)} = \sum\limits_{m=-1}^{1} f_{30}f_{2m} \lrpar{\psi^{30}\cdot\psi^{2m}}^{(1\tilde m)} = f_{30}f_{2\tilde m} \underbrace{\lrpar{\psi^{30}\cdot\psi^{2\tilde m}}^{(1\tilde m)}}_{\neq 0},
\end{aligned}\label{EQFourierFactPG}
\end{align}
where we applied Lemma \ref{LEMCGcoeff}, using that by Lemma \ref{LEMfourier},
\begin{align*}
\begin{aligned}
\lrpar{\psi^{2m} \cdot \psi^{2m'}}^{(1\tilde{m})} = \frac{1}{6}\lrpar{ Y^{2m}Y^{2m'} }^{1\tilde{m}}, \,\,\lrpar{\psi^{3m} \cdot \psi^{2m'}}^{(1\tilde{m})} = \frac{\sqrt{5}}{3} \lrpar{ Y^{3m}Y^{2m'} }^{1\tilde{m}}.
\end{aligned}
\end{align*}

\ni On the other hand, as the $\psi^{lm}$ form an orthonormal basis, 
\begin{align}
\begin{aligned}
\lrpar{\vert \mfW_0\vert_\gac^2}^{(0)} = \frac{1}{\sqrt{4\pi}} \lrpar{f_{30}^2 + \sum\limits_{m=-1}^1 f_{2m}^2 + f^2_{22}}.
\end{aligned}\label{EQFourierFactE}
\end{align}
\subsection{Fourier theory for the gluing of $\mathbf{L}$} \label{SECFourierWNDW} 

In this section we analyze the integrals
\begin{align*}
\begin{aligned}
\int (H^{1\tilde{m}})^C \Nd_A \lrpar{(\Psi^{30})^{AB}(\phi^{2 m'})_{B}^{\,\,\,\,C}} \,\,\, \text{ and } \,\,\, \int (H^{1\tilde{m}})^C \Nd_C (\Psi^{30})_{AD}(\phi^{2 m'})^{AD}.
\end{aligned}
\end{align*}
On the one hand, by $\Divd \Psi^{30} = \sqrt{5} E^{30}$ and Lemmas \ref{LEMfourier}, \ref{LEMbilinearFourierID} and \ref{LEMibp1},
\begin{align*}
\begin{aligned}
\int (H^{1\tilde{m}})^C \Nd_A \lrpar{(\Psi^{30})^{AB}(\phi^{2 m'})_{BC}} =& \sqrt{5} \int (H^{(1\tilde m)})^C (E^{30})^B \phi^{2m}_{BC} + \underbrace{\int (H^{(1\tilde m)})^C (\Psi^{30})^{AB} \Nd_A \phi_{BC}^{2m'}}_{=0} \\
=& -\frac{\sqrt{5}}{3\sqrt{2}} \lrpar{Y^{30}Y^{2m'}}^{(1\tilde m)},
\end{aligned}
\end{align*}
On the other hand, by Lemmas \ref{LEMfourier}, \ref{LEMbilinearFourierID} and \ref{LEMibp1} and using that $\Psi^{30}=\frac{1}{\sqrt{5}} \DDd_2^\ast(E^{30})$ by Section \ref{SECdefTENSORspherical},
\begin{align*}
\begin{aligned}
&\int (H^{1\tilde{m}})^C \Nd_C (\Psi^{30})_{AD}(\phi^{2 m'})^{AD}\\
=& -\frac{1}{\sqrt{5}} \int (H^{(1\tilde m)})^C \Nd_C \Nd_A E^{30}_D (\phi^{2m'})^{AD}\\
=& -\frac{1}{\sqrt{5}} \int (H^{(1\tilde m)})^C \Nd_A \Nd_C E^{30}_D (\phi^{2m'})^{AD} -\frac{1}{\sqrt{5}} \int (H^{(1\tilde m)})^C (\Nd_C \Nd_A-\Nd_A\Nd_C) E^{30}_D (\phi^{2m'})^{AD} \\
=& -\frac{1}{\sqrt{5}} \lrpar{-\sqrt{5} \int (H^{(1\tilde m)})^C\Nd_A \Psi^{30}_{CD}(\phi^{2m'})^{AD}+\sqrt{2}\lrpar{Y^{30}Y^{2m'}}^{(1\tilde m)}} -\frac{1}{\sqrt{5}} \lrpar{-\frac{1}{3\sqrt{2}}\lrpar{Y^{30}Y^{2m'}}^{(1\tilde m)}}\\
=& -\frac{1}{\sqrt{5}} \lrpar{-\frac{5}{\sqrt{2}}+\sqrt{2}-\frac{1}{3\sqrt{2}}},
\end{aligned}
\end{align*}
where we used that on the round unit sphere,
\begin{align*}
\begin{aligned}
(\Nd_C \Nd_A-\Nd_A\Nd_C) (E^{30})_D = \lrpar{\gac_{DC} \gac_{BA}-\gac_{DA}\gac_{BC}} (E^{30})^B = \gac_{DC}E^{30}_A - \gac_{DA} E^{30}_C.
\end{aligned}
\end{align*}

\ni Combining the above two calculations, we get in particular that
\begin{align}
\begin{aligned}
&\half \lrpar{\, \int (H^{1\tilde{m}})^C \Nd_A \lrpar{(\Psi^{30})^{AB}(\phi^{2 m'})_{B}^{\,\,\,\,C}}} -\frac{1}{4} \lrpar{\, \int (H^{1\tilde{m}})^C \Nd_C (\Psi^{30})_{AD}(\phi^{2 m'})^{AD}}\\
=& -\frac{\sqrt{5}}{3\sqrt{2}} \lrpar{Y^{30}Y^{2m'}}^{(1\tilde m)}.
\end{aligned}\label{EQFourierLTheorySurjectivity}
\end{align}
By Lemma \ref{LEMCGcoeff}, the right-hand side of \eqref{EQFourierLTheorySurjectivity} is non-vanishing for $m'=\tilde{m}$. 

\section{Derivation of null transport equations for $\Ef,\Pf,\Lf,\Gf$} \label{AppDerivation}

\ni In the following we first derive null transport equations for $\mathfrak{m}$ and $\mfb$ defined in \eqref{EQdefMUBETA}. Projecting these null transport equations onto the Fourier modes $l=0$ and $l=1$ in accordance with \eqref{EQdefNonlinCharges} then yields the transport equations for the charges $\Ef,\Pf,\Lf,\Gf$. Here we use that the Fourier decomposition is defined with respect to the unit round metric $\gac$ on each sphere and is thus $v$-invariant.\\

\ni \textbf{Transport equation for $\mfb$.} We claim that 
\begin{align}
\begin{aligned}
D\mfb =& \frac{\phi^3}{2\Om^2} \lrpar{-(\eta-2\di\log\Om) \vert \Om\chih\vert^2_\gd - \di \lrpar{\vert \Om\chih\vert^2_\gd} +\Om\trchi \Divd\lrpar{\Om\chih}}.
\end{aligned}\label{EQBequation1}
\end{align}
\ni Indeed, from \eqref{EQdefRicciINTRO}, \eqref{EQfirstvariation001}, \eqref{EQcompactRAY} and \eqref{EQnulleta}, and using that $[D,\di]=0$,
\begin{align*}
\begin{aligned}
D\mfb =& D\lrpar{\frac{\phi^3}{2\Om} \lrpar{\di\trchi+\trchi(\eta-\di\log\Om)}} \\
=& \frac{\phi^3}{2\Om}\lrpar{\frac{3\Om\trchi}{2}-\om} \lrpar{\di\trchi+\trchi(\eta-\di\log\Om)} \\
&+\frac{\phi^3}{2\Om} \lrpar{\di \lrpar{D\trchi}+ (D\trchi)(\eta-\di\log\Om) +\trchi \lrpar{D\eta-\di\om}}\\
=& \frac{\phi^3}{2\Om} \lrpar{-\di\lrpar{\Om\vert\chih\vert^2} -\Om\vert\chih\vert^2 (\eta-\di\log\Om) +\Om\trchi\Divd\chih + \Om\trchi\chih\cdot\di\log\Om} \\
=& \frac{\phi^3}{2\Om^2} \lrpar{-(\eta-2\di\log\Om) \vert\Om\chih\vert^2 - \di\lrpar{\vert\Om\chih\vert^2} + \Om\trchi \Divd (\Om\chih)}.
\end{aligned}
\end{align*}
From \eqref{EQdefNonlinCharges} and \eqref{EQBequation1} we explicitly get that
\begin{align}
\begin{aligned}
D\Lf =& \lrpar{\frac{\phi^3}{2\Om^2} \lrpar{-(\eta-2\di\log\Om) \vert \Om\chih\vert^2_\gd - \di \lrpar{\vert \Om\chih\vert^2_\gd} +\Om\trchi \Divd\lrpar{\Om\chih}}}_H^{(1m)},\\
D\Gf =& \lrpar{\frac{\phi^3}{2\Om^2} \lrpar{-(\eta-2\di\log\Om) \vert \Om\chih\vert^2_\gd - \di \lrpar{\vert \Om\chih\vert^2_\gd} +\Om\trchi \Divd\lrpar{\Om\chih}}}^{(1m)}_E.
\end{aligned}\label{EQLGtransportEquations}
\end{align}

\ni \textbf{Transport equation for $\mathfrak{m}$.} We claim that
\begin{align}
\begin{aligned}
D\mathfrak{m}=& \phi^3 \Divd\Divd \lrpar{\Om\chih} - \frac{\phi}{2} \Divd \lrpar{\frac{\Om\trchi \phi^3}{\Om^2} \Divd\lrpar{\Om\chih}} \\
&-\frac{\phi^3}{2} \Ld \lrpar{\Om\trchi} + \frac{\Om\trchi\phi}{4} \Divd\lrpar{\frac{\phi^3}{\Om^2}\di\lrpar{\Om\trchi}}\\
&-\frac{\Om\trchi\phi^3}{2} \Divd\lrpar{\eta-2\di\log\Om} + \frac{\Om\trchi\phi}{4}\Divd\lrpar{\frac{\phi^3\Om\trchi}{\Om^2}\lrpar{\eta-2\di\log\Om}} \\
&-\frac{\Om\trchib\phi^3}{4\Om^2} \vert \Om\chih \vert^2 + \frac{\phi}{2} \Divd \lrpar{\frac{\phi^3}{\Om^2}\di\lrpar{\vert\Om\chih\vert^2}} \\
&+\frac{\Om\trchi \phi^3}{2} \vert \eta-2\di\log\Om\vert^2 + 2\phi \Divd \lrpar{\Om\chih \cdot \mfb} +\frac{\phi}{2}\Divd\lrpar{\frac{\phi^3}{\Om^2} (\eta-2\di\log\Om) \vert\Om\chih\vert^2}.
\end{aligned}\label{EQnulltransportMU}
\end{align}
Indeed, we start by writing
\begin{align}
\begin{aligned}
D\mathfrak{m} =& D\lrpar{\phi^3\lrpar{K+\frac{1}{4}\trchi\trchib}}-D\lrpar{\phi\Divd \mfb}. 
\end{aligned}\label{EQmueqderiv0}
\end{align}
On the one hand, using that the Gauss curvature $K$ satisfies (see, for example, (5.28) in \cite{ChrForm})
\begin{align*}
\begin{aligned}
D\lrpar{\phi^2 K} = \phi^2 \lrpar{\Divd\Divd \lrpar{\Om\chih} -\half \Ld (\Om\trchi)}
\end{aligned}
\end{align*}
and that by \eqref{EQcompactRAY} and \eqref{EQtrchib}, 
\begin{align*}
\begin{aligned}
\frac{1}{4} D\lrpar{\phi^3 \trchi\trchib} =& \frac{1}{4} D\lrpar{\lrpar{\frac{\phi\trchi}{\Om} }\lrpar{\Om\phi^2\trchib}}\\
=& \frac{1}{4} \lrpar{-\phi\vert\chih\vert^2} \Om\phi^2\trchib \\
&+\frac{1}{4}{ \frac{\phi\trchi}{\Om} \lrpar{-2\phi^2\Om^2\Divd\lrpar{\eta-2\di\log\Om} + 2\phi^2\Om^2\vert \eta-2\di \log \Om\vert^2-2\Om^2\phi^2K}},
\end{aligned}
\end{align*}
we get that, using again \eqref{EQcompactRAY},
\begin{align}
\begin{aligned}
&D\lrpar{\phi^3\lrpar{K+\frac{1}{4}\trchi\trchib}}\\
=&\frac{\Om\trchi\phi}{2} \phi^2 K + \phi^3 \lrpar{\Divd\Divd \lrpar{\Om\chih} -\half \Ld (\Om\trchi)} +\frac{1}{4} \lrpar{-\phi\vert\chih\vert^2} \Om\phi^2\trchib \\
&+ \frac{1}{4}{ \frac{\phi\trchi}{\Om} \lrpar{-2\phi^2\Om^2\Divd\lrpar{\eta-2\di\log\Om} + 2\phi^2\Om^2\vert \eta-2\di \log \Om\vert^2-2\Om^2\phi^2K}} \\
=& \phi^3 \lrpar{\Divd\Divd \lrpar{\Om\chih} -\half \Ld (\Om\trchi)} + \frac{1}{4} \lrpar{-\phi\vert\chih\vert^2} \Om\phi^2\trchib \\
&+\frac{1}{4}{ \frac{\phi\trchi}{\Om} \lrpar{-2\phi^2\Om^2\Divd\lrpar{\eta-2\di\log\Om} + 2\phi^2\Om^2\vert \eta-2\di \log \Om\vert^2}}.
\end{aligned}\label{EQmueqderiv1}
\end{align}
On the other hand, using that for $1$-forms $\xi$ (see, for example, (6.107) in \cite{ChrForm})
\begin{align*}
\begin{aligned}
D\Divd \xi - \Divd D \xi = -2\Divd \lrpar{\Om\chih \cdot \xi} - \Om\trchi \Divd \xi,
\end{aligned}
\end{align*}
and \eqref{EQBequation1}, we have that by \eqref{EQcompactRAY},
\begin{align}
\begin{aligned}
&-D\lrpar{\phi \Divd \mfb} \\
=&- \frac{\Om\trchi\phi}{2} \Divd \mfb - \phi \lrpar{\Divd \lrpar{D\mfb}-2\Divd \lrpar{\Om\chih\cdot \mfb} -\Om\trchi \Divd \mfb} \\
=&  \frac{\Om\trchi\phi}{2} \Divd \mfb + 2\phi \Divd \lrpar{\Om\chih\cdot \mfb} \\
&-\phi \Divd \lrpar{\frac{\phi^3}{2\Om^2} \lrpar{-(\eta-2\di\log\Om) \vert \Om\chih\vert^2_\gd - \di \lrpar{\vert \Om\chih\vert^2_\gd} +\Om\trchi \Divd\lrpar{\Om\chih}}}.
\end{aligned}\label{EQmueqderiv2}
\end{align}

\ni Plugging \eqref{EQmueqderiv1} and \eqref{EQmueqderiv2} into \eqref{EQmueqderiv0} and using \eqref{EQdefMUBETA} finishes the proof of \eqref{EQnulltransportMU}.

Projecting \eqref{EQnulltransportMU} onto the Fourier modes \eqref{EQdefNonlinCharges} yields the null transport equations for $\Ef$ and $\Pf$, respectively. We omit writing out these projections of \eqref{EQnulltransportMU} explicitly.



\begin{thebibliography}{99}

\bibitem{ACR1}
S.~Aretakis, S.~Czimek, I.~Rodnianski.
\newblock {\em The characteristic gluing problem for the Einstein equations and applications}. 
\newblock arXiv:2107.02441, 31 pages.

\bibitem{ACR2}
S.~Aretakis, S.~Czimek, I.~Rodnianski.
\newblock {\em The characteristic gluing problem for the Einstein vacuum equations. Linear and non-linear analysis}. 
\newblock arXiv:2107.02449, 102 pages.

\bibitem{ACR3}
S.~Aretakis, S.~Czimek, I.~Rodnianski.
\newblock {\em Characteristic gluing to the Kerr family and application to spacelike gluing}. 
\newblock arXiv:2107.02456, 88 pages.
\bibitem{CarlottoSchoen}
A.~Carlotto, R.~Schoen. 
\newblock {\em Localizing solutions of the Einstein constraint equations}. 
\newblock Invent. Math., 205(3):559-615, 2016.
\bibitem{ChristodoulouMASS}
D.~Christodoulou.
\newblock {\em Reversible and Irreversible Transformations in Black-Hole Physics}. 
\newblock Phys. Rev. Lett., vol. 25, n. 22, 1970.

\bibitem{CHRmassaspect}
D.~Christodoulou.
\newblock {\em Nonlinear Nature of Gravitation and Gravitational-Wave Experiments}. 
\newblock Phys. Rev. Lett., vol. 67, n. 12, 1991.

\bibitem{ChrForm}
D.~Christodoulou.
\newblock {\em The formation of black holes in general relativity}. 
\newblock European Mathematical Society (EMS), Z\"urich, 2009. x+589 pp.


\bibitem{ChrKl}
D.~Christodoulou, S.~Klainerman.
\newblock {\em The global nonlinear stability of the Minkowski space}. 
\newblock Princeton University Press, 1993.

\bibitem{ChruscielDelay1}
P.~Chru\'sciel, E.~Delay.
\newblock {\em Existence of non-trivial, vacuum, asymptotically simple spacetimes}. 
\newblock Classical and Quantum Gravity, 19(9):L71, 2002.
\bibitem{ChruscielDelay}
P.~Chru\'sciel, E.~Delay.
\newblock {\em On mapping properties of the general relativistic constraints operator in weighted function spaces, with applications}. 
\newblock M\'em. Soc. Math. Fr. (N.S.), (94):vi+103, 2003.

\bibitem{CIP1}
P.~Chru\'sciel, J.~Isenberg, D.~Pollack.
\newblock {\em Gluing initial data sets for general relativity}. 
\newblock Physical review letters, 93(8):081101, 2004.
%
\bibitem{CIP2}
P.~Chru\'sciel, J.~Isenberg, D.~Pollack.
\newblock {\em Initial data engineering}. 
\newblock Comm. Math. Phys., 257(1):29-42, 2005.
\bibitem{ChruscielMazzeo}
P.~Chru\'sciel, R.~Mazzeo.
\newblock {\em On 'many-black-hole' vacuum spacetimes}. 
\newblock Classical and Quantum Gravity, 20(4):729, 2003.
\bibitem{ChruscielPollack}
P.~Chru\'sciel, D.~Pollack.
\newblock {\em Singular Yamabe metrics and initial data with exactly Kottler-Schwarzschild-de Sitter ends}. 
\newblock Ann. Henri Poincar\'e, 9(4):639-654, 2008.
\bibitem{Cortier}
J.~Cortier.
\newblock {\em Gluing construction of initial data with Kerr-de Sitter ends}. 
\newblock Ann. Henri Poincar\'e, 14(5):1109-1134, 2013.
%
%
\bibitem{Corvino}
J.~Corvino.
\newblock {\em Scalar Curvature Deformation and a Gluing Construction for the Einstein Constraint Equations}.
\newblock Comm. Math. Phys. 214 (2000), 137-189.
%
\bibitem{CorvinoSchoen}
J.~Corvino, R.~Schoen.
\newblock {\em On the Asymptotics for the Vacuum Einstein Constraint Equations}.
\newblock J. Differential Geom. 73 (2006), no. 2, 185-217.



\bibitem{Czimek1}
S.~Czimek.
\newblock {\em An extension procedure for the constraint equations}. 
\newblock Ann. PDE (2018), 4:2, 122 pages. 

\bibitem{DHR}
M. Dafermos, G. Holzegel, I. Rodnianski
\newblock{\em The linear stability of the Schwarzschild solution to gravitational perturbations}.
\newblock Acta Math. 222 (2019) no. 1, pp. 1-214.

\bibitem{DeLellisSz}
C.~De Lellis, L.~Sz\'ekelyhidi.
\newblock {\em The Euler equations as a differential inclusion}. 
\newblock Ann. of Math., 170 (2009) 3, p. 1417-1436.

\bibitem{Gromov}
M.~Gromov.
\newblock {\em Partial differential relations}.
\newblock Ergebnisse der Mathematik und ihrer Grenzgebiete, Vol. 9, Springer-Verlag, 1986, ix+363 pp.

\bibitem{GromovL}
M.~Gromov, H.~Lawson.
\newblock {\em The classification of simply connected manifolds of positive scalar curvature}.
\newblock Ann. of Math. (2) 111 (1980), no. 3, 423-434.
%
\bibitem{Hintz}
P.~Hintz.
\newblock {\em Black hole gluing in de Sitter space}.
\newblock CPDE 46 (2021) no. 7, 1280-1318.

\bibitem{IMP3}
J.~Isenberg, D.~Maxwell, D.~Pollack.
\newblock {\em A gluing construction for non-vacuum solutions of the Einstein-constraint equations}. 
\newblock Adv. Theor. Math. Phys., 9(1):129-172, 2005.
\bibitem{IMP1}
J.~Isenberg, R.~Mazzeo, D.~Pollack.
\newblock {\em Gluing and wormholes for the Einstein constraint equations}. 
\newblock Comm. Math. Phys., 231(3):529-568, 2002.
%
\bibitem{IMP2}
J.~Isenberg, R.~Mazzeo, D.~Pollack.
\newblock {\em On the topology of vacuum spacetimes}. 
\newblock Ann. Henri Poincar\'e, volume 4, pages 369-383. Springer, 2003.

\bibitem{LukChar}
J.~Luk.
\newblock {\em On the local existence for the characteristic initial value problem in general relativity}. 
\newblock Int. Math. Res. Not. IMRN 2012, no. 20, 4625-4678.
\bibitem{LukRod1}
J.~Luk, I.~Rodnianski.
\newblock {\em Local propagation of impulsive gravitational waves}.
\newblock Comm. Pure Appl. Math. 68 (2015), no. 4, 511-624.


\bibitem{LukRodnianskiHIGH}
J.~Luk, I.~Rodnianski.
\newblock {\em High-frequency limits and null dust shell solutions in general relativity}. 
\newblock arXiv:2009.08968, 95 pages.


\bibitem{Nash}
J.~Nash.
\newblock {\em $C^1$ isometric imbeddings}. 
\newblock Ann. of Math., 60 (1954) no. 3, pp. 383-396.

\bibitem{Nash1}
J.~Nash.
\newblock {\em
The imbedding problem for Riemannian manifolds}. 
\newblock Annals of Math., 63 (1956) no. 1, pp. 20-63.

\bibitem{PenroseNuovo}
R.~Penrose.
\newblock {\em Gravitational Collapse: The Role of General Relativity}.
\newblock Rivista del Nuovo Cimento, Numero Speziale I, 257 (1969).


\bibitem{Rendall}
A.~Rendall.
\newblock {\em Reduction of the Characteristic Initial Value Problem to the Cauchy Problem and Its Applications to the Einstein Equations.}
\newblock Proc. Roy. Soc. London Ser. A 427 (1990), no. 1872, 221-239.
\bibitem{SachsIVP}
R.~Sachs.
\newblock {\em On the Characteristic Initial Value Problem in Gravitational Theory}.
\newblock J. Math. Phys. 3, 908 (1962).
\bibitem{SchoenYauPSCM}
R.~Schoen, S.~Yau.
\newblock {\em On the structure of manifolds with positive scalar curvature}.
\newblock Manuscripta Math. 28 (1979), no. 1-3, 159-183.
\bibitem{SchoenYau1}
R.~Schoen, S.~Yau.
\newblock {\em The Energy and the Linear
Momentum of Space-Times in General Relativity}.
\newblock Comm. Math. Phys. 79 (1981), 47-51.

\bibitem{Witten}
E.~Witten.
\newblock {\em A new proof of the positive energy theorem}.
\newblock Comm. Math. Phys. 80 (3), 1981, 381-402.

\end{thebibliography}
\end{document}